\documentclass[12pt,english]{article}

\usepackage{geometry}
\usepackage{color}
\usepackage{babel}
\usepackage{amsmath}
\usepackage{amsthm}
\usepackage{amssymb}
\usepackage{graphicx}
\usepackage{setspace}
\usepackage[longnamesfirst,sort,comma]{natbib}
\usepackage[unicode=true,pdfusetitle,
 bookmarks=true,bookmarksnumbered=false,bookmarksopen=false,
 breaklinks=false,pdfborder={0 0 1},colorlinks=true]{hyperref}
\hypersetup{linkcolor=red, citecolor=blue}
\usepackage{enumitem}
\usepackage[all]{hypcap}

\theoremstyle{plain}
\newtheorem{thm}{\protect\theoremname}[section]
\theoremstyle{plain}
\newtheorem{prop}[thm]{\protect\propositionname}
\ifx\proof\undefined
\newenvironment{proof}[1][\protect\proofname]{\par
	\normalfont\topsep6\p@\@plus6\p@\relax
	\trivlist
	\itemindent\parindent
	\item[\hskip\labelsep\scshape #1]\ignorespaces
}{%
	\endtrivlist\@endpefalse
}
\providecommand{\proofname}{Proof}
\fi
\theoremstyle{remark}

\theoremstyle{definition}
\newtheorem{defn}[thm]{\protect\definitionname}
\theoremstyle{plain}
\newtheorem{lem}[thm]{\protect\lemmaname}
\newtheorem{hp}{Assumption}

\providecommand{\definitionname}{Definition}
\providecommand{\lemmaname}{Lemma}
\providecommand{\propositionname}{Proposition}
\providecommand{\remarkname}{Remark}
\providecommand{\theoremname}{Theorem}

\newcommand{\cvd}{\mbox{$\stackrel{d}{\longrightarrow}\,$}}

\newcommand{\Var}{\operatorname{Var}}

\def\1{\mathbb{I}}
\def\Real{\mathbb{R}}

\setcitestyle{citesep={;}}

\renewcommand{\baselinestretch}{1.2}

\begin{document}
\title{A Bootstrap Specification Test for Semiparametric Models with Generated Regressors 
}
\author{Elia Lapenta\thanks{\emph{CREST and ENSAE. Email: elia.lapenta@ensae.fr}. Address correspondence: CREST, 5 Avenue Le Chatelier, 91120 Palaiseau, FRANCE. 
} \hspace{.1cm}  }
\date{First version: November 2019. This version: October 2023}

\maketitle
\thispagestyle{empty}
\vspace*{-.5cm}
\normalsize

\begin{abstract}
This paper provides a specification test for semiparametric models with nonparametrically generated regressors. Such variables are not observed by the researcher but are nonparametrically identified and estimable. Applications of the test include models with endogenous regressors identified by control functions, semiparametric sample selection models, or binary games with incomplete information. The statistic is built from the residuals of the semiparametric model. A novel wild bootstrap procedure is shown to provide valid critical values. We consider nonparametric estimators with an automatic bias correction that makes the test implementable without undersmoothing. In simulations the test exhibits good small sample performances, and an application to women's labor force participation decisions shows its implementation in a real data context.  
\end{abstract}
\textbf{Keywords}: Hypothesis Testing, Bootstrap, Generated Regressors, Semiparametric Model, Control Function, Bias Correction.
\vspace{0.25cm}\\
\textbf{JEL Classification}: C01, C12, C14

\newpage
\setcounter{page}{2}

\section{Introduction}
Checking the correct specification of a model is empirically relevant, as a misspecified model can yield biased and inconsistent estimates and provide a misleading counterfactual analysis. In this paper we contribute to the literature by providing a specification test for semiparametric models with nonparametrically generated regressors. Such regressors are not observed by the researcher but are nonparametrically identified and estimable. Examples of semiparametric models with generated variables are common in empirical frameworks. They include endogenous models with control function  \citep{rivers_limited_1988, blundell_endogeneity_2004, newey_nonparametric_1999}, semiparametric sample selection models, extension of tobit models \citep{escanciano_identification_2016}, or semiparametric empirical games with incomplete information  \citep{aradillas-lopez_pairwise-difference_2012, lewbel_identification_2015}.  \\
Let $Y\in \mathbb R $, $Z\in \mathbb R ^p$, and $X\in \mathbb R^{p_X} $. Our goal is to test the null hypothesis
\begin{equation}\label{eq: null hypothesis}
    \mathcal H _0\,:\, \mathbb E \{Y|\nu(X,H(Z))\}=\mathbb E \{Y|q(\beta_0,X,H(Z))\}\text{ for some }\beta_0 \in B
\end{equation} 
against its logical complement $\mathcal H _1=\mathcal H _0 ^c$, where $\nu:\mathbb{R}^{p_X+1}\mapsto\mathbb{R}^{p_\nu}$ and $q:B\times\mathbb{R}^{p_X+1}\mapsto\mathbb{R}^d$ are known vector-valued functions, the $\sigma$-field of $q(\beta,X,H(Z))$ is contained in the $\sigma$-field of $\nu(X,H(Z))$ for all $\beta\in B$, $B$ is a real set, and
\begin{equation}\label{eq: 1st step NP regression}
    H(Z)=\mathbb E \{D|Z\}\, 
\end{equation}
is a nonparametric function, with $D\in \mathbb{R}$. The null hypothesis is featured by the presence of the generated variable $H(Z)$. This means that $H(Z)$ is latent, but it is identified and estimable in the first-step nonparametric regression (\ref{eq: 1st step NP regression}). The conditional moment restriction in (\ref{eq: null hypothesis}) arises from the above mentioned empirical models, see Section \ref{sec: Main Applications} for details. 

Our main contributions are twofold. First, we develop a test for the null hypothesis in (\ref{eq: null hypothesis}) featured by the presence of generated variables. Due to the presence of generated variables, the null hypothesis in (\ref{eq: null hypothesis}) cannot be tested by using existing methods. Second, we construct and show the validity of a novel wild bootstrap procedure to compute the critical values. Our wild bootstrap procedure (i) uses the information under the null hypothesis and (ii) differently from existing procedures for inference in the presence of generated variables does not require estimating a nonparametric derivative. 

Our first main contribution is thus to propose a test for the conditional moment restriction in (\ref{eq: null hypothesis}). Such a moment restriction is semiparametric, in the sense that the conditional expectations of $Y$ on both sides of the equation are not restricted to have a specific functional form. If this moment condition did not contain the generated variable $H(Z)$, a specification test could be based on procedures already available in the literature for models where all regressors are observed, see e.g. \citet{delgado_significance_2001}, \citet{xia_goodness--fit_2004}. However, since the generated variable  $H(Z)$ is not observed, these tests cannot be implemented in the setting considered here.\\
To construct our test statistic we proceed in three steps: in a first step we estimate the generated variable $H(Z)$, in a second step we estimate the parameter  $\beta_0$, and in a third step we estimate the right hand side of Equation (\ref{eq: null hypothesis}). This allows us to compute the residuals of the semiparametric model. Our test statistic is then based on an {\itshape empirical process} involving the estimated residuals, see Section \ref{sec: The Test} for details. A specific feature of this setting is that replacing the generated regressor $H(Z)$ with its estimate in the first step introduces an estimation error that impacts on the asymptotic behavior of our test statistic. Hence, it needs to be taken into account. A similar result is obtained in \citet{hahn_asymptotic_2013}, \citet{mammen_semiparametric_2016}, and  \citet{hahn_nonparametric_2018} in an estimation context. Our context however is different from theirs because we have to deal with an empirical process. Moreover, compared to existing specification tests based on empirical processes \citep{delgado_significance_2001, xia_goodness--fit_2004}, when establishing the asymptotic behavior of our test statistic we are faced with the technical challenge of dealing with an estimation error coming from the generated variables that is not present in existing specification tests.\\
A distinguishing feature of our test is the presence of bias corrections for the nonparametric estimators of the first and third step. Bias corrections have been proposed for the purpose of nonparametric estimation in, e.g., \citet{buhlmann_boosting_2003}, \citet{di_marzio_boosting_2008}, and \citet{park_l2_2009}. We show that in our setting these bias corrections guarantee a {\itshape small bias property} of our test statistic. This means that the bias of our test statistic converges to zero faster than the bias of the nonparametric estimators on which it is based, see also \citet{newey_twicing_2004}. To the best of our knowledge, {\itshape  this is a novel approach in semiparametric models with generated variables}. The small bias property of our statistic is both theoretically appealing and practically relevant. It is theoretically appealing, as it avoids undersmoothing. This means that the bandwidths ``optimal" for estimation can be employed in our specification test.
The small bias property is also practically relevant: in our simulation experiment we show that, thanks to the bias corrections and the small bias property, our test is more stable to the selection of smoothing parameters than a test that does not use bias corrections and does not have a small bias property. Alternative approaches developed in \citet{chernozhukov_locally_2016} or \citet{escanciano_asymptotic_2018-2} could be used in our context to obtain a small bias property of the test statistic. However, as we explain in Section \ref{sec:Assumptions and Theory}, implementing these approaches in our context would be difficult: {\itshape due to the presence of generated variables}, they would require estimating a nonparametric derivative. Since nonparametric derivatives have slow convergence rates, such an estimation would complicate the practical implementation of the specification test. Differently, our bias corrections do not require estimating nonparametric derivatives and are easy to implement in the presence of generated variables. In other words, a further contribution of this paper is to show how to obtain a small bias property in a semiparametric context with generated variables, without estimating nonparametric derivatives.

Our second main contribution is to construct a novel wild bootstrap procedure to compute valid critical values for our test. We show that asymptotically our statistic converges to an intricate distribution depending on unknown features of the data generating process. So, the asymptotic distribution cannot be directly employed to obtain the critical values. This problem also arises in the semiparametric specification tests of, e.g., \citet{delgado_significance_2001} and \citet{xia_goodness--fit_2004} who develop bootstrap procedures to obtain the critical values. However, their bootstrap procedure are not valid in our context due to the presence of generated regressors. Since their bootstrap procedures are developed for settings {\itshape without} generated variables, they cannot replicate the estimation error arising from the generated regressors. Differently, in our setting the generated variables introduce an estimation error that impacts on the asymptotic behavior of the test statistic, and hence it needs to be taken into account. \\
Other procedures are developed in the literature to adjust semiparametric methods for the presence of generated variables, see \citet[Chapter  6]{wooldridge_econometric_2010}, \citet{hahn_asymptotic_2013}, \citet{hahn_nonparametric_2018}. Such methods consist in estimating the adjustment terms due the generated variables in the influence function representation of the statistic, and then correct the standard errors of the test. Using this approach in our testing problem would be difficult for two reasons. First, it requires the statistic to be asymptotically pivotal, while this property does not hold in our setting. Second, in our framework estimating the adjustment terms coming from the generated variables requires estimating a nonparametric derivative. As highlighted earlier, this is not appealing in practice.\\
Thus, in this paper we construct a novel wild bootstrap procedure for obtaining valid critical values. In particular, our contribution is to develop a wild bootstrap method that does not require estimating a nonparametric derivative and can replicate the estimation error coming from the generated variables. To the best of our knowledge this is a novel contribution in the literature on generated regressors. The bootstrap we develop here imposes the null hypothesis when resampling the observations. This feature is attractive, as it means that we are using all the information available when constructing the critical values for testing. We develop the wild bootstrap test so that also the bootstrapped statistic has a small bias property, as well as the sample statistic.  

Our last contribution is to develop an empirical method for the bandwidth selection in our testing problem. The method we propose chooses the bandwidth minimizing the ``distance" between the estimated semiparametric model and the null hypothesis. In our simulation experiment, we show that this method provides a reliable inference in moderate samples.\\

\noindent \textbf{Related literature}. 
Beyond the studies already cited, there is a large literature analyzing the problem of estimation with generated variables. Early work includes \citet{pagan_econometric_1984}, \citet{ahn_distribution_1993}, \citet{ahn_semiparametric_1997}. More recent results are presented in \citet{li_semiparametric_2002}, \citet{chen_estimation_2003}, \citet{rothe_semiparametric_2009}, \citet{sperlich_note_2009}, \citet{mammen_nonparametric_2012}, \citet{gutknecht_testing_2016}, \citet{mammen_semiparametric_2016}, \citet{bravo_two-step_2020},
\citet{hahn_nonparametric_2018}, \citet{vanhems_estimation_2019}. The impact of generated regressors on the asymptotic distribution of a finite dimensional estimator is analyzed in \citet{newey_asymptotic_1994} and \citet{hahn_asymptotic_2013}. \citet{escanciano_uniform_2014} obtain an expansion of the residuals from a regression involving variables estimated in a preliminary step. These papers focus on estimation and do not address specification testing in the presence of generated regressors. \\
From a statistical point of view, our work is related to  \citet{escanciano_uniform_2014} with two important distinctions. First, in this paper we develop a wild bootstrap test, while  \citet{escanciano_uniform_2014} are not concerned with constructing a bootstrap procedure. Second, the residuals at the basis of our statistic involve bias correction terms which are not present in their context. This allows us to avoid undersmoothing and to obtain a small bias property of our statistic.\\
Finally, our work is related to the literature on specification testing in semiparametric models, see \citet{bierens_consistent_1982}, \citet{bierens_consistent_1990}, \citet{fan_consistent_1996}, \citet{bierens_asymptotic_1997}, \,\citet{stute_nonparametric_1997}\,, \citet{delgado_consistent_2006}, \citet{einmahl_specification_2008}, \citet{lavergne_breaking_2008}\,,\, \citet{delgado_distribution-free_2008}\,,\, \citet{escanciano_testing_2010}\,,\, \citet{neumeyer_estimating_2010}, \citet{lavergne_significance_2015}. We use an approach similar to \citet{bierens_asymptotic_1997}, \citet{stinchcombe_consistent_1998}, and \citet{delgado_significance_2001}.  The distinctive feature of our work is the presence of generated variables that introduce extra terms in the asymptotic expansion of our statistic.\\

\noindent \textbf{Organization of the paper}. Section \ref{sec: The Test} starts by constructing the test statistic. Then, it sets up the assumptions and obtains the asymptotic behavior of our statistic. Section \ref{sec: The Bootstrap Test} details the construction of our wild bootstrap procedure and  shows its validity for computing the critical values. The main applications of our test are in Section \ref{sec: Main Applications}. Section \ref{sec: Empirical Implementation and Applications} first describes the practical implementation of our test and our method for the bandwidth selection. Then, it provides evidence about the finite sample behavior of our test and finally shows its application  to a real data set. All proofs are relegated to Appendix \ref{sec:Asymptotic Analysis} and \ref{sec: Bootstrap Analysis}. The supplementary material contains auxiliary results.

\section{The Test}\label{sec: The Test}

To use a more compact notation, we define 
\begin{equation}
    W(\beta):=q(\beta,X,H(Z))\, \text{ and } \, G_{W(\beta)}(w):=\mathbb E \{Y|W(\beta)=w\}\, .
\end{equation}
Let us assume that $W(\beta)$ is continuously distributed for any $\beta \in B$, and let $f_{W(\beta)}$ be the density function of $W(\beta)$ with respect to the Lebesgue measure. Since the $\sigma$-field of $W(\beta_0)$ is a subset of the $\sigma$-field of $\nu(X,H(Z))$, the null hypothesis  in (\ref{eq: null hypothesis}) is equivalent to 
\begin{equation}\label{eq: null hypothesis II}
        \mathcal H _0\,:\, \mathbb E \,\left\{\,[Y-G_{W(\beta_0)}(W(\beta_0))]\,f_{W(\beta_0)}(W(\beta_0))\,|\,\nu(X,H(Z))\,\right\}=0
\end{equation}
for some $\beta_0\in B$.
We have introduced the density function $f_{W(\beta_0)}$ to avoid a random denominator in the estimation of $G_{W(\beta_0)}$, see Section \ref{subsec: The Test Statistic} for details.  Equation (\ref{eq: null hypothesis II}) is a conditional moment restriction. To test such an equation we consider an  {\itshape equivalent continuum} of {\itshape unconditional} moments. Let us assume that $\nu(X,H(Z))$ is a bounded random variable. Then, the null hypothesis in (\ref{eq: null hypothesis II}) is equivalent to \begin{equation}\label{eq: null hypothesis continuum}
    \mathbb E \,\,[Y-G_{W(\beta_0)}(W(\beta_0))]\,f_{W(\beta_0)}(W(\beta_0))\, \varphi_s(X,H(Z))\,=\,0\text{ for all }s\in \mathcal S\, ,
\end{equation}
where $\mathcal S$ is a set containing a neighborhood of the origin, $\varphi_s(X,H(Z)):=\varphi(s^T \nu(X,H(Z)))$, and $\varphi$ is a suitably chosen weighting function. In Assumption \ref{Assumption: iid and phi} we report the precise conditions that $\varphi$ must satisfy to guarantee the equivalence between (\ref{eq: null hypothesis II}) and (\ref{eq: null hypothesis continuum}). For example, \citet{bierens_consistent_1982} shows such an equivalence with $\varphi(u)=\exp(\,\sqrt{-1}\, u\,)$, \citet{bierens_consistent_1990} uses the real exponential $\varphi(u)=\exp(u)$, while \citet{stinchcombe_consistent_1998} use more general weighting functions. Now, the formulation of the null hypothesis into a continuum of unconditional moments is quite convenient, as it allows us to construct a test without estimating $\mathbb E \{Y|\nu(X,H(Z))\}$. Differently, constructing a test based on the conditional moments in (\ref{eq: null hypothesis II}) or (\ref{eq: null hypothesis}) would require estimating $\mathbb E \{Y|\nu(X,H(Z))\}$. From a practical point of view this would complicate the implementation of the test, as we should select additional smoothing parameters, say bandwidths. Moreover, in the presence of an important dimension of $\nu(X,H(Z))$ such a test would suffer from a curse of dimensionality. \\
Let us assume to observe an iid sample $\{Y_i,X_i,Z_i,D_i\}_{i=1}^n$ from $(Y,X,Z,D)$, and let us denote with $\mathbb P_n$ the empirical mean operator, i.e. $\mathbb P_n g(Y,X,Z,D)=n^{-1}\sum_{i=1}^n g(Y_i,X_i,Z_i,D_i)$ for any function $g$. Also, let $\mu$ be a finite measure supported on $\mathcal S $. Since the null hypothesis in (\ref{eq: null hypothesis continuum}) is equivalent to 
\begin{equation*}
    \int_{ } \left|\mathbb{E}[Y-G_{W(\beta_0)}(W(\beta_0))]\,f_{W(\beta_0)}(W(\beta_0))\, \varphi_s(X,H(Z))\right|^2\, d \mu(s)\, =0\, ,
\end{equation*}
if we knew $G_{W(\beta_0)}$ and $f_{W(\beta_0)}$ we could test (\ref{eq: null hypothesis continuum}) by using the Cramer-Von Mises functional
\begin{equation}\label{eq: Oracle Cramer Von Mises}
    n\,\int_{ }\left|\mathbb P _n [Y-G_{W(\beta_0)}(W(\beta_0))]\,f_{W(\beta_0)}(W(\beta_0))\,\varphi_s(X,H(Z))\right|^2 d \mu(s)\, .
\end{equation}
  However,
 such an  ``oracle" statistic cannot be used in practice, as $G_{W(\beta_0)}$ and $f_{W(\beta_0)}$ are unknown and must be estimated. Now, if we could observe the generated regressor $H(Z)$ a test for (\ref{eq: null hypothesis continuum}) could be constructed in two steps. In a first step, we could estimate $\beta_0$ by say $\widehat \beta$. In a second step, given $W(\widehat \beta)=q(\widehat \beta,X,H (Z))$, we could regress nonparametrically $Y$  on $W(\widehat \beta)$ to get an estimate of $G_{W(\beta_0)}$, and we could construct a nonparametric density estimator of $f_{W(\beta_0)}$ based on $W(\widehat \beta)$. Then, a statistic as in (\ref{eq: Oracle Cramer Von Mises}) based on such estimators could be used for testing (\ref{eq: null hypothesis continuum}). A test of this type is proposed in \citet{xia_goodness--fit_2004} and \citet{delgado_significance_2001} who develop bootstrap tests for equations similar to (\ref{eq: null hypothesis continuum}) {\itshape when all regressors are observed}. However, since the generated variable $H(Z)$ is not observed, such tests cannot be implemented in our context.\\
 Due to the presence of the generated variable $H(Z)$, we compute a {\itshape feasible} test statistic by a {\itshape three} step estimation. In a first step, on the basis of (\ref{eq: 1st step NP regression}), we regress nonparametrically $D$ on $Z$ so as to get an estimate of $H$, say $\widetilde H$. In a second step, we obtain an estimator of $\beta_0$, say $\widehat \beta$. Finally, in a third step, given an estimate of the generated regressor $\widetilde W(\widehat \beta)=q(\widehat \beta, X, \widetilde H (Z))$,  we regress nonparametrically $Y$ on $\widetilde W(\widehat \beta)$ to get an estimate of $G_{W(\beta_0)}$, and we construct a nonparametric density estimator of $f_{W(\beta_0)}$ based on $\widetilde W (\widehat \beta)$. Plugging these estimators into (\ref{eq: Oracle Cramer Von Mises}) gives a feasible test statistic. 
 In the next section we detail such a three-step estimation procedure.

\subsection{The Test Statistic}\label{subsec: The Test Statistic}

\textbf{First step estimation}. Let $K_H$ be a kernel on $\mathbb{R}^p$ and $h_H$ be a bandwidth.  We define 
\begin{align}\label{eq: That X e and fhat}
    \widehat T ^{D}_{Z}(z):=\frac{1}{n h_H^p}\sum_{i=1}^n D_i K_H\left(\frac{z-Z_i}{h_H}\right)\text{ and } \widehat{f}_Z(z):=\frac{1}{n h_H^p}\sum_{i=1}^n  K_H\left(\frac{z-Z_i}{h_H}\right)\, .
\end{align}
The estimation of the generated variable $H(Z)$ could be based on $\widehat{H}(z):=\widehat{T}^{D}_Z(z)/\widehat{f}_Z(z)$. However, a test based on the first-step estimator $\widehat H$ would necessitate some undersmoothing, see Section \ref{sec:Assumptions and Theory} for details. This would complicate the implementation of the test, as it would make tricky the choice of the first step bandwidth. To avoid these problems and allow for a {\itshape small-bias property} of our test statistic, we use an $L_2$ boosted or, equivalently, a bias corrected estimator in the first step. These types of estimators have been introduced in \citet{buhlmann_boosting_2003} for spline smoothing and \citet{di_marzio_boosting_2008} and \citet{park_l2_2009} for kernel smoothing. \citet{xia_goodness--fit_2004} and \citet{lapenta_encompassing_2022} used similar bias corrections for specification tests {\itshape when all regressors are observed}. The procedure works as follows. Given the initial estimator $\widehat H$, we compute the residuals $D_i-\widehat H (Z_i)$ for $i=1,\ldots,n$. Then, we construct a kernel estimator $\widehat{T}^{D-\widehat H}_Z(z)/\widehat f_Z (z)$, where $\widehat{T}^{D-\widehat H}_Z$ is defined similarly as in (\ref{eq: That X e and fhat}), with the difference that the projected variable is $D-\widehat{H}(Z)$ instead of $D$. Here $-\widehat{T}^{D-\widehat H}_Z(z)/\widehat f_Z (z)$ is an estimate of the bias of $\widehat H(z)$, see \citet{xia_goodness--fit_2004}. Finally, we consider the updated or bias-corrected estimator $\widehat H + \widehat{T}^{D-\widehat H}_Z/\widehat f_Z $. Such a procedure can be iterated a finite number of times, as proposed in \citet{buhlmann_boosting_2003}, \citet{di_marzio_boosting_2008}, and \citet{park_l2_2009}. Here, we restrict ourselves to a one-step boosting: this is enough for our purpose to avoid undersmoothing and guarantee a small-bias property of the test statistic.  \\
To precisely define the estimator, let use denote with $\mathbb I \{A\}$ the indicator function of the event $A$, and let us assume that $(X,Z)$ admits a joint density $f$ with respect to the Lebesgue measure. We define
\begin{align*}
    \widehat{T}^{D-\widehat H}_Z(z):=&\frac{1}{n h_H^p}\sum_{i=1}^n [D_i - \widehat H (Z_i)]\, \widehat t_i \, K_H\left(\frac{z-Z_i}{h_H}\right)\, \\
    \text{ with }\widehat t_i :=\mathbb{I}\{\widehat f (X_i,Z_i)\geq \tau_n\}& \text{ and }\widehat f(x,z):=\frac{1}{n h_0^{\dim(X,Z)}}\sum_{i=1}^n K_0\left(\frac{(x,z)-(X_i,Z_i)}{h_0}\right)\, .
\end{align*}
$(\tau_n)_n$ is a sequence of positive numbers converging to zero, whose features are specified in Section \ref{sec:Assumptions and Theory}. Indeed, $K_0$ and $h_0$ denote a kernel and a bandwidth, while $\widehat f$ is the kernel estimator of the joint density of $(X,Z)$. We introduce the trimming $\widehat t$  in $\widehat{T}^{D - \widehat H}_Z$ to control for the random denominator of $\widehat H$. Then, we estimate $H$ by 
\begin{align}\label{eq: definition of Htilde}
    \widetilde {H}(z):=\widehat H (z) + \widehat T^{D - \widehat H }_Z(z)/\widehat{f}_Z(z)\, .
\end{align}
As noticed in \citet{xia_goodness--fit_2004},  $-\widehat T ^{D-\widehat H}_Z/\widehat f_Z$ can be interpreted as an estimator of the bias of $\widehat H$. So, $\widetilde H$ can be thought of as a bias-corrected estimator. \citet{di_marzio_boosting_2008} show that when bias corrections or, equivalently, boosting iterations are applied to kernel estimators, the bias decreases exponentially fast while the variance increases exponentially slow. Thus, when applied for estimation purposes these estimators are practically useful, as a kernel bias corrected estimator displays a lower mean-squared error than a kernel estimator that is not bias corrected.\footnote{Precise conditions can be found in \citet{di_marzio_boosting_2008} and \citet{park_l2_2009}.} In our case, we do {\itshape not} apply these bias corrections for estimation purposes. As we argue in Section \ref{sec:Assumptions and Theory}, our motivation for using these bias corrections is to obtain a {\itshape small bias property} of the test statistic and to avoid undersmoothing.\\
\vspace{0.5cm}


\noindent \textbf{Second-step estimation}. In the second step, we estimate $\beta_0$ on the basis of the estimator of the generated variable $H(Z)$. Let $\widetilde W_i (\beta):=q(\beta,X_i, \widetilde{H}(Z_i))$ for $i=1,\ldots,n$, and let us recall that $d$ is the dimension of $\widetilde W(\beta)$. For a kernel $K$ and a bandwidth $h$,  we define  
\begin{align}\label{eq: definition of That Y Wtilde beta}
    \widehat{T}^Y_{\widetilde W (\beta)}(w):=\frac{1}{n h^d}\sum_{i=1}^n Y_i \,\widehat t _i \, K\left(\frac{w-\widetilde W _i(\beta)}{h}\right) \text{ and }\widehat{f}_{\widetilde W (\beta)}(w):=\frac{1}{n h^d}\sum_{i=1}^n \widehat t _i \, K\left(\frac{w-\widetilde W _i(\beta)}{h}\right)\, .
\end{align}
Once again, the trimming $\widehat t$ is used in the above expressions to control for a random denominator in $\widetilde W(\beta)$, i.e. in $\widetilde H$.
Then, an estimator of $\mathbb E \{Y|W(\beta)=w\}$ is 
\begin{equation*}
\widehat{G}_{\widetilde{W}(\beta)}(w):=\widehat{T}^Y_{\widetilde W (\beta)}(w)/\widehat{f}_{\widetilde W (\beta)}(w)\, .
\end{equation*} 
We estimate $\beta_0$ by a Semiparametric Least-Squares (SLS) principle: 
\begin{equation}\label{eq: definition of betahat}
    \widehat \beta := \arg \min _{\beta\in B}\,  \mathbb{P}_n\,  [Y-\widehat{G}_{\widetilde W (\beta)}(\widetilde W (\beta))]^2\,  \widehat t\, .
\end{equation}
We remark that in the above objective function only the first step estimator $\widetilde H$ is bias corrected, while $\widehat G _{\widetilde W (\beta)}$ does not contain bias corrections. As we show in Lemma \ref{lem: IFR for betahat}, this is enough to obtain a small-bias property of our test statistic. \\
A SLS estimator is also used in the specification test proposed by \citet{xia_goodness--fit_2004}, where all variables are observed. Differently from \citet{xia_goodness--fit_2004}, our context is featured by the generated variable $H(Z)$, so the SLS estimator used in our test is based on the first-step nonparametric estimator $\widetilde H$. This implies that the estimation error from $\widetilde H$ will appear in the influence-function representation of $\widehat \beta$, see Lemma \ref{lem: IFR for betahat}. A SLS estimator for index models with generated variables was also proposed in \citet{escanciano_identification_2016}. However, in our case $\widehat \beta$ is based on the bias-corrected estimator $\widetilde H$: this allows us to avoid undersmoothing in the first step estimation and to obtain a small-bias property of the test statistic. \\
\vspace{0.5cm}

\noindent \textbf{Third-step estimation}.
For $\widehat \varepsilon _i := Y_i - \widehat{G}_{\widetilde W (\widehat \beta)}(\widetilde W _i (\widehat \beta))$ we  let
\begin{equation}\label{eq: definition of That eps hat Wtilde betahat }
    \widehat T ^{\widehat \varepsilon}_{\widetilde W (\widehat \beta)}(w)=\frac{1}{n h^d}\sum_{i=1}^n \widehat \varepsilon _i \widehat t _i K\left(\frac{w-\widetilde W _i (\widehat \beta)}{h}\right) \, .
\end{equation}
The role of the trimming $\widehat t$ in the above display is to control for a random denominator in $\widehat G _{\widetilde W (\widehat \beta)}$ and $\widetilde W(\widehat \beta)$. Similarly as in the first step estimation, $-\widehat{T}^{\widehat \varepsilon}_{\widetilde W (\widehat \beta)}/\widehat f _{\widetilde W (\widehat \beta)}$ is an estimate of the bias of $\widehat G _{\widetilde W (\widehat \beta)}$.
So, we define the bias-corrected estimator of $G_{W(\beta_0)}$ as 
\begin{equation}\label{eq: definition of G tilde W tilde betahat}
    \widetilde G _{\widetilde W (\widehat \beta)}(w):=\widehat G _{\widetilde W (\widehat \beta)}(w)+\widehat T ^{\widehat \varepsilon}_{\widetilde W (\widehat \beta)}(w)/\widehat f _{\widetilde W (\widehat \beta)}(w)\, .
\end{equation}
Denoting the bias-corrected residuals as $\widetilde \varepsilon _i :=Y_i-\widetilde{G}_{\widetilde W (\widehat \beta)}(\widetilde W_i (\widehat \beta))$,
our feasible statistic is given by
\begin{equation}\label{eq: feasible test statistic}
    S_n:=n \int_{ }\,\left|\,\mathbb P _n\, \widetilde \varepsilon\,\,\widehat f_{\widetilde W (\widehat \beta)}(\widetilde W (\widehat \beta))\,\, \varphi_s(X,\widetilde H (Z))\,\,\widehat t\,\right|^2\,d \mu(s)\, .
\end{equation}
Let us comment on the above expression. The weighting by the estimated density  $\widehat f_{\widetilde W(\widehat \beta)}$ allows us to get rid of the random denominator in $\widetilde G _{\widetilde W (\widehat \beta)}$. This has a double role. First, it ``stabilizes" the behavior of $S_n$ by avoiding a denominator that can be close to zero if observations on the ``tails" are included in the computation of $S_n$. Second, avoiding a random denominator in $\widetilde G _{\widetilde W (\widehat \beta)}$ simplifies our proofs for obtaining the asymptotic behavior of $S_n$. However, {\itshape due to the presence of generated variables}, we still have to control for a random denominator in $\widetilde W (\widehat \beta)$, i.e. in $\widetilde H (Z)$. This is the reason why we include the trimming $\widehat t$ in the computation of $S_n$.\\
The asymptotic behavior of $S_n$ is determined by the empirical process $\sqrt{n} \, \mathbb P _n\, \widetilde \varepsilon$ $\,\widehat f_{\widetilde W (\widehat \beta)}(\widetilde W (\widehat \beta))$ $\widehat t\,\, \varphi_s(X,\widetilde H (Z))$.  
As mentioned earlier, a Cramer-Von Mises statistic similar to $S_n$ has already been used for testing equations similar to (\ref{eq: null hypothesis continuum}) {\itshape without generated variables}, see \citet{xia_goodness--fit_2004} and \citet{delgado_significance_2001}. Differently, our context is featured by the presence of the generated regressors $H(Z)$ that are replaced by the first-step estimates $\widetilde H (Z)$. Such a replacement introduces an estimation error that needs to be taken into account. So, when obtaining the influence function representation of the empirical process at the basis of $S_n$, we are faced with the challenge of controlling for an estimation error that is not present in \citet{xia_goodness--fit_2004} and \citet{delgado_significance_2001}. The estimation error from $\widetilde H (Z)$ is not negligible, as it will appear in the influence function representation of the empirical process at the basis of $S_n$, see the next section for details.

\subsection{Assumptions and Asymptotic Behavior of the Test Statistic}
\label{sec:Assumptions and Theory}

Let us start by introducing some definitions. We define the differential operator
\[
\partial^{l}g(w)  =
\frac{\partial^{|l|}}{\partial^{l_{1}}w_{1} \ldots
\partial^{l_{d}}w_{d}} g(w)\, ,
\qquad l=(l_{1},..,l_{d})'\, , \qquad |l| = l_{1}+..+l_{d}
\, .
\]
\begin{defn}\label{defn: class of functions}
(a)
$\mathcal{G}_\lambda(\mathcal{A}) = \left\{g:\mathcal{A}\mapsto\mathbb{R} : \quad \sup_{a\in\mathcal{A}}|\partial^l g(a)|< M \text{ for all } |l|\leq \lambda \right\}$.
(b) $\mathcal{K}_{\lambda}^{r}(\mathbb R ^d)$ $:=\{$ $(u_1,\ldots,u_d)\mapsto \prod_{j=1}^d k(u_j)$ :   $k$ is a univariate
kernel of order $r$, $\lambda$
times continuously differentiable with uniformly bounded derivatives,
symmetric about zero, and with bounded support $\}$.
\end{defn}
We next introduce the assumptions and then we comment on them.
\begin{hp}
\label{Assumption: iid and phi}
(i) $\{Y_i, X_i, Z_i, D_i\}_{i = 1}^n$ is an iid sample of bounded random variables. (ii) $\mathbb{E}\{Y|X,Z\}=\mathbb{E}\{Y|\nu(X,H(Z))\}$. (iii) $\mathcal S \subset \mathbb{R}^{\dim(\nu(X,Z))}$ is a compact set containing a neighborhood of the origin. (iv) $\varphi_s(x,H(z))=\varphi(s^T\,\nu(x,H(z)))$ with $s\in\mathcal{S}$ and $\varphi$ analytic non-polynomial function such that  $\partial^l \varphi(0)\neq 0$ for all $l\in \mathbb{N}$. 
\end{hp}

\begin{hp}
\label{Assumption: smoothness}
(i) $f_{W(\beta)}\, ,\, \partial_\beta f_{W(\beta)}\,,\, G_{W(\beta)}\, , \partial_\beta G_{W(\beta)}\, \in \mathcal{G}_{r}({\Real^d})$ uniformly in $\beta\in B$ with $r\geq \lceil (d+1)/2\rceil+1$.\footnote{$\lceil x \rceil$ denotes the smallest integer
 above $x$.}
(ii) $\mathbb{E}\{\varphi_s(X,H(Z))|W(\beta)=\cdot\}\in \mathcal{G}_{r}(\mathbb{R}^d)$ uniformly in $s
 \in {\cal S}$ and $\beta\in B$, with $r\geq \lceil (d+1)/2\rceil+1$.
 (iii) $f_Z\, , H\, \in \mathcal{G}_{r_H}(\mathbb{R}^p)$ for $r_H \geq \lceil (p+1)/2 \rceil $. (iv) $f \in \mathcal{G}_{r_0}(\Real ^{\dim(X,Z)})$ with $r_0\geq 2$. (v) $\nu(x,u)$, $\partial_u \nu(x,u)$, $q(\beta,x,u)$, $\partial_u q(\beta,x,u)$, $\partial_\beta q(\beta,x,u)$, $\partial_u \partial_\beta q(\beta,x,u)$, and $\partial^2_{\beta \beta^T} q(\beta,x,u)$ are Lipschitz in $(\beta,x,u)$. (vi) $\partial_w G_{W(\beta)}(w)$, $\partial_\beta G_{W(\beta)}(w)$, $\partial_\beta \partial_w G_{W(\beta)}(w)$, $\partial^2_{w w^T} G_{W(\beta)}(w)$, and $\partial^2_{\beta \beta^T}G_{W(\beta)}(w)$, $f_{W(\beta)}(w)$, $\mathbb{E}\{\varphi_s(X,H(Z))|W(\beta)=w\}$, and $\partial_w\mathbb{E}\{\varphi_s(X,H(Z))|W(\beta)=w\}$ are Lipschitz and uniformly bounded in $(\beta,w,s)$.\footnote{More in detail, since $G_{W(\beta)}(w)=\mathbb{E}\{Y|W(\beta)=w\}$,  $\partial_w G_{W(\beta)}(w)=\partial_w \mathbb{E}\{Y|W(\beta)=w\}$ and  $\partial_\beta G_{W(\beta)}(w)=\partial_\beta \mathbb{E}\{Y|W(\beta)=w\}$. }
\end{hp}

\begin{hp}\label{Assumption : kernels}
(i) $K\in \mathcal{K}^r_{\lambda}(\mathbb R ^d)$ with $\lambda \geq \lceil (d+1)/2\rceil +6 $. (ii) $K_H\in \mathcal{K}^{r_H}_\lambda(\mathbb R^p)$ with $\lambda \geq \lceil (p+1)/2 \rceil $. (iii) $K_0\in \mathcal{K}^{r_0}_0(\mathbb R ^{\dim(X,Z)})$.
\end{hp}
Let us define the following convergence rates
\begin{align*}
d_0:=\sqrt{\frac{\log n}{n h_0^{\dim(X,Z)}}}+h_0^{r_0}\, ,\,
 d_{H}&  := \sqrt{\frac{\log n}{n h_H^{p}}} + h_H^{r_H}\, , \text{ and } \, d_G:=\sqrt{\frac{\log n}{n h^d}}+h^r
 \, .
\end{align*}
The rate $d_0$ denotes the uniform convergence rate of the kernel density estimator $\widehat{f}$.
The rate $d_H$ is the uniform convergence rate of the first step kernel density estimator $\widehat{f}_Z$. Finally, $d_G$ is the uniform convergence rate of the third step kernel density estimator if $H(Z)$ was observed.

\begin{hp}
\label{Assumption: bandwidth}
(i) $d_G h^{-1} \tau_n^{-2}=o(n^{-1/4})$ and $d_G h^{-|l|}\tau_n^{-2|l|}=o(1)$ for $|l|=\lceil (d+1)/2 \rceil +1$.
(ii) $d_H \tau_n^{-5} h^{-2}=o(n^{-1/4})$ and $d_H h_H^{-|l|}\tau_n^{-2-|l|}=o(1)$ for $|l|=\lceil (p+1)/2 \rceil $\,.
\end{hp}

Since the framework is featured by generated regressors, we need some conditions on the rates at which the densities of the variables go to zero on the tails. So, let $p_n^{W(\beta_0)}:=P(f_{W(\beta_0)}(W(\beta_0))\leq 3 \tau_n /2)$ , $p_n^Z:=P(f_Z(Z)\leq 3 \tau_n /2)$, $p_n^{X,Z}:=P(f(X,Z)\leq 3 \tau_n /2)$, and $p_n:=\max(p_n^{W(\beta_0)}, p_n^Z, p_n^{X,Z})$ . Also, let us define the sets
\begin{align}\label{eq: definition of cal W, cal Z and cal U}
    \mathcal{W}^{\delta}_{n,\beta}:=\left\{w\,:\,f_{W(\beta)}(w)\geq \delta \tau_n\right\}\,,&\, \mathcal{Z}_{n}^\delta:=\left\{z\,:\, f_{Z}(z)\geq \delta \tau_n\right\}\,,\nonumber \\
    \text{ and }\mathcal{U}_{n}^\delta:=\{(x,z)\,:\,f&(x,z)\geq \delta \tau_n\}\, .
\end{align}

\begin{hp}
\label{Assumption: trimming}
(i) For any $\delta>0$ there exists $\eta(\delta)>0$ and $N(\delta)$ such that for all $n\geq N(\delta)$ $\mathcal{U}_n^\delta\subset \{(x,z)\,:\,q(\beta,x,H(z))\in \mathcal{W}_{n,\beta}^{\eta(\delta)}\}$ for all $\beta\in B$ and $\mathcal{U}_n^\delta\subset \{(x,z)\,:\, z\in \mathcal{Z}_n^{\eta(\delta)}\}$. 
(ii) $d_0 \tau_n^{-1}=o(1)$, $h_H \tau_n^{-1}=o(1)$, and $h \tau_n^{-1}=o(1)$.
(iii) $p_n=o(n^{-1/2})$, $p_n h^{-d-1}\tau_n^{-2}=o(n^{-1/4})$, $p_n h^{-|l|-d}\tau_n^{-1-|l|}=o(1)$, and $p_n h_H^{-p} h^{-1-|l|}\tau_n^{-3-|l|}=o(1) $ for $|l|=\lceil (d+1)/2 \rceil +1 $. (iv) $p_n h_H^{-p}h^{-2}\tau_n^{-4}=o(n^{-1/4})$  and $p_n h^{-p-|l|}_H\tau_n^{-1-|l|}=o(1)$ for $|l|=\lceil (p+1)/2 \rceil$.
(v) For any $\delta$ there exists $N(\delta)$ such that  $\mathcal{W}_{n,\beta_0}^\delta$ and $\mathcal{Z}_n^\delta$ are convex for all $n\geq N(\delta)$. 
\end{hp}

\begin{hp}\label{Assumption: parametric}
(i) $B$ is a compact convex set with non-empty interior and $\arg \min_{\beta\in B} \mathbb{E}\{Y-G_{W(\beta)}(W(\beta))\}^2=\{\beta^*\}$ for some $\beta^*\in Int(B)$. (ii) $\mathbb{E}\{G_{W(\beta_1)}(W(\beta_1))|W(\beta_2)=w\}$ is Lipschitz and uniformly bounded in $(\beta_1,\beta_2,w)$. 
\end{hp}

Assumption \ref{Assumption: iid and phi}(ii) is an index restriction satisfied in the applications of our test described in Section \ref{sec: Main Applications}. In general, such an index restriction is implied by distributional assumptions on unobservables that can be justified by economic arguments, see Section \ref{sec: Main Applications} for details. 

Assumptions \ref{Assumption: iid and phi}(iii)(iv) set the conditions that the weighting function $\varphi_s$ must satisfy to guarantee the equivalence between the conditional moment restriction in (\ref{eq: null hypothesis II}) and the continuum of unconditional moments in (\ref{eq: null hypothesis continuum}). By building on \citet{bierens_consistent_1982}, \citet[Theorem 2.2]{bierens_econometric_2017} shows that Assumptions \ref{Assumption: iid and phi}(iii)(iv) and the boundedness of $\nu(X,H(Z))$ (implied by Assumption \ref{Assumption: iid and phi}(i) and \ref{Assumption: smoothness}(iii)(v)) are sufficient for such an equivalence to hold. Several choices of $\varphi_s$ have already been discussed in Section \ref{sec: The Test}. More choices can be found in \citet{stinchcombe_consistent_1998}. Assumption \ref{Assumption: smoothness} imposes smoothing conditions that are common in the literature on semi and nonparametric testing. \\
Assumption \ref{Assumption: bandwidth} sets the main conditions on the bandwidths for the first step estimation of $H$ and the third step estimation of $G_{W(\beta_0)}$. These conditions have a twofold role. First, together with Assumptions \ref{Assumption: smoothness} and \ref{Assumption : kernels} they imply that the first step estimator of $H$ and the third step estimator of    $G_{W(\beta_0)}$ are contained in a smooth class of functions with a bounded entropy. This is needed to obtain the asymptotic stochastic equicontinuity of the empirical process at the basis of the statistic $S_n$. Second, the bandwidth conditions in Assumption \ref{Assumption: bandwidth} guarantee that the first step estimator $\widetilde H$, the third step estimator $\widehat G _{\widetilde W (\widehat \beta)}$, and its first-order derivative  $\partial \widehat G _{\widetilde W (\widehat \beta)}$ have an $n^{-1/4}$ convergence rate towards their targets. While the $n^{-1/4}$ consistency of the nonparametric estimators is common in the literature on semiparametric testing, the $n^{-1/4}$ consistency of the first-order derivative of $\widehat G _{\widetilde W (\widehat \beta)}$ is specific to our framework with generated regressors. Heuristically, to handle the estimation error from the generated regressor $\widetilde H(Z)$, we combine the $n^{-1/4}$ consistency of the first step estimator $\widetilde H$ with the $n^{-1/4}$ consistency of the third step derivative $\partial \widehat G _{\widetilde W (\widehat \beta)}$ to get a first-order approximation of the type 
\begin{align*}
    \sqrt{n} \mathbb P _n \widetilde G _{\widetilde W (\widehat \beta)}(\widetilde W (\widehat \beta))\, \varphi_s(X,\widetilde{H}(Z))\approx &    \sqrt{n} \mathbb P _n \widetilde G _{\widetilde W (\widehat \beta)}(W (\beta_0)) \, \varphi_s(X,H(Z))\\
    &+ \sqrt{n} \mathbb P _n \varphi_s(X,H(Z))\, \partial^T G_{W(\beta_0)}(W(\beta_0))\, [  \widetilde W (\widehat \beta) - W(\beta_0)]\, .
\end{align*}
The above expansion allows us to disentangle the estimation error coming from the third step and the estimation error due to the generated variable $\widetilde H(Z)$.

An important aspect is that Assumption \ref{Assumption: bandwidth} avoids undersmoothing and allows for a {\itshape small bias property} of the empirical process at the basis of $S_n$. To explain these features, let us abstract from the appearance of the trimming rate $\tau_n$. Then, Assumption \ref{Assumption: bandwidth} requires that $n h^{4(r-1)}=o(1)$ and $n h_H^{4 r_H}=o(1)$. These conditions have two consequences. First, we avoid undersmoothing in the sense that the test can be implemented with bandwidths that are ``optimal" in terms of estimation.\footnote{ Here by ``optimal" bandwidths we refer to the bandwidths minimizing the Mean Squared Errors of the nonparametric estimators. So, for the first step estimation the ``optimal" bandwidth is proportional to $n^{-1/(2r_H + p)}$, while for the third step estimation the ``optimal" bandwidth is proportional to $n^{-1/(2 r + d)}$. See \citet[Chapter  1]{li_nonparametric_2006}. } Second, 
the empirical process at the basis of the statistic has a {\itshape small bias property}. This means that such an empirical process remains $\sqrt{n}$ consistent although the biases of the nonparametric estimators on which it is based do not converge to zero at a $n^{-1/2}$ rate. See \citet{newey_twicing_2004}. This feature is possible thanks to the bias corrections (or equivalently the $L_2$ boosting corrections) used in the first step and the third step estimation. Without such bias corrections, the conditions $n h^{4(r-1)}=o(1)$ and $n h_H^{4 r_H}=o(1)$ should be ``augmented" by $n h^{2\,r}=o(1)$ and $n h_H^{2 r_H}=o(1)$ to guarantee that the bias of the empirical process disappears fast enough, see e.g. \citet{delgado_significance_2001} and \citet{escanciano_uniform_2014}. With such rates we would have to use undersmoothed bandwidths, and the empirical process at the basis of $S_n$ would no longer have a small-bias property. \\
The bias corrections used here are similar to those used in \citet{xia_goodness--fit_2004} and \citet{lapenta_encompassing_2022} who develop tests where all regressors are observed. Here, we show that these bias corrections (or equivalently the $L_2$ boosting corrections) guarantee a small bias property in a {\itshape semiparametric context} featured by the presence of {\itshape nonparametrically generated regressors}.\footnote{Due to the presence of nonparametrically generated regressors, in our framework the bias corrections of the first step estimator enter the bias corrections of the third step estimator. Thus, the way we handle the bias corrections to obtain the expansion of Proposition \ref{prop: Asymptotic Test}  is fundamentally different from \citet{xia_goodness--fit_2004} and \citet{lapenta_encompassing_2022} who consider tests where {\itshape all regressors are observed}.}\\
The small bias property of the empirical process does not only have a theoretical appeal, but it is also attractive from a practical point of view: as we show in Section \ref{sec: small sample behavior}, this property implies that when the bandwidth is perturbed our test is more stable in terms of size control than a test without a small bias property. See also \citet{newey_twicing_2004}. In other words, from a practical standpoint, thanks to the small bias property we avoid an excessive sensitivity of the test with respect to the smoothing parameter. \\
Assumption \ref{Assumption: trimming}(i)-(iv) gathers the conditions guaranteeing that the trimming disappears fast enough to avoid a bias coming from trimming. Similar conditions can be found in \citet{escanciano_uniform_2014}. Assumption \ref{Assumption: trimming}(v) implies a convexity feature needed to limit the entropy of the class of functions that asymptotically contains the nonparametric estimators. \\
Finally Assumption \ref{Assumption: parametric} imposes the existence of a unique pseudo-true value of the finite dimensional parameter $\beta_0$.
Such a condition is common to any specification test for nonlinear models where the finite dimensional parameter is estimated by optimizing a nonlinear objective function,
see \citet{bierens_consistent_1982}, \citet{lavergne_smooth_2013}, \citet{escanciano_asymptotic_2018-2}. Assumption \ref{Assumption: parametric} plays a double role: under $\mathcal H _0$ it ensures the identification of $\beta_0$, while under $\mathcal H _1$ it ensures that the estimator $\widehat \beta$ has a well defined limit in probability. Now, under $\mathcal H_0$ Assumption \ref{Assumption: parametric} can be shown by using normalization conditions on $B$ and support conditions on $X,Z$. A common normalization condition is that the first component of $\beta_0$, or equivalently the first coordinate of $B$, is set to one, see \citet{blundell_endogeneity_2004}, \citet{rothe_semiparametric_2009}, \citet{escanciano_identification_2016}. Differently, under $\mathcal H_1$ what is needed for our test to have power is a well defined limit in probability of $\widehat \beta$. We could directly impose these conditions, but we choose to keep Assumption \ref{Assumption: parametric} for presentation purposes.\\

To present the asymptotic behavior of the statistic, we need to define several objects. In the definitions below we drop the arguments of the functions when such  arguments are clear from the context.\footnote{In (\ref{eq: components of the IFR for the empirical process at the basis of Sn}) for notational simplicity we use: $\partial G_{W(\beta)}(W(\beta)):=\partial_w G_{W(\beta)}(w)|_{w=W(\beta)}$, $\nabla_{\beta^T} G_{W(\beta)}(W(\beta))=\partial_{\beta^T} G_{W(\beta)}(W(\beta))+\partial^T G_{W(\beta)}(W(\beta))\partial_{\beta^T} W(\beta) $, and $\partial_H q(\beta,X,H(Z)):=\partial_u q(\beta,X,u)|_{u=H(Z)}$. }
\begin{align}\label{eq: components of the IFR for the empirical process at the basis of Sn}
    \varphi_s^\perp (X,Z)=& \varphi_s(X,H(Z)) - \mathbb E \{\varphi_s(X,H(Z)) | W(\beta_0)\}\nonumber \\
    a_s(Z)=& \mathbb E \{\partial^T G_{W(\beta_0)}\, f_{W(\beta_0)}\, \varphi_s^\perp \partial_{H} q | Z\}\nonumber \\
    \Sigma=& \mathbb E \{\nabla_\beta G_{W(\beta_0)} \nabla_{\beta^T} G_{W(\beta_0)}\}\nonumber \\
    b_s^1(X,Z)=& \mathbb E \{\partial^T G_{W(\beta_0)}\, f_{W(\beta_0)}\,\varphi_s^\perp \, \partial_{\beta^T} q \} \Sigma\,\nabla_\beta G_{W(\beta_0)}(W(\beta_0))\nonumber \\
    b_s^2(Z)=& \mathbb E \{\partial^T G_{W(\beta_0)}\, f_{W(\beta_0)}\, \varphi_s^\perp \partial_{\beta^T} q\}\,\Sigma\, \mathbb E \{\nabla_\beta G_{W(\beta_0)}\, \partial^T G_{W(\beta_0)}\, \partial_{H} q | Z\}
\end{align}

The following proposition shows the influence function representation (IFR)
of the empirical processes at the basis of the statistic $S_n$.

\begin{prop}
\label{prop: Asymptotic Test}
Let Assumptions \ref{Assumption: iid and phi}-\ref{Assumption: parametric} hold. Then, under $\mathcal H _0$
\begin{align*}
    \sqrt{n}\, \mathbb P _n\, \widetilde \varepsilon \, \widehat f_{\widetilde W(\widehat \beta)}(\widetilde W(\widehat \beta))\, \varphi_s(X,\widetilde H(Z))\, \widehat t=& \sqrt{n} \mathbb P _n \varepsilon \, \left[f_{W(\beta_0)}\,\varphi^\perp_s - b_s^1(X,Z) \right]\\
    &- \sqrt{n} \mathbb P _n \left[a_s(Z) - b_s^2(Z)\right] \, (D - H(Z)) + o_P(1)
\end{align*}
uniformly in $s \in {\cal S}$.
\end{prop}

The above proposition is proved in Appendix \ref{sec:Asymptotic Analysis}. A similar IFR is obtained in \citet{escanciano_uniform_2014} but under different conditions on the bandwidths and using different estimators. In particular, compared to  \citet{escanciano_uniform_2014} our estimators contain bias corrections and our empirical process has a small bias property. Accordingly, the way we handle the bias terms appearing in the expansion of our empirical process is fundamentally different from the methods used in \citet{escanciano_uniform_2014}.  

The IFR obtained in Proposition \ref{prop: Asymptotic Test} allows us to get more insights on the connection between the small bias property of our empirical process and the presence of generated variables. The second term in the IFR of the empirical process is due to the generated variable: the fact that we replace the unobserved regressor $H(Z)$ with the generated regressor $\widetilde H (Z)$ introduces additional terms in the first-order asymptotics of our empirical process. Thus, the initial empirical process is not {\itshape locally robust} with respect to the first step estimation, in the sense that the Hadamard derivative of the moment condition in (\ref{eq: null hypothesis continuum}) with respect to $H$ is non null. In other words, such a moment condition is not Neyman-orthogonal with respect to the first step estimation, see \citet{chernozhukov_locally_2016} and \citet{escanciano_asymptotic_2018-2}. Now, as an alternative to the bias ($L_2$ boosting) corrections, we could get a small bias property of our empirical process by using a ``locally robust" approach, as proposed in \citet{chernozhukov_locally_2016} and \citet{escanciano_asymptotic_2018-2}. This would require estimating the second term appearing in the IFR of Proposition \ref{prop: Asymptotic Test} (due to the generated variable) and subtracting it to the empirical process at the basis of $S_n$. The new empirical process would be ``locally robust", in the sense that it would be the empirical counterpart of a moment condition whose Hadamard derivative with respect to $H$ is zero. Then, as a test statistic we could consider the Cramer-Von Mises functional of such a {\itshape locally robust empirical process}. This new empirical process would have a small bias property as well as ours, see \citet{chernozhukov_locally_2016} for more insights. However, the implementation of such a locally robust approach in our semiparametric framework with generated variables would be difficult. Constructing a locally robust empirical process would require estimating the term $a_s$. This would require (i) to select additional smoothing parameters and (ii) to estimate the {\itshape nonparametric} derivative $\partial G _{W(\beta_0)}$. Since in practice it is not convenient to estimate nonparametric derivatives, 
such an estimation could complicate the practical implementation of our test. Differently, our $L_2$ boosting procedure does not require estimating nonparametric derivatives. Thus, we are proposing an attractive and convenient method to obtain the small bias property in a semiparametric context with generated variables.\\

From Proposition \ref{prop: Asymptotic Test} we can get the asymptotic distribution of $S_n$ under $\mathcal H _0$.  Since $\varphi_s$ is Lipschitz in $s$, the IFR appearing in Proposition \ref{prop: Asymptotic Test} is Donsker. Hence, the empirical process at the basis of the statistic will converge weakly to a tight zero-mean Gaussian process $\mathbb G _s$ valued in $\ell^\infty (\mathcal S)$, with  $\ell^{\infty}(\mathcal S)$ denoting the space of functionals on $\mathcal S$ endowed with the uniform norm. The Gaussian process $\mathbb{G}_s$ is characterized by the collection of covariances 
\begin{align}\label{eq: collection of covariances}
    \mathbb E\, \Big\{\,\varepsilon [f_{W(\beta_0)}\,\varphi_s^\perp - b_s^1]& - (D-H(Z))[a_s(Z)-b^2_s(Z)]\,\Big\} \nonumber \\
    &\Big\{\,\varepsilon [f_{W(\beta_0)}\,\varphi_l^\perp - b_l^1] - (D-H(Z))[a_l(Z)-b^2_l(Z)]\,\Big\}\text{ : }s,l\in\mathcal S\, .
\end{align}
Hence, by the continuity of the Cramer-Von Mises functional we get that under $\mathcal H _0$ 
\begin{equation*}
    S_n  \cvd  \int_{ }|\mathbb G _s|^2 d \mu(s)\, .
\end{equation*}
From the previous two displays, the asymptotic distribution of $S_n$ depends on the unknown data generating process and we cannot use it for testing.\\
This problem is often encountered in semiparametric testing. As highlighted in the previous section, \citet{delgado_significance_2001} and \citet{xia_goodness--fit_2004} test an equation similar to (\ref{eq: null hypothesis continuum}) but without generated variables. Their test statistics  also converge to a functional of a zero-mean Gaussian process, and to obtain critical values they propose a bootstrap procedure. Differently from theirs, our context is featured by the presence of generated regressors. As shown in Proposition \ref{prop: Asymptotic Test} the generated variables introduce an estimation error that impacts on the asymptotic behavior of our test statistic. Thus, the bootstrap procedures proposed in \citet{delgado_significance_2001} and \citet{xia_goodness--fit_2004} cannot be applied in our testing problem, as they can't reproduce the estimation error due to the generated regressors.  
So, in the next section we motivate and develop a novel bootstrap procedure to get the critical values.

\section{The Bootstrap Test}\label{sec: The Bootstrap Test}
Before developing our bootstrap test, we will discuss some alternative methods that could be used to construct the critical values. This allows us to better motivate our bootstrap procedure. There are general methods available in the literature to adjust semiparametric tests for the presence of generated covariates, see \citet{pagan_econometric_1984}, \citet[Chapter  6]{wooldridge_econometric_2010}, \citet{hahn_asymptotic_2013}, \citet{hahn_nonparametric_2018}. They essentially work in two steps: (i) estimate the adjustment terms in the IFR of the statistic that are due to the generated variables; (ii) by using these estimated adjustment terms, correct the ``standard errors" of the test statistic and construct the critical values for the test. We could adapt such a procedure to our framework, but this would be difficult for two reasons. First, this general procedure holds for a ``studentized" statistic that is asymptotically pivotal, while in our empirical process framework constructing an asymptotically pivotal statistic is tricky, see \citet{song_testing_2010}. Second, given the IFR in Proposition \ref{prop: Asymptotic Test}, implementing such a procedure in our context requires estimating the weighting function $a_s$ defined in (\ref{eq: components of the IFR for the empirical process at the basis of Sn}). As noticed in the previous section, this requires selecting additional smoothing parameters and estimating a nonparametric derivative. So, such a procedure would be practically difficult to implement in our context for the reasons discussed earlier.\\
A second alternative would be to construct a weighted bootstrap procedure. See, e.g., \citet{huang_flexible_2016}. By letting $\{\zeta_i\,:\,i=1,\ldots,n\}$ be a sequence of iid weights with mean zero, unit variance, and with a known distribution, the weighted bootstrap would be based on 
\begin{align*}
    \sqrt{n} \mathbb P _n \zeta \widehat \varepsilon \, \left[\widehat f_{\widehat W(\widehat \beta)}\,\widehat \varphi^\perp_s - \widehat b_s^1(X,Z) \right]
    - \sqrt{n} \mathbb P _n \zeta\left[\widehat a_s(Z) - \widehat b_s^2(Z)\right] \, (D - \widehat H(Z))\, ,
\end{align*}
where the ``hatted" elements stand for estimated objects. The previous process is just a re-weighted version of the IFR in Proposition \ref{prop: Asymptotic Test}. Conditionally on the sample data, the above process is expected to mimic the null behavior of the empirical process at the basis of our statistic. So, the idea would be to use it to construct valid critical values. However, such a procedure is not attractive in our case: due to the presence of generated variables, $a_s$ involves a nonparametric derivative that needs to be estimated, see Equation (\ref{eq: components of the IFR for the empirical process at the basis of Sn}).\\

In view of the previous remarks, our goal in this section is to develop a bootstrap method that (i) does not require estimating a nonparametric derivative, (ii) that can reproduce/mimic the estimation error due to the generated regressors, (iii) that resamples the observations by imposing the null hypothesis in the resampling scheme, and (iv) that is based on an empirical process with a small bias property. \\
To this end, let $\widehat W_i (\beta):=q(\beta,X_i,\widehat H (Z_i))$, where $\widehat H$ has been introduced in Section \ref{subsec: The Test Statistic}. We define
\begin{align*}
    \widehat G _{\widehat W (\beta)}(w):= \widehat T ^Y _{\widehat W (\beta)}(w) / \widehat f _{\widehat W (\beta)}(w)\, ,
\end{align*}
where $\widehat T ^Y _{\widehat W (\beta)}$ and $\widehat f _{\widehat W (\beta)}$ are defined similarly as in (\ref{eq: definition of That Y Wtilde beta}) with $\widehat W (\beta)$ replacing $\widetilde W (\beta)$.\\
Let $\{\xi_i\,:\,i=1,\ldots,n\}$ be iid copies  of a random variable $\xi$ having a known distribution with $\mathbb E \xi=0$ and $\mathbb E \xi^2=1$, and let $\{\xi_i\,:\,i=1,\ldots,n\}$ be independent from the sample data $\{Y_i,X_i,Z_i,D_i\}_{i=1}^n$.  The bootstrap Data Generating Process (DGP) is
\begin{align}\label{eq: bootstrap resampling scheme}
    Y_i^*=&\widehat G _{\widehat W (\widehat \beta)}(\widehat W_i (\widehat \beta)) + \xi_i\,(\,Y_i-\widehat G _{\widehat W (\widehat \beta)}(\widehat W_i (\widehat \beta))\,)\nonumber \\
    D_i^*=& \widehat H (Z_i)+ \xi_i\,(D_i - \widehat H (Z_i))\, 
\end{align}
for $i=1,\ldots,n$. In the ``bootstrap world" only the weights $\{\xi_i\,:\,i=1,\ldots,n\}$ are random and the sample data is fixed.
The bootstrap sample is $\{Y_i^*,D_i^{*},X_i,Z_i\}_{i=1}^n$. Below we describe how to construct the bootstrap version of the statistic by using such a sample. For the moment, let us highlight some important features of the bootstrap DGP in (\ref{eq: bootstrap resampling scheme}). First, the bootstrap DGP is made by two equations. If there were no generated variables, the first equation alone would be enough to construct valid critical values. However, due to the presence of generated variables, a bootstrap DGP based only on the first equation would provide a misspecified inference, as it could not mimic the estimation error coming from the generated regressors. This is the reason for introducing the second equation in (\ref{eq: bootstrap resampling scheme}): its role is to reproduce the estimation error from the generated variable $\widetilde H (Z)$ and to mimic the behavior of the second term in the IFR of Proposition \ref{prop: Asymptotic Test}.\\
Second, in some of the applications described in Section \ref{sec: Main Applications}, the variable $D$ is a component of $X$. Thus, $D$ enters the true variable $W(\beta_0)=q(\beta_0,X,H(Z))$ in the population DGP. However, in the bootstrap DGP the variable $\widehat W (\widehat \beta)=q(\widehat \beta,X,\widehat H (Z))$ will not contain the bootstrap version of $D$, i.e. $D^*$, but only $D$. As we show in Proposition  \ref{prop: bootstrap test}, this is enough for providing a valid bootstrap inference. \\
Third, the observations $Y^*_i$ and $D^{*}_i$ cannot be resampled independently, but they must be generated by the {\itshape same} bootstrap weight $\xi_i$. This is an important feature, as it guarantees that the covariance  of the bootstrap errors $[Y^*_i-\widehat G _{\widehat W (\widehat \beta)}(\widehat W _i (\widehat \beta))$ $ , D^*_i - \widehat H (Z_i)]$ has the same structure as the covariance of the population errors $[\,Y_i-G_{W(\beta_0)}(W_i(\beta_0))\,,$ $D_i-H(Z_i)\,]$. Since the former errors determine the behavior of the bootstrap version of $S_n$ and the latter errors determine the null behavior of $S_n$, such a feature is fundamental to guarantee the validity of the bootstrap inference. See Proposition \ref{prop: bootstrap test} ahead.\\
Fourth, the bootstrap DGP in (\ref{eq: bootstrap resampling scheme}) does not involve the estimators based on bias corrections, i.e. $\widetilde H (Z)$ and $\widetilde G _{\widetilde W (\widehat \beta)}$, but only the non-corrected estimators $\widehat H(Z)$ and $\widehat G _{\widehat W (\widehat \beta)}$. As we show in Proposition \ref{prop: bootstrap test}, this guarantees that also the bootstrapped empirical process will enjoy a small bias property. Differently, if we included the bias corrected estimators in (\ref{eq: bootstrap resampling scheme}) the bootstrapped empirical process would not enjoy the small bias property.\\
Finally, we notice that the bootstrap DGP in (\ref{eq: bootstrap resampling scheme}) guarantees the validity of the null hypothesis in the ``bootstrap world". In particular, let us denote with $\mathbb E ^*$  the expectation in the bootstrap world, where only the weights $\{\xi_i\,:\,i=1,\ldots,n\}$ are random while the sample data is fixed. Then, (\ref{eq: bootstrap resampling scheme}) implies that $\mathbb E ^* \{Y^*|\nu(X,\widehat H (Z))\}$ $=\mathbb E ^* \{Y^*|q(\widehat \beta, X,\widehat H (Z))\}$. This equality represents bootstrap version of $\mathcal H _0$. The fact that we impose the null hypothesis in the bootstrap world is an attractive feature, as it means that we are using all the ``available information" when resampling our observations. \\
In what follows we show how to construct the bootstrap version of $S_n$. Similarly to Section \ref{sec: The Test}, this is done by a three steps procedure.\\

\noindent\textbf{First step estimation for the bootstrap}. Let
\begin{align}
    \widehat T _Z^{D^{*}}(z):=\frac{1}{n h_H^p}\sum_{i=1}^n D^{*}_i\,\widehat t _i K_H&\left(\frac{z-Z_i}{h_H}\right) \text{ , }\widehat H^*(z):= \widehat T _Z^{D^{*}}(z) / \widehat f_{Z}(z)\, ,\nonumber\\
    \widehat T ^{D^{*}-\widehat H ^*}_Z(z):=&\frac{1}{n h_H ^p}\sum_{i=1}^n [D_i^{*}-\widehat H ^* (Z_i)]\, \widehat t _i \, K_H\left(\frac{z-Z_i}{h_H}\right)\, .
\end{align}
We use the trimming $\widehat t$ in the above expressions to control for a random denominator in $D^{*}$, see (\ref{eq: bootstrap resampling scheme}). Then, the bootstrap counterpart of $\widetilde H$ is 
\begin{equation}
    \widetilde H ^* (z):=\widehat H ^* (z) + \widehat T ^{D^{*}-\widehat H ^*}_Z(z)/\widehat f _Z (z)\, . 
\end{equation}

\noindent\textbf{Second step estimation for the bootstrap}. Let $\widetilde W ^* _i (\beta):=q(\beta,X_i,\widetilde H ^* (Z_i))$. Then, we define 

\begin{equation}\label{eq: definition of That Ystar Wtilde star beta}
    \widehat T ^{Y^*}_{\widetilde W^* (\beta)}(w):=\frac{1}{n h^d}\sum_{i=1}^n Y_i^* \, \widehat t _i\,K\left(\frac{w-\widetilde W _i^*(\beta)}{h}\right)
\end{equation}
 and let $\widehat f_{\widetilde W ^* (\beta)}$ be defined similarly as in Equation (\ref{eq: definition of That Y Wtilde beta}), obviously with $\widetilde W ^* (\beta)$ replacing $\widetilde W (\beta)$. Again, we use the trimming $\widehat t$ to control for the random denominator in $Y^*$ and $\widetilde W ^*(\beta)$. Then, the bootstrap counterpart of $\widehat G _{\widetilde W (\beta)}$ is 
 \begin{equation}
     \widehat G ^*_{\widetilde W ^* (\beta)}(w):=\widehat T ^{Y^*}_{\widetilde W ^* (\beta)}(w)/\widehat f _{\widetilde W ^* (\beta)}(w)\, .
 \end{equation}
 So, we let the bootstrap counterpart of $\widehat \beta$ be 
 \begin{equation*}
     \widehat \beta ^* :=\arg \min _{\beta\in B}\mathbb P _n [Y^*-\widehat G ^*_{\widetilde W ^* (\beta)}(\widetilde W ^* (\beta))]^2\, \widehat t\, .
 \end{equation*}
 
 \noindent \textbf{Third step estimation for the bootstrap}. Let us define $\widehat \varepsilon ^* _i := Y_i^* - \widehat G ^* _{\widetilde W ^* (\widehat \beta ^*)}(\widetilde W ^* _i (\widehat \beta ^*))$ for $i=1,\ldots,n$. Let $\widehat T ^{\widehat \varepsilon^*}_{\widetilde W ^* (\beta)}$ be defined similarly as in (\ref{eq: definition of That Ystar Wtilde star beta}), obviously with $\widehat \varepsilon^*$ replacing $Y^*$. Then, the bootstrap counterpart of $\widetilde G _{\widetilde W (\widehat \beta)}$ is 
 \begin{equation}\label{eq: definition of G tilde star beta hat star}
     \widetilde G ^* _{\widetilde W ^* (\widehat \beta^*)}(w):= \widehat G ^*_{\widetilde W ^* (\widehat \beta ^*)}(w)+ \widehat T ^{\widehat \varepsilon^*}_{\widetilde W ^* (\widehat \beta^*)}(w)/\widehat f_{\widetilde W ^* (\widehat \beta^*)}(w)\, .
 \end{equation}
 The bootstrap counterpart of the bias corrected residual will be $\widetilde \varepsilon ^*_i:=Y_i^*-\widetilde G ^*_{\widetilde W ^* (\widehat \beta ^*)}(\widetilde W ^*_i(\widehat \beta ^*))$. So, the bootstrap counterpart of the statistic is 
 \begin{equation}
     S_n^* := n \int_{ }\,\left|\mathbb P _n \widetilde \varepsilon ^* \,\widehat f _{\widetilde W ^* (\widehat \beta ^*)}(\widetilde W ^*(\widehat \beta ^*)) \, \varphi_s(X,\widetilde H ^* (Z))\, \widehat t\right|^2\, d\mu(s)\, .
 \end{equation}
 
Let us denote with $\Pr^*$ the probability where only the bootstrap weights $\{\xi_i\,:\,i=1,\ldots,n\}$ are random while the sample data is fixed. A test at the $\alpha$ nominal level will be based on the quantile 
\begin{equation}
    \widehat c_{1-\alpha}=\inf\left\{\,c\,:\, \Pr \,^*\,(S_n^*\leq c)\geq 1-\alpha \,\right\}\, .
\end{equation}
In practice, one can compute $\widehat c_{1-\alpha}$ by a Monte Carlo procedure that goes as follows:
\begin{enumerate}
    \item Generate $n$ iid bootstrap weights $\{\xi_i\,:\,i=1,\ldots,n\}$ from the distribution of $\xi$ and compute the observations $\{Y_i^*,D^{*}_i\,:\,i=1,\ldots,n\}$ as in (\ref{eq: bootstrap resampling scheme}) 
    \item Given the bootstrap sample $\{Y_i^*,D^{*}_i,X_i,Z_i\}_{i=1}^n$ compute the test statistic $S_n^*$ according to the three steps procedure previously described
    \item Repeat steps 1-2 $B$ times, with $B$ large, and  obtain the collection of bootstrapped statistics $\{S^*_{n,b}\,:\,b=1,\ldots,B\}$
    \item Compute $\widehat c_{1-\alpha}$ as the $1-\alpha$ quantile of $\{S^*_{n,b}\,:\,b=1,\ldots,B\}$. 
\end{enumerate}
If the statistic $S_n$ is larger than $\widehat c_{1-\alpha}$, then the null $\mathcal H _0$ will be rejected at the $\alpha$ nominal level. 
We recommend choosing $B$ such that $(1-\alpha)B$ is an integer, see \cite{davison1997bootstrap}. \\

The proposition below shows the validity of the bootstrap test just proposed.

\begin{prop}\label{prop: bootstrap test}
Let $\{\xi_i\,:\,i=1\ldots,n\}$ be an iid sequence of bounded variables independent from the sample, with $\mathbb E \xi=0$ and $\mathbb E \xi^2=1$, and let Assumptions \ref{Assumption: iid and phi}-\ref{Assumption: parametric} hold. Then,
\begin{enumerate}[label=(\roman*)]
    \item under $\mathcal H _0$
    \begin{align*}
         \sqrt{n}\, \mathbb P _n\, \widetilde \varepsilon^* \, \widehat f_{\widetilde W^*(\widehat \beta^*)}(\widetilde W^*(\widehat \beta^*))\, \varphi_s(X,&\widetilde H^*(Z))\, \widehat t\\
         =& \sqrt{n} \mathbb P _n\, \xi\,\varepsilon \, \left[f_{W(\beta_0)}\,\varphi^\perp_s - b_s^1(X,Z) \right]\\
    &- \sqrt{n} \mathbb P _n \, \xi\,(D - H(Z))\,\left[a_s(Z) - b_s^2(Z)\right] + o_P(1)\, 
    \end{align*}
    uniformly in $s\in\mathcal S$, where $\varphi_s^\perp$, $a_s$, $b_s^1$, and $b_s^2$ are defined in (\ref{eq: components of the IFR for the empirical process at the basis of Sn}), and the probability space is the joint probability on the random bootstrap weights and the sample data.
    \item Under $\mathcal H _0$,  $\Pr(S_n > \widehat c_{1-\alpha})\rightarrow \alpha$ . 
    \item Under $\mathcal H _1$,  $\Pr (S_n > \widehat c_{1-\alpha})\rightarrow 1$ . 
\end{enumerate}
\end{prop}

The above result is proved in Appendix \ref{sec: Bootstrap Analysis}. Let us comment on the bootstrap expansion obtained in $(i)$. The IFR of the bootstrapped empirical process is a reweighted version of the IFR obtained in Proposition \ref{prop: Asymptotic Test}. The weights are represented by the bootstrap weights $\{\xi_i\,:\,i=1,\ldots,n\}$. Notice that the covariance function of the bootstrapped IFR is the same as the covariance function in (\ref{eq: collection of covariances}). This guarantees that under $\mathcal H _0$ the bootstrapped IFR will have the same behavior as the population IFR of Proposition \ref{prop: Asymptotic Test}. Thus, the bootstrapped statistic $S_n^*$ will mimic the null behavior of $S_n$ and the bootstrap inference will be valid. For such a feature to hold, it is fundamental that the bootstrap weights in the two equations in (\ref{eq: bootstrap resampling scheme}) are the same. \\
Finally, we notice that Proposition \ref{prop: bootstrap test} is obtained under the bandwidth conditions of Assumption \ref{Assumption: bandwidth}. Thus, also the bootstrapped empirical process at the basis of $S_n^*$ has a small bias property, in the sense that it remains $\sqrt{n}$ consistent even though the bias of the nonparametric estimators on which it is based does not convergence to zero at a $n^{-1/2}$ rate.

\section{Main Applications}\label{sec: Main Applications}
In this section we present the main semiparametric models our test can be applied to.\\

\noindent \textbf{Binary choice models with control functions}. Control function methods are popular in econometrics for estimating binary choice models with endogenous regressors. They have been first introduced by \citet{rivers_limited_1988}. They consider 
\begin{equation}\label{eq: binary choice variable}
    Y=\mathbb I \{g_1(X)\geq u\}\, ,
\end{equation}
where $Y$ denotes the binary choice of an agent, $g_1(X)=\theta_0^T X$ for a parameter $\theta_0$, $X=(D,Z_1^T)^T$, and the unobserved error $u$ is independent from $Z_1$ but correlated with $D$. To handle such a correlation, \citet{rivers_limited_1988} assume that there exists an instrument $Z_2$ such that for $Z=(Z_1^T,Z_2^T)^T$
\begin{equation}\label{eq: CF assumption for single index}
u=g_2(V,\epsilon)\text{ , with }V=D-\mathbb E \{D|Z\}\text{ and }\epsilon\perp X,Z\, .
\end{equation}
The residual $V$ is called {\itshape control function}, as it allows controlling for the endogeneity of $D$, and $g_2(V,\epsilon)=-\gamma_0 V + \epsilon$.
In addition to this structure, \citet{rivers_limited_1988} assume that $(\epsilon,V)$ are jointly normal and $\mathbb E \{D|Z\}$ is a linear function of $Z$. See also Section 4 in \citet{wooldridge_control_2015}. Instead, one could consider a model given by (\ref{eq: binary choice variable}) and (\ref{eq: CF assumption for single index}), without imposing the normality of $(\epsilon,V)$ and the linearity of $\mathbb E \{D|Z\}$. This could be seen as a semiparametric version of \citet{rivers_limited_1988}'s setting. Equations (\ref{eq: binary choice variable}) and (\ref{eq: CF assumption for single index}) imply that $\mathbb E \{Y|X,Z\}=\mathbb E \{Y|X,V\}$ and that
\begin{align*}
    \mathbb E \{Y|X,V\}=\mathbb E \{Y|\theta^T_0 X + \gamma V\}\, .
\end{align*}
The equation $\mathbb E \{Y|X,Z\}=\mathbb E \{Y|X,V\}$ is an exclusion restriction that is a direct  consequence of (\ref{eq: binary choice variable}) and the control function assumption in (\ref{eq: CF assumption for single index}). It is obtained without imposing the parametric restrictions on $g_1$ and $g_2$, and it can be justified by economic arguments. 
Differently, the equation in the above display is implied by the parametric restrictions on $g_1$ and $g_2$ introduced to limit the curse of dimensionality. Our test can be applied to check the correct specification of such an equation. In particular, the equation in the previous display is a specific version of Equation (\ref{eq: null hypothesis}), where $\nu(X,H(Z))=(X,D-H(Z))$, $q((\theta_0,\gamma_0),X,H(Z))=(\theta_0^T X + \gamma_0(D-H(Z)))$, and the generated variable is $H(Z)=\mathbb E \{D|Z\}$. \\

In the same spirit, \citet{blundell_endogeneity_2004} consider a model given by Equation (\ref{eq: binary choice variable}), $X=(D,Z_1^T)^T$, $Z=(Z_1^T,Z_2^T)^T$, and the following exclusion restrictions
\begin{equation}\label{eq: CF assumption from Blundell and Powell}
    u|X,Z\sim u|X,V\sim u|V\, ,
\end{equation}
where $\sim$ stands for equality in distribution, and $V$ is as in (\ref{eq: CF assumption for single index}). Moreover, they assume that $g_1(X)=\theta_0^T X$ for a parameter $\theta_0$. The above display,  (\ref{eq: binary choice variable}), and the parametric structure of $g_1$ imply that $\mathbb E \{Y|X,Z\}=\mathbb E \{Y|X,V\}$ and that 
\begin{align}\label{eq: implication from Blundell and Powell}
    \mathbb E \{Y|X,V\}=\mathbb E \{Y|\theta^T_0 X , V\}\, .
\end{align}
The equation $\mathbb E \{Y|X,Z\}=\mathbb E \{Y|X,V\}$ is a direct consequence of (\ref{eq: binary choice variable}) and the conditional independence restrictions in (\ref{eq: CF assumption from Blundell and Powell}). It does not need the single index structure of $g_1$. Hence, it can be justified by economic arguments. Differently, the parametric restriction on $g_1$ implies Equation (\ref{eq: implication from Blundell and Powell}). 
Such a restriction is introduced for convenience (i.e. to limit the curse of dimensionality) and cannot be justified by economic arguments. Our test can be used to check the validity of such a restriction. In particular, this restriction is a particular case of (\ref{eq: null hypothesis}), with $\nu(X,H(Z))=(X,D-H(Z))$, $q(\theta_0,X,H(Z))=(\theta_0^T X,D-H(Z))$, and  $H(Z)=\mathbb E \{D|Z\}$ as the generated variable . \\

\noindent \textbf{Separable single index models with endogenous regressors}. Let $Y$ be a continuous variable, $X=(D,Z_1^T)^T$, and $Z=(Z_1^T,Z_2^T)^T$. \citet{newey_nonparametric_1999} consider the model 
\begin{align}\label{eq: implication from Newey, Powell, and Vella}
    Y=&G_0(X^T \beta_0)+\epsilon \nonumber \\
    \text{ with }& \mathbb E \{\epsilon |X,Z\}= \mathbb E \{\epsilon | V\}
\end{align}
and $V=D-\mathbb E \{D|Z\}$. Here $\epsilon$ is an unobserved error term, $G_0$ is a nonparametric function, $\beta_0$ is a vector of parameters of interest, and $D$ is an endogenous variable possibly correlated with $\epsilon$. The second part of the previous display is a mean-independence condition allowing to control for the endogeneity of $D$. Identification of $G_0$ and $\beta_0$ is discussed in \citet{newey_nonparametric_1999}. Similarly to the previous subsection, this model implies that $\mathbb E \{Y|X,Z\}=\mathbb E \{Y|X,V\}$ and that $\mathbb E \{Y|X,V\}=\mathbb E \{Y|X^T\beta_0,V\}$. The equation $\mathbb E \{Y|X,Z\}=\mathbb E \{Y|X,V\}$ does not depend on the index restriction on $G_0$ and is a direct consequence of the mean independence condition in (\ref{eq: implication from Newey, Powell, and Vella}). Hence, it can be justified by economic arguments. Differently, the equation $\mathbb E \{Y|X,V\}=\mathbb E \{Y|X^T\beta_0,V\}$ is a parametric restriction that can be checked by using our test, with $H(Z)=\mathbb E \{D|Z\}$ as the generated variable, $\nu(X,H(Z))=(X,D-H(Z))$, and $q(\beta_0,X,H(Z))=(\beta_0^T X,D-H(Z))$.\\

\noindent \textbf{Models with sample selection}. \citet{escanciano_identification_2016} consider a semiparametric model with sample selection and possibly truncation on the response variable. Let $\widetilde Y \in \mathbb R$ be a scalar random variable denoting an agent's decision, $(X,Z)$ be a vector of covariates, $D\in \{0,1\}$ be a selection variable, and $(\epsilon,u)$ be unobserved error terms. The model is 
\begin{align}\label{eq: sample selection model}
    \widetilde Y = \psi(g(X), \epsilon)\, , \, D=\mathbb I \{H(Z)\geq u\}\, ,\, (\epsilon,u)\perp (X,Z)\, ,
\end{align}
where $\psi$ is a known function, $g(X)=\beta_0^T X$, but $H$ is nonparametric. Without loss of generality we can assume $u\sim U[0,1]$, so $H$ is identified as $H(Z)=\mathbb E \{D|Z\}$ and denotes the propensity score. Because of selection, only $Y=\widetilde Y D$ is observed. Let us denote with $F$ the joint distribution of $(\epsilon,u)$. With $F$ nonparametric, this setup is a semiparametric generalization of Heckman's sample selection model. Putting $\psi(g(X) ,\epsilon)=\mathbb I \{g(X) \geq \epsilon\}$ gives a binary choice model with sample selection. If instead $\psi(g(X) ,\epsilon)=\max \{0,g(X) + \epsilon\}$ we get a truncated regression model with sample selection. This model is also known as double hurdle model, see \citet{escanciano_identification_2016} and \citet{cragg_statistical_1971}. By the independence between $(\epsilon,u)$ and $(X,Z)$ we have $\mathbb E \{Y|X,Z\}=\int \psi(g(X), \epsilon)\, $ $\mathbb I \{H(Z)\geq u\} d\,F(u,\epsilon)$. Hence, the model implies $\mathbb E \{Y|X,Z\}= \mathbb E \{Y|X,H(Z)\}$ and 
\begin{equation*}
    \mathbb E \{Y|X,H(Z)\}=\mathbb E \{Y|\beta_0^T X,H(Z)\}\, .
\end{equation*}
The equation  $\mathbb E \{Y|X,Z\}= \mathbb E \{Y|X,H(Z)\}$ is obtained from (\ref{eq: sample selection model}) without assuming the linear structure of $g$. Differently, the equation in the above display comes from the parametric restriction on $g$. Such an equation can be checked by using our test. In terms of the notation of Section \ref{sec: The Test}, our test can be applied to this framework by putting $\nu(X,H(Z))=(X,H(Z))$,  $q(\beta_0,X,H(Z))=(\beta_0^T X, H(Z))$, and setting the generated variable as $H(Z)=\mathbb E \{D|Z\}$.\\

\noindent \textbf{Semiparametric games with incomplete information}. \citet{aradillas-lopez_pairwise-difference_2012} and \citet{lewbel_identification_2015} study identification and estimation in binary games with incomplete information and linear payoffs. For simplicity, let us consider two players indexed by $j\in \{1,2\}$. Each must take a binary decision, say $a_j\in \{0,1\}$. Let $X_j$ be the vector of observed covariates entering agent $j$'s payoff function and let $Z=(X_1^T,X_2^T)^T$. Each player has a private information $u_j$ that neither the other player nor the researcher can observe. It is assumed that $(u_1,u_2)\perp Z$ and $u_1\perp u_2$. The payoff function of player $j$ is 
\begin{equation*}
    a_j\, [g_j(X_j, a_{-j})-u_j]\, .
\end{equation*}
with $g_j(X_j, a_{-j})=\gamma_j^T X_j + \alpha_j a_{-j} $.
We let $H_j(Z)=\mathbb E \{a_j|Z\}$. Under the hypothesis that a unique Bayesian-Nash equilibrium is played, the model implies that $\mathbb E \{a_j |X_1,X_2\}=\mathbb E \{a_j|X_j,$ $H_{-j}(Z)\}$ and that 
\begin{equation*}
    \mathbb E \{a_j|X_j,H_{-j}(Z)\}=\mathbb E \{a_j|\gamma_j^T X_j+ \alpha_j\,H_{-j}(Z)\}
\end{equation*}
for $j=1,2$. The equation $\mathbb E \{a_j |X_1,X_2\}=\mathbb E \{a_j|X_j,H_{-j}(Z)\}$ follows directly from the assumption that a unique Bayesian-Nash equilibrium is played and does not need the linearity of the profit function $g_j$. This can be justified by economic arguments. Differently, the  equation in the previous display follows from a parametric restriction on $g_j$. Our test can be used to check such a condition. In terms of the notation of Section \ref{sec: The Test}, our test can be applied to this model with $\nu(X,H(Z))=(X_j,H_{-j}(Z))$  $q((\gamma_j,\alpha_j), X_j,H_{-j}(Z))=\gamma_j^T X _j + \alpha_j \, H_{-j}(Z)$, $D=a_{-j}$, and $H_{-j}(Z)=\mathbb{E}\{a_{-j}|Z\}$ as a generated variable.

\section{Empirical Implementation and Applications}\label{sec: Empirical Implementation and Applications}
In the first part of this section we show how to implement our test in practice. In the second part, we study its behavior in small samples. Finally, in the third part we apply our test to a real data example. 

\subsection{Implementation of the test}\label{sec: Implementation of the test}
To show the practical implementation of our test, we provide the computational details for the three steps procedure described in Section \ref{subsec: The Test Statistic}. We suggest to prior trim the 1\% more extreme  the observations on $X,Z$.\footnote{ Specifically, we trim those observations falling beyond the  99\% quantile of the empirical distribution of $(|X|,|Z|)$. } Let us now describe the first step. The bandwidth for $\widetilde H $ is set according to the Silverman's rule of thumb, so $h_H=\widehat {\sigma}_Z\, n^{1/(2+2\,r_H)}$ with $\widehat \sigma _Z$ denoting the estimated standard deviations of the components of $Z$. The kernel $K_H$ is set to a second order Gaussian kernel. Then, we compute $\widetilde H $ as in (\ref{eq: definition of Htilde}). \\
For the second step of Section \ref{subsec: The Test Statistic}, we need to compute the estimator $\widehat \beta$ by minimizing the objective function in (\ref{eq: definition of betahat}). Notice that the estimator $\widehat G_{\widetilde W (\beta)}$ in (\ref{eq: definition of betahat}) depends on both $\beta$ and the bandwidth $h$. We follow the standard approach in semiparametric estimation, see \citet{delecroix_semiparametric_2006}, \citet{rothe_semiparametric_2009}, \citet{escanciano_identification_2016}, \citet{maistre_nonparametric_2018}, and in practice we choose $\widehat \beta$ by solving the joint minimization problem  
\begin{equation}\label{eq: leave one out SLS}
    (\widehat \beta , \overline h)=\arg \min_{\beta,h}\sum_{i=1}^n[Y_i - \widehat G _{\widetilde W (\beta)}^{-i}(\widetilde W _i(\beta))]^2\, \widehat t _i
\end{equation}
where $\widehat G_{\widetilde W (\beta)}^{-i}$ denotes the leave-$i$-out version of $\widehat G_{\widetilde W (\beta)}$ defined as\footnote{For notational convenience we are dropping the dependence of $\widehat G _{\widetilde W (\beta)}^{-i}$ from $h$. } 
\begin{equation}\label{eq: leave i out version of Ghat Wtilde beta}
    \widehat G _{\widetilde W (\beta)}^{-i}(w):=\frac{\sum_{j\neq i} Y_j K\left(\frac{w-\widetilde W _j(\beta)}{h}\right)}{\sum_{j\neq i}  K\left(\frac{w-\widetilde W _j(\beta)}{h}\right)}\, .
\end{equation}
The $\widehat \beta$ from (\ref{eq: leave one out SLS}) is the estimator of $\beta_0$. 
\\
Let us now describe the computational details of the third step of Section \ref{subsec: The Test Statistic}. To compute the statistic in (\ref{eq: feasible test statistic}) we need to select (i) a weighting function $\varphi_s$, a measure $\mu$ , and a set $\cal S$, and (ii)  a bandwidth $h$ and a kernel $K$ for $\widetilde G_{\widetilde W (\widehat \beta)}$ and $\widehat f _{\widetilde W (\widehat \beta)}$. We select $\varphi_s$, $\mu$, and $\cal S$ to make the computation of the integral in (\ref{eq: feasible test statistic}) fast and simple, see below for details. So, we set $\varphi_s(\nu(X,H(Z)))=\exp(\,\sqrt{-1}\,s^T\, \nu(X,H(Z))\,)$, $\mu$ to the standard multivariate Gaussian density, and $\mathcal {S}=\mathbb{R}^{\dim\nu(X,H(Z))}$. We set the kernel $K$ to the second order Gaussian kernel. Finally, to choose the bandwidth $h$ we minimize the distance between the semiparametric model and the null hypothesis, so as to compute $S_n$ in an ``optimistic" way. To this end, let us define the leave-$i$-out version of $\widetilde G_{\widetilde W (\widehat \beta)}$ as follows
\begin{align*}
    \widetilde G ^{-i}_{\widetilde W (\widehat \beta)}(w):=\widehat G ^{-i}_{\widetilde W (\widehat \beta)}(w)+\widehat T ^{\widehat \varepsilon,-i}_{\widetilde W (\widehat \beta)}(w)\, / \, \widehat f^{-i}_{\widetilde W (\widehat \beta)}(w)\, ,
\end{align*}
where $\widehat G ^{-i}_{\widetilde W (\widehat \beta)}$ is defined in (\ref{eq: leave i out version of Ghat Wtilde beta}), $\widehat f^{-i}_{\widetilde W (\widehat \beta)}$ is the leave-$i$-out version of $\widehat f_{\widetilde W (\widehat \beta)}$ (defined similarly as in (\ref{eq: leave i out version of Ghat Wtilde beta})), and 
\begin{equation*}
    \widehat T ^{\widehat \varepsilon,-i}_{\widetilde W (\widehat \beta)}(w):=\frac{1}{n h^d}\sum_{j\neq i} [Y_j - \widehat G^{-j}_{\widetilde W (\widehat \beta)}(\widetilde W _j(\widehat \beta))]\, \widehat t_j\, K\left(\frac{w-\widetilde W _j(\widehat \beta)}{h}\right)\, .
\end{equation*}
Then, the bandwidth $h$ for testing is selected as follows
\begin{align}\label{eq: bandwidth choice}
    \widehat h = & \arg \min_h \int_{ } \Big|\frac{1}{n}\sum_{i=1}^n [Y_i-\widetilde G_{\widetilde W (\widehat \beta)}^{-i}(\widetilde W _i(\widehat \beta))]\, \, \varphi_s(X_i,\widetilde H (Z_i))\, \widehat t_i\,\Big|^2 d\,\mu(s)\, .
\end{align}
The intuition is simple: we are choosing the bandwidth for testing that makes the model ``close" to $\mathcal H _0$, where the distance is measured by a leave-one-out version of the Cramer-Von Mises metric. The procedure can be seen as an adaptation of the Cross-validation principle to our testing problem: while the Cross Validation chooses the bandwidth minimizing the distance between the leave-one-out residuals and zero, in our testing problem we choose the bandwidth to minimize the distance between the semiparametric model and the null hypothesis. This distance is measured by a leave-one-out version of the Cramer-Von Mises functional. As we explain below, this way of choosing the bandwidth allows us to compare an ``optimistic" statistic computed with the sample data to an ``optimistic" statistic computed with the bootstrap data, see below for details.\footnote{As an alternative, we could choose the bandwidth to maximize the value of the statistic in the hope to increase the power. This approach is suggested in \citet{escanciano_uniform_2014} for an asymptotically pivotal test statistic. However, such a procedure would appear to be tricky in our case, since (i) our statistic is not asymptotically pivotal and (ii) we also have to choose the bandwidth in the bootstrap sample.}\\
Thanks to the choice of $\varphi_s$, $\mu$, and $\mathcal{S}$, the integral in the optimization problem (\ref{eq: bandwidth choice}) has a simple closed form expression given by 
\begin{align*}
    \frac{1}{n^2}\sum_{i,j}\left[Y_i-\widetilde G_{\widetilde W (\widehat \beta)}^{-i}(\widetilde W _i(\widehat \beta))\right]\widehat t _i\,&\left[Y_j-\widetilde G_{\widetilde W (\widehat \beta)}^{-j}(\widetilde W _j(\widehat \beta))\right]\widehat t_j\,\\
    & \cdot\, \exp (-\|\nu(X_i,\widetilde H (Z_i))-\nu(X_j,\widetilde H (Z_j))\|^2/2)\, ,
\end{align*}
where $\exp(-\|\cdot\|^2/2)$ is the characteristic function of the standard multivariate normal. Since the above double sum depends non linearly on $h$, the optimization in (\ref{eq: bandwidth choice}) must be carried out numerically.\\
Once $\widehat h$ has been obtained from the minimization in (\ref{eq: bandwidth choice}), we can use it to compute the third step estimators $\widetilde{G}_{\widetilde W (\widehat \beta)}$ and $\widehat f _{\widetilde{W}(\widehat \beta)}$ as in Section \ref{subsec: The Test Statistic}. Thus, given the residuals  $\widetilde \varepsilon_i =Y_i - \widetilde G_{\widetilde W (\widehat \beta)}(\widetilde W_i (\widehat \beta))$ with $i=1,\ldots,n$,
we get the test statistic:
\begin{align*}
    S_n=& \int_{ }\Big|\frac{1}{\sqrt{n}}\sum_{i=1}^n \widetilde \varepsilon _i \widehat f_{\widetilde W (\widehat \beta)}(\widetilde W_i (\widehat \beta))\, \widehat t _i \varphi_s(X_i,\widetilde H (Z_i))\Big|^2 d\, \mu(s)\\
    =& \frac{1}{n}\sum_{i,j} \widetilde \varepsilon _i \widehat f_{\widetilde W (\widehat \beta)}(\widetilde W_i (\widehat \beta))\,\widehat t _i\, \widetilde \varepsilon _j \widehat f_{\widetilde W (\widehat \beta)}(\widetilde W_j (\widehat \beta))\, \widehat t _j\, \exp(-\|\nu(X_i,\widetilde H (Z_i))-\nu(X_j,\widetilde H (Z_j))\|^2/2)\, .
\end{align*}
Let us now detail the computations for the bootstrapped statistic $S_n^*$ described in Section \ref{sec: The Bootstrap Test}. First, to generate bootstrap observations from (\ref{eq: bootstrap resampling scheme}) we need to construct the estimates $\widehat H$ and $\widehat G_{\widehat W (\widehat \beta)}$, and we need to define a distribution for the bootstrap weights $\{\xi_i\,:\,i=1,\ldots,n\,\}$. The estimators $\widehat H$ and $\widehat G_{\widehat W (\widehat \beta)}$ entering the bootstrap DGP in (\ref{eq: bootstrap resampling scheme}) are computed by using second order Gaussian kernels and the bandwidths $(h_H,\widehat h)$ obtained from the sample data. We suggest to generate the bootstrap weights from a two-points Rademacher distribution: $\Pr(\xi=1)=\Pr(\xi=-1)=1/2$. The choice of this distribution for the bootstrap weights is motivated by its good performance in other contexts, see \citet{davidson_wild_2008}, \citet{djogbenou_asymptotic_2019}.\footnote{ I thank an anonymous referee for having suggested the use of these weights. }
Then, given the bootstrap sample $\{Y_i^*,D_i^{*},X_i,Z_i\, :\, i=1,\ldots,n\}$ obtained from (\ref{eq: bootstrap resampling scheme}), we use {\itshape the same procedure as in the sample} to compute the bootstrap version of the test statistic, $S_n^*$. Notice that for each bootstrap sample/iteration we will have to minimize the bootstrap version of (\ref{eq: leave one out SLS}) to get $\widehat \beta^*$ and the bootstrap version of (\ref{eq: bandwidth choice}) to get $\widehat h ^*$.  Finally, the critical values are obtained by the Monte Carlo procedure described in Section \ref{sec: The Bootstrap Test}. \\
At this point, it is worth to spend some words on the bandwidth selection procedure we have proposed. As previously  highlighted, $\widehat h$ is selected to minimize the distance between the model and $\mathcal H _0$. This gives an ``optimistic" value for the sample statistic $S_n$. However, also in the ``bootstrap world" the bandwidth $h$ is chosen to minimize the distance between $\mathcal H _0$ and the semiparametric model. Thus, also in the bootstrap world the statistic $S_n^*$ is computed with an optimistic view. Now, the optimistic view of $S_n$ will reflect the reality only when the null is true. Differently, the optimistic view of  $S_n^*$ will {\itshape always} reflect the reality in the bootstrap world: since we are imposing the null $\mathcal H _0$ in the resampling scheme, the null is always satisfied in the bootstrap world. This implies that when $\mathcal H _0$ holds, the behavior of $S_n$ will reflect the behavior of $S_n^*$, allowing to control the size of the test. However, when the null hypothesis $\mathcal{H}_0$ does not hold, $S_n$ will be too large as compared to the distribution of $S_n^*$, and the test will reject $\mathcal H _0$.

\subsection{Small Sample Behavior}\label{sec: small sample behavior}
In this section we study the small sample performances of our test in a Monte Carlo experiment. 
We specify a Data Generating Process (DGP) in line with the binary choice model with a control function described in Section \ref{sec: Main Applications}, and in particular in Equations (\ref{eq: binary choice variable}) and (\ref{eq: CF assumption for single index}). So, 
\begin{align}\label{eq: DGP I}
    Y =& \mathbb{I}\{D + Z^{in}\geq u\}\,\quad \, \quad u=\epsilon+ V + \frac{p}{4}(V^2 - 1)\nonumber \\
    D =& H(Z^{in},Z^{ex}) + V\,\quad \, \quad  H(Z^{in},Z^{ex})=(1,Z^{in} , Z^{ex})\alpha \,.
\end{align}
 $(Z^{in}, Z^{ex} , V,\epsilon)$ are mutually independent. The error terms $(u,V,\epsilon)$ are not observed. $u$ is correlated with the endogenous regressor $D$ through $V$, so $V$ plays the role of a control function. The regressors $Z^{in}$ and $Z^{ex}$ are exogenous: the former is included in the equation for $Y$, while the latter is an excluded variable playing the role of an instruments. The control variable $V$ is unobserved but can be estimated by estimating $H$. The functional form of $H$ is unknown to the researcher, so $H$ must be estimated nonparametrically. \\
 We set $V\sim \mathcal{N}(0,1)$, $\alpha=(1,1,1)/\sqrt{2}$, $Z^{in}\sim \mathcal{N}(0,1)$,  and we draw $Z^{ex}$ from an exponential distribution truncated from above at 3 and standardized to have zero mean and unit variance. The parameter $p$ controls the departure from the null hypothesis, see below for details. \\
To check the robustness of our test  with respect to different DGPs, we consider three different specifications for the distribution of $\epsilon$:
\begin{enumerate}[leftmargin=4em]
    \item[DGP 1)] $\epsilon\sim \mathcal{N}(0,sd=\sqrt{7})$
    \item[DGP 2)] $\epsilon\sim \sqrt{7/(2\cdot 5)}\, (\chi_{5}^2 - 5)$
    \item[DGP 3)] $\epsilon\sim 0.8\cdot \mathcal{N}(-2.5 , sd=\sqrt{3.5})+0.2\cdot \mathcal{N}(2.5,sd=1)$ .
\end{enumerate}
 DGP 1 delivers a rescaled probit model with a distribution of $\epsilon$ that is unimodal and symmetric about zero. DGP 2 gives a unimodal distribution of $\epsilon$ with positive asymmetry and left skewness. Finally,  DGP 3 generates $\epsilon$ according to a mixture between two Gaussians, delivering a distribution of $\epsilon$  that is bimodal and left-skewed. The three distributions for $\epsilon$ and (\ref{eq: DGP I}) are built so as to guarantee  realistic features of the simulated data. Across  these DGPs it holds that $\Var (u)\approx 8$ , Corr$(u,V)\approx 0.35$ , Corr$(u,D)\approx 0.25$ , and $\Var (D + Z^{in})\approx 4.5$. \\
 We assume that the researcher specifies the model as 
 \begin{equation*}
     Y = \mathbb {I}\{D + \theta_0 Z^{in} +  \geq u\}\,,\,\text{ with }u=-\gamma_0 V + \epsilon\, , 
 \end{equation*}
 $V=D-\mathbb{E}\{D|Z\}$, and $\epsilon$ independent from all the other variables.
 The distribution of $(\epsilon,V)$ is nonparametric. Similarly to Section \ref{sec: Main Applications}, to check the correct specification of this model the researcher needs to test 
 \begin{equation}\label{eq: H0 in simulations}
     \mathcal{H}_0 : \mathbb{E}\{Y|X,D - H(Z)\}=\mathbb{E}\{Y|D + \theta_0 Z^{in} + \gamma_0 (D - H(Z))\}\, \text{ versus }\mathcal{H}_1:\mathcal{H}_0^c
 \end{equation}
 where $X=(D,Z^{in})$. Given the DGP in (\ref{eq: DGP I}), the null hypothesis holds if $p=0$. As long as  $p\neq 0 $ we are under the alternative $\mathcal{H}_1$. The magnitude of $p$ measures the departure from the null hypothesis. \\
 
 In the reminder of this section our goal is threefold. First, we want to analyze the practical advantages of the bias corrections in our statistic guaranteeing the small-bias properties of our test. Second, we want to check the capacity of our test to correctly control the size under $\mathcal H _0$. Third, we want to analyze its power properties. \\
 Let us start from the first task. To check the advantages of using bias corrections and having a small bias property of our statistic, we will compare the performances of two tests under different bandwidth choices. The first test is the one proposed in this paper that uses the {\itshape bias corrected} estimators $\widetilde H$ and $\widetilde G _{\widetilde W (\widehat \beta)}$. We will call it {\itshape BC Test}. The second test employs the procedure described in this paper, with the only exception that it does {\itshape not} use the bias corrected estimators, but it employs the {\itshape uncorrected} estimators $\widehat H$ and $\widehat G _{\widehat W (\widehat \beta)}$. We will call this test {\itshape UN Test}. We  compare the ability of these two tests to control the size under different bandwidth choices. So, in the third step estimation we will use the bandwidth rule $h=\,C\, \widehat \sigma(D + \widehat \theta Z^{in} + \widehat \gamma (D-\widetilde H (Z))) n^{-1/6}$, where $\widehat \sigma$ is the estimated standard deviation of its argument and $C$ is a constant. This allows us to study the stability of each test under different bandwidth choices by letting the constant $C$ vary.\footnote{We stress  that for this comparison the bandwidth $h$ in the third step is not computed as in (\ref{eq: bandwidth choice}), but according to the rule $h=C\, \widehat \sigma(D + \widehat \theta Z^{in} + \widehat \gamma (D-\widetilde H (Z)))\, n^{-1/6}$.} We have ran 10000 simulations under the null $\mathcal H _0$, i.e. by setting $p=0$, for $n=400$. Since the procedure is intense from a computational point of view, see Section \ref{sec: Implementation of the test}, we used the Warp-Speed method proposed by \citet{davidson_improving_2007} and studied in \citet{giacomini_warp-speed_2013}. This consists in drawing one bootstrap sample for each Monte Carlo simulation, and then use the entire set of bootstrapped statistics to compute the bootstrap p-values associated to each original statistic. We let the constant $C$ vary in the range of values $\{0.5,1,1.5,2,2.5\}$. For each of these values and for each simulated sample, the {\itshape BC Test} and the {\itshape UN Test} give two respective p-values and hence two respective decisions about the rejection of the null $\mathcal H _0$. The results are reported in Figures \ref{Fig 1} and \ref{Fig 2} for DGP 1. For a reason of space,  we omit the results for DGP 2 and DGP 3 since they are qualitatively similar to DGP 1. For each value of $C$, Figures \ref{Fig 1} and \ref{Fig 2}  contain the Error in Rejection Probability (ERP) of each test: this is the difference between the empirical rejection frequency and the nominal size of the test under the null $\mathcal H _0$. An ideal test would display a null ERP for each nominal size and for each value of $C$. This comparison allows us to evaluate (i) the ability of each test in controlling the size under different bandwidth choices, and (ii) the importance of employing bias corrections for estimating the null distribution of the statistic by wild bootstrapping. In other words, this gives us a measure of the stability/robustness of each bootstrap test with respect to the bandwidth. For $C=0.5$ up to $C=1$ both the {\itshape BC Test} and the {\itshape UN Test} display a good empirical size control for any nominal size. For $C=1.5$ and $C=2$ the size control starts deteriorating for the {\itshape UN Test} but remains quite stable for the {\itshape BC Test} at any nominal size. This means that the bootstrap does {\itshape not} manage to provide a good estimation of the null distribution of the statistic when bias corrections are {\itshape not} employed. This feature becomes even more clear for $C=2.5$: in this case the {\itshape UN Test} has a large ERP, while the {\itshape BC Test} shows a good size control for nominal levels up to 10 percent. These results show that employing bias corrections  enables  the bootstrap test to be stable with respect to the bandwidth choice. Also, these bias corrections allow for a satisfactory approximation of the null distribution of the statistic by our wild bootstrap procedure. Thus, such corrections do not only have the attractive theoretical feature of guaranteeing a small bias property of the test statistic, but they also have remarkable practical advantages.\\
 Let us now turn to our second and third task. Here, we analyze the performance of our test in terms of size control and power properties when the third step bandwidth $h$ is selected as in  (\ref{eq: bandwidth choice}). We label such a test as {\itshape BC Test ($\widehat h$)}. We compare the performance of this test with the {\itshape BC Test}  described previously in this subsection and applied with $C=1$. This bandwidth rule corresponds to the Silverman's Rule of Thumb. We label this test as {\itshape BC Test($C=1$)}. We ran 10000 Monte Carlo replications for $n=400$ and $n=800$. Also in this  case the procedure is intense from a computational point of view: we need to run a nonlinear optimization for obtaining $\widehat \beta$ and a nonlinear optimization for obtaining $\widehat h$, for both the sample and the bootstrap data. So, to speed up computations we use the warp speed method previously described. We compute the empirical rejection frequencies for the {\itshape BC Test ($\widehat h$)} and the {\itshape BC Test($C=1$)} under $\mathcal H _0$, i.e. by setting $p=0$. This allows us to evaluate the performance of the test in terms of size control for each sample size. To evaluate the performance  of the test in terms of power, we use $p=3$ and $p=4$. Indeed, the larger $p$ the more we depart from the null hypothesis. The results are reported in Tables \ref{table: DGP1}, \ref{table: DGP2}, and \ref{table: DGP3} for the nominal sizes of 5\% and 10\%. We notice that  the {\itshape BC Test ($\widehat h$)} shows a good size control under $\mathcal H _0$ across the different DGPs considered, also for the limited sample size of $n=400$. The test also displays a satisfactory power across the different DGPs. For a given sample size, the power of the test increases as we depart from the null hypothesis, i.e. when $p$ increases. Similarly, for a given level of $p$, the power of the test increases as the sample size grows. This simulation experiment shows that the bootstrap {\itshape BC Test ($\widehat h$)} based on the empirical bandwidth selection in (\ref{eq: bandwidth choice}) manages to control well the size under $\mathcal H _0$ and has a satisfactory power.

\begin{table}[ht]\label{table: DGP1}
\begin{center}
\begin{tabular}{rrrrrrrl}
\hline
& & \multicolumn{2}{c}{$\mathcal{H}_0$} & \multicolumn{2}{c}{$\mathcal{H}_1(p=3)$} & \multicolumn{2}{c}{$\mathcal{H}_1(p=4)$} \\
\hline
    &  & \textbf{0.05} & \textbf{0.10} & \textbf{0.05} & \textbf{0.10} & \textbf{0.05} & \textbf{0.10} \\ 
    \hline
    \\
   $n$=400 & {\itshape BC Test} ($\widehat h$)\hspace{0.85cm} & 0.0518 & 0.0884 & 0.3822 & 0.7886 & 0.7502 & 0.9558 \\ 
   
    & {\itshape BC Test} ($C=1$) & 0.0497 & 0.0868 & 0.3226 & 0.7874 & 0.7208 & 0.9592 \\ 
   \\
   $n$=800 & {\itshape BC Test} ($\widehat h$)\hspace{0.85cm} & 0.0407 & 0.0884 & 0.9833 & 0.9956 & 0.9996 & 1 \\ 
   
    & {\itshape BC Test} ($C=1$) & 0.0416 & 0.0883 & 0.9838 & 0.9944 & 0.9994 & 1 \\ 
   \hline
\end{tabular}
\begin{minipage}{\textwidth}
\linespread{1.1}\selectfont
\footnotesize
\vspace{0.2cm}
Notes: Empirical rejection frequencies of the tests for DGP 1 under $\mathcal{H}_0$ (corresponding to $p=0$ in (\ref{eq: DGP I})) and under the two alternatives (corresponding to $p=3$ and $p=4$ in (\ref{eq: DGP I})). {\itshape BC Test ($\widehat h$)} is the test with bias corrections implemented with the bandwidth from Equation (\ref{eq: bandwidth choice}), while {\itshape BC Test($C=1$)} is the test with bias corrections implemented with the Silverman's rule of thumb bandwidth. The nominal sizes of the tests are in bold. The number of Monte Carlo replication is 10.000.
\vspace{0.2cm}
\end{minipage}
\caption{Empirical rejection frequencies of the tests for DGP 1}
\end{center}
\end{table}

\begin{table}[ht]\label{table: DGP2}
\begin{center}
\begin{tabular}{rrrrrrrrr}
\hline
& & \multicolumn{2}{c}{$\mathcal{H}_0$} & \multicolumn{2}{c}{$\mathcal{H}_1(p=3)$} & \multicolumn{2}{c}{$\mathcal{H}_1(p=4)$} \\
\hline
    &  & \textbf{0.05} & \textbf{0.10} & \textbf{0.05} & \textbf{0.10} & \textbf{0.05} & \textbf{0.10} \\ 
    \hline
    \\
   $n$=400 & {\itshape BC Test} ($\widehat h$)\hspace{0.85cm} & 0.0447 & 0.0893 & 0.5352 & 0.886 & 0.8512 & 0.9814 \\ 
    & {\itshape BC Test} ($C=1$) & 0.043 & 0.0829 & 0.4708 & 0.8894 & 0.8268 & 0.9852 \\ 
   \\
   $n$=800 & {\itshape BC Test} ($\widehat h$)\hspace{0.85cm} & 0.038 & 0.0881 & 0.9956 & 0.9992 & 1 & 1 \\ 
    & {\itshape BC Test} ($C=1$) & 0.0405 & 0.0868 & 0.9974 & 0.999 & 1 & 1 \\ 
   \hline
\end{tabular}
\begin{minipage}{\textwidth}
\linespread{1.1}\selectfont
\footnotesize
\vspace{0.2cm}
Notes: Empirical rejection frequencies of the tests for DGP 2 under $\mathcal{H}_0$ (corresponding to $p=0$ in (\ref{eq: DGP I})) and under the two alternatives (corresponding to $p=3$ and $p=4$ in (\ref{eq: DGP I})). {\itshape BC Test ($\widehat h$)} is the test with bias corrections implemented with the bandwidth from Equation (\ref{eq: bandwidth choice}), while {\itshape BC Test($C=1$)} is the test with bias corrections implemented with the Silverman's rule of thumb bandwidth. The nominal sizes of the tests are in bold. The number of Monte Carlo replication is 10.000.
\vspace{0.2cm}
\end{minipage}
\caption{Empirical rejection frequencies of the tests for DGP 2}
\end{center}
\end{table}

\begin{table}[ht]\label{table: DGP3}
\begin{center}
\begin{tabular}{rrrrrrrrr}
\hline
& & \multicolumn{2}{c}{$\mathcal{H}_0$} & \multicolumn{2}{c}{$\mathcal{H}_1(p=3)$} & \multicolumn{2}{c}{$\mathcal{H}_1(p=4)$} \\
\hline
    &  & \textbf{0.05} & \textbf{0.10} & \textbf{0.05} & \textbf{0.10} & \textbf{0.05} & \textbf{0.10} \\ 
    \hline
    \\
    $n$=400 & {\itshape BC Test} ($\widehat h$)\hspace{0.85cm} & 0.0493 & 0.0858 & 0.3848 & 0.7804 & 0.7034 & 0.9534 \\ 
     & {\itshape BC Test} ($C=1$) & 0.0511 & 0.0819 & 0.2718 & 0.7612 & 0.5914 & 0.9418 \\ 
    \\
   $n$=800 & {\itshape BC Test} ($\widehat h$)\hspace{0.85cm} & 0.0366 & 0.0905 & 0.9798 & 0.996 & 0.9998 & 1 \\ 
    & {\itshape BC Test} ($C=1$) & 0.0372 & 0.0896 & 0.9856 & 0.9968 & 1 & 1 \\ 
   \hline
\end{tabular}
\begin{minipage}{\textwidth}
\linespread{1.1}\selectfont
\footnotesize
\vspace{0.2cm}
Notes: Empirical rejection frequencies of the tests for DGP 3 under $\mathcal{H}_0$ (corresponding to $p=0$ in (\ref{eq: DGP I})) and under the two alternatives (corresponding to $p=3$ and $p=4$ in (\ref{eq: DGP I})). {\itshape BC Test ($\widehat h$)} is the test with bias corrections implemented with the bandwidth from Equation (\ref{eq: bandwidth choice}), while {\itshape BC Test($C=1$)} is the test with bias corrections implemented with the Silverman's rule of thumb bandwidth. The nominal sizes of the tests are in bold. The number of Monte Carlo replication is 10.000.
\vspace{0.1cm}
\end{minipage}
\caption{Empirical rejection frequencies of the tests for DGP 3}
\end{center}
\end{table}

\begin{table}[ht]\label{table: empirical application}
\begin{center}
\begin{tabular}{rr}
  \hline
 $S_n$ & $\widehat c _{0.99}$  \\ 
  \hline
 &  \\ 
  1229 & 110 \\ 
   \hline
\end{tabular}
\begin{minipage}{\textwidth}
\linespread{1.1}\selectfont
\footnotesize
\vspace{0.2cm}
 Notes: The test statistic is $S_n$ and its 99 percent bootstrap quantile is $\widehat c _{0.99}$ for the real-data application. The test performed is the {\itshape BC test}($\widehat h$), where the bandwidths is computed according to Equation (\ref{eq: bandwidth choice}). The sample size is 2947 and the number of bootstrap replications is 999.
 \end{minipage}
\caption{Test statistic and its 99 percent bootstrap quantile for the data application}
\end{center}
\end{table}

\begin{figure}[ht]
\begin{center}
\includegraphics[width=7cm,height=7.5cm]{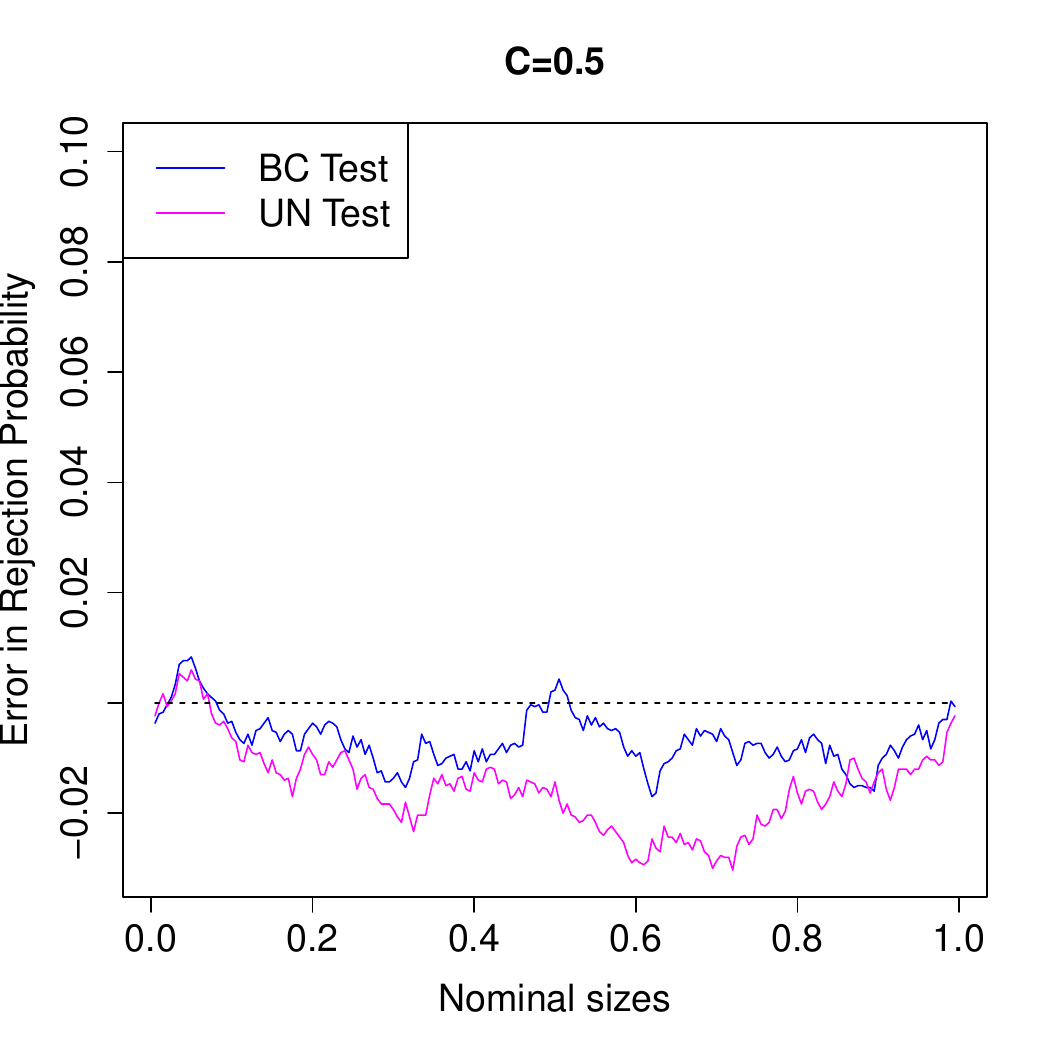}
\includegraphics[width=7cm,height=7.5cm]{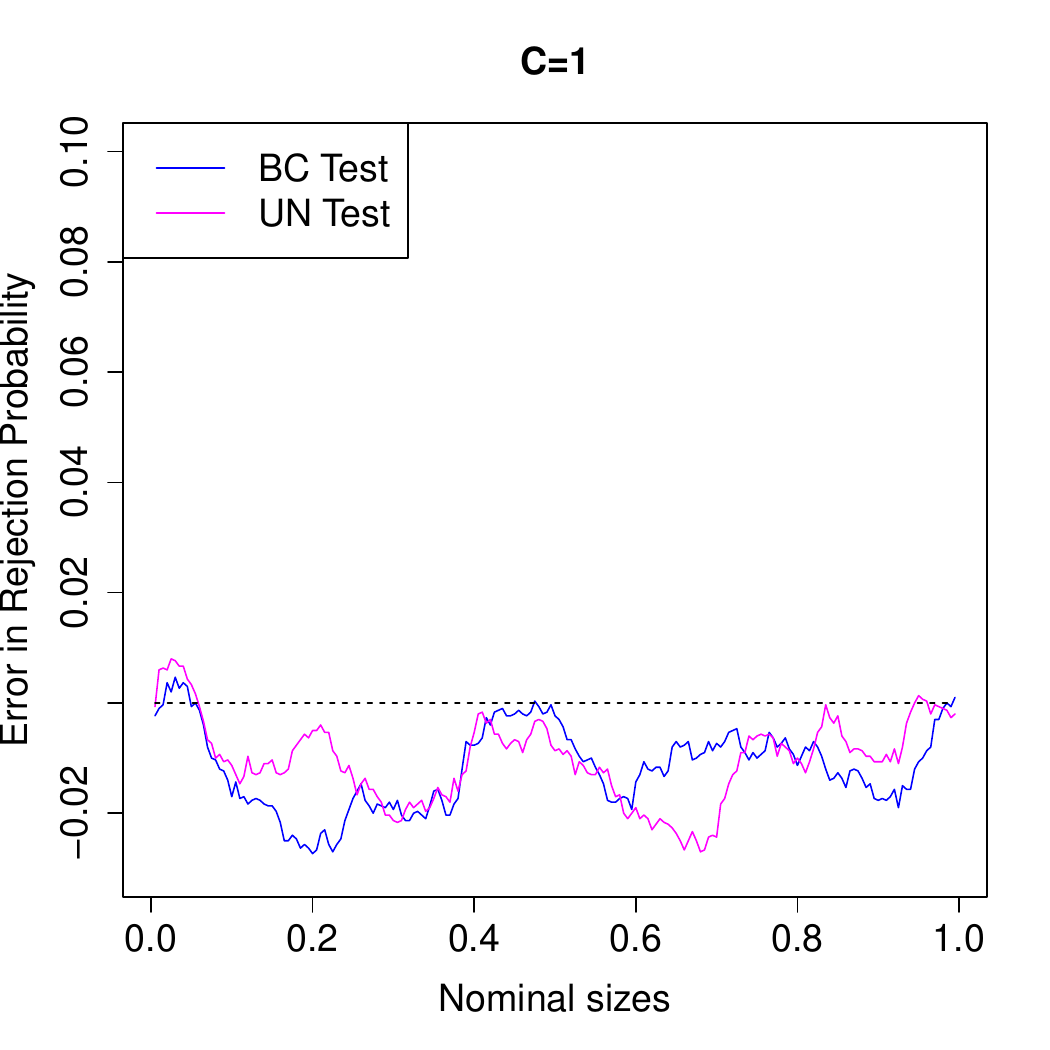}
\includegraphics[width=7cm,height=7.5cm]{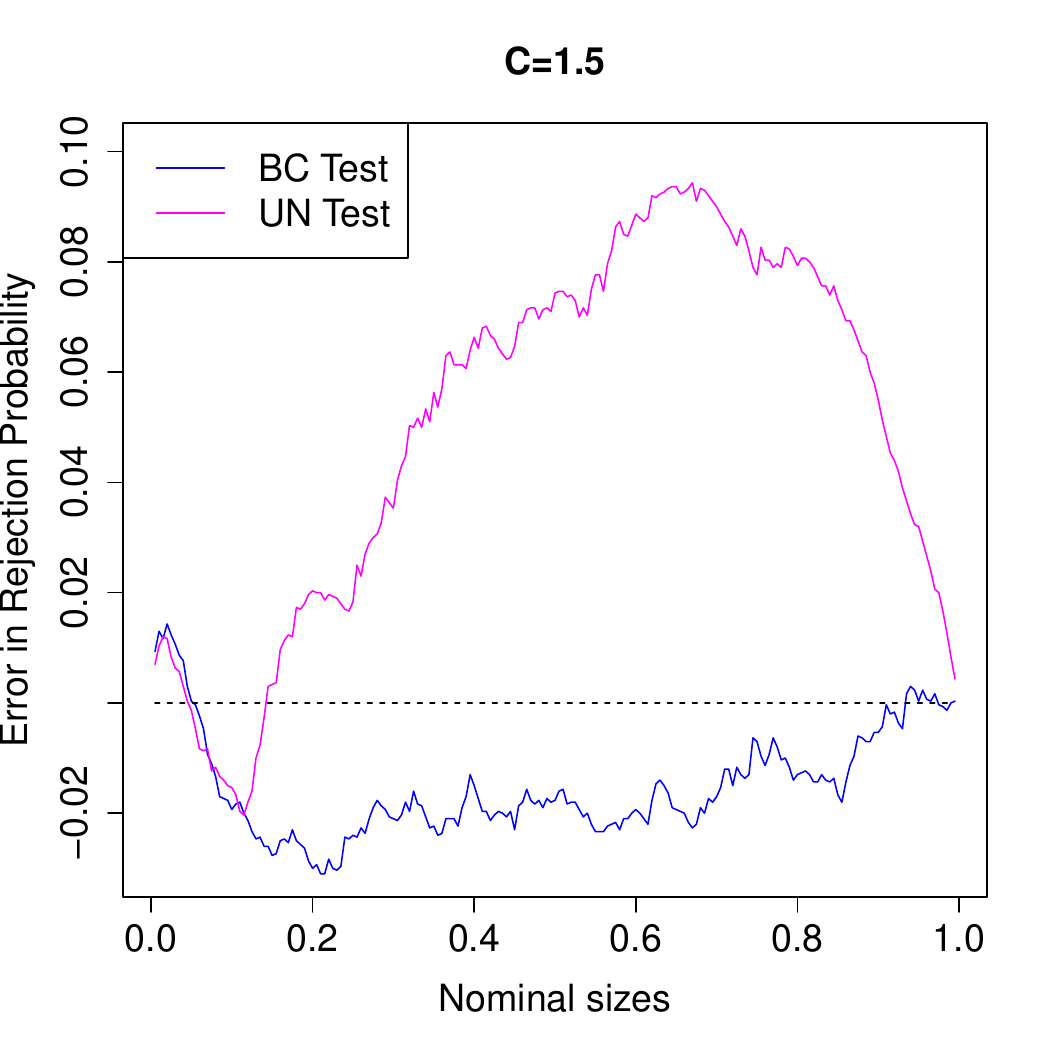}
\includegraphics[width=7cm,height=7.5cm]{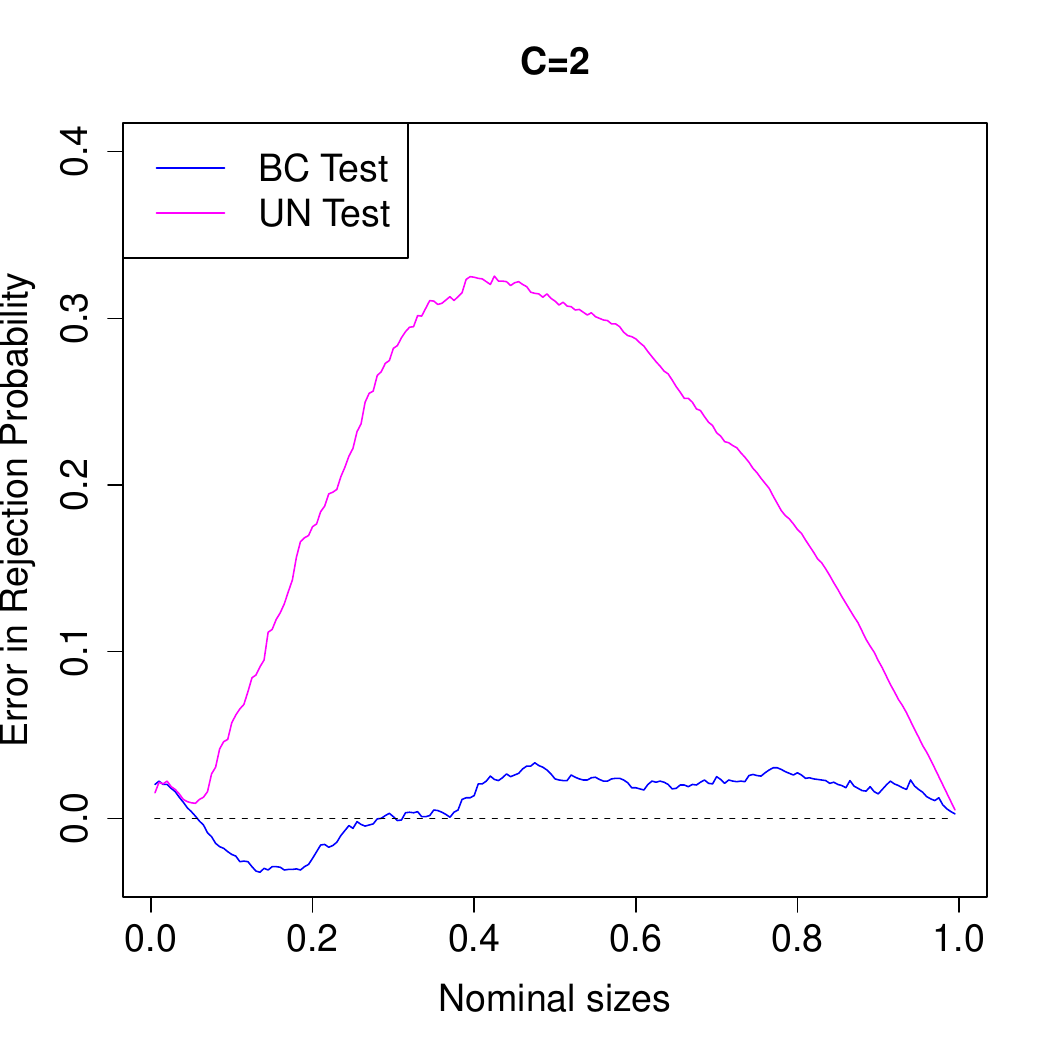}
\begin{minipage}{\textwidth}
\footnotesize
Notes: Error in Rejection Probability (=Empirical rejection frequency of the test - nominal size of the test) as a function of the nominal size for the Bias corrected test (BC Test) and the Uncorrected test (UN Test), each applied with the bandwidth $C\cdot h_S$, where $h_S$ is the Silverman's rule of thumb bandwidth and $C\in\{0.5,1,1.5, 2\}$. The sample size is $n=400$.
\end{minipage}
\caption{\small Errors in Rejection Probabilities for the Bias corrected and the Uncorrected test.}
\end{center}
\label{Fig 1}
\end{figure}

\begin{figure}[ht]
\begin{center}
\includegraphics[width=7cm,height=7.5cm]{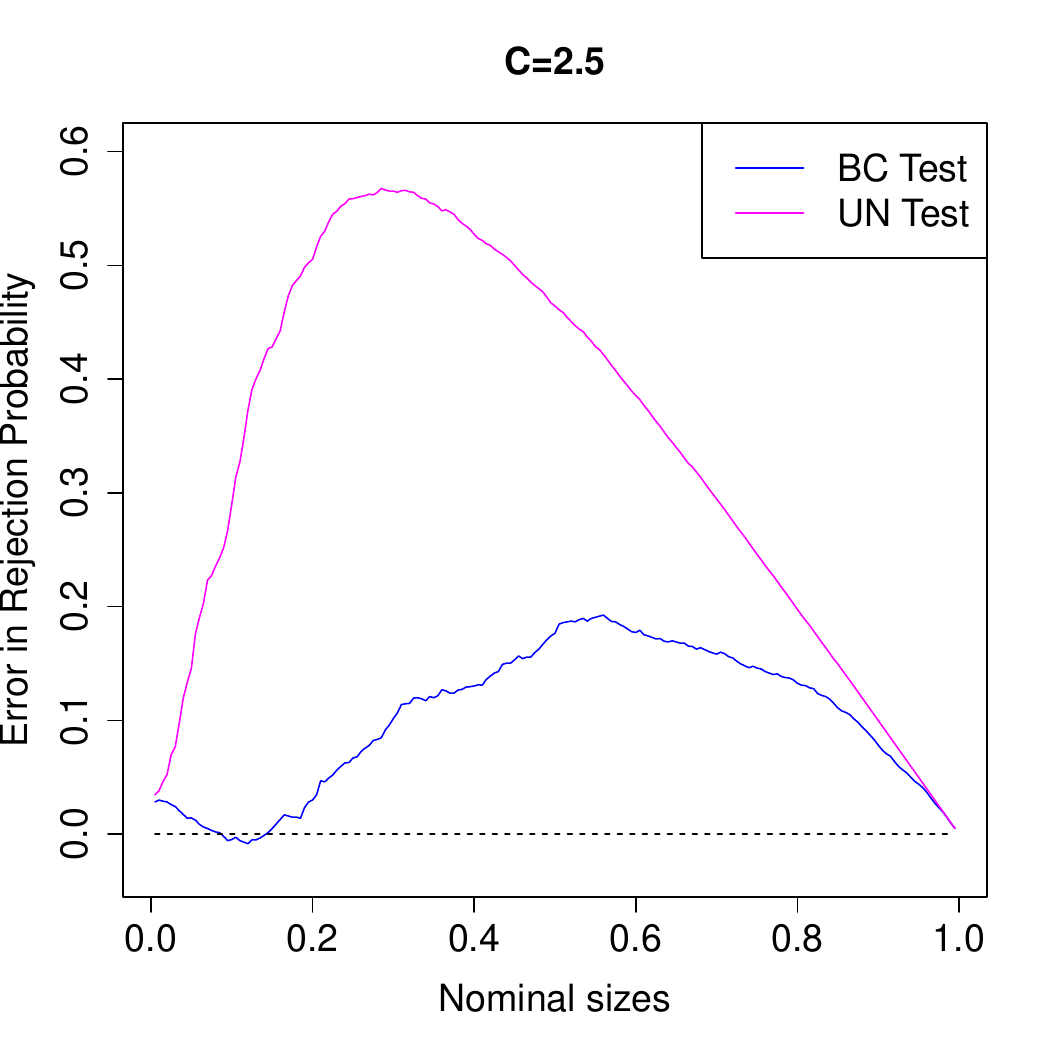}
\begin{minipage}{\textwidth}
\footnotesize
Notes: Error in Rejection Probability (=Empirical rejection frequency of the test - nominal size of the test) as a function of the nominal size for the Bias corrected test (BC Test) and the Uncorrected test (UN Test), each applied with the bandwidth $C\cdot h_S$, where $h_S$ is the Silverman's rule of thumb bandwidth and $C=2.5$. The sample size is $n=400$. 
\end{minipage}
\caption{\small Errors in Rejection Probabilities for the Bias corrected and the Uncorrected test.}
\end{center}
\label{Fig 2}
\end{figure}

\subsection{An Application to Real Data}
We illustrate the practical implementation of our procedure in a real-data example. We test the specification of a semiparametric model describing women's labor force participation, see \citet{wooldridge_control_2015}. Such a model studies the impact of non-labor income on women's labor force participation decisions. The data we use are taken from the 1991 wave of the Current Population Survey.\footnote{I am very grateful to Jeffrey Wooldridge for having shared his data set.} The sample consists of 2,947 women who do not have kids under the age of 6 and have a positive non labor income. We set up the model and the notation similarly as in 
Section \ref{sec: small sample behavior}. The dependent variable $Y$ is an indicator that equals one if the woman participates to the labor force and 0 otherwise. $D$ represents the household's ``other sources of income". The vector of controls $Z^{in}$ includes  the logarithm of the woman's experience and a dummy indicating college or above education for the woman. We treat $D$ (other sources of income) as endogenous. Following \citet{wooldridge_control_2015}, we set $Z^{ex}$ to the husband's level of education. This is a dummy indicating college or above education, and it is used as an instrument to control for the endogeneity of $D$. We normalize to unity the coefficient attached to $D$. \\
We apply our test to check the correct specification of a semiparametric index model similar to the one described in the previous section. Thus, we test the hypothesis in (\ref{eq: H0 in simulations}). To implement the test in this real data example, we apply the steps described in Section \ref{sec: Implementation of the test}. We consider the {\itshape BC Test}($\widehat h$) where the bandwidth is computed as in Equation (\ref{eq: bandwidth choice}). Notice that the nonlinear optimization for the bandwidth has to be done with the sample data and for each bootstrap replication. The number of bootstrap iterations is set to $B=999$. In Table \ref{table: empirical application}, we report the value of the test statistic and the 99\% bootstrap quantile. This is the quantile of the bootstrap distribution of $S_n$. Since the value of the statistic is well above the 99\% bootstrapped quantile, the correct specification of the model must be rejected at the 1\% nominal level. 

\vspace{5cm}
\section*{Acknowledgements}
This is a revised version of a chapter of my PhD thesis. I thank my supervisor Pascal Lavergne and the components of my PhD committee Jean-Pierre Florens, Ingrid Van Keilegom, and Juan Carlos Escanciano for their comments and suggestions. I also thank Xavier d'Haultfoeuille and Jad Beyhum for their comments. Finally, I thank three anonymous referees 
for their comments. 
This research has received financial support from the European Research Council under the European Community's Seventh Program FP7/2007-2013 grant agreement N. 295298.
\section*{Appendix}
\appendix 
This Appendix is divided into two parts. In Appendix \ref{sec:Asymptotic Analysis}, we prove the asymptotic expansion of Proposition \ref{prop: Asymptotic Test}. In Appendix \ref{sec: Bootstrap Analysis}, we prove Proposition \ref{prop: bootstrap test} and hence the validity of our wild bootstrap test. Auxiliary lemmas are gathered in Appendix \ref{sec:Auxiliary-Lemmas} contained in a supplementary material.
\section{Asymptotic Analysis}
\label{sec:Asymptotic Analysis}
We use the notation $\|g\|_{\infty,\mathcal{A}}=\sup_{a\in\mathcal{A}}|g(a)|$ for any function $g$ defined on a set $\cal A$. When the support of the argument of $g$ is clear from the context, we will simply denote $\sup_{a\in\mathcal{A}}|g(a)|$ with $\|g\|_\infty$. We also use the convention that $\widehat G _{\widetilde W (\widehat \beta)}(w)=T^Y_{\widetilde W (\widehat \beta)}(w)/\widehat f_{\widetilde W (\widehat \beta)}(w)$ is set to 0 whenever $\widehat f_{\widetilde W (\widehat \beta)}(w)=0$. The same will hold for all the other quantities involving a random denominator. We use the acronym ``wpa1" to mean ``with probability approaching one", the acronym ``RHS" to mean ``right hand side", and the acronym ``LHS" to mean ``left hand side".

\subsection{Proof of Proposition \ref{prop: Asymptotic Test}}
By Lemma \ref{lem: trimming implications}(ii), 
wpa1 $\widehat t (X_i,Z_i)\leq t^\delta(X_i,Z_i)$ uniformly in $i=1,\ldots,n$, where $\delta=1/2$. Assumption \ref{Assumption: trimming}(i) ensures that there exists $\eta(\delta)>0$ such that wpa1 $t^\delta(X_i,Z_i)\leq $ $\mathbb I \{f_{W(\widehat \beta)}(q(\widehat \beta,X_i,$ $H(Z_i))\geq \eta(\delta)\tau_n\}$ for all $i=1,\ldots,n$. Also, by Lemma \ref{lem: trimming implications}(iv) wpa1 $t^\delta(X_i,Z_i)$ $\mathbb I \{f_{W(\widehat \beta)}(q(\widehat \beta,$ $X_i,H(Z_i))\geq \eta(\delta)\tau_n\}$ $\leq t^\delta (X_i,Z_i)$ $\mathbb I \{f_{W(\widehat \beta)}(q(\widehat \beta,X_i,\widetilde H (Z_i)))\geq \eta(\delta) \tau_n/2\}$ for all $i=1,\ldots,n$, and 
Lemma \ref{lem: trimming implications}(iii) gives that wpa1 $\mathbb I \{f_{W(\widehat \beta)}(q(\widehat \beta,$ $X_i,\widetilde H (Z_i)))\geq \eta(\delta) \tau_n/2\}$ $\leq \mathbb I \{\widehat f_{\widetilde W(\widehat \beta)}(q(\widehat \beta,$ $X_i,\widetilde H (Z_i)))\geq \eta(\delta) \tau_n/4\}$ for all $i=1,\ldots,n$. Gathering results, wpa1 
\begin{equation}\label{eq: trimming inequality for simplification}
    \widehat t (X_i,Z_i)\leq \mathbb I \{\widehat f _{\widetilde W (\widehat \beta)}(\widetilde W _i (\widehat \beta))\geq \eta(\delta)\tau_n /4\}
    \end{equation}
uniformly in $i=1,\ldots,n$. So, by using the definition of $\widetilde{G}_{\widetilde W (\widehat \beta)}$ we have 
\begin{align}\label{eq: 1st decomposition of Asy Emp Process}
    \sqrt{n} \mathbb P _n [Y - \widetilde{G}_{\widetilde W (\widehat \beta)}(\widetilde W (\widehat \beta))]\, \widehat{f}_{\widetilde W (\widehat \beta)}(\widetilde W (\widehat \beta))\, & \varphi_s (X,\widetilde H (Z))\, \widehat t \nonumber  \\
    = \sqrt{n} \mathbb P _n & \widehat{\varepsilon}  \, \widehat{f}_{\widetilde W (\widehat \beta)}(\widetilde W (\widehat \beta))\, \varphi_s (X , \widetilde H (Z))\, \widehat{t} \nonumber \\
     -  \sqrt{n} \mathbb P _n & \widehat{T}^{\widehat{\varepsilon}}_{\widetilde W (\widehat \beta)}(\widetilde W (\widehat \beta))\, \varphi_s(X, \widetilde{H}(Z))\, \widehat{t}\, 
\end{align}
wpa1. Next, we introduce
\begin{align*}
\widehat{T}^{\varphi_s}_{\widetilde W (\widehat \beta)}(w):=&\frac{1}{n h^d}\sum_{i=1}^n \varphi_s(X_i,\widetilde{H}(Z_i))\, \widehat t _i\, K\left(\frac{w-\widetilde{W}_i (\widehat \beta)}{h}\right)\, ,\\
    \widehat{\iota}_s(w):=\widehat{T}^{\varphi_s}_{\widetilde W (\widehat \beta)}(w)/\widehat{f}_{\widetilde W (\widehat \beta)}(w)\, , &\text{ and }\iota_s(w):=\mathbb E \{\varphi_s(X,H(Z))|W(\beta_0)=w\}\, .
\end{align*}
In view of (\ref{eq: trimming inequality for simplification}), we can use the definition of $\widehat{T}^{\widehat{\varepsilon}}_{\widetilde W (\widehat \beta)}$ and exchange sums to get that the second term on the RHS of (\ref{eq: 1st decomposition of Asy Emp Process}) equals $\sqrt{n}\mathbb P _n\, \widehat{\varepsilon}\,\, \widehat{\iota}_s(\widetilde W (\widehat \beta))\, \widehat{f}_{\widetilde W (\widehat \beta)}(\widetilde W (\widehat \beta))\,\widehat{t}$ wpa1. Thus, wpa1 the RHS of (\ref{eq: 1st decomposition of Asy Emp Process}) is
\begin{align}\label{eq: 2nd decomposition of Asy Emp Process}
    \sqrt{n} \mathbb P _n \widehat{\varepsilon}\, \widehat{f}&_{\widetilde W (\widehat \beta)}(\widetilde W (\widehat \beta))\,[\varphi_s(X,\widetilde{H}(Z))-  \widehat{\iota}_s(\widetilde W (\widehat \beta))]\, \widehat{t} \nonumber \\
    = & \sqrt{n} \mathbb P _n \varepsilon \widehat{f}_{\widetilde W (\widehat \beta)}(\widetilde W (\widehat \beta))  \, [\varphi_s(X,\widetilde{H}(Z))-  \widehat{\iota}_s(\widetilde W (\widehat \beta))] \widehat t \\
    & + \sqrt{n} \mathbb{P}_n [G_{W(\beta_0)}(W(\beta_0))  -\widehat{G}_{\widetilde W (\widehat \beta)}(\widetilde W ( \widehat \beta)) ] \widehat{f}_{\widetilde W (\widehat \beta)}(\widetilde W (\widehat \beta)) [\varphi_s(X,\widetilde{H}(Z))-\widehat{\iota}_s(\widetilde W (\widehat \beta))] \widehat{t}\, . \nonumber 
    \end{align}
    To prove the desired result, we will obtain the IFR of the RHS of the above display. We will handle each term separately. Let us start from the first one. By Lemma \ref{lem: trimming and convergence rates for 1st step}(iii) $|\widetilde H(Z_i)|(\widehat t _i + t_i)$ is bounded in probability uniformly in $i=1,\ldots,n$. Thus, $|\varphi_s(X_i,\widetilde{H}(Z_i))|(\widehat{t}_i+t_i)$ is also bounded in probability, uniformly in $s\in\mathcal{S}$ and $i=1,\ldots,n$. Also, Lemma \ref{lem: IFR for betahat} and Lemma \ref{lem: convergence rates for Ghat fhat and iotahat}(xiii)(xv) ensure that $|\widehat f_{\widetilde W (\widehat \beta)}(\widetilde W _i (\widehat \beta))|(\widehat t_i + t_i)$ and $|\widehat \iota_s(\widetilde W _i (\widehat \beta))|(\widehat t_i + t_i)$ are  bounded in probability, uniformly in $i=1,\ldots,n$ and $s$. So, since $\widehat{t}-t=(\widehat{t}+t)(\widehat{t}-t)$ we get 
\begin{align}\label{eq: trimming replacement}
   \sup_{s\in \cal S}  \left|\sqrt{n} \mathbb P _n \varepsilon \widehat{f}_{\widetilde W (\widehat \beta)}(\widetilde W (\widehat \beta))  \, [\varphi_s(X,\widetilde{H}(Z))-  \widehat{\iota}_s(\widetilde W (\widehat \beta))] (\widehat t-t)\right|\leq O_P(\sqrt{n} \mathbb{P}_n |\widehat{t}-t|)=o_P(1)\, ,
\end{align}
where in the last equality we have used $\sqrt{n}\mathbb{P}_n |\widehat t - t|=o_P(1)$ from Lemma \ref{lem: trimming and convergence rates for 1st step}(i).
Accordingly, in the first term on the RHS of (\ref{eq: 2nd decomposition of Asy Emp Process}) we can replace the trimming $\widehat t$ with $t$ at the cost of an $o_P(1)$ reminder, and obtain that uniformly in $s$
\begin{align}\label{eq: 1st term RHS of the Asymptotic decomposition}
    \sqrt{n}\mathbb{P}_n\varepsilon \widehat f_{\widetilde W (\widehat \beta)}(\widetilde W (\widehat \beta))&[\varphi_s(X,\widetilde H (Z))-\widehat \iota_s(\widetilde W (\widehat \beta))]\widehat t \nonumber\\
    =&     \sqrt{n}\mathbb{P}_n\varepsilon \widehat f_{\widetilde W (\widehat \beta)}(\widetilde W (\widehat \beta))[\varphi_s(X,\widetilde H (Z))-\widehat \iota_s(\widetilde W (\widehat \beta))]\, t+o_P(1)\, . 
\end{align}

Now, Lemma \ref{lem: IFR for betahat} and Lemma \ref{lem: trimming and convergence rates for 1st step}(iii) ensure that $\|\widehat{\beta}-\beta_0\|=O_P(n^{-1/2})$ and $\max_{i=1,\ldots,n}$ $|\widetilde{H}(Z_i)$ $-H(Z_i)|t_i=o_P(n^{-1/4})$. So, 
by a Mean-Value expansion of $q(\widehat{\beta},X_i,\widetilde{H}(Z_i))$ around $(\beta_0,X_i,H(Z_i))$ and the Lipschitz continuity of $\partial_\beta q$ and $\partial_H q$ (see Assumption \ref{Assumption: smoothness}(v)), we get\footnote{For notational simplicity, we are using $\partial_H q(\beta,X,H(Z)):=\partial_u q(\beta,X,u)|_{u=H(Z)}$.}
\begin{align}\label{eq: mean value expoansion of q}
    [\,\widetilde{W}_i(\widehat \beta)-W_i(\beta_0)\,]\,t_i=&[q(\widehat{\beta},X_i, \widetilde{H}(Z_i))-q(\beta_0,X_i, H(Z_i))]t_i \nonumber\\
    =& \partial_{\beta^T} q(\beta_0,X_i, H(Z_i)) t_i (\widehat{ \beta} - \beta_0)\\
    &+ \partial_{H} q(\beta_0,X_i,H(Z_i))(\widetilde{H}(Z_i)-H(Z_i)) t_i + o_P(n^{-1/2}) \nonumber 
\end{align}
uniformly in $i=1,\ldots,n$. This implies that $\max_{i=1,\ldots,n}|\widetilde W_i (\widehat \beta)-W_i(\beta_0)| t_i$ $=O_P(\|(\widetilde H - H) t\|_\infty + \|\widehat \beta - \beta_0\|)$ $=o_P(n^{-1/4})$. Next, Lemma \ref{lem: convergence rates for Ghat fhat and iotahat}(vii)(ix) ensures that  both $|\widehat f _{\widetilde W (\widehat \beta)}(W_i(\beta_0))|t_i$ and $|\widehat \iota _s(W_i(\beta_0))|t_i$ are bounded in probability uniformly in $i=1,\ldots,n$ and $s$. Using these results and plugging the stochastic expansions from  Lemma \ref{lem: stochastic expansions}(ii)(iii)(iv) into (\ref{eq: 1st term RHS of the Asymptotic decomposition}) leads to\footnote{In (\ref{eq: approximation of 1st term of 2nd decomposition of Asy Emp Process}) for notational simplicity we use: $\partial_H \varphi_s(X,H(Z)):=\partial_u \varphi_s(X,u)|_{u=H(Z)}$ and $\partial \iota_s(W(\beta)):=\partial_w \iota_s(w)|_{w=W(\beta)}$. }
\begin{align}\label{eq: approximation of 1st term of 2nd decomposition of Asy Emp Process}
    \sqrt{n} \mathbb P _n \varepsilon\, \widehat{f}_{\widetilde W (\widehat \beta)}(\widetilde W (&\widehat \beta))  \, [\varphi_s(X,\widetilde{H}(Z))  -  \widehat{\iota}_s(\widetilde W (\widehat \beta))]\, t \nonumber \\
    =& \sqrt{n} \mathbb{P}_n \varepsilon\, \widehat{f}_{\widetilde W (\widehat \beta)}( W ( \beta_0))\,\varphi_s(X,H(Z)) \, t \nonumber \\
    & - \sqrt{n} \mathbb{P}_n \varepsilon\, \widehat{f}_{\widetilde W (\widehat \beta)}( W ( \beta_0))\,\widehat{\iota}_s( W ( \beta_0))] \, t \nonumber \\
    &+ \sqrt{n}\mathbb{P}_n \varepsilon\, \widehat{f}_{\widetilde W (\widehat \beta)}(W(\beta_0))\, \partial_{H} \varphi_s(X,H(Z))\, [\widetilde{H}(Z)-H(Z)] \, t \nonumber \\
    &- \sqrt{n} \mathbb{P}_n \varepsilon\, \widehat{f}_{\widetilde W (\widehat \beta)}(W(\beta_0))\,\partial^T \iota_s(W(\beta_0))[\widetilde W (\widehat \beta)-W(\beta_0)] \, t \nonumber \\
    & + \sqrt{n}\mathbb{P}_n \varepsilon\, \partial^T f_{W(\beta_0)}(W(\beta_0))\,[\widetilde W (\widehat \beta) - W(\beta_0)]\,[\varphi_s(X,H(Z))-\widehat{\iota}_s(W(\beta_0))] \, t \nonumber \\
    & +o_P(1)  \, 
\end{align}
uniformly in $s\in \mathcal{S}$. Now, let us deal with the first term on the RHS of (\ref{eq: approximation of 1st term of 2nd decomposition of Asy Emp Process}). For any $\eta>0$ we define 
\begin{equation}\label{eq: definition of t W beta_0 eta}
    t_{W(\beta_0)}^\eta(w)=\mathbb I \{f_{W(\beta_0)}(w)\geq \eta \tau_n\}\, .
\end{equation}
We have (see the comments below) 
\begin{align}\label{eq: first ASE for asy behavior}
\text{1}^{st}\text{ term RHS of (\ref{eq: approximation of 1st term of 2nd decomposition of Asy Emp Process}) } = & \sqrt{n} \mathbb{P}_n \varepsilon \widehat{f}_{\widetilde{W}(\widehat \beta)}(W(\beta_0))\varphi_s(X,H(Z)) t_{W(\beta_0)}^\eta(W(\beta_0)) \nonumber   \\
    &+ \sqrt{n} \mathbb{P}_n \varepsilon \widehat{f}_{\widetilde{W}(\widehat \beta)}(W(\beta_0))\varphi_s(X,H(Z)) t^\eta_{W(\beta_0)}(W(\beta_0)) (t-1)\, \nonumber \\
    =& \sqrt{n} \mathbb{P}_n \varepsilon \widehat{f}_{\widetilde{W}(\widehat \beta)}(W(\beta_0))\varphi_s(X,H(Z)) t_{W(\beta_0)}^\eta(W(\beta_0)) \nonumber \\
    &+ o_P(1)\, 
\end{align}
uniformly in $s\in \mathcal{S}$.
By Assumption \ref{Assumption: trimming}(i), there exists $\eta>0$ such that for each large $n$ we have $t(X_i,Z_i)\leq t(X_i,Z_i)\mathbb{I}\{f_{W(\beta_0)}(W_i(\beta_0))\geq \eta \tau_n \}$ for all $i=1\ldots,n$. This gives the first equality. Lemma \ref{lem: convergence rates for Ghat fhat and iotahat}(x) implies that $|\widehat{f}_{\widetilde{W}(\widehat \beta)}(W_i(\beta_0))|$ $t^\eta_{W(\beta_0)}(W_i(\beta_0))$ is bounded in probability uniformly in $i=1,\ldots,n$. So, the second term after the first equality is $O_P(\sqrt{n} \mathbb P _n |t-1|)$. Since $\sqrt{n} \mathbb P _n |t-1|=o_P(1)$ from Lemma \ref{lem: trimming and convergence rates for 1st step}(i), the second equality follows.\\
Notice that under $\mathcal{H}_0$ the leading term of (\ref{eq: first ASE for asy behavior}) is a centered empirical process. Lemma \ref{lem: belonging conditions}(iv) ensures that $\widehat{f}_{\widetilde{W}(\widehat \beta)} - f_{W(\beta_0)}\in \mathcal{G}_\lambda(\mathcal{W}_{n,\beta_0}^{\eta/2})$ with probability approaching one, where $\lambda=\lceil (d+1)/2 \rceil$. Also, by Lemma \ref{lem: convergence rates for Ghat fhat and iotahat}(x) $\|(\widehat{f}_{\widetilde{W}(\widehat \beta)}-f_{W(\beta_0)})t^\eta_{W(\beta_0)}\|_{\infty,\mathcal{W}}=o_P(1)$. So, the stochastic equicontinuity result in Lemma \ref{lem: ASE}(i) yields 
\begin{align*}
    \sqrt{n} \mathbb{P}_n \varepsilon \widehat{f}_{\widetilde{W}(\widehat \beta)}(W(\beta_0))\varphi_s(X,H(Z)) & t_{W(\beta_0)}^\eta(W(\beta_0))\\
    = \sqrt{n} \mathbb{P}_n &\varepsilon f_{W( \beta_0)}(W(\beta_0))\varphi_s(X,H(Z)) t_{W(\beta_0)}^\eta(W(\beta_0))+o_P(1)
\end{align*}
uniformly in $s\in\cal S$. Finally, since $\sqrt{n}\mathbb{P}_n|t^\eta_{W(\beta_0)}-1|=o_P(1)$ from Lemma \ref{lem: trimming implications}(i), we can replace $t^\eta_{W(\beta_0)}$ with 1 at the cost of an $o_P(1)$ reminder. Gathering results,
\begin{equation}
     \text{1}^{st}\text{ term RHS of  (\ref{eq: approximation of 1st term of 2nd decomposition of Asy Emp Process})}=   \sqrt{n} \mathbb{P}_n \varepsilon f_{W( \beta_0)}(W(\beta_0))\varphi_s(X,H(Z)) +o_P(1)
\end{equation}
uniformly in $s\in\cal S$. Similar arguments lead to 
\begin{equation}
 \text{2}^{nd}\text{ term RHS of  (\ref{eq: approximation of 1st term of 2nd decomposition of Asy Emp Process})}=     \sqrt{n} \mathbb{P}_n \varepsilon f_{ W ( \beta)}( W ( \beta_0))\iota_s( W ( \beta_0))+o_P(1)\, 
\end{equation}
uniformly in $s\in \mathcal{S}$.
Next, we deal with the 3rd term on the RHS of  (\ref{eq: approximation of 1st term of 2nd decomposition of Asy Emp Process}). We have (see the comments below)  
\begin{align*}
    \text{3}^{rd} \text{ term RHS of  (\ref{eq: approximation of 1st term of 2nd decomposition of Asy Emp Process})} = & \sqrt{n}\mathbb{P}_n \varepsilon f_{W(\beta_0)}(W(\beta_0)) \partial_{H} \varphi_s(X,H(Z))\, [\widetilde H (Z) - H(Z)] t + o_P(1)\\
    &= o_P(1)\, 
\end{align*}
uniformly in $s$. To obtain the first equality we have combined the $n^{-1/4}$ rates in Lemma \ref{lem: convergence rates for Ghat fhat and iotahat}(vii) and Lemma \ref{lem: trimming and convergence rates for 1st step}(iii). The second equality is obtained after applying the stochastic expansion in Lemma \ref{lem: ASE for 1st step}(i) and noticing that under $\mathcal{H}_0$ and from Assumption \ref{Assumption: iid and phi}(ii) we have $\mathbb{E}\{\varepsilon|X,Z\}=0$.\\ Let us now deal with the 4th term on the RHS of (\ref{eq: approximation of 1st term of 2nd decomposition of Asy Emp Process}). Below we show that
\begin{align}\label{eq: approximation  of W.tilde(betahat) - W(beta0)}
    \text{4}^{th} \text{ term RHS} &\text{  of (\ref{eq: approximation of 1st term of 2nd decomposition of Asy Emp Process})}=
    \sqrt{n}\mathbb{P}_n \varepsilon f_{W(\beta_0)}(W(\beta_0))\,\partial^T \iota_s(W(\beta_0)) [\widetilde W (\widehat \beta) - W(\beta_0)] t \nonumber \\
    =&\mathbb P _n \varepsilon f_{W(\beta_0)}(W(\beta_0)) \partial^T \iota_s (W(\beta_0))\, t\, \partial q_{\beta^T}(\beta_0,X,H(Z)) \sqrt{n}(\widehat{\beta}-\beta_0) \nonumber \\
    &+ \sqrt{n} \mathbb{P}_n \varepsilon  f_{W(\beta_0)}(W(\beta_0))\, \partial^T \iota_s (W(\beta_0))\, \partial_{H^T} q(\beta_0,X,H(Z)) [\widetilde{H}(Z)-H(Z)] t \nonumber  \\
    & +o_p(1) \nonumber  \\
    =& o_P(1)\, .
\end{align}
For the first equality we have used $\max_{i=1,\ldots,n}|\widetilde{W}_i(\widehat \beta) - W_i(\beta_0)|t_i=o_P(n^{-1/4})$, as already noticed earlier in this proof, and $\max_{i=1,\ldots,n}|\widehat{f}_{\widetilde{W}(\widehat{\beta})}(W_i(\beta_0))-f_{W(\beta_0)}(W_i(\beta_0))|t_i=o_P(n^{-1/4})$ from Lemma \ref{lem: convergence rates for Ghat fhat and iotahat}(vii). Next, using (\ref{eq: mean value expoansion of q}) gives the second equality in (\ref{eq: approximation  of W.tilde(betahat) - W(beta0)}). For the third equality, notice that Lemma \ref{lem: trimming and convergence rates for 1st step}(i) allows us to replace $t$ with 1 in the first leading term, at the cost of an $o_P(1)$ reminder. Next, a Glivenko-Cantelli property of $\partial \iota_s$ gives
\begin{equation*}
    \mathbb{P}_n \varepsilon f_{W(\beta_0)} \partial^T \iota_s \partial_{\beta^T} q(\beta_0,X,H(Z)) =\mathbb E \varepsilon f_{W(\beta_0)}\partial^T \iota_s \partial_{\beta^T} q(\beta_0,X,H(Z))+o_P(1)
\end{equation*}
uniformly in $s\in\mathcal{S}$, where for simplicity of notation we have omitted the argument of $f_{W(\beta_0)}$ and $\iota_s$.  The expectation in the previous equation is null, as $\mathbb E \{\varepsilon |X,Z\}=0$ under $\mathcal{H}_0$ and from Assumption \ref{Assumption: iid and phi}(ii). Thus, since $\sqrt{n}(\widehat \beta - \beta_0)=O_P(1)$ from Lemma \ref{lem: IFR for betahat}, the first leading term after the second equality in (\ref{eq: approximation  of W.tilde(betahat) - W(beta0)}) is $o_P(1)$.  The second leading term is also negligible thanks to the expansion in Lemma \ref{lem: ASE for 1st step}(i) and $\mathbb E \{\varepsilon |X,Z\}=0$ under $\mathcal H _0$. Thus, the third equality in (\ref{eq: approximation  of W.tilde(betahat) - W(beta0)}) follows. \\
Arguing similarly as in (\ref{eq: approximation  of W.tilde(betahat) - W(beta0)}) leads to 
\begin{align*}
    \text{5}^{th} \text{ term RHS of (\ref{eq: approximation of 1st term of 2nd decomposition of Asy Emp Process})} = o_P(1)\, .
\end{align*}
So, by gathering the previous results we get that uniformly in $s$ 
\begin{equation}\label{eq: 1st term of the final asymptotic expansion}
    \text{1}^{st}\text{ term RHS of (\ref{eq: 2nd decomposition of Asy Emp Process})}=\sqrt{n} \mathbb P _n \varepsilon f_{W(\beta_0)}(W(\beta_0))\, \varphi_s^{\perp} + o_P(1)\, ,
\end{equation}
where $\varphi_s^\perp(X,Z)=\varphi_s(X,H(Z))-\mathbb E \{\varphi_s(X,H(Z))|W(\beta_0)\}$ and we have dropped the argument of $\varphi_s^\perp$ for simplicity of notation.\\
In view of the previous display, to obtain the desired result it suffices to obtain an IFR for the second term on the RHS of (\ref{eq: 2nd decomposition of Asy Emp Process}). Since $\|(\widetilde H - H)\widehat t\|_\infty=o_P(n^{-1/4})$, see Lemma \ref{lem: trimming and convergence rates for 1st step}(iii), by Assumptions \ref{Assumption: iid and phi}(iv) and \ref{Assumption: smoothness}(v) we have that $|\varphi_s(X_i,\widetilde H (Z_i))-\varphi_s(X_i,H(Z_i))|\widehat t _i =O_P(\|(\widetilde H - H)\widehat t\|_\infty)=o_P(n^{-1/4})$ uniformly in $i=1,\ldots,n$ and $s$. Using this and the $n^{-1/4}$ rates in Lemma \ref{lem: convergence rates for Ghat fhat and iotahat}(xiii)(xiv)(xv) leads to
\begin{align*}
    \text{2}^{nd} \text{ term RHS of (\ref{eq: 2nd decomposition of Asy Emp Process})}= \sqrt{n}\mathbb{P}_n [G_0(W(\beta_0))-\widehat{G}_{\widetilde{W}(\widehat \beta)}(\widetilde{W}(\widehat \beta))]\, f_{W(\beta_0)}(W(\beta_0))\, \varphi_s^{\perp} \widehat{t} +o_P(1)
\end{align*}
uniformly in $s$, where for simplicity of notation we have omitted the argument of $\varphi_s^\perp$. By proceeding as in (\ref{eq: trimming replacement}), the trimming $\widehat{t}$ can be replaced with $t$ at the cost of an $o_P(1)$ reminder. So, 
\begin{align*}
    \sqrt{n}\mathbb{P}_n [G_0(W(\beta_0))-\widehat{G}_{\widetilde{W}(\widehat \beta)}&(\widetilde{W}(\widehat \beta))]\, f_{W(\beta_0)}(W(\beta_0))\, \varphi_s^{\perp} \widehat{t}\\
    =& \sqrt{n}\mathbb{P}_n [G_0(W(\beta_0))-\widehat{G}_{\widetilde{W}(\widehat \beta)}(\widetilde{W}(\widehat \beta))]\, f_{W(\beta_0)}(W(\beta_0))\, \varphi_s^{\perp} t+o_P(1)
\end{align*}
uniformly in $s$. Next, by Lemma \ref{lem: stochastic expansions}(i) we get that uniformly in $s\in\mathcal{S}$
\begin{align}\label{eq: 2nd term RHS of 2nd decomposition of Asy Emp Process}
    \sqrt{n}\mathbb{P}_n [G_{W(\beta_0)}&(W(\beta_0))-\widehat{G}_{\widetilde{W} (\widehat \beta)}(\widetilde{W}(\widehat \beta))]\,  f_{W(\beta_0)}(W(\beta_0))\, \varphi_s^{\perp}t \nonumber \\
    =& \sqrt{n}\mathbb{P}_n [G_{W(\beta_0)}(W(\beta_0))-\widehat{G}_{\widetilde{W}(\widehat \beta)}(W(\beta_0))]\, f_{W(\beta_0)}(W(\beta_0))\, \varphi_s^{\perp} t \nonumber  \\
    &- \sqrt{n} \mathbb{P}_n f_{W(\beta_0)}(W(\beta_0)) \varphi_s^\perp \partial^T G_{W(\beta_0)}(W(\beta_0))\, [\widetilde{W}(\widehat \beta)-W(\beta_0)]\, t + o_P(1) \, .
\end{align}
 Now, for the first term on the RHS of (\ref{eq: 2nd term RHS of 2nd decomposition of Asy Emp Process})  we have  (see the comments below)
\begin{align*}
    \sqrt{n}\mathbb{P}_n [G_{W(\beta_0)}(W(\beta_0))-&\widehat{G}_{\widetilde{W}(\widehat \beta)}(W(\beta_0))]  \, f_{W(\beta_0)}(W(\beta_0))\, \varphi_s^{\perp} t \\
    =& \sqrt{n}\mathbb{P}_n [G_0(W(\beta_0))-\widehat{G}_{\widetilde{W}(\widehat \beta)}(W(\beta_0))]\, f_{W(\beta_0)}(W(\beta_0))\, \varphi_s^{\perp} t^\eta_{W(\beta_0)}+o_P(1)\\
    =& o_P(1),
\end{align*}
where $\eta>0$. The first equality follows by the same arguments as in (\ref{eq: first ASE for asy behavior}). To get the second equality,  notice first that the leading term is a centered empirical process, as $\mathbb{E}\{\varphi_s^\perp |$ $W(\beta_0)\}=0$. Also,  $G_{W(\beta_0)}-\widehat{G}_{\widetilde{W}(\widehat \beta)}\in \mathcal{G}_\lambda(\mathcal W _{n,\beta_0}^{\eta/2})$ wpa1, with $\lambda=\lceil (d+1)/2\rceil$ 
(see Lemma \ref{lem: belonging conditions}(v)),  and  $||\widehat{G}_{\widetilde{W}(\widehat{\beta})}-G_{W(\beta_0)}]t^\eta_{W(\beta_0)}||_\infty=o_P(1)$ (see Lemma \ref{lem: convergence rates for Ghat fhat and iotahat}(xi)). Hence, the stochastic equicontinuity result in Lemma \ref{lem: ASE}(i) delivers the second equality in the previous display. \\
Finally, to handle the second term on the RHS of (\ref{eq: 2nd term RHS of 2nd decomposition of Asy Emp Process}) we can use the same arguments as for (\ref{eq: approximation  of W.tilde(betahat) - W(beta0)}) and obtain  that uniformly in $s$
\begin{align*}
     \sqrt{n}\mathbb P _n f_{W(\beta_0)}(W(\beta_0))\,\varphi_s^\perp \partial^T &G_{W(\beta_0)}(W(\beta_0))\, [\widetilde W (\widehat \beta)-W(\beta_0)]\, t\\
     =& \mathbb E \, \{\, f_{W(\beta_0)}(W(\beta_0))\varphi_s^\perp \partial^T G_{W(\beta_0)}(W(\beta_0))\partial_{\beta^T} q(\beta_0,X,H(Z))\,\}\,\\
     &\hspace{9cm}\cdot\sqrt{n}(\widehat \beta - \beta_0)\\
     &+ \sqrt{n} \mathbb P _n a_s(Z)\,[D-H(Z)]+o_P(1)\, ,\\
     \text{ where }a_s(Z)=&\mathbb E \{f_{W(\beta_0)}\varphi_s^\perp \partial^T G_{W(\beta_0)}\partial_{H} q|Z\}\, .
\end{align*}
By gathering together the previous five displays we get 
\begin{align*}
      \text{2}^{nd} \text{ term RHS of (\ref{eq: 2nd decomposition of Asy Emp Process})}     =& \mathbb E\,\{\,  f_{W(\beta_0)}\varphi_s^\perp \partial^T G_{W(\beta_0)}\partial_{\beta^T} q\,\}\, \sqrt{n}(\widehat \beta - \beta_0)\\
     &+ \sqrt{n} \mathbb P _n a_s(Z)\,[D-H(Z)]+o_P(1)\, 
\end{align*}
uniformly in $s\in\mathcal{S}$. Finally, by the above display, (\ref{eq: 1st term of the final asymptotic expansion}), (\ref{eq: 2nd decomposition of Asy Emp Process}), and the IFR of $\sqrt{n}(\widehat \beta-\beta_0)$ from Lemma \ref{lem: IFR for betahat} we obtain the desired result.
\begin{flushright}
{\itshape [Q.E.D.]}
\end{flushright}
 \vspace{0.5cm}
Let $\delta\in(0,1]$ and let us recall the definitions of $\mathcal{W}_{n,\beta}^{\delta}$,  $\mathcal{Z}_n^\delta$, and $\mathcal{U}_n^\delta$ from (\ref{eq: definition of cal W, cal Z and cal U}). 
We slightly simplify our notation by introducing 
\begin{equation}\label{eq: definition of cal W beta0 n}
    \mathcal{W}_{n}^\delta:=\mathcal{W}_{n,\beta_0}^\delta\, .
\end{equation}
From Assumption \ref{Assumption: trimming}(i), there exists $\eta(\delta)$ and $N(\delta)$ such that for any $n\geq N(\delta)$  we have $\mathcal{U}_n^\delta\subset\{(x,z)\,:\,z\in \mathcal{Z}_n^{\eta(\delta)}\}$ and $\mathcal{U}_n^\delta\subset \{(x,z)\,:\, q(\beta_0,x,H(z))\in \mathcal{W}_{n}^{\eta(\delta)}\}$.
For such $\eta(\delta)$ we let  
\begin{equation}\label{eq: definition of Psi class of functions}
    \Psi_{n}^\delta:=\left\{\psi\in\mathcal{G}_\lambda(\mathcal{Z}_{n}^{\eta(\delta)})\text{ : }f_{W(\beta_0)}(q(\beta_0,x,\psi(z)))\geq \eta(\delta)\tau_n/2\text{ for all }(x,z)\in\mathcal{U}_{n}^{\delta}\right\}\, ,
\end{equation}
with $\lambda=\lceil (p+1)/2\rceil$ and $\mathcal{G}_\lambda$ from Definition \ref{defn: class of functions}. For the next lemma, we write $a\lesssim b$ to mean that $a\leq C b$ where $C$ is a universal constant. Also, for a generic class of functions $\mathcal H$ endowed with a metric $\|\cdot\|_{\mathcal{H}}$, we denote by $N(\epsilon,\mathcal H ,\|\cdot\|_{\mathcal{H}})$ its $\epsilon$ covering number, with $\epsilon>0$, see \citet[Chapter 19]{vaart_asymptotic_1998}.

\begin{lem}\label{lem: entropy bounds for asynotm of betahat}
Fix $\delta\in (0,1]$ and let 
\begin{equation*}
    \mathcal{F}_n^\delta:=\left\{(x,z)\mapsto g(q(\beta_0,x,\psi(z)))\text{ : }g\in\mathcal{G}_\lambda(\mathcal{W}^{\eta(\delta)/3}_{n})\, \text{ and }\, \psi\in \Psi_n^\delta\right\}\, 
\end{equation*}
with $\lambda=\lceil (d+1)/2 \rceil$. Then, for $\upsilon\in (0,2)$ 
\begin{enumerate}[label=(\roman*)]
    \item $\log N(\epsilon,\mathcal{F}_n^\delta,\|\cdot\|_{\infty,\mathcal{U}_n^\delta})\lesssim \epsilon^{-\upsilon}$ for all $\epsilon\in (0,1)$
    \item $\log N(\epsilon,\mathcal{F}_n^\delta\cdot \mathcal{G}_{\lambda}(\mathcal{Z}_n^{\eta(\delta)}),\|\cdot\|_{\infty,\mathcal{U}_n^\delta})\lesssim \epsilon^{-\upsilon}$ for all $\epsilon\in(0,1)$ and for $\lambda=\lceil (p+1)/2 \rceil$.
\end{enumerate}
\end{lem}

\begin{proof}
$(i)$ Given $\delta\in (0,1]$, Assumption \ref{Assumption: trimming}(i)(v) ensures that there exists $\eta(\delta)>0$ and $N(\delta)$ such that for any $n\geq N(\delta)$ we have (i) $\mathcal{U}_n^\delta\subset \{(x,z)\,:\, z\in \mathcal{Z}_n^{\eta(\delta)}\}$ and (ii) $\mathcal{W}_n^{\eta(\delta)/3}$ and  $\mathcal{Z}_n^{\eta(\delta)}$ are convex. Let us fix $n$ so that these two conditions are met. 
Given the convexity of $\mathcal{Z}_n^{\eta(\delta)}$,  \citet[Theorem 2.7.1]{van_der_vaart_weak_1996}   ensures that $\log N(\epsilon,\mathcal{G}_\lambda(\mathcal{Z}_n^{\eta(\delta)}),\|\cdot\|_{\infty,\mathcal{Z}_n^{\eta(\delta)}})\lesssim \epsilon^{-\upsilon}$  with  $\upsilon=p/\lceil (p+1)/2 \rceil\in(0,2)$  and $ \lambda=\lceil (p+1)/2 \rceil\,$. Since $\Psi_n^\delta\subset \mathcal{G}_{\lceil (p+1)/2 \rceil} (\mathcal{Z}_n^{\eta(\delta)})$, such an entropy bound implies 

\begin{equation}\label{eq: first entropy bound for betahat}
\log N(\epsilon,\Psi_n^\delta,\|\cdot\|_{\infty,\mathcal{Z}_n^{\eta(\delta)}})\lesssim \epsilon^{-\upsilon}\text{ for all }\epsilon\in(0,1)\, .
\end{equation}
Similarly, 
\begin{equation}\label{eq: second entropy bound for betahat}
    \log N(\epsilon,\mathcal{G}_{\lceil (d+1)/2\rceil }(\mathcal{W}_n^{\eta(\delta)/3}),\|\cdot\|_{\infty,\mathcal{W}_n^{\eta(\delta)/3}}) \lesssim \epsilon^{-\upsilon}\text{ for all }\epsilon \in (0,1)\, ,
\end{equation}
with $\upsilon\in(0,2)$. Now, let us fix $\epsilon\in (0,1)$ and let us consider an $\epsilon$ cover for $\Psi_n^\delta$, say
\begin{equation*}
    \Psi^{\epsilon}:=\{\psi^{\kappa}:\kappa=1,\ldots,N(\epsilon,\Psi_n^\delta,\|\cdot\|_{\infty,\mathcal{Z}_n^{\eta(\delta)}})\}\, ,
\end{equation*}
and an $\epsilon$ cover for $\mathcal{G}_{\lceil (d+1)/2\rceil }(\mathcal{W}_n^{\eta(\delta)/3})$, say
\begin{equation*}
   \mathcal{G}^\epsilon:=\{g^\kappa : \kappa=1,\ldots,N(\epsilon,\mathcal{G}_{\lceil (d+1)/2\rceil }(\mathcal{W}_n^{\eta(\delta)/3}) , \|\cdot\|_{\infty,\mathcal{W}_n^{\eta(\delta)/3}})\}. 
\end{equation*}
We can assume wlog that $\Psi^\epsilon \subset \Psi_n^\delta$ and $\mathcal{G}^\epsilon \subset \mathcal{G}_{\lceil (d+1)/2\rceil }(\mathcal{W}_n^{\eta(\delta)/3}) $.
Pick an arbitrary element of $\mathcal{F}_n^\delta$, say a mapping $(x,z)\mapsto g(q(\beta_0,x,\psi(z)))$ with $g\in \mathcal{G}_{\lceil (d+1)/2\rceil }(\mathcal{W}_n^{\eta(\delta)/3})$ and $\psi\in \Psi_n^{\delta}$. Then, $\|g-g^\kappa\|_{\infty,\mathcal{W}_n^{\eta(\delta)/3}}<\epsilon$ for some $g^\kappa \in \cal G^\epsilon$, and $\|\psi-\psi^\kappa\|_{\infty,\mathcal{Z}_n^{\eta(\delta)}}<\epsilon$ for some $\psi^\kappa \in \Psi^\epsilon$. Now, 
\begin{align}\label{eq: inequality for the entropy for betahat}
    \sup_{(x,z)\in \mathcal{U}_n^\delta}\,|g(q(\beta_0,x,\psi(z)))&-g^\kappa(q(\beta_0,x,\psi^\kappa(z)))\,|\nonumber \\
    \leq & \sup_{(x,z)\in \mathcal{U}_n^\delta} \left|g(q(\beta_0,x,\psi(z))) - g^\kappa(q(\beta_0,x,\psi(z)))\right| \nonumber \\
    &+ \sup_{(x,z)\in \mathcal{U}_n^\delta} \left| g^\kappa(q(\beta_0,x,\psi(z))) - g^\kappa(q(\beta_0,x,\psi^\kappa(z)))\right|\, .
\end{align}
By definition of the class $\Psi_n^\delta$, the first term on the RHS is bounded by $\|g-g^\kappa\|_{\infty,\mathcal{W}_n^{\eta(\delta)/3}}<\epsilon$. To bound the second term, notice that since $\psi,\psi^\kappa\in \Psi_n^\delta$, it must be that $q(\beta_0,x,$ $\psi(z))\,,\, q(\beta_0,x,\psi^\kappa(z)) \in \mathcal{W}_n^{\eta(\delta)/2} $ for all $(x,z)\in \mathcal{U}_n^\delta$. Thus, by convexity of $\mathcal{W}_n^{\eta(\delta)/3}$, we can apply the Mean-Value Theorem to the second term on the RHS of (\ref{eq: inequality for the entropy for betahat}) and bound it by 
\begin{align*}
    \|\partial g^\kappa\|_{\infty,\mathcal{W}^{\eta(\delta)/3}_{n}}\, \sup_{(x,z)\in \mathcal{U}_n^\delta}|q(\beta_0,x,\psi(z))&-q(\beta_0,x,\psi^\kappa(z))|\\
    \leq & M C \sup_{(x,z)\in \mathcal{U}_n^\delta}|\psi(z)-\psi^\kappa(z)|\\
    \leq & M\, C\ \|\psi-\psi^\kappa\|_{\infty,\mathcal{Z}_n^{\eta(\delta)}}\, ,
\end{align*}
where $M$ is a constant that depends on the class $\mathcal{G}_{\lceil (d+1)/2 \rceil }(\mathcal{W}_n^{\eta(\delta)/3})$, $C$ is the Lipschitz constant of $q$ (see Assumption \ref{Assumption: smoothness}(v)), and the last inequality follows from  $\mathcal{U}_n^\delta\subset \{(x,z):z\in\mathcal{Z}_n^{\eta(\delta)}\}$. By definition of $\psi^\kappa$, we have $\|\psi-\psi^\kappa\|_{\infty,\mathcal{Z}_n^{\eta(\delta)}} < \epsilon$. Gathering results, 
\begin{equation*}
    \sup_{(x,z)\in \mathcal{U}_n^\delta}\,|g(q(\beta_0,x,\psi(z)))-g^\kappa(q(\beta_0,x,\psi^\kappa(z)))\,| < (1+C M )\epsilon\, .
\end{equation*}
Since $\epsilon\in(0,1)$ was arbitrary, we conclude that
\begin{equation*}
    N(\epsilon(1+CM),\mathcal{F}_n^\delta,\|\cdot\|_{\infty,\mathcal{U}_n^\delta})\leq N(\epsilon,\Psi_n^\delta,\|\cdot\|_{\infty,\mathcal{Z}_n^{\eta(\delta)}})\,\, N(\epsilon,\mathcal{G}_{\lceil (d+1)/2\rceil }(\mathcal{W}_n^{\eta(\delta)/3}),\|\cdot\|_{\infty,\mathcal{W}_n^{\eta(\delta)/3}})
    \end{equation*}
    for all $\epsilon\in(0,1)$. So, (\ref{eq: first entropy bound for betahat}) and (\ref{eq: second entropy bound for betahat}) give $(i)$.\\
    
$(ii)$ Fix $\delta\in(0,1]$.  
By definition of the class $\mathcal{F}_n^\delta$ we have $\sup_{(x,z)\in\mathcal{U}_n^\delta}|f(x,z)|\leq M$ for all $f\in \mathcal{F}_n^\delta$ and for a fixed constant $M$, i.e. a constant envelope function. Also, 
From Assumption \ref{Assumption: trimming}(i) there exists $\eta(\delta)>0$ and $N(\delta)$ such that for all $n\geq N(\delta)$ we have $\mathcal{U}_n^\delta\subset \{(x,z):z\in\mathcal{Z}_n^{\eta(\delta)}\}$. Thus, for all $\psi\in\mathcal{G}_{\lceil (p+1)/2 \rceil}(\mathcal{Z}_n^{\eta(\delta)})$ we have  $\sup_{(x,z)\in\mathcal{U}_n^{\delta}}|\psi(z)|\leq M$, i.e. a constant envelope function. Using these bounds, the entropy condition on $\mathcal{F}_n^\delta$ (from $(i)$ of the present lemma), and the entropy condition on $\mathcal{G}_{\lceil (p+1)/2 \rceil}(\mathcal{Z}_n^{\eta(\delta)})$ (see the proof of $(i)$ of the present lemma) leads to the desired result.
\end{proof}
\vspace{0.5cm}
The following lemma obtains the IFR for $\widehat \beta$. This IFR also gives the asymptotic normality of $\sqrt{n}(\widehat \beta - \beta_0)$. The asymptotic normality result could alternatively be obtained by proving the high-level conditions in \citet{chen_estimation_2003}. For our purposes, however, we need an expression for the IFR of $\widehat \beta$, so we provide a complete proof for this. 

\begin{lem}\label{lem: IFR for betahat}
Let Assumptions \ref{Assumption: iid and phi}-\ref{Assumption: trimming} hold. Then, under $\mathcal H _0$
\begin{align*}
    \sqrt{n}(\widehat \beta - \beta_0)=& \sqrt{n} \mathbb P _n\, \varepsilon\, \Sigma\, \nabla_\beta G_{W(\beta_0)}(W(\beta_0))\\
    & - \sqrt{n} \mathbb P _n\, a(Z)\,[D-H(Z)]+o_P(1)\, ,
    \end{align*}
    where 
    \begin{align*}
        a(Z)=\mathbb{E}\{\,\Sigma\, \nabla_\beta& G_{W(\beta_0)}(W(\beta_0))\, \partial^T G_{W(\beta_0)}\,\partial_H q\, |\, Z\,\}\\
    \text{ and }\Sigma=\mathbb{E}\{\,\nabla_\beta& G_{W(\beta_0)}(W(\beta_0))\,\, \nabla_{\beta^T} G_{W(\beta_0)}(W(\beta_0))\,\}\, .
\end{align*}

\end{lem}

\begin{proof}
We will first prove the consistency of $\widehat \beta$. From Lemma \ref{lem: convergence rates for Ghat fhat and iotahat}(i) we get that $\max_{i=1,\ldots,n}$ $|\widehat{G}_{\widetilde{W}(\beta)}(\widetilde{W}_i(\beta))$ $-G_{W(\beta)}(W_i(\beta))|\widehat t _i=o_P(1)$ uniformly in $\beta\in B$. Thus,  $\mathbb P _n [Y-\widehat G _{\widetilde W (\beta)}(\widetilde W (\beta))]^2\,\widehat t$ $=\mathbb P _n [Y- G _{ W (\beta)}( W (\beta))]^2\,\widehat t$ $+o_P(1)$ uniformly in $\beta\in B$. Lemma \ref{lem: trimming and convergence rates for 1st step}(i) allows replacing $\widehat t$ with 1 on the RHS of the previous equation, at the cost of an $o_P(1)$ reminder. Assumption \ref{Assumption: smoothness}(i)(v) implies that $G_{W(\beta)}(q(\beta,x,H(z)))$ is Lipschitz in $\beta$ uniformly in $(x,z)\in Supp(X,Z)$. So, by a Glivenko-Cantelli theorem $\mathbb P _n [Y- G _{ W (\beta)}( W (\beta))]^2$ $=\mathbb E [Y- G _{ W (\beta)}( W (\beta))]^2$ $+o_P(1)$ uniformly in $\beta\in B$. Gathering results, 
\begin{equation*}
    \mathbb P _n [Y-\widehat G _{\widetilde W (\beta)}(\widetilde W (\beta))]^2\widehat t=\mathbb E [Y-G_{W(\beta)}(W(\beta))]^2+o_P(1)
\end{equation*}
uniformly in $\beta\in B$. By Assumption \ref{Assumption: parametric}(i), the RHS of the above display is uniquely minimized at $\beta_0$. So, by \citet[Theorem 5.7]{vaart_asymptotic_1998} 
$$\widehat \beta=\beta_0+o_P(1)\, .$$
Let us now obtain the influence function representation for $\widehat \beta$. By Lemma \ref{lem: trimming implications}(ii), wpa1 $\widehat t (X_i,Z_i)\leq t^\delta(X_i,Z_i)$ uniformly in $i=1,\ldots,n$, where $\delta=1/2$. Assumption \ref{Assumption: trimming}(i) ensures that there exists $\eta(\delta)>0$ such that wpa1 $t^\delta(X_i,Z_i)\leq $ $\mathbb I \{f_{W( \beta)}(q( \beta,X_i,$ $H(Z_i))\geq \eta(\delta)\tau_n\}$ for all $i=1,\ldots,n$ and $\beta\in B$. Also, by Lemma \ref{lem: trimming implications}(iv) wpa1 $t^\delta(X_i,Z_i)$ $\mathbb I \{f_{W( \beta)}(q( \beta,X_i,H(Z_i))\geq \eta(\delta)\tau_n\}$ $\leq t^\delta (X_i,Z_i)$ $\mathbb I \{f_{W( \beta)}(q( \beta,X_i,\widetilde H (Z_i)))\geq \eta(\delta) \tau_n/2\}$ for all $i=1,\ldots,n$ and $\beta\in B$. Finally, by using Lemma \ref{lem: trimming implications}(iii) we have that wpa1 $\mathbb I \{f_{W( \beta)}(q( \beta,X_i,\widetilde H (Z_i)))\geq \eta(\delta) \tau_n/2\}$ $\leq \mathbb I \{\widehat f_{\widetilde W( \beta)}(q( \beta,X_i,\widetilde H (Z_i)))\geq \eta(\delta) \tau_n/4\}$ for all $i=1,\ldots,n$ and $\beta\in B$. Gathering results gives that wpa1 
\begin{equation}\label{eq: trimming implication for foc of betahat}
    \widehat t (X_i,Z_i) \leq t^\delta (X_i,Z_i) \, \mathbb I \{\widehat f_{\widetilde W(\beta)}(q(\beta,X_i,\widetilde H(Z_i)))\geq \eta(\delta) \tau_n/4\}
\end{equation}
uniformly in $i=1,\ldots,n$ and $\beta\in B$, with $\delta=1/2$. This implies that wpa1 the objective function in (\ref{eq: definition of betahat}) is differentiable in $\beta$ over $Int(B)$. Thus, since $\widehat \beta= \beta_0+o_P(1)$ with $\beta_0\in Int(B)$ we get 
\begin{equation*}
    \mathbb P _n [Y-\widehat G _{\widetilde W (\widehat \beta)}(\widetilde W (\widehat \beta))]\,\nabla_\beta \widehat G _{\widetilde W (\widehat \beta)}(\widetilde W (\widehat \beta)) \,\widehat t=o_P(n^{-1/2})\, .
\end{equation*}
Now, from Equation (\ref{eq: trimming implication for foc of betahat}) we can use the Mean-Value Theorem to expand the LHS of the previous display around $\beta_0$. This gives 
\begin{align}\label{eq: mean-value expansion for betahat}
    \widehat \Sigma (\overline \beta)\, \sqrt{n}(\widehat \beta - \beta_0)=& \sqrt{n}\mathbb P _n [Y-\widehat G _{\widetilde W (\beta_0)}(\widetilde W (\beta_0))]\, \nabla_\beta \widehat G _{\widetilde W (\beta_0)}(\widetilde W (\beta_0))\, \,\widehat t+o_P(1)\,
\end{align}
where $\overline \beta$ lies on the segment joining $\widehat \beta$ and $\beta_0$ and 
\begin{align*}
        \widehat \Sigma (\beta)=&- \mathbb P _n [Y-\widehat G _{\widetilde W (\beta)}(\widetilde W (\beta))]\, \nabla_{\beta \beta^T}^2 \widehat G _{\widetilde W (\beta)}(\widetilde W (\beta))\, \,\widehat t \nonumber \\
    & + \mathbb P _n \nabla_{\beta } \widehat G _{\widetilde W (\beta)}(\widetilde W (\beta))\, \nabla_{\beta^T } \widehat G _{\widetilde W (\beta)}(\widetilde W (\beta)) \,\widehat t\, . \nonumber 
\end{align*}
Let us now obtain a limit for $\widehat \Sigma (\overline \beta)$. Combining Lemma \ref{lem: convergence rates for Ghat fhat and iotahat}(i) and Lemma \ref{lem: convergence of the derivatives at the data points}(i)(ii) with the arguments used at the beginning of this proof gives $\widehat \Sigma (\beta)$ $=-\mathbb E [Y-G_{W(\beta)}(W(\beta))]$ $\nabla^2_{\beta \beta^T}G_{W(\beta)}(W(\beta))$ $+ \mathbb E $ $\nabla_\beta G_{W(\beta)}(W(\beta))$ $ \nabla_{\beta^T} G_{W(\beta)}(W(\beta))$ $+o_P(1)$ uniformly in $\beta\in B$. By using this result,  $\overline \beta = \beta_0+o_P(1)$, and since $G_{W(\beta)}(W(\beta))$, $\nabla_\beta G_{W(\beta)}(W(\beta))$, and $\nabla^2_{\beta \beta^T} $ $G_{W(\beta)}(W(\beta))$ are Lipschitz in $\beta$ (see Assumption \ref{Assumption: smoothness}(i)(v)(vi)),  we get\footnote{Notice that under $\mathcal{H}_0$, by Assumption \ref{Assumption: iid and phi}(ii) we have $\mathbb{E}[Y-G_{W(\beta_0)}]\nabla^2_{\beta \beta^T} G_{W(\beta_0)}(W(\beta_0))=0$ .}
\begin{equation}\label{convergence of Sigmahat betahat}
    \widehat \Sigma (\overline \beta)=\mathbb E\, \nabla_{\beta} G_{W(\beta_0)}(W(\beta_0))\, \nabla_{\beta^T} G_{W(\beta_0)}(W(\beta_0)) + o_P(1)=\Sigma+o_P(1)\, .
\end{equation}
So, to show the result of the lemma it suffices to obtain an IFR for the leading term on the RHS of (\ref{eq: mean-value expansion for betahat}). To this end, let us decompose such a term as 
\begin{align}\label{eq: decomposition of FOC for betahat}
    \sqrt{n}\mathbb P _n [Y-\widehat G _{\widetilde W (\beta_0)}&(\widetilde W (\beta_0))]\, \nabla_\beta \widehat G _{\widetilde W (\beta_0)}(\widetilde W (\beta_0))\, \,\widehat t\nonumber \\
&= \sqrt{n}\mathbb P _n  \varepsilon\, \nabla_\beta \widehat G _{\widetilde W (\beta_0)}(\widetilde W (\beta_0)) \widehat t \nonumber \\
&+ \sqrt{n} \mathbb P _n [G_{W(\beta_0)}(W(\beta_0))-\widehat G _{\widetilde W (\beta_0)}(\widetilde W (\beta_0))]\, \nabla_\beta \widehat G _{\widetilde W (\beta_0)}(\widetilde W (\beta_0)) \, \widehat t\, .
\end{align}
Let us consider the first term on the RHS. From Lemma \ref{lem: convergence of the derivatives at the data points}(i) $\max_{i=1,\ldots,n}|\nabla_\beta \widehat{G}_{\widetilde W (\beta_0)}(\widetilde W_i(\beta_0)) |(\widehat t _i$ $ + t_i)=O_P(1)$. Then,  and arguing as in (\ref{eq: trimming replacement}), the trimming $\widehat t$ can be replaced with $t$ at the cost of an $o_P(1)$ reminder.  So, 
\begin{align}\label{eq: first term in decomposition of FOC for betahat}
    \text{1}^{st}\text{ term RHS of (\ref{eq: decomposition of FOC for betahat})}^T=& \sqrt{n}\mathbb P _n \varepsilon \partial_{\beta^T} \widehat G_{\widetilde W (\beta_0)}( q(\beta_0,X,\widetilde H (Z)))\, t \nonumber \\
    &+ \sqrt{n} \mathbb P _n \varepsilon \partial^T \widehat G _{\widetilde W (\beta_0)}(q(\beta_0,X,\widetilde H (Z)))\, \partial_{\beta^T} q(\beta_0,X,\widetilde H (Z))\, t \nonumber \\
    &+ o_P(1)\, .
\end{align}
To handle the first term on the RHS of the previous display we will prove the conditions of Lemma \ref{lem: ASE}(ii). So, let us fix $\delta=1/2$. By Assumption \ref{Assumption: trimming}(i) there exists $\eta(\delta)>0$ such that for each $n$ large enough $t^\delta(x,z)\leq t^\delta (x,z)\, \mathbb I \{f_{W(\beta_0)}(q(\beta_0,x,H(z)))\geq \eta(\delta)\tau_n\}$ for all $(x,z)\in Supp(X,Z)$. Also, Lemma \ref{lem: trimming implications}(iv) ensures that wpa1  $t^\delta(x,z)\, \mathbb I \{f_{W(\beta_0)}(q(\beta_0,x,H(z)))\geq \eta(\delta)\tau_n\}$ $\leq t^\delta(x,z)\, \mathbb I \{f_{W(\beta_0)}(q(\beta_0,x,\widetilde H(z)))\geq \eta(\delta)\tau_n/2\}$  for all $(x,z)\in Supp(X,Z)$. The previous two inequalities imply that wpa1 $t^\delta(x,z)\leq t^\delta(x,z)$
 $\mathbb I \{f_{W(\beta_0)}(q(\beta_0,x,\widetilde H(z)))\geq \eta(\delta)\tau_n/2\}$ for all $(x,z)\in Supp(X,Z)$. Also, from Lemma \ref{lem: belonging conditions}(xiii) $\Pr(\widetilde H \in \mathcal{G}_{\lambda}(\mathcal{Z}_n^{\eta(\delta)}))\rightarrow 1$, with $\lambda=\lceil (p+1)/2 \rceil$.  So, we obtain that \begin{equation*}
     \Pr(\widetilde H \in \Psi_n^\delta)\rightarrow 1\  ,
\end{equation*}
where $\Psi_n^\delta$ is the class of functions defined in (\ref{eq: definition of Psi class of functions}) and $\delta=1/2$. Moreover, from Lemma \ref{lem: belonging conditions}(iii) 
\begin{equation*}
    \Pr (\partial_{\beta_l}\widehat G_{\widetilde W (\beta_0)}\in \mathcal{G}_\lambda(\mathcal{W}_{n}^{\eta(\delta)/3}))\rightarrow 1\, \text{ with } \lambda=\lceil (d+1)/2 \rceil\, ,
\end{equation*}
$\mathcal{W}_{n}^{\eta(\delta)/3}:=\mathcal{W}_{n,\beta_0}^{\eta(\delta)/3}$\,, and $\delta=1/2$. The previous two displays imply that the function $(x,z)\mapsto \partial_{\beta_l} \widehat G _{\widetilde W (\beta_0)}(q(\beta_0,x,\widetilde H (z)))$ belongs  to the class $\mathcal{F}_n^\delta$ wpa1, where $\mathcal{F}_n^\delta$ is defined in Lemma \ref{lem: entropy bounds for asynotm of betahat} and $\delta=1/2$. From Lemma \ref{lem: entropy bounds for asynotm of betahat}(i), $\log N(\epsilon,\mathcal{F}_n^\delta,\|\cdot\|_{\infty,\mathcal{U}_n^\delta})\lesssim \epsilon^{-\upsilon}$ for all $\epsilon\in(0,1)$ with $\upsilon\in(0,2)$. This and $\sup_{x,z}|\partial_{\beta_l}\widehat G _{\widetilde{W}(\beta_0)}(q(\beta_0,x,\widetilde H (z)))- \partial_{\beta_l} G _{W(\beta_0)}(q(\beta_0,x,H (z)))|t=o_P(1)$ from Lemma \ref{lem: convergence of the derivatives at the data points}(iv) show that the conditions of Lemma \ref{lem: ASE}(ii) hold for the first term on the RHS of (\ref{eq: first term in decomposition of FOC for betahat}), with  $\mathcal I _n = \mathcal{F}_{n}^{1/2}$ and $t^{2\delta}=t$. Thus,  
\begin{equation}\label{eq: 1st component of the 1st addendum in IFR of betahat}
    \mathbb G _n \varepsilon \partial_\beta \widehat G_{\widetilde W (\beta_0)}(q(\beta_0,X,\widetilde H (Z)))\, t=     \mathbb G _n \varepsilon \partial_\beta  G_{ W (\beta_0)}(q(\beta_0,X,H (Z)))\, t+o_P(1)\, ,
\end{equation}
where for any function $g(Y,X,Z,D)$, $\mathbb{G}_n g(Y,X,Z,D):=\sqrt{n}(\mathbb{P}_n g-\int g(Y,X,Z,D)$ $dP(Y,$ $X,Z,D))$.
By a similar reasoning, 
we obtain\footnote{More in detail, by the Lipschitz continuity of $\partial_u \partial_{\beta^T}q(\beta,x,u)$ (see Assumption \ref{Assumption: smoothness}(v)), $\|(\widetilde H - H)t\|_\infty=o_P(n^{-1/4})$, and a Mean-Value expansion, we have $\partial_{\beta^T}q(\beta_0,X_i,\widetilde{H}(Z_i))\,t_i$ $=\partial_{\beta^T}q(\beta_0,X_i,H(Z_i))\,t_i$ $+\partial_H\partial_{\beta^T}q(\beta_0,X_i,H(Z_i))\,t_i\,[\widetilde{H}(Z_i)-H(Z_i)]$ $+o_P(n^{-1/2})$ uniformly in $i=1,\ldots,n$. We can then plug this expression into the LHS of (\ref{eq: 2nd component of the 1st addendum in IFR of betahat}), argue similarly as for (\ref{eq: 1st component of the 1st addendum in IFR of betahat}), and use Lemma \ref{lem: convergence of the derivatives at the data points}(iii), Lemma \ref{lem: trimming and convergence rates for 1st step}(iii), Lemma \ref{lem: belonging conditions}(ii), and Lemma \ref{lem: entropy bounds for asynotm of betahat}(i)(ii)   to finally get the equality in (\ref{eq: 2nd component of the 1st addendum in IFR of betahat}).  }
\begin{align}\label{eq: 2nd component of the 1st addendum in IFR of betahat}
    \sqrt{n} \mathbb G _n \varepsilon \partial^T \widehat G _{\widetilde W (\beta_0)}(q(\beta_0&,X,\widetilde H (Z)))\, \partial_{\beta^T} q(\beta_0,X,\widetilde H (Z))\, t \nonumber  \\
    & = \sqrt{n} \mathbb G _n \varepsilon \partial^T  G _{ W (\beta_0)}(q(\beta_0,X, H (Z)))\, \partial_{\beta^T} q(\beta_0,X, H (Z))\, t + o_P(1)\, .
\end{align}
Next, since under $\mathcal{H}_0$ and from Assumption \ref{Assumption: smoothness}(ii)  $\mathbb{E}\{\varepsilon|X,Z\}=0$, we can  replace the operator $\mathbb{G}_n$ with $\mathbb P _n$ in the previous two displays. Also, 
from Lemma \ref{lem: trimming and convergence rates for 1st step}(i), the trimming $t$ in the RHSs of the previous two displays can be replaced with 1 at the cost of an $o_P(1)$ reminder. 
Thus, plugging the previous two displays into (\ref{eq: first term in decomposition of FOC for betahat}) gives 
\begin{equation}\label{eq: first term in decomposition of FOC for betahat II}
    \text{1}^{st}\text{ term RHS of (\ref{eq: decomposition of FOC for betahat})}=\sqrt{n} \mathbb P _n \varepsilon \nabla_\beta G_{W(\beta_0)}(W(\beta_0))+o_P(1)\, .
\end{equation}
Let us now consider the second term RHS of (\ref{eq: decomposition of FOC for betahat}). We have (see the comments below)
\begin{align*}
    \text{2}^{nd}\text{ term RHS of (\ref{eq: decomposition of FOC for betahat})}=& \sqrt{n} \mathbb P _n [G_{W(\beta_0)}(W(\beta_0))- \widehat G _{\widetilde W (\beta_0)}(\widetilde W (\beta_0))]\, \nabla_\beta G_{W(\beta_0)}(W(\beta_0))\, \widehat t\\
 =& \sqrt{n} \mathbb P _n [G_{W(\beta_0)}(W(\beta_0))- \widehat G _{\widetilde W (\beta_0)}(\widetilde W (\beta_0))]\, \nabla_\beta G_{W(\beta_0)}(W(\beta_0))\,  t + o_P(1)\, .
\end{align*}
The first equality follows from the $n^{-1/4}$ rates in Lemma \ref{lem: convergence of the derivatives at the data points}(i) and Lemma \ref{lem: convergence rates for Ghat fhat and iotahat}(i). To obtain the second equality, we have used Lemma \ref{lem: convergence rates for Ghat fhat and iotahat}(i) and the same arguments as in (\ref{eq: trimming replacement}). Now, from \citet[pages 401-403]{klein_efficient_1993} $\mathbb E \{\nabla_\beta G_{W(\beta_0)}(W(\beta_0))|W(\beta_0)\}=0$. So, we can use the same reasoning as for the LHS of Equation (\ref{eq: 2nd term RHS of 2nd decomposition of Asy Emp Process}),  replacing $f_{W(\beta_0)}\varphi^\perp_s$ with $\nabla_\beta G_{W(\beta_0)}(W(\beta_0))$ and $\widehat \beta$ with $\beta_0$, to get 
\begin{align*}
    \sqrt{n} \mathbb P _n [G_{W(\beta_0)}(W(\beta_0))- \widehat G _{\widetilde W (\beta_0)}(\widetilde W (\beta_0))]&\, \nabla_\beta G_{W(\beta_0)}(W(\beta_0))\,  t \\
    &= -\sqrt{n}\mathbb P _n a(Z)[D-H(Z)]+o_P(1)\, ,\\
    \text{ where }a(Z):=\mathbb{E}\{\nabla_\beta G_{W(\beta_0)}(W(\beta_0))\, \partial^T G_{W(\beta_0)}&(W(\beta_0))\,\partial_Hq(\beta_0,X,H(Z))\, | \, Z\}\, .
\end{align*}
Gathering the previous three displays and (\ref{eq: decomposition of FOC for betahat}) concludes the proof.

\end{proof}

\section{Bootstrap Analysis}\label{sec: Bootstrap Analysis}
In this Appendix, we consider statements relative to the joint probability measure of the bootstrap weights and the sample data. 

\subsection{Proof of Proposition \ref{prop: bootstrap test}}
$(i)$ Let us define 
\begin{align}\label{eq: definition of iota hat star}
    \widehat{\iota}^*_s(w):=&\widehat{T}^{\varphi_s}_{\widetilde W^* (\widehat \beta^*)}(w)/\widehat{f}_{\widetilde W^* (\widehat \beta^*)}(w)\, \\
    \text{ and }\widehat{T}^{\varphi_s}_{\widetilde W^* (\widehat \beta^*)}(w):=&\frac{1}{n h^d}\sum_{i=1}^n \varphi_s(X_i,\widetilde{H}^*(Z_i))\,\widehat t _i \,K\left(\frac{w-\widetilde{W}_i^*(\widehat \beta^*)} {h}\right)\, . \nonumber
\end{align}
We have that wpa1(see the comments below)
\begin{align}\label{eq: decomposition of bootstrapped Emp process}
    \sqrt{n} \mathbb P _n [Y^* -& \widetilde{G}^*_{\widetilde W^* (\widehat \beta^*)}(\widetilde W^* (\widehat \beta^*))]\, \widehat{f}_{\widetilde W^* (\widehat \beta^*)}(\widetilde W^* (\widehat \beta^*))\, \varphi_s(X,\widetilde{H}^*(Z))\,\widehat t \nonumber \\
    = &\sqrt{n}\mathbb P _n \widehat{\varepsilon}^*  \widehat{f}_{\widetilde W^* (\widehat \beta^*)}(\widetilde W^* (\widehat \beta^*))\, \varphi_s(X,\widetilde{H}^*(Z))\, \widehat{t} \nonumber  \\
    &- \sqrt{n} \mathbb P _n \widehat{T}^{\widehat \varepsilon ^*}_{\widetilde W^* (\widehat \beta^*)}(\widetilde W^* (\widehat \beta^*))\,  \varphi_s(X,\widetilde{H}^*(Z))\, \widehat t \nonumber \\
    =& \sqrt{n}\mathbb{P}_n \widehat{\varepsilon}^* \widehat{f}_{\widetilde W^* (\widehat \beta^*)}(\widetilde W^* (\widehat \beta^*))\, [\varphi_s(X,\widetilde{H}^*(Z))-\widehat{\iota}^*_s(\widetilde W^* (\widehat \beta^*))]\, \widehat{t}\nonumber \\
    =& \sqrt{n} \mathbb P _n \xi [Y-\widehat{G}_{\widehat{W}(\widehat{\beta})}(\widehat W (\widehat \beta))]\, \widehat{f}_{\widetilde W^*(\widehat \beta^*)}(\widetilde W^*(\widehat \beta^*))\, [\varphi_s(X,\widetilde{H}^*(Z))-\widehat{\iota}_s^*(\widetilde W^* (\widehat \beta^*))] \widehat t \nonumber  \\
    &+ \sqrt{n} \mathbb P _n [\widehat{G}_{\widehat W (\widehat \beta)}(\widehat W (\widehat \beta))-\widehat{G}^*_{\widetilde W^* (\widehat \beta^*)}(\widetilde W^*(\widehat \beta^*))]\widehat{f}_{\widetilde W^*(\widehat \beta^*)}(\widetilde W^*(\widehat \beta^*)) \\
    &\qquad \hspace{7cm} \cdot [\varphi_s(X,\widetilde{H}^*(Z))-\widehat{\iota}_s^*(\widetilde W^*(\widehat \beta^*))]\, \widehat t\, . \nonumber 
\end{align}
Lemma \ref{lem: trimming implications}(ii) ensures that wpa1 $\widehat t (X_i,Z_i)\leq t^\delta(X_i,Z_i)$ uniformly in $i=1,\ldots,n$, where $\delta=1/2$. From Assumption \ref{Assumption: trimming}(i) there exists $\eta(\delta)>0$ such that wpa1 $t^\delta(X_i,Z_i)\leq $ $\mathbb I \{f_{W(\widehat \beta^*)}(q(\widehat \beta^*,$ $X_i,H(Z_i))\geq \eta(\delta)\tau_n\}$ for all $i=1,\ldots,n$. Lemma \ref{lem: trimming implications}(iv) implies that wpa1 $t^\delta(X_i,Z_i)$ $\mathbb I \{f_{W(\widehat \beta^*)}(q(\widehat \beta^*,X_i,H(Z_i))\geq \eta(\delta)\tau_n\}$ $\leq t^\delta (X_i,Z_i)$ $\mathbb I \{f_{W(\widehat \beta^*)}(q(\widehat \beta^*,X_i,\widetilde H^* (Z_i)))\geq \eta(\delta) \tau_n/2\}$ for all $i=1,\ldots,n$ , and using Lemma \ref{lem: trimming implications}(iii) we have that wpa1 $\mathbb I \{f_{W(\widehat \beta^*)}(q(\widehat \beta^*,X_i,\widetilde H^* (Z_i)))\geq \eta(\delta) \tau_n/2\}$ $\leq \mathbb I \{\widehat f_{\widetilde W^*(\widehat \beta^*)}(q(\widehat \beta^*,X_i,\widetilde H ^* (Z_i)))\geq \eta(\delta) \tau_n/4\}$ for all $i=1,\ldots,n$. Gathering results, wpa1 
\begin{equation}\label{eq: trimming inequality for simplification in Bootstrap}
 \widehat t(X_i,Z_i)\leq \mathbb I \{\widehat f _{\widetilde W ^*(\widehat \beta ^*)}(\widetilde W ^* _i(\widehat \beta ^*))\geq \eta(\delta)\tau_n/4\}   
\end{equation}
 uniformly in $i=1,\ldots,n$. This and the definition of $\widetilde G ^* _{\widetilde W ^* (\widehat \beta)}$ in (\ref{eq: definition of G tilde star beta hat star}) gives the first equality in (\ref{eq: decomposition of bootstrapped Emp process}). Turning to the second equality in (\ref{eq: decomposition of bootstrapped Emp process}), by the definition of $\widehat{T}^{\widehat{\varepsilon}^*}_{\widetilde W^* (\widehat \beta^*)}(\widetilde W^* (\widehat \beta^*))$, we can exchange sums and then use (\ref{eq: trimming inequality for simplification in Bootstrap}) and the definition of $\widehat \iota^*_s$ in (\ref{eq: definition of iota hat star}) to get that wpa1
\begin{align*}
    \sqrt{n} \mathbb P _n \widehat{T}^{\widehat \varepsilon ^*}_{\widetilde W^* (\widehat \beta^*)}(\widetilde W^* (\widehat \beta^*))\,  \varphi_s(X,\widetilde{H}^*(Z))\, \widehat t
     = \sqrt{n} \mathbb P _n \widehat{\varepsilon}^* \widehat{f}_{\widetilde W^* (\widehat \beta^*)}(\widetilde W^* (\widehat \beta^*))\, \widehat{\iota}^*_s(\widetilde W^* (\widehat \beta^*))\, \widehat t\, .
\end{align*}
Then, rearranging terms yields the second equality in (\ref{eq: decomposition of bootstrapped Emp process}). Finally, we obtain the third equality in (\ref{eq: decomposition of bootstrapped Emp process}) from the definition of $\widehat{\varepsilon}^*$ and $Y^*$. \\
To show the desired result, we will obtain the IFR of the RHS of (\ref{eq: decomposition of bootstrapped Emp process}). 
We will handle each term separately. Let us start from the first one. By Lemma \ref{lem: trimming and convergence rates for 1st step}(iv) $|\widetilde H ^*(Z_i)|(\widehat t _i + t_i)$ is bounded in probability uniformly in $i=1,\ldots,n$, so $|\varphi_s(X_i,\widetilde H ^* (Z_i))|(\widehat t_i + t_i)$ will also be bounded in probability uniformly in $i=1,\ldots,n$ and $s$. From Lemmas \ref{lem: IFR for betahat} and \ref{lem: IFR for betahat star}, $\widehat \beta-\beta_0=O_P(n^{-1/2})$ and $\widehat \beta^*-\beta_0=O_P(n^{-1/2})$. So, thanks to Lemma \ref{lem: convergence rates for Ghat fhat and iotahat}(i)(xxiii)(xxv), we have that $|\widehat G _{\widehat W (\widehat \beta)}(\widehat W _i (\widehat \beta))|(\widehat t_i + t_i)$ , $|\widehat f _{\widetilde W^* (\widehat \beta^*)}(\widetilde W^* _i(\widehat \beta^*))|(\widehat t_i + t_i)$, and $|\widehat \iota_s^*(\widetilde W ^*_i(\widehat \beta^*))|(\widehat t_i+t_i)$ are all bounded in probability uniformly in $i=1,\ldots,n$ and $s$. Then,
by using the same reasoning as in (\ref{eq: trimming replacement}) we can replace the trimming $\widehat{t}$ with $t$ at the cost of an $o_P(1)$ reminder and obtain that uniformly in $s$
\begin{align}\label{eq: trimming replacement in 1st term of the bootstrap decomposition}
    \sqrt{n} \mathbb P _n \xi &[Y-\widehat{G}_{\widehat{W}(\widehat{\beta})}(\widehat W (\widehat \beta))]\, \widehat{f}_{\widetilde W^*(\widehat \beta^*)}(\widetilde W^*(\widehat \beta^*))\, [\varphi_s(X,\widetilde{H}^*(Z))-\widehat{\iota}_s^*(\widetilde W^* (\widehat \beta^*))] \widehat t \\
    =&  \sqrt{n} \mathbb P _n \xi [Y-\widehat{G}_{\widehat{W}(\widehat{\beta})}(\widehat W (\widehat \beta))]\, \widehat{f}_{\widetilde W^*(\widehat \beta^*)}(\widetilde W^*(\widehat \beta^*))\, [\varphi_s(X,\widetilde{H}^*(Z))-\widehat{\iota}_s^*(\widetilde W^* (\widehat \beta^*))]  t+o_P(1)\, .  \nonumber 
\end{align}
Now, from Lemma \ref{lem: trimming and convergence rates for 1st step}(iv) $\max_{i=1,\ldots,n}|\widetilde{H}^*(Z_i)-H(Z_i)|t_i=o_P(n^{-1/4})$, while from Lemmas \ref{lem: IFR for betahat star} and \ref{lem: IFR for betahat} $\widehat{\beta}^*-\beta_0=O_P(n^{-1/2})$. Hence, by a Mean-Value expansion of $q(\widehat{\beta}^*,X_i,\widetilde{H}^*(Z_i))$ around $(\beta_0,X_i,H(Z_i))$ and the Lipschitz continuity of $\partial_\beta q$ and $\partial_H q $ (see Assumption \ref{Assumption: smoothness}(v)), we get
\begin{align}\label{eq: bootstrap mean value expansion of q}
    q(\widehat{\beta}^*,X_i,\widetilde{H}^*(Z_i))-q(\beta_0,X_i,H(Z_i))]t_i=&\partial_{\beta^T} q(\beta_0,X_i,H(Z_i))(\widehat \beta ^* - \beta_0) t_i \\
    &+ \partial_H q(\beta_0,X_i,H(Z_i)) [\widetilde{H}^*(Z_i)-H(Z_i)]t_i + o_P(n^{-1/2})\nonumber 
\end{align}
uniformly in $i=1,\ldots,n$. This yields
\begin{align*}
    \max_{i=1,\ldots,n}|\widetilde W ^*_i (\widehat \beta^*)-W_i(\beta_0)|t_i=&O_P\left(\max_{i=1,\ldots,n}|\widetilde H ^*(Z_i)-H(Z_i)|t_i\right)+O_P(\|\widehat \beta^* - \beta_0\|)\\ =&o_P(n^{-1/4})\, .\nonumber 
\end{align*}
Similarly, 
\begin{align*}
   \max_{i=1,\ldots,n}|\widehat W _i (\widehat \beta)-W_i(\beta_0)|t_i=o_P(n^{-1/4})\, . 
\end{align*}
 Also, Lemma \ref{lem: convergence rates for Ghat fhat and iotahat}(viii)(xvii)(xix) implies that $|\widehat G _{\widehat W (\widehat \beta)}( W_i ( \beta_0))|t_i$ , $|\widehat f _{\widetilde W ^* (\widehat \beta ^*)}(W_i(\beta_0))|t_i$, and $|\widehat \iota ^* _{s}(W_i(\beta_0))|t_i$ are bounded in probability uniformly in $i=1,\ldots,n$ and $s$. Hence, using these results, the previous two displays, and plugging the stochastic expansions from Lemma \ref{lem: stochastic expansions}(i)(vii)(viii)(ix) into (\ref{eq: trimming replacement in 1st term of the bootstrap decomposition}) gives 
\begin{align}\label{eq: decomposition of the 1st term of the bootstrap Emp Proc}
\sqrt{n} \mathbb P _n \xi [Y-&\widehat{G}_{\widehat{W}(\widehat{\beta})}(\widehat W (\widehat \beta))]\, \widehat{f}_{\widetilde W^*(\widehat \beta^*)}(\widetilde W^*(\widehat \beta^*))\, [\varphi_s(X,\widetilde{H}^*(Z))-\widehat{\iota}_s^*(\widetilde W^* (\widehat \beta^*))]  t \nonumber  \\ 
=& \sqrt{n} \mathbb P _n \xi[Y-\widehat{G}_{\widehat{W}(\widehat \beta)}(W(\beta_0))]\widehat{f}_{\widetilde W ^* (\widehat \beta ^*)}(W(\beta_0)) \varphi_s(X,H(Z)) t \nonumber \\
    &- \sqrt{n} \mathbb P _n \xi[Y-\widehat{G}_{\widehat{W}(\widehat \beta)}(W(\beta_0))]\widehat{f}_{\widetilde W ^* (\widehat \beta ^*)}(W(\beta_0)) \widehat{\iota}_s^*(W(\beta_0))  t \nonumber \\
 &+ \sqrt{n} \mathbb P_n \xi[Y-\widehat{G}_{\widehat{W}(\widehat \beta)}(W(\beta_0))]  [\varphi_s(X,H(Z))-\widehat{\iota}^*_s(W(\beta_0))] t \nonumber \\
 &\hspace{8cm}\partial^T f_{W(\beta_0)}(W(\beta_0))[\widetilde{W}^* (\widehat \beta ^*)-W(\beta_0)]\nonumber \\
     &+ \sqrt{n} \mathbb P_n \xi[Y-\widehat{G}_{\widehat{W}(\widehat \beta)}(W(\beta_0))] \widehat{f}_{\widetilde W ^* (\widehat \beta ^*)}(W(\beta_0)) \partial_H \varphi_s(X,H(Z))\, [\widetilde{H}^*(Z)-H(Z)] \,  t \nonumber \\
     &-\sqrt{n} \mathbb P_n \xi[Y-\widehat{G}_{\widehat{W}(\widehat \beta)}(W(\beta_0))] \widehat{f}_{\widetilde W ^* (\widehat \beta ^*)}(W(\beta_0)) \partial^T \iota_s(W(\beta_0))\, [\widetilde W ^* (\widehat \beta^*)-W(\beta_0)]  t \nonumber \\
     &- \sqrt{n} \mathbb{P}_n \xi  \widehat{f}_{\widetilde W ^* (\widehat \beta ^*)}(W(\beta_0)) [\varphi_s(X,H(Z))-\widehat{\iota}^*_s(W(\beta_0))]  t \nonumber \\
     &\hspace{8cm}\partial^T G_{W(\beta_0)}(W(\beta_0)) [\widehat W (\widehat \beta)-W(\beta_0)]\nonumber \\
     &+ o_P(1) .
\end{align}
Let us deal with each term on the RHS of (\ref{eq: decomposition of the 1st term of the bootstrap Emp Proc}). For the first one we have (see the comments below)
\begin{align}\label{eq: 1st term RHS of  decomposition of the 1st term of the bootstrap Emp Proc}
    \text{1}^{st} \text{ term RHS of (\ref{eq: decomposition of the 1st term of the bootstrap Emp Proc})}=& \sqrt{n}\mathbb{P}_n \xi [Y - \widehat{G}_{\widehat W (\widehat \beta)}(W(\beta_0))]\widehat{f}_{\widetilde{W}^*(\widehat \beta ^*)}(W(\beta_0))\, \varphi_s(X,H(Z)) t^\eta_{W(\beta_0)}+o_P(1)\nonumber\\
    =& \sqrt{n}\mathbb{P}_n \xi \varepsilon f_{W(\beta_0)} \varphi_s(X,H(Z))\,t^\eta_{W(\beta_0)}+o_P(1)\nonumber \\
    =& \sqrt{n}\mathbb{P}_n \xi \varepsilon f_{W(\beta_0)} \varphi_s(X,H(Z))\,+o_P(1)\, 
\end{align}
uniformly in $s\in\cal S$. The first equality follows from the same arguments as in (\ref{eq: first ASE for asy behavior}). To obtain the second equality, notice first that the leading term is a centered empirical process, as $\mathbb E \xi=0$ and the bootstrap weights $\{\xi_i:i=1,\ldots,n\}$ are independent from the sample data. By Lemma \ref{lem: convergence rates for Ghat fhat and iotahat}(xi)(xx) $\|[\widehat{G}_{\widehat W  (\widehat \beta)}-G_{W(\beta_0)}]t^\eta_{W(\beta_0)}\|_{\infty,\mathcal{W}}=o_P(1)$ and $\|[\widehat{f}_{\widetilde W ^* (\widehat \beta^*)}-f_{W(\beta_0)}]t^\eta_{W(\beta_0)}\|_{\infty, \mathcal{W}}=o_P(1)$. Also, $\widehat{G}_{\widehat W (\widehat \beta)}, \widehat{f}_{\widetilde W^* (\widehat \beta ^*)}\in \mathcal{G}_\lambda(\mathcal{W}_{n,\beta_0}^{\eta/2})$ wpa1 with $\lambda=\lceil (d+1)/2\rceil $, see Lemma \ref{lem: belonging conditions}(v)(xi). So, we can apply the stochastic equicontinuity result in Lemma \ref{lem: ASE}(i) to get the second equality. Finally, by Lemma \ref{lem: trimming implications}(i) we get the third equality.

By the same reasoning we get 
\begin{align}
    \text{2}^{nd} \text{ term RHS of (\ref{eq: decomposition of the 1st term of the bootstrap Emp Proc})}= \sqrt{n}\mathbb{P}_n \xi \varepsilon f_{W(\beta_0)} \iota_s(W(\beta_0))+o_P(1) \, .
\end{align}
Let us now consider the third term on the RHS of (\ref{eq: decomposition of the 1st term of the bootstrap Emp Proc}). Below we show that 
\begin{align}\label{eq: 3rd term RHS of 1st term of Bootstrap Emp proc}
    \text{3}^{rd}  \text{ term RHS of (\ref{eq: decomposition of the 1st term of the bootstrap Emp Proc})}=&\sqrt{n} \mathbb P _n \xi \varepsilon \partial^T f_{W(\beta_0)} [\widetilde W ^* (\widehat \beta ^*) - W(\beta_0)] t \varphi_s^\perp+o_P(1) \nonumber \\
    =& \mathbb P _n \xi \varepsilon \varphi_s^\perp \partial^T f_{W(\beta_0)} \partial_{\beta^T} q(\beta_0,X,H(Z)) t \sqrt{n}(\widehat \beta^* - \beta_0) \nonumber  \\
    & + \sqrt{n} \mathbb P _n \xi \varepsilon \varphi_s^\perp \partial^T f_{W(\beta_0)} \partial_H q(\beta_0,X,H(Z))\, [\widetilde{H}^*(Z)-H(Z)]\ t + o_P(1) \nonumber \\
    = & o_P(1) \, .
\end{align}
As already obtained earlier in this proof, $\max_{i=1,\ldots,n} |\widetilde W^*_i (\widehat \beta ^*)-W_i(\beta_0)| t_i=o_P(n^{-1/4})$. Also, by Lemma \ref{lem: convergence rates for Ghat fhat and iotahat}(viii)(xix) $|\widehat{G}_{\widehat W (\widehat \beta)}(W_i(\beta_0))-G_{0}(W_i(\beta_0))|t_i=o_P(n^{-1/4})$ and 
$|\widehat{\iota}_s^*(W_i(\beta_0))-\iota_s(W_i(\beta_0))|t_i=o_P(n^{-1/4})$ uniformly in $i=1,\ldots,n$ and $s$. This gives the first equality in (\ref{eq: 3rd term RHS of 1st term of Bootstrap Emp proc}). To get the second equality in (\ref{eq: 3rd term RHS of 1st term of Bootstrap Emp proc}), we use the expansion in (\ref{eq: bootstrap mean value expansion of q}). Let us now focus on the third equality in (\ref{eq: 3rd term RHS of 1st term of Bootstrap Emp proc}). By Lemma \ref{lem: trimming and convergence rates for 1st step}(i) $\mathbb P _n \xi \varepsilon \varphi^\perp _s $ $\partial^T f_{W(\beta_0)} \partial_{\beta^T} q(\beta_0,X,H(Z))t$ $=\mathbb P _n \xi \varepsilon $ $ \varphi^\perp _s \partial^T f_{W(\beta_0)} \partial_{\beta^T} q(\beta_0,X,H(Z))+o_P(1)$ uniformly in $s$. Then, a Glivenko-Cantelli property of $\varphi_s^\perp$ gives\footnote{ Assumption \ref{Assumption: iid and phi}(iv) ensures that $\varphi_s$ is uniformly Lipschitz in $s\in\mathcal{S}$, and hence so is $\varphi_s^\perp$. }
\begin{align*}
\mathbb P _n \xi \varepsilon \varphi^\perp _s \partial^T f_{W(\beta_0)} \partial_{\beta^T} q(\beta_0,X,H(Z))+o_P(1)=&\mathbb E \xi \varepsilon \varphi^\perp _s \partial^T f_{W(\beta_0)} \partial_{\beta^T} q(\beta_0,X,H(Z))+o_P(1)\\
=&o_P(1)  \, 
\end{align*}
uniformly in $s\in\mathcal{S}$,
where the last equality follows from $\mathbb E \xi=0$ and the independence of $\{\xi_i:i=1,\ldots,n\}$ from the sample data. Thus,  since $\sqrt{n}(\widehat \beta ^* - \beta_0)=O_P(1)$ from Lemmas \ref{lem: IFR for betahat} and \ref{lem: IFR for betahat star}, the first term after the second equality of (\ref{eq: 3rd term RHS of 1st term of Bootstrap Emp proc}) is $o_P(1)$. Finally, using the stochastic expansions in Lemma \ref{lem: ASE for 1st step}(ii)(iii) and $\mathbb E \xi =0$ gives that the second term after the second equality of (\ref{eq: 3rd term RHS of 1st term of Bootstrap Emp proc}) is also $o_P(1)$.   \\
Arguments similar to those used for (\ref{eq: 3rd term RHS of 1st term of Bootstrap Emp proc}) yield
\begin{align}
    \text{4}^{th}  \text{ term RHS of (\ref{eq: decomposition of the 1st term of the bootstrap Emp Proc})}=&o_P(1)\\
    \text{5}^{th}  \text{ term RHS of (\ref{eq: decomposition of the 1st term of the bootstrap Emp Proc})}=&o_P(1)\, .
\end{align}
Finally, for the sixth term on the RHS of (\ref{eq: decomposition of the 1st term of the bootstrap Emp Proc}) we have (see the comments below)
\begin{align}\label{eq: 6th term RHS of eq: 1st term RHS of  decomposition of the 1st term of the bootstrap Emp Proc}
    \text{6}^{th}  \text{ term RHS of (\ref{eq: decomposition of the 1st term of the bootstrap Emp Proc})}=& \sqrt{n} \mathbb P _n \xi \partial^T G_{W(\beta_0)} [\widehat W (\widehat \beta) - W(\beta_0)] f_{W(\beta_0)}\varphi_s^\perp +o_P(1) \nonumber \\
    =& o_P(1)\, .
\end{align}
As already noticed earlier, 
$\max_{i=1,\ldots,n}|\widehat W_i (\widehat \beta)-W_i(\beta_0)|t_i=o_P(n^{-1/4})$. Also, from Lemma \ref{lem: convergence rates for Ghat fhat and iotahat}(xvii)(xix) $|\widehat f _{\widetilde W^* (\widehat \beta ^*)}(W_i(\beta_0)) - f_{W(\beta_0)}(W_i(\beta_0))|t_i=o_P(n^{-1/4})$ and 
$|\widehat \iota_s^*(W_i(\beta_0)) - \iota_s(W_i(\beta_0))|t_i$ $=o_P(n^{-1/4})$ uniformly in $i=1,\ldots,n$ and $s$. Using these rates leads to the first equality of (\ref{eq: 6th term RHS of eq: 1st term RHS of  decomposition of the 1st term of the bootstrap Emp Proc}). To get the second equality of (\ref{eq: 6th term RHS of eq: 1st term RHS of  decomposition of the 1st term of the bootstrap Emp Proc}), we can proceed as in (\ref{eq: approximation  of W.tilde(betahat) - W(beta0)}) and use the fact that the bootstrap weights $\{\xi_i:i=1,\ldots,n\}$ are mean-zero and independent from the sample data.  

Plugging (\ref{eq: 1st term RHS of  decomposition of the 1st term of the bootstrap Emp Proc})-(\ref{eq: 6th term RHS of eq: 1st term RHS of  decomposition of the 1st term of the bootstrap Emp Proc}) into (\ref{eq: decomposition of the 1st term of the bootstrap Emp Proc}) and then using (\ref{eq: trimming replacement in 1st term of the bootstrap decomposition}) gives 
\begin{align}\label{eq: IFR of the 1st term of the bootstrap decomposition}
    \text{1}^{st} \text{ term RHS of (\ref{eq: decomposition of bootstrapped Emp process})}=\sqrt{n}\mathbb P _n \xi \varepsilon f_{W(\beta_0)} \varphi_s^\perp +o_P(1)\, .
\end{align}
In view of the above display, to prove the desired result we only need to get an expansion for the second term on the RHS of (\ref{eq: decomposition of bootstrapped Emp process}). By Lemma \ref{lem: trimming and convergence rates for 1st step}(iv) $\|(\widetilde H ^* - H)\widehat t\|_\infty=o_P(n^{-1/4})$. Thus,  by Assumptions \ref{Assumption: iid and phi}(iv) and \ref{Assumption: smoothness}(v) 
 $|\varphi_s(X_i,\widetilde H ^* (Z_i))-\varphi_s(X_i,H(Z_i))|\widehat t _i=O_P(\|(\widetilde H ^* - H)\widehat t\|_\infty)=o_P(n^{-1/4})$. Combining this with Lemma \ref{lem: convergence rates for Ghat fhat and iotahat}(i)(xxiii)(xxiv)(xxv) gives 
\begin{align}\label{eq: 2nd term RHS of eq: decomposition of bootstrapped Emp process}
    \text{2}^{nd} \text{ term RHS of (\ref{eq: decomposition of bootstrapped Emp process})}=& \sqrt{n} \mathbb P _n [\widehat G _{\widehat W (\widehat \beta)}(\widehat W (\widehat \beta))-\widehat G ^* _{\widetilde W ^* (\widehat \beta ^*)}(\widetilde W ^* (\widehat \beta ^*))]f_{W(\beta_0)} \varphi^\perp _s \widehat t + o_P(1)\, .
\end{align}
Next, using Lemma \ref{lem: convergence rates for Ghat fhat and iotahat}(i)(xxiv) and the same reasoning as in (\ref{eq: trimming replacement})
 allows us to replace the trimming $\widehat t$ with $t$ at the cost of an $o_P(1)$ reminder. Thus, uniformly in $s$
 \begin{align*}
     \sqrt{n} \mathbb P _n [\widehat G _{\widehat W (\widehat \beta)}(\widehat W (\widehat \beta)&)-\widehat G ^* _{\widetilde W ^* (\widehat \beta ^*)}(\widetilde W ^* (\widehat \beta ^*))]f_{W(\beta_0)} \varphi^\perp _s \widehat t\\
     &= \sqrt{n} \mathbb P _n [\widehat G _{\widehat W (\widehat \beta)}(\widehat W (\widehat \beta))-\widehat G ^* _{\widetilde W ^* (\widehat \beta ^*)}(\widetilde W ^* (\widehat \beta ^*))]f_{W(\beta_0)} \varphi^\perp _s  t+o_P(1)\, .
 \end{align*}
 Next, an application of Lemma \ref{lem: stochastic expansions}(i)(vi) gives
\begin{align}\label{eq: 2nd term of decomposition of bootstrapped Emp process II }
    \sqrt{n} \mathbb P _n [\widehat G _{\widehat W (\widehat \beta)}(\widehat W (\widehat \beta))-\widehat G ^* _{\widetilde W ^* (\widehat \beta ^*)}(\widetilde W ^*& (\widehat \beta ^*))]f_{W(\beta_0)} \varphi^\perp _s t \nonumber  \\
    = & \sqrt{n} \mathbb P _n  [\widehat G _{\widehat W (\widehat \beta)}( W ( \beta_0))-\widehat G ^* _{\widetilde W ^* (\widehat \beta ^*)}( W ( \beta_0))]f_{W(\beta_0)} \varphi^\perp _s t \nonumber \\
    &- \sqrt{n} \mathbb P _n \partial^T G_{W(\beta_0)}(W(\beta_0)) [\widetilde W ^* (\widehat{\beta}^*) - \widehat W(\widehat \beta)] t f_{W(\beta_0)} \varphi_s^\perp \nonumber \\
    &+ o_P(1)\, .
\end{align}
Let us consider the first leading term on the RHS of the above display. We have that uniformly in $s$ (see the comments below)
\begin{align*}
    \sqrt{n} \mathbb P _n  & [\widehat G _{\widehat W (\widehat \beta)}( W ( \beta_0))-\widehat G ^* _{\widetilde W ^* (\widehat \beta ^*)}( W ( \beta_0))]f_{W(\beta_0)} \varphi^\perp _s t\\
    =& \sqrt{n} \mathbb P _n  [\widehat G _{\widehat W (\widehat \beta)}( W ( \beta_0))-\widehat G ^* _{\widetilde W ^* (\widehat \beta ^*)}( W ( \beta_0))]f_{W(\beta_0)} \varphi^\perp _s t^\eta_{W(\beta_0)}+o_P(1)\\
    =& o_P(1)\, .
\end{align*}
The same arguments as in (\ref{eq: first ASE for asy behavior}) allow us to replace the trimming $t$ with $t^\eta_{W(\beta_0)}$ at the cost of an $o_P(1)$ reminder. This gives the first equality in the previous display.
For the second equality, notice that the leading term is a centered empirical process, as $\mathbb{E}\{ \varphi^\perp_s|W(\beta_0)\}=0$. By Lemma \ref{lem: belonging conditions}(v)(x) $\widehat G _{\widehat W (\widehat \beta)}-\widehat G ^* _{\widetilde W ^* (\widehat \beta ^*)}\in  \mathcal{G}_\lambda(\mathcal{W}_{n,\beta_0}^{\eta/2})$ wpa1, with $\lambda=\lceil (d+1)/2 \rceil$, while by Lemma \ref{lem: convergence rates for Ghat fhat and iotahat}(xi)(xxi)  $\|(\widehat G _{\widehat W (\widehat \beta)}-\widehat G ^* _{\widetilde W ^* (\widehat \beta ^*)})t^\eta_{W(\beta_0)}\|_{\infty,\mathcal{W}}=o_P(1)$. So, the conditions of Lemma \ref{lem: ASE}(i) are satisfied and the above empirical process is $o_P(1)$ uniformly in $s\in \cal S$. This proves the second equality of the previous display. 

Let us now consider the second term on the RHS of (\ref{eq: 2nd term of decomposition of bootstrapped Emp process II }). A Mean-Value expansion of $q(\widehat{\beta}^*,X_i,\widetilde{H}^*(Z_i))$ around $(\widehat{\beta},X_i, \widehat{H}(Z_i))$ gives
\begin{align*}
    [q(\widehat{\beta}^*,X_i,\widetilde{H}^*(Z_i)) - q(\widehat{\beta},X_i, \widehat{H}(Z_i)) ] t_i =& \partial_{\beta^T} q(\overline{\beta}, X_i, \overline{H(Z_i)}) (\widehat{\beta}^*-\widehat{\beta})t_i\\ 
    &+ \partial_H q(\overline{\beta}, X_i, \overline{H(Z_i)}) t_i [\widetilde{H}^*(Z_i) - \widehat{H}(Z_i)]
\end{align*}
where $\overline{H(Z_i)}$ lies on the segment between $\widetilde{H}^*(Z_i)$ and $\widehat{H}(Z_i)$, and $\overline{\beta}$ lies on the segment between $\widehat \beta ^*$ and $\widehat \beta$.  Now, 
$\max_{i=1,\ldots,n} |\widetilde{H}^*(Z_i)$ $- H(Z_i)|t_i$ $=o_P(n^{-1/4})$ and $ \max_{i=1,\ldots,n}|\widehat{H}(Z_i)$ $- H(Z_i)|t_i$ $=o_P(n^{-1/4})$, see Lemma \ref{lem: trimming and convergence rates for 1st step}(ii)(iv). Thus, $\max_{i=1,\ldots,n} |\overline{H(Z_i)}$ $- H(Z_i)|t_i$ $=o_P(n^{-1/4})$.
Also, $\widehat{\beta}^*-\beta_0=O_P(n^{-1/2})$ and $\widehat{\beta}-\beta_0=O_P(n^{-1/2})$, see Lemmas \ref{lem: IFR for betahat}-\ref{lem: IFR for betahat star}. Thus, $\overline{\beta}-\beta_0=O_P(n^{-1/2})$.
So, by the previous display and the Lipschitz continuity of $\partial_H q$ and $\partial_\beta q$ (see Assumption \ref{Assumption: smoothness}(v)) we have 
\begin{align*}
    [q(\widehat{\beta}^*,X_i,\widetilde{H}^*(Z_i)) - q(\widehat{\beta},X_i, \widehat{H}(Z_i)) ] t_i =& \partial_{\beta^T} q(\beta_0, X_i, H(Z_i)) (\widehat{\beta}^*-\widehat{\beta})t_i\\ 
    &+ \partial_H q(\beta_0, X_i, H(Z_i)) t_i [\widetilde{H}^*(Z_i) - \widehat{H}(Z_i)]+o_P(n^{-1/2})
\end{align*}
uniformly in $i=1,\ldots,n$. Such an expansion gives 
\begin{align*}
    \sqrt{n} \mathbb P _n \partial^T G_{W(\beta_0)} [\widetilde W ^*& (\widehat \beta ^*) - \widehat W (\widehat \beta)] t f_{W(\beta_0)} \varphi_s^\perp\\ =&  \mathbb P _n f_{W(\beta_0)} \varphi_s^\perp \partial^T G_{W(\beta_0)} \partial_{\beta^T} q(\beta_0,X,H(Z)) t \sqrt{n} (\widehat \beta^* - \widehat \beta) \nonumber \\
    &+ \sqrt{n} \mathbb P _n f_{W(\beta_0)}  \varphi_s^\perp \partial^T G_{W(\beta_0)} \partial_H q(\beta_0,X,H(Z)) t [\widetilde{H}^*(Z_i) - \widehat{H}(Z_i)]\\
    &+ o_P(1)\, 
\end{align*}
uniformly in $s\in\mathcal{S}$. 
By arguments used earlier in this proof, we obtain $\mathbb P _n f_{W(\beta_0)} \varphi_s^\perp \partial^T G_{W(\beta_0)} $ $\partial_{\beta^T} q$ $ t$ $=\mathbb E f_{W(\beta_0)} \varphi_s^\perp \partial^T G_{W(\beta_0)} \partial_{\beta^T} q$ $+o_P(1)$ uniformly in $s\in\mathcal{S}$. Then, using $\sqrt{n}(\widehat{\beta}^*-\widehat \beta)=O_P(1)$ for the first leading term (see  Lemma \ref{lem: IFR for betahat star}) and Lemma \ref{lem: ASE for 1st step}(iii) for the second leading term gives that, uniformly in $s\in\mathcal{S}$, the RHS of the previous display equals 
\begin{align*}
    &\mathbb E \{f_{W(\beta_0)}\varphi_s^\perp \partial^T G_{W(\beta_0)}\partial_{\beta^T} q (\beta_0,X,H(Z))\}\sqrt{n}(\widehat \beta ^* - \widehat \beta)\\
    &+\sqrt{n}\mathbb P _n a_s(Z)\xi[D-H(Z)]+o_P(1)\\
     \text{ with }a_s(Z)&=\mathbb E \{f_{W(\beta_0)}\varphi_s^\perp \partial^T G_{W(\beta_0)}\partial_H q |Z\}\, .
\end{align*}
Gathering results, 
\begin{align*}
    \text{2}^{nd} \text{ term RHS of (\ref{eq: decomposition of bootstrapped Emp process})}=& -\mathbb E \{f_{W(\beta_0)}\varphi_s^\perp \partial^T G_{W(\beta_0)}\partial_{\beta^T} q (\beta_0,X,H(Z))\}\sqrt{n}(\widehat \beta ^* - \widehat \beta)\\
    &-\sqrt{n}\mathbb P _n a_s(Z)\xi[D-H(Z)]+o_P(1)\, 
\end{align*}
uniformly in $s\in\mathcal{S}$. Finally, using this expansion, (\ref{eq: IFR of the 1st term of the bootstrap decomposition}), (\ref{eq: decomposition of bootstrapped Emp process}), and the IFR of $\sqrt{n}(\widehat \beta ^* - \widehat \beta)$ from Lemma \ref{lem: IFR for betahat star}  concludes the proof.\\

$(ii)$ From Assumption \ref{Assumption: iid and phi}(iii)(iv) $\varphi_s$ is Lipschitz in $s\in\mathcal{S}$. So, the class of functions indexing the IFR in Proposition \ref{prop: Asymptotic Test} is Donsker. Thus, by the expansion in Proposition \ref{prop: Asymptotic Test}, the empirical process at the basis of $S_n$ converges weakly to $\mathbb{G}_s$ in $\ell^\infty(\mathcal{S})$, where $\mathbb{G}_s$ is a zero-mean tight Gaussian process  characterized by the collection of covariances in (\ref{eq: collection of covariances}). This implies, by the  continuity of the Cramer-Von Mises functional, that $S_n$ converges in distribution to $\int_{ }|\mathbb{G}_s|^2 d\,\mu(s)$. 
As shown in  \citet{bentkus_asymptotic_1993}, $\int_{ }|\mathbb{G}_s|^2 d\,\mu(s)$ has a continuous distribution. \\
Thanks to \citet[Theorem 2.9.6]{van_der_vaart_weak_1996}, we obtain that, conditionally on the initial sample, the IFR in Proposition \ref{prop: bootstrap test}(i) converges weakly in probability to $\mathbb{G}_s$. So, using the expansion in Proposition \ref{prop: bootstrap test}(i), the bootstrapped statistic $S_n^*$ converges weakly in probability to $\int_{ }|\mathbb{G}_s|^2 d\,\mu(s)$ conditionally on the initial sample. Thus, (ii) follows.\\

$(iii)$ Let us begin by obtaining the asymptotic behavior of $\widehat \beta$ under $\mathcal{H}_1$. By Assumption \ref{Assumption: parametric} there exists a unique $\beta\in B$ that minimizes $\mathbb E [Y-G_{W(\beta)}(W(\beta))]^2$ under $\mathcal H _1$. Let us denote such a pseudo-true value with $\beta^{pt}$. Then,  the arguments used at the beginning of Lemma \ref{lem: IFR for betahat} give that $\widehat \beta = \beta^{pt}+o_P(1)$. Specifically, such arguments do not rely on $\mathcal{H}_0$ and thus remain valid also under $\mathcal{H}_1$. \\
We will now get the asymptotic behavior of $S_n/n$ under $\mathcal{H}_1$. Since the arguments in (\ref{eq: 1st decomposition of Asy Emp Process}) and (\ref{eq: 2nd decomposition of Asy Emp Process}) remain valid also under $\mathcal H _1$, we get that uniformly in $s\in\mathcal{S}$
\begin{align*}
    \mathbb P _n \widetilde{\varepsilon}\, \widehat{f}&_{\widetilde W (\widehat \beta)}(\widetilde W (\widehat \beta))\,\varphi_s(X,\widetilde{H}(Z))\, \widehat{t} \nonumber \\
    = &  \mathbb P _n [Y - \widehat G _{\widetilde W (\widehat \beta)}(\widetilde W (\widehat \beta))] \widehat{f}_{\widetilde W (\widehat \beta)}(\widetilde W (\widehat \beta))  \, [\varphi_s(X,\widetilde{H}(Z))-  \widehat{\iota}_s(\widetilde W (\widehat \beta))] \widehat t \, .
\end{align*}
Below we show that uniformly in $s\in\mathcal{S}$
\begin{align*}
    \mathbb P _n [Y-\widehat G _{\widetilde W (\widehat \beta)}&(\widetilde W (\widehat \beta))]\,\widehat f _{\widetilde W (\widehat \beta)}(\widetilde W (\widehat \beta))\,[\varphi_s(X,\widetilde H (Z))-\widehat \iota_s(\widetilde W (\widehat \beta))]\, \widehat t\\
    =& \mathbb P _n [Y-G_{W(\widehat \beta)}(W(\widehat \beta))]\,f_{W(\widehat \beta)}(W(\widehat \beta))\,[\varphi_s(X,H(Z))-\mathbb{E}\{\varphi_s|W(\beta)\}|_{\beta=\widehat \beta}]\, \widehat t + o_P(1)\\
    =& \mathbb P _n [Y-G_{W(\widehat \beta)}(W(\widehat \beta))]\,f_{W(\widehat \beta)}(W(\widehat \beta))\,[\varphi_s(X,H(Z))-\mathbb{E}\{\varphi_s|W(\beta)\}|_{\beta=\widehat \beta}] + o_P(1)\\
    =& \mathbb E [Y-G_{W(\widehat \beta)}(W(\widehat \beta))]\, f_{W(\widehat \beta)}(W(\widehat \beta))\,[\varphi_s(X,H(Z))-\mathbb{E}\{\varphi_s|W( \beta)\}|_{\beta=\widehat \beta}] + o_P(1)\\
    =& \mathbb E [Y-G_{W(\beta^{pt})}(W(\beta^{pt}))]\, f_{W(\beta^{pt})}(W(\beta^{pt}))\,[\varphi_s(X,H(Z))-\iota_s(W(\beta^{pt}))]\\
    & + o_P(1)\, ,
\end{align*}
where, with an abuse of notation, the expectation after the third equality considers as random only $(Y,X,Z)$ but not $\widehat \beta$.
 To obtain the first equality we have used $\sup_{s\in\mathcal{S}}\max_{i=1,\ldots,n}$ $|\varphi_s(X_i,$$\widetilde H (Z_i))-\varphi_s(X_i,H(Z_i))|\widehat t _i=O_P(\|(\widetilde H - H)\widehat t_i\|_\infty)$ $=o_P(n^{-1/4})$, as noticed earlier, and Lemma \ref{lem: convergence rates for Ghat fhat and iotahat}(i)(ii)(iii). For the second equality we have replaced $\widehat t$ with 1 at the cost of an $o_P(1)$ reminder, thanks to Lemma \ref{lem: trimming and convergence rates for 1st step}(i). Turning to the third equality, we have used a slight abuse of notation, as the expectation $\mathbb{E}$ only integrates with respect to $Y,X,Z$ but considers $\widehat \beta$ fixed. By Assumptions \ref{Assumption: smoothness}(i)(v), $G_{W(\beta)}(W(\beta))$ is Lipschitz in $\beta$, and by Assumption \ref{Assumption: smoothness}(v)(vi) $f_{W(\beta)}(W(\beta))$ and $\mathbb{E}\{\varphi_s|W(\beta)\}$ are also Lipschitz in $\beta,s$. Then, the third equality follows from a Glivenko-Cantelli Theorem.  
 Finally, the last equality is obtained from $\widehat \beta=\beta^{pt}+o_P(1)$ and the Lipschitz conditions just mentioned. \\ 
 Next, by the above display and since the Cramer-Von Mises functional is continuous on $\ell^\infty(\mathcal S)$ with respect to the uniform metric, a Continuous Mapping Theorem gives 
\begin{equation}\label{eq: convergence of S n under H1}
    \frac{S_n}{n} =\int_{ }\left|\mathbb E [Y-G_{W(\beta^{pt})}(W(\beta^{pt}))]\, f_{W(\beta^{pt})}(W(\beta^{pt}))\,[\varphi_s(X,H(Z))-\iota_s(W(\beta^{pt}))]\right|^2 d \mu(s) + o_P(1)\, .
\end{equation}
From \citet[Theorem 2.2]{bierens_econometric_2017}, under $\mathcal{H}_1$ the argument of the above integral is non-null for almost all $s\in \mathcal S$, so the integral will be strictly positive.\\
Let us now obtain the behavior of $\widehat \beta ^*$ under $\mathcal H _1$. We have that (see the comments below)
\begin{align}\label{eq: limit of the bootstrapped SLS criterion}
    \mathbb P _n [Y^*-\widehat G ^*_{\widetilde W ^* (\beta)}(\widetilde W ^* (\beta))]^2 \, \widehat t =& \mathbb P _n [G_{W(\beta^{pt})}(W(\beta^{pt})) + \xi (Y-G_{W(\beta^{pt})}(W(\beta^{pt})))\nonumber \\
    &\hspace{2cm}-\mathbb E \{G_{W(\beta^{pt})}(W(\beta^{pt}))|W(\beta)\}]^2\, \widehat t \, +o_P(1)\nonumber \\
    =& \mathbb E [G_{W(\beta^{pt})}(W(\beta^{pt})) + \xi (Y-G_{W(\beta^{pt})}(W(\beta^{pt})))\nonumber \\
    &\hspace{2cm}-\mathbb E \{G_{W(\beta^{pt})}(W(\beta^{pt}))|W(\beta)\}]^2\,  +o_P(1)\nonumber \\
    =& \mathbb E [G_{W(\beta^{pt})}(W(\beta^{pt})) - \mathbb E \{G_{W(\beta^{pt})}(W(\beta^{pt}))|W(\beta)\} ]^2 \nonumber \\
    &\hspace{2cm} + \mathbb E [Y-G_{W(\beta^{pt})}(W(\beta^{pt}))]^2+o_P(1)\nonumber \\
    =:& Q(\beta)+o_P(1)
\end{align}
uniformly in $\beta \in B$. Let us show the first equality. By Lemma \ref{lem: convergence rates for Ghat fhat and iotahat}(i), $\widehat \beta = \beta^{pt}+o_P(1)$, and since $G_{W(\beta)}(W(\beta))$ is Lipschitz in $\beta$ (as noticed earlier),  we obtain that $|\widehat G _{\widehat W (\widehat \beta )}(\widehat W _i (\widehat \beta))-G_{W(\beta^{pt})}(W_i(\beta^{pt})|\widehat t _i=o_P(1)$ uniformly in $i=1,\ldots,n$. Also, Lemma \ref{lem: convergence rates for Ghat fhat and iotahat}(v) implies that $|\widehat G ^*_{\widetilde W ^* (\beta)}(\widetilde W ^* _i (\beta))-\mathbb E \{G_{W(\widehat \beta)}(W (\widehat \beta))|W(\beta)=W_i(\beta)\}|\widehat t _i=o_P(1)$ uniformly in $i=1,\ldots,n$ and $\beta$. Since $\widehat \beta=\beta^{pt}+o_P(1)$,  we also get that $\mathbb E \{G_{W(\widehat \beta)}(W (\widehat \beta))|W(\beta)=W_i(\beta)\}$ $=\mathbb E \{G_{W(\beta^{pt})}(W (\beta^{pt}))|$ $W(\beta)=W_i(\beta)\}|$ $+o_P(1)$ uniformly in  $i=1,\ldots,n$ and $\beta$. Gathering results gives the first equality of the previous display. The second equality follows from replacing the trimming $\widehat t$ with 1, thanks to Lemma \ref{lem: trimming and convergence rates for 1st step}(i), and then using a Glivenko-Cantelli Theorem. Such a Glivenko-Cantelli Theorem applies thanks to the Lipschitz condition in Assumption \ref{Assumption: parametric}(ii). Finally, we obtain the third equality by developing the square and using the fact that $\mathbb E \xi=0$, $\mathbb E \xi^2=1$, and $\xi$ is independent from the sample data.\\ Next, we show that the leading term on the RHS of the previous display is uniquely minimized in $\beta^{pt}$. To this end, notice that
\begin{equation*}
    \arg \min_{\beta\in B}Q(\beta)=\{\beta\in B\,:\,G_{W(\beta^{pt})}(W(\beta^{pt}))=\mathbb E \{G_{W(\beta^{pt})}(W(\beta^{pt}))|W(\beta)\}\text{ almost surely }\}\, .
\end{equation*}
Now, pick up a generic $\beta$ in such a set. Then, since $G_{W(\beta)}(W(\beta))=\mathbb E \{Y|W(\beta)\}$, we have $\mathbb E |Y-G_{W(\beta)}(W(\beta))|^2$ $\leq \mathbb E |Y-\mathbb E \{G_{W(\beta^{pt})}(W(\beta^{pt})) | W (\beta)\}|^2 $ $=\mathbb E |Y-G_{W(\beta^{pt})}(W(\beta^{pt}))|^2$. So, by Assumption \ref{Assumption: parametric}(i) we obtain that $\beta=\beta^{pt}$. Hence, $\arg \min_{\beta\in B} Q(\beta)=\{\beta^{pt}\}$. By this and (\ref{eq: limit of the bootstrapped SLS criterion}) we can apply \citet[Theorem 5.7]{vaart_asymptotic_1998} to obtain that
\begin{equation*}
\widehat \beta^* = \beta^{pt}+o_P(1)\, .    
\end{equation*}
Next, let us obtain the limit of the bootstrapped statistic. Notice first that the arguments for (\ref{eq: decomposition of bootstrapped Emp process}) remain valid under $\mathcal H _1$. So, 
\begin{align}\label{eq: decomposition of bootstrapped Emp process 2}
    \mathbb P _n \widetilde \varepsilon ^* \widehat f_{\widetilde W ^* (\widehat \beta ^*)}&(\widetilde W ^* (\widehat \beta ^*))\, \varphi_s(X,\widetilde H ^* (Z))\, \widehat t \nonumber \\
    =&  \mathbb P _n \xi [Y-\widehat{G}_{\widehat{W}(\widehat{\beta})}(\widehat W (\widehat \beta))]\, \widehat{f}_{\widetilde W^*(\widehat \beta^*)}(\widetilde W^*(\widehat \beta^*))\, [\varphi_s(X,\widetilde{H}^*(Z))-\widehat{\iota}_s^*(\widetilde W^* (\widehat \beta^*))] \widehat t \nonumber  \\
    &+  \mathbb P _n [\widehat{G}_{\widehat W (\widehat \beta)}(\widehat W (\widehat \beta))-\widehat{G}^*_{\widetilde W^* (\widehat \beta^*)}(\widetilde W^*(\widehat \beta^*))]\widehat{f}_{\widetilde W^*(\widehat \beta^*)}(\widetilde W^*(\widehat \beta^*)) \\
    &\qquad \hspace{7cm} \cdot [\varphi_s(X,\widetilde{H}^*(Z))-\widehat{\iota}_s^*(\widetilde W^*(\widehat \beta^*))]\, \widehat t\, . \nonumber 
\end{align}
From Lemma \ref{lem: convergence rates for Ghat fhat and iotahat}(i) and since $G_{W(\beta)}(W(\beta))$ is Lipschitz in $\beta$ (as already noticed earlier), we get that $\max_{i=1,\ldots,n}|\widehat G _{\widehat{W}(\widehat \beta)}(\widehat{W}_i(\widehat \beta)) - G_{W(\beta^{pt})}(W_i(\beta^{pt}))|\widehat t_i=o_P(1)$. By Lemma \ref{lem: convergence rates for Ghat fhat and iotahat}(iv), $\widehat \beta ^*=\beta^{pt}+o_P(1)$, and since $f_{W(\beta)}(W(\beta))$ is Lipschitz in $\beta$ (see Assumption \ref{Assumption: smoothness}(v)(vi)) we get that $\max_{i=1,\ldots,n}|\widehat f_{\widetilde{W}^*(\widehat \beta ^*)}(\widetilde{W}^*_i(\widehat \beta ^*))-f_{W(\beta^{pt})}(W_i(\beta^{pt}))|\widehat t_i=o_P(1)$.
 As already obtained earlier, $\sup_{s\in\mathcal{S}}\max_{i=1,\ldots,n}|\varphi_s(X_i,\widetilde H ^* (Z_i))-\varphi_s(X_i,H(Z_i))|\widehat t_i=o_P(1)$. Also, by Lemma \ref{lem: convergence rates for Ghat fhat and iotahat}(v),  $\widehat \beta=\beta^{pt}+o_P(1)$, $\widehat \beta^* = \beta^{pt}+o_P(1)$, and the Lipschitz property in Assumption \ref{Assumption: parametric}(ii) we get that $\max_{i=1,\ldots,n}|\widehat G ^* _{\widetilde W ^* (\widehat \beta ^*)}(\widetilde W ^*_i (\widehat \beta ^*))- G_{W(\beta^{pt})}(W_i(\beta^{pt}))|\widehat t _i=o_P(1)$. 
  Finally, by Lemma \ref{lem: convergence rates for Ghat fhat and iotahat}(vi), $\widehat \beta ^*=\beta^{pt}+o_P(1)$, and the Lipschitz property from Assumption \ref{Assumption: smoothness}(v)(vi), we get that $\sup_{s\in\mathcal{S}}\max_{i=1,\ldots,n}|\widehat \iota_s^*(\widetilde{W}^*_i(\widehat \beta ^*)) - \mathbb{E}\{\varphi_s|W_i(\beta^{pt})\}|\widehat t _i=o_P(1)$. Gathering results,  
 
\begin{align}\label{eq: bootstrap empirical process under H1}
    \mathbb P _n \widetilde \varepsilon ^* \widehat f_{\widetilde W ^* (\widehat \beta ^*)}(\widetilde W ^* (\widehat \beta ^*))&\, \varphi_s(X,\widetilde H ^* (Z))\, \widehat t \nonumber \\
    =& \mathbb P _n \xi [Y-G_{W(\beta^{pt})}(W(\beta^{pt}))]\, f_{W(\beta^{pt})}(W(\beta^{pt}))\, [\varphi_s(X,H(Z))-\iota_s(W(\beta^{pt}))]\,\widehat t\nonumber \\
    &+ o_P(1)\, ,
\end{align}
uniformly in $s\in\mathcal{S}$. Next, we have (see the comments below)  
\begin{equation}\label{eq: bootstrap empirical process under H1 (ii)}
    \mathbb P _n \xi [Y-G_{W(\beta^{pt})}(W(\beta^{pt}))]\, f_{W(\beta^{pt})}(W(\beta^{pt}))\, [\varphi_s(X,H(Z))-\iota_s(W(\beta^{pt}))]\,\widehat t= o_P(1)
\end{equation}
uniformly in $s\in\mathcal{S}$. To see this, first use Lemma \ref{lem: trimming and convergence rates for 1st step}(i) to replace $\widehat t$ with 1. Then, since $\varphi_s$ is Lipschitz in $s\in\mathcal{S}$, by Glivenko-Cantelli's Theorem and $\mathbb E \xi=0$
\begin{align*}
    \mathbb P _n \xi [Y-G_{W(\beta^{pt})}&(W(\beta^{pt}))]\, f_{W(\beta^{pt})}(W(\beta^{pt}))\,  [\varphi_s(X,H(Z))-\iota_s(W(\beta^{pt}))] \\
    =& \mathbb E \xi [Y-G_{W(\beta^{pt})}(W(\beta^{pt}))]\, f_{W(\beta^{pt})}(W(\beta^{pt}))\, [\varphi_s(X,H(Z))-\iota_s(W(\beta^{pt}))]+o_P(1)\\
    =& o_P(1)\, 
\end{align*}
uniformly in $s\in\mathcal{S}$. By (\ref{eq: bootstrap empirical process under H1}), (\ref{eq: bootstrap empirical process under H1 (ii)}), and the continuity of the Cramer-Von Mises functional we obtain 
\begin{equation}\label{eq: convergence of S n star under H1}
    \frac{S_n^*}{n}=o_P(1)\, .
\end{equation}
We can finally show that $\Pr(S_n>\widehat c_{1-\alpha})=1+o(1)$ under $\mathcal H _1$.  
To this end, fix any $b>0$ smaller than the integral on the RHS of (\ref{eq: convergence of S n under H1}) (which is strictly positive under $\mathcal{H}_1$, as noticed earlier). For any such $b$ by (\ref{eq: convergence of S n star under H1}) we have $\Pr^*(S_n^*/n\leq b)=1+o_P(1)$, so that $\Pr^*(S_n^*/n\leq b)\geq 1-\alpha$ wpa1. Thus, by definition of $\widehat c_{1-\alpha}$ we also find that $\widehat c_{1-\alpha}/n\leq b$ wpa1.  Since by (\ref{eq: convergence of S n under H1}) $S_n/n> b$ wpa1, we also obtain that $S_n/n>\widehat c_{1-\alpha}/n$ wpa1. This concludes the proof.

\begin{flushright}
{\itshape [Q.E.D.]}
\end{flushright}
 \vspace{0.5cm}

\begin{lem}\label{lem: IFR for betahat star}
Under Assumptions \ref{Assumption: iid and phi}-\ref{Assumption: parametric} and $\mathcal H _0$ we have 
\begin{align*}
    \sqrt{n}(\widehat \beta^* - \widehat \beta)=& \sqrt{n} \mathbb P _n\, \xi\,\varepsilon\, \Sigma\, \nabla_\beta G_{W(\beta_0)}(W(\beta_0))\\
    & - \sqrt{n} \mathbb P _n\, a(Z)\,\xi\,[D-H(Z)]+o_P(1)\, ,
    \end{align*}
    where 
    \begin{align*}
        a(Z)=\mathbb{E}\{\,\Sigma\, \nabla_\beta& G_{W(\beta_0)}(W(\beta_0))\, \partial^T G_{W(\beta_0)}\,\partial_{\beta^T} q\, |\, Z\,\}\\
    \text{ and }\Sigma=\mathbb{E}\{\,\nabla_\beta& G_{W(\beta_0)}(W(\beta_0))\,\, \nabla_{\beta^T} G_{W(\beta_0)}(W(\beta_0))\,\}\, .
\end{align*}
\end{lem}

\begin{proof}
We first show that $\widehat \beta ^* = \beta_0+o_P(1)$. To this end, notice that the arguments used for $\widehat \beta ^*=\beta^{pt}+o_P(1)$ in the proof of Proposition \ref{prop: bootstrap test}(iii) remain valid with $\beta^{pt}$ replaced by $\beta_0$. In particular, such arguments do not rely on $\mathcal{H}_1$. Thus, 
\begin{equation*}
    \widehat \beta ^* = \beta_0+o_P(1)\, .
\end{equation*}
Let us now obtain the influence function representation for $\widehat \beta^*$. The same arguments as in (\ref{eq: mean-value expansion for betahat}) lead to  
\begin{equation}\label{eq: mean value expansion of betahatstar}
    \widehat \Sigma ^* (\overline \beta ^*) \, \sqrt{n}(\widehat \beta ^* - \widehat \beta)=\sqrt{n} \mathbb P _n [Y^* - \widehat G ^*_{\widetilde W ^* (\widehat \beta)}(\widetilde W ^* (\widehat \beta))]\, \nabla_\beta \widehat G ^* _{\widetilde W ^* (\widehat \beta)}(\widetilde W ^* (\widehat \beta))\, \widehat t + o_P(1)\, ,
\end{equation}
where $\overline \beta ^*$ lies on the segments joining $\widehat \beta^*$ and $\widehat \beta$. For $\widehat \Sigma ^* (\beta)$ we have  (see the comments below)  
\begin{align*}
    \widehat \Sigma ^* (\beta)=&-\mathbb P _n [Y^*-\widehat G ^*_{\widetilde W ^*(\beta)}(\widetilde W ^* (\beta))]\,\nabla^2_{\beta \beta ^T}\widehat G ^*_{\widetilde W ^* (\beta)}(\widetilde W ^* (\beta))\, \widehat t\\
    &+ \mathbb P _n \nabla_{\beta}\widehat G ^*_{\widetilde W ^* (\beta)}(\widetilde W ^* (\beta))\, \nabla_{\beta}\widehat G ^*_{\widetilde W ^* (\beta)}(\widetilde W ^* (\beta))\, \widehat t\\
    =&  -\mathbb E [G_{W(\beta_0)}(W(\beta_0))+\xi (Y-G_{W(\beta_0)}(W(\beta_0))) -G_{W(\beta)}(W(\beta))]  \nabla^2_{\beta \beta^T} G_{W(\beta)}(W(\beta))\\
    &+ \mathbb E \nabla_\beta G_{W(\beta)}(W(\beta)) \nabla_{\beta^T} G_{W(\beta)}(W(\beta))+o_P(1)\\
    =&  -\mathbb E [G_{W(\beta_0)}(W(\beta_0)) -G_{W(\beta)}(W(\beta))]  \nabla^2_{\beta \beta^T} G_{W(\beta)}(W(\beta))\\
    &+ \mathbb E \nabla_\beta G_{W(\beta)}(W(\beta)) \nabla_{\beta^T} G_{W(\beta)}(W(\beta))+o_P(1)
\end{align*}
uniformly in $\beta\in B$. Now, the first equality is the definition of $\widehat \Sigma ^* (\beta)$. Let us obtain the second equality. Since under $\mathcal{H}_0$ $\mathbb{E}\{G_{W(\beta_0)}(W(\beta_0))|W(\beta)\}=G_{W(\beta)}(W(\beta))$,   by Lemma \ref{lem: convergence rates for Ghat fhat and iotahat}(xvi) $|\widehat G ^*_{\widetilde{W}^*(\beta)}(\widetilde{W}^*_i(\beta))-G_{W(\beta)}(W_i(\beta))|\widehat t _i=o_P(1)$ uniformly in $i=1,\ldots,n$ and $\beta \in B$. Since $Y^*_i=\widehat G _{\widehat W (\widehat \beta)}(\widehat W_i (\widehat \beta))+ \xi_i(Y_i-\widehat G _{\widehat W (\widehat \beta)}(\widehat W_i (\widehat \beta)))$, by using Lemma \ref{lem: convergence rates for Ghat fhat and iotahat}(i) and $\widehat \beta=\beta_0+O_P(n^{-1/2})$ we get that
$\max_{i=1,\ldots,n}|Y^*_i-G_{W(\beta_0)}(W_i(\beta_0))-\xi_i(Y_i-G_{W(\beta_0)}(W_i(\beta_0)))|\widehat t_i=o_P(1)
$. By Lemma \ref{lem: convergence of the derivatives at the data points}(v)(vi), $\|\nabla_\beta \widehat G^*_{\widetilde{W}^*(\beta)}(\widetilde{W}_i^*(\beta))$ $-\nabla_\beta G_{W(\beta)}(W_i(\beta))\|\widehat t_i=o_P(1)$ and $\|\nabla^2_{\beta \beta^T} $$\widehat G^*_{\widetilde{W}^*(\beta)}(\widetilde{W}_i^*(\beta))$ $-\nabla^2_{\beta \beta^T} $ $G_{W(\beta)}(W_i(\beta))$$\|\widehat t_i=o_P(1)$ uniformly in $i=1,\ldots,n$ and $\beta\in B$. Thus, using these results and arguing as at the beginning of the proof of Lemma \ref{lem: IFR for betahat} gives the second equality.  
Finally, to obtain the third equality we have used $\mathbb E \xi=0$ and the fact that $\xi$ is independent from the sample data. \\
As noticed in the proof of Lemma \ref{lem: IFR for betahat}, $G_{W(\beta)}(W(\beta))$, $\nabla_\beta G_{W(\beta)}(W(\beta))$, and $\nabla^2_{\beta \beta^T} G_{W(\beta)}(W(\beta))$  are Lipschitz in $\beta\in B$. Thus, by the previous display and $\overline \beta ^*=\beta_0+o_P(1)$ we get that
\begin{equation*}
    \widehat \Sigma ^* (\overline \beta ^*)=\mathbb E \, \nabla_\beta G_{W(\beta_0)}(W(\beta_0))\,\nabla_\beta G_{W(\beta_0)}(W(\beta_0))+o_P(1)
    =\Sigma+o_P(1)\, .
\end{equation*}
Now, to prove the desired result it suffices to get an IFR for the leading term on the RHS of (\ref{eq: mean value expansion of betahatstar}). To this end, we use the following decomposition 
\begin{align}\label{eq: leading term of mean value expansion of betahat star}
    \text{1}^{st}\text{ term RHS of (\ref{eq: mean value expansion of betahatstar})}=& \sqrt{n} \mathbb P _n \xi \varepsilon \nabla_\beta \widehat G ^*_{\widetilde W ^* (\widehat \beta)}(\widetilde W ^* (\widehat \beta))\, \widehat t \nonumber \\
    &+ \sqrt{n} \mathbb P _n \xi [G_{W(\beta_0)}(W(\beta_0)) - \widehat G _{\widehat W (\widehat \beta)}(\widehat W (\widehat \beta))]\, \nabla_\beta \widehat G ^*_{\widetilde W ^* (\widehat \beta)}(\widetilde W ^* (\widehat \beta))\, \widehat t \nonumber \\
    &+\sqrt{n} \mathbb P _n [\widehat G _{\widehat W (\widehat \beta)}(\widehat W (\widehat \beta)) - \widehat G^* _{\widetilde W ^*(\widehat \beta)}(\widetilde W^* (\widehat \beta))] \nabla_\beta \widehat G ^*_{\widetilde W ^* (\widehat \beta)}(\widetilde W ^* (\widehat \beta))\, \widehat t\, .
\end{align}
Let us handle separately each term of the above decomposition. For the first term, we have (see the comments below) 
\begin{align}
\sqrt{n}\mathbb P _n \xi \varepsilon \nabla_\beta \widehat G ^*_{\widetilde W ^* (\widehat \beta)}(\widetilde W ^* (\widehat \beta ))\, \widehat t=& \sqrt{n}\mathbb P _n \xi \varepsilon \nabla_\beta \widehat G ^*_{\widetilde W ^* ( \beta_0)}(\widetilde W ^* ( \beta_0 ))\, \widehat t \nonumber \\
    &+ \mathbb P _n \xi \varepsilon \nabla^2_{\beta \beta^T} \widehat G ^*_{\widetilde W ^* (\overline \beta)}(\widetilde W ^* (\overline \beta ))\, \widehat t\,\sqrt{n}(\widehat \beta - \beta_0) \nonumber  \\
    =& \sqrt{n}\mathbb P _n \xi \varepsilon \nabla_\beta \widehat G ^*_{\widetilde W ^* ( \beta_0)}(\widetilde W ^* ( \beta_0 ))\, \widehat t + o_P(1) \nonumber  \\
    =& \sqrt{n}\mathbb P _n \xi \varepsilon \nabla_\beta G _{ W  ( \beta_0)}( W  ( \beta_0 )) + o_P(1)\nonumber \, ,
\end{align}
with $\overline \beta\,$ lying on the segment joining $\widehat \beta$ and $\beta_0$. The first equality follows from a mean-value expansion around $\beta_0$.\footnote{By the arguments used in (\ref{eq: trimming implication for foc of betahat}), we get that wpa1 $\widehat t_i \leq $ $\mathbb{I}\{\widehat f_{\widetilde{W}^*(\beta)}(\widetilde{W}^*_i(\beta))\geq \eta(\delta)\tau_n/4\}$  for all $i=1,\ldots,n$ and $\beta\in B$. Thus, from Assumption \ref{Assumption : kernels}(i) we have that wpa1 $\nabla_\beta \widehat{G}^*_{\widetilde{W}^*(\beta)}(\widetilde{W}^*_i(\beta))$ is differentiable in $\beta\in Int(B)$ whenever $\widehat t_i=1$.
}
Turning to the second equality, by Lemma \ref{lem: convergence of the derivatives at the data points}(vi), $\overline \beta = \beta_0 +o_P(1)$, and since $\nabla^2_{\beta \beta^T}G_{W(\beta)}(W(\beta))$ is Lipschitz in $\beta$ (as recalled earlier) 
we have $\mathbb P _n \xi \varepsilon \nabla^2_{\beta \beta^T}\widehat G^*_{\widetilde W^*(\overline \beta)}(\widetilde W^*(\overline \beta))\, \widehat t$ $=\mathbb P _n \xi \varepsilon \nabla^2_{\beta \beta^T}G_{W(\beta_0)}(W(\beta_0))\, \widehat t+o_P(1)$. Lemma \ref{lem: trimming and convergence rates for 1st step}(i) allows replacing $\widehat t$ with 1 at the cost of an $o_P(1)$ reminder. Then, by a Law of Large numbers $\mathbb P _n \xi \varepsilon \nabla^2_{\beta \beta^T}G_{W(\beta_0)}(W(\beta_0))$ $=\mathbb E \xi \varepsilon \nabla^2_{\beta \beta^T}G_{W(\beta_0)}(W(\beta_0))+o_P(1)=o_P(1)$, as $\mathbb E \xi=0$ and $\xi$ is independent from the sample data. Thus, $\mathbb P _n \xi \varepsilon $ $\nabla^2_{\beta \beta^T} \widehat G ^*_{\widetilde W ^* (\overline \beta)}(\widetilde W ^* (\overline \beta ))\,$ $ \widehat t\,\sqrt{n}(\widehat \beta - \beta_0)$ $=o_P(\sqrt{n}(\widehat \beta - \beta_0))=o_P(1)$. This gives the second equality. To get the third equality we have combined the arguments used for (\ref{eq: first term in decomposition of FOC for betahat II}) with Lemma \ref{lem: convergence of the derivatives at the data points}(vii)(viii) and Lemma \ref{lem: belonging conditions}(viii)(ix)(xiv).\\
Let us now consider the second term on the RHS of (\ref{eq: leading term of mean value expansion of betahat star}). We have (see the comments below)
\begin{align*}\label{eq: leading term of mean value expansion of betahat star II}
    \sqrt{n}\mathbb P _n \xi [G_{W(\beta_0)}(W(\beta_0)) -& \widehat G _{\widehat W (\widehat \beta)}(\widehat W (\widehat \beta))]\,\nabla_\beta \widehat G ^*_{\widetilde W ^*(\widehat \beta)}(\widetilde W ^*(\widehat \beta))\, \widehat t\nonumber \\
    =&  \sqrt{n}\mathbb P _n \xi [G_{W(\beta_0)}(W(\beta_0)) - \widehat G _{\widehat W (\widehat \beta)}(\widehat W (\widehat \beta))]\,\nabla_\beta G_{W(\beta_0)}(W(\beta_0)) \,\widehat t +o_P(1) \nonumber \\
    =& \sqrt{n}\mathbb P _n \xi [G_{W(\beta_0)}(W(\beta_0)) - \widehat G _{\widehat W (\widehat \beta)}(\widehat W (\widehat \beta))]\,\nabla_\beta G_{W(\beta_0)}(W(\beta_0)) \, t +o_P(1)\nonumber \\
    =& o_P(1)\, .
\end{align*}
The first equality is a direct consequence of Lemma \ref{lem: convergence rates for Ghat fhat and iotahat}(i), $\widehat \beta=\beta_0+O_P(n^{-1/2})$, and Lemma \ref{lem: convergence of the derivatives at the data points}(v).  To obtain the second equality, we have used Lemma \ref{lem: convergence rates for Ghat fhat and iotahat}(i) and the same arguments as in (\ref{eq: trimming replacement}). Finally, the third equality can be obtained by the arguments used for handling  (\ref{eq: 2nd term RHS of 2nd decomposition of Asy Emp Process}) and by recalling that $\mathbb E \xi=0$ and $\xi$ is independent from the sample data.\\
Let us now handle the last term on the RHS of (\ref{eq: leading term of mean value expansion of betahat star}). We have (see the comments below) 
\begin{align*}
    \sqrt{n} \mathbb P _n [\widehat G _{\widehat W (\widehat \beta)}(\widehat W (\widehat \beta)) - &\widehat G^* _{\widetilde W ^*(\widehat \beta)}(\widetilde W^* (\widehat \beta))] \nabla_\beta \widehat G ^*_{\widetilde W ^* (\widehat \beta)}(\widetilde W ^* (\widehat \beta))\, \widehat t \\
    &= \sqrt{n} \mathbb P _n [\widehat G _{\widehat W (\widehat \beta)}(\widehat W (\widehat \beta)) - \widehat G^* _{\widetilde W ^*(\widehat \beta)}(\widetilde W^* (\widehat \beta))] \nabla_\beta  G _{ W (\beta_0)}( W  ( \beta_0))\, \widehat t + o_P(1)\\
    &=\sqrt{n} \mathbb P _n [\widehat G _{\widehat W (\widehat \beta)}(\widehat W (\widehat \beta)) - \widehat G^* _{\widetilde W ^*(\widehat \beta)}(\widetilde W^* (\widehat \beta))] \nabla_\beta  G _{ W (\beta_0)}( W  ( \beta_0))\,  t + o_P(1)\\
    &= -\sqrt{n} \mathbb P _n a(Z)\, \xi[D-H(Z)]+o_P(1)\,,\\
    \text{ where }a(Z)&=\mathbb E \{\nabla_\beta G_{W(\beta_0)}(W(\beta_0))\, \partial^T G_{W(\beta_0)}\, \partial_H q(\beta_0,X,H(Z))|Z\}
\end{align*}
The first equality is obtained from Lemma \ref{lem: convergence rates for Ghat fhat and iotahat}(i)(xvi), Lemma \ref{lem: convergence of the derivatives at the data points}(v), and $\widehat \beta = \beta_0+O_P(n^{-1/2})$. To get the second equality we have used Lemma \ref{lem: convergence rates for Ghat fhat and iotahat}(i)(xvi) and the arguments in (\ref{eq: trimming replacement}). Finally, to get the last equality we have combined Lemma \ref{lem: stochastic expansions}(i)(v) with the arguments used for the RHS of (\ref{eq: 2nd term of decomposition of bootstrapped Emp process II }), thanks to the fact that $\mathbb E \{\nabla_\beta G_{W(\beta_0)}(W(\beta_0))|W(\beta_0)\}=0$. \\
Finally, gathering results gives an IFR for the leading term on the RHS of (\ref{eq: mean value expansion of betahatstar}). This concludes the proof.
\end{proof}


\section{Auxiliary Lemmas\label{sec:Auxiliary-Lemmas}}
We use the notation $\|g\|_{\infty,\mathcal{A}}=\sup_{a\in\mathcal{A}}|g(a)|$ for any function $g$ defined on a set $\cal A$. When the support of the argument of $g$ is clear from the context, we will simply denote with $\|g\|_\infty$ the supremum norm of $g$ taken over the support of its argument.\\

\noindent For notational simplicity we introduce 
\begin{equation}\label{eq: definition of dtilde H}
    \widetilde d _H :=\frac{d_H}{\tau_n^3}+\frac{p_n}{h_H^p \tau_n^2}\, .
\end{equation}
From Assumptions \ref{Assumption: bandwidth}(ii) and \ref{Assumption: trimming}(iv) we have 
\begin{equation}\label{eq: n 1/4 convergence of dtilde H}
    \widetilde d_H = o(n^{-1/4})\,\text{ and } \frac{\widetilde d _H}{\tau_n}=o(1).
\end{equation}

\begin{lem}\label{lem: trimming and convergence rates for 1st step}
Let    $t^\delta(x,z):=\mathbb{I}\{f(x,z)\geq \delta \tau_n\}$ for $\delta \in (0,1]$. 
Under Assumptions \ref{Assumption: iid and phi}-\ref{Assumption: trimming}
\begin{enumerate}[label=(\roman*)]
\item $\mathbb P _n |t^\delta - 1|=O_P(p_n)=o_P(n^{-1/2})$ and $\mathbb P _n |\widehat t - t|=O_P(p_n)=o_P(n^{-1/2})$\,,
\item $\|(\widehat H - H)t^\delta\|_\infty =O_P(\widetilde d_H)=o_P(n^{-1/4})$ and $\|(\widehat H - H) \widehat t\|_\infty =O_P(\widetilde d_H)=o_P(n^{-1/4})$\,,
\item $\|(\widetilde{H}-H)t^\delta\|_\infty=O_P(\widetilde d_H)=o_P(n^{-1/4})$ and $\|(\widetilde{H}-H)\widehat t\|_\infty=O_P(\widetilde d_H)=o_P(n^{-1/4})$\,,
\item  $\|(\widetilde{H}^*-H)t^\delta\|_\infty=O_P(\widetilde d_H)=o_P(n^{-1/4})$ and $\|(\widetilde{H}^*-H)\widehat t\|_\infty=O_P(\widetilde d_H)=o_P(n^{-1/4})$\,.
\end{enumerate}
\end{lem}

\begin{proof}
$(i)$ The proof of $(i)$ is contained in Lemma 8.1(i) of \citet{lapenta_encompassing_2022}.\\
To show $(ii)$, by Assumption \ref{Assumption: trimming}(i) for any large $n$ we have $t^\delta(x,z)\leq \mathbb{I}\{f_Z(z)\geq \eta(\delta)\tau_n\}$ for all $(x,z)$. Hence, we get $\|(\widehat{H}-H)t^\delta\|_{\infty}\leq$ $ \|(\widehat{H}-H)\mathbb{I}\{f_Z(\cdot)\geq \eta(\delta)\tau_n\}\|_{\infty}$ for any large $n$. From Equation (21) in \citet{lapenta_encompassing_2022} $ \|(\widehat{H}-H)\mathbb{I}\{f_Z(\cdot)\geq \eta(\delta)\tau_n\}\|_{\infty}=O_P(d_H\tau_n^{-1})=O_P(\widetilde d_H)$, and from (\ref{eq: n 1/4 convergence of dtilde H}) we obtain $\widetilde d_H=o(n^{-1/4})$. This proves the first part of $(ii)$. To show the second part of $(ii)$, from Equation (20) in \citet{lapenta_encompassing_2022} and from Assumption \ref{Assumption: trimming}(ii) we have that wpa1 $\widehat{t}(x,z)\leq \mathbb{I}\{f(x,z)\geq \tau_n/2\}=t^{(1/2)}(x,z)$ for all $(x,z)\in Supp(X,Z)$. So, wpa1 $\|(\widehat{H}-H)\widehat t\|_\infty\leq$ $\|(\widehat{H}-H)t^{(1/2)}\|_\infty$ and the desired result follows. \\
The proofs of $(iii)$ and $(iv)$ follow by combining the arguments used in the proof of $(ii)$ with the arguments of the proof of Lemma 8.1 in \citet{lapenta_encompassing_2022}. 
\end{proof}

\begin{lem}\label{lem: ASE for 1st step}
Let $\{(x,z)\mapsto g_s(x,z):s\in\cal S\}$ be a collection of functions such that $\sup_s \|g_s\|_{\infty}<\infty$ and $\sup_{s_1,s_2}\|g_{s_1}-g_{s_2}\|_{\infty}\leq C \|s_1-s_2\|$. Let $a_s(Z):=\mathbb E \{g_s(X,Z)|Z\}$ . Then, under Assumptions \ref{Assumption: iid and phi}-\ref{Assumption: trimming} 
\begin{enumerate}[label=(\roman*)]
    \item $\sqrt{n} \mathbb P _n g_s(X,Z) [\widetilde{H}(Z)-H(Z)]t=\sqrt{n}\mathbb P _n a_s(Z)[D - H(Z)]+o_P(1)$\,,
    \item If $a_s(Z)=0$ then $\sqrt{n}\mathbb P _n g_s(X,Z)[\widehat H (Z) - H(Z)] t=o_P(1)$\,,
    \item $\sqrt{n} \mathbb P _n g_s(X,Z)[\widetilde{H}^*(Z)-\widehat H(Z)]t=\sqrt{n}\mathbb P _n a_s(Z) \xi [D - H(Z)]+o_P(1)$\,.
\end{enumerate}
\end{lem}

\begin{proof}
Let us start with the proof of $(i)$. From Assumption \ref{Assumption: trimming}(i) for any large $n$ we have $t(x,z)\leq t(x,z) \mathbb I \{f_Z(z)\geq \eta(1) \tau_n\}$ for all $x,z$. So, for any large $n$
\begin{equation*}
\sqrt{n}\mathbb P _n g_s(X,Z)[\widetilde H (Z) - H(Z)]t=\sqrt{n}\mathbb P _n g_s(X,Z)[\widetilde H (Z) - H(Z)]\,t\, \mathbb I \{f_Z(Z)\geq \eta(1) \tau_n\} \, .
\end{equation*}
Now, from arguments analogous to those in Equations (21) and (24) in \citet{lapenta_encompassing_2022} we have $\|(\widetilde H - H)\mathbb I \{f_Z(\cdot)\geq \eta(1) \tau_n\}\|_\infty=o_P(1)$. So,
\begin{align*}
|\sqrt{n}\mathbb P _n g_s(X,Z)\,[\widetilde H (Z) - H(Z)]\,\mathbb I \{&f_Z(Z)\geq  \eta(1)\tau_n\}\,[t-1]|\\
\leq & C \|(\widetilde H - H)\mathbb I \{f_Z(\cdot)\geq \eta(1)\tau_n\}\|_\infty \sqrt{n}\mathbb P _n |t-1|\\ =&O_P(\sqrt{n}\mathbb P _n |t-1|)\\
=& o_P(1)\, ,
\end{align*}
where the last equality follows from $\sqrt{n}\mathbb P _n |t-1|=o_P(1)$, see Lemma \ref{lem: trimming and convergence rates for 1st step}(i). Thus, 
\begin{equation}\label{eq: trimming replacement for empirical process of first step}
    \sqrt{n}\mathbb P _n g_s(X,Z)[\widetilde H (Z) - H(Z)]t  =\sqrt{n}\mathbb P _n g_s(X,Z)[\widetilde H (Z) - H(Z)]\,\mathbb{I} \{f_Z(Z)\geq \eta(1)\tau_n\}+o_P(1)
\end{equation}
uniformly in $s\in \mathcal S$.
Finally, reasoning as in Proposition 4.2 in \citet{lapenta_encompassing_2022} leads to 
\begin{equation*}
    \sqrt{n}\mathbb P _n g_s(X,Z)[\widetilde H (Z) - H(Z)]\,\mathbb{I} \{f_Z(Z)\geq \eta(1)\tau_n\}=\sqrt{n} \mathbb P _n a_s(Z)\, [D- H(Z)]+o_P(1)
\end{equation*}
 uniformly in $s$. So, $(i)$ is proved. The proof of $(ii)$ follows from similar arguments. Finally, the proof of $(iii)$ can be obtained by combining the arguments used for
 (\ref{eq: trimming replacement for empirical process of first step}) with the arguments of Proposition 5.1 in \citet{lapenta_encompassing_2022}.
 
\end{proof}

For notational simplicity we define 
\begin{equation}\label{eq: definition of dtilde G}
    \widetilde d _G := d_G + \frac{p_n}{h^d} + \sum_{j=1}^4 \frac{\widetilde d _H^j}{h^j} + \frac{\widetilde d _H^5}{h^{d+5}}\, .
\end{equation}
From Assumptions \ref{Assumption: bandwidth} and \ref{Assumption: trimming}(iii)(iv)  we have\footnote{\label{footnote: 1st negligibility of convergence rate}To see that (\ref{eq: n 1/4 convergence of dtilde G / tau 2}) holds, notice first that from Assumption \ref{Assumption: bandwidth}(i) $d_G/\tau_n^2=o(n^{-1/4})$ and from Assumption \ref{Assumption: trimming}(iii) $p_n/(h^d \tau_n^2)=o(n^{-1/4})$. Also, from Assumption \ref{Assumption: bandwidth}(ii) $d_H/(h \tau_n^5)=o(n^{-1/4})$, while from Assumption \ref{Assumption: trimming}(iv) $p_n/(h^p_H h \tau_n^4)=o(n^{-1/4})$. In view of the last two equalities and (\ref{eq: definition of dtilde H}) we get $\widetilde{d}_H/(h \tau_n^2)=o(n^{-1/4})$ and hence $\sum_{j=1}^4 \widetilde{d}_H^j/(h^j\tau_n^2)=\sum_{j=1}^4[\widetilde{d}_H/(h \tau_n^2)]^j \tau_n^{2(j-1)}=o(n^{-1/4})$. Finally, to show that $\widetilde{d}_H^5/(h^{d+5}\tau_n^2)=o(n^{-1/4})$, notice that from Assumption \ref{Assumption: bandwidth}(i) we have $d_G/h^{\lceil(d+1)/2\rceil +1}=o(1)$, which implies that $n h^{2d +2}\rightarrow \infty$. Thus, 
$$\frac{\widetilde{d}_H^5}{h^{d+5}\tau_n^2}=\left(\frac{\widetilde{d}_H}{h \tau_n^2}\right)^4 \frac{\widetilde{d}_H \tau_n^6}{h^{d+1}}=o(n^{-1})\frac{\widetilde{d}_H \tau_n^6}{h^{d+1}}=o(n^{-1/2})\frac{\widetilde{d}_H \tau_n^6}{\sqrt{n h^{2d+2}}}=o(n^{-1/2})\, .$$   } 
\begin{equation}\label{eq: n 1/4 convergence of dtilde G / tau 2}
    \frac{\widetilde d _G}{\tau^2_n}=o(n^{-1/4})\, .
\end{equation}
We also define the set
\begin{equation}\label{eq: definition of cal W}
    \mathcal{W}:=\left\{w\,:\,w\in\,Supp(W(\beta))\text{ for some }\beta\in B\right\}\, .
\end{equation}

\begin{lem}
\label{lem: trimming implications}
Let Assumptions \ref{Assumption: iid and phi}-\ref{Assumption: trimming} hold. Then, 
\begin{enumerate}[label=(\roman*)]
    \item  $\mathbb P _n |t^\eta_{W(\beta_0)}-1|=O_P(p_n)=o_P(n^{-1/2})$ for any $\eta\in(0,1]$\,, where $t_{W(\beta_0)}^\eta(w):=\mathbb{I}\{f_{W(\beta_0)}(w)\geq \eta \tau_n\}$\,,
    \item wpa1 $\mathbb{I}\{\widehat f (x,z) \geq  \tau_n \}\leq \mathbb{I}\{f(x,z)\geq \tau_n/2\}$ for all $(x,z)\in Supp(X,Z)$\,,
    \item wpa1 $\mathbb{I}\{f_{W(\beta)}(w)\geq \delta \tau_n\}\leq \mathbb I \{\widehat f _{\widetilde W (\beta)}(w)\geq \delta \tau_n /2\}$ for all $(w,\beta)\in\mathcal W \times B$, and the same result holds by replacing $\widetilde W$ with $\widehat W$ or $\widetilde W ^*$\,,
    \item wpa1 $t^\delta(x,z) \mathbb I\{f_{W(\beta)}(q(\beta,x,H(z)))\geq \delta_2 \tau_n\}$ $\leq t^\delta(x,z)\mathbb I \{f_{W(\beta)}(q(\beta,x,\widetilde H (z)))\geq \delta_2 \tau_n /2\}$ for all $(x,z,\beta)\in Supp(X,Z)\times B$, and the same result holds by replacing $\widetilde H$ with $\widehat H$ or $\widetilde H ^*$\,.
    \end{enumerate}
    If moreover $\widehat{\beta}-\beta_0=O_P(n^{-1/2})$ then
    \begin{enumerate}[label=(\roman*)]
        \setcounter{enumi}{4}
    \item wpa1 $\mathbb I \{f_{W(\beta_0)}(w)\geq \delta \tau_n\}\leq \mathbb I \{\widehat f _{\widetilde W (\widehat \beta)}(w)\delta \tau_n/2\}$ for all $w\in\mathcal W$ and the same result holds by replacing $\widetilde W$ with $\widehat W$ or $\widetilde W ^*$\,,
    \item wpa1 $t^\delta (x,z) \mathbb I \{f_{W(\beta_0)}(q(\beta_0,x,H(z)))\geq \delta_2 \tau_n \}$ $\leq t^\delta (x,z) \mathbb I \{f_{W(\beta_0)}(q(\overline{\beta},x,\overline{H(z)}))\geq \delta_2 \tau_n /2\}$ for all $(x,z)\in Supp(X,Z)$, with $\overline{\beta}\in [\widehat{\beta},\beta_0]$ and $\overline{H(z)}\in[\widetilde H (z),H(z)]$, and the same result holds by replacing $\widetilde{H}$ with $\widehat H$ or $\widetilde H ^*$\,.
\end{enumerate}
\end{lem}

\begin{proof}
$(i)$ Since $t^\eta_{W(\beta_0)}(W(\beta_0))-1=\mathbb I \{f_{W(\beta_0)}(W(\beta_0))< \eta \tau_n\}$, we have $\mathbb E\, \mathbb P _n |t^\eta_{W(\beta_0)}-1|=$ $\mathbb E \, \mathbb I \{f_{W(\beta_0)}(W(\beta_0))< \eta \tau_n\}$ $=P(f_{W(\beta_0)}(W(\beta_0))< \eta \tau_n )=O(p_n)=o(n^{-1/2})$, where the last equality follows from Assumption \ref{Assumption: trimming}(iii). Hence, by Markov's inequality we get $\mathbb P _n |t^\eta_{W(\beta_0)}-1|=O_P(p_n)=o(n^{-1/2})$.


$(ii)$ The proof follows from the same arguments as those for Equation (20) in \citet{lapenta_encompassing_2022}.

$(iii)$ Notice first that by the Lipschitz condition in Assumption \ref{Assumption: smoothness}(v) we have $\max_{i=1,\ldots,n}$ $\|\widetilde W _i(\beta)-W_i(\beta)\|\widehat t_i=O_P(\|(\widetilde H - H)\widehat t\|_\infty)$ uniformly in $\beta\in B$. So,
by a 4th order Taylor expansion and Lemma \ref{lem : Li Racine} we get
\begin{equation*}
    \widehat{f}_{\widetilde{W}(\beta)}(w)=h^{-d}\,\mathbb P _n \widehat t K_h(w-W(\beta)) + O_P\left(\frac{\|(\widetilde{H}-H)\widehat t\|^5 _\infty}{h^{d+5}}+\sum_{j=1}^4\frac{\|(\widetilde H - H)\widehat t\|^{j}_\infty}{h^j}\right)
\end{equation*}
uniformly in $(w,\beta)\in\mathcal W \times B$. By Lemma \ref{lem: trimming and convergence rates for 1st step}(i) the trimming $\widehat t$ in the first leading term can be replaced with $1$ at the cost of an $O_P(p_n h^{-d})$ reminder. Also, by Lemma \ref{lem : Li Racine} $h^{-d}\,\mathbb P _n  K_h(w-W(\beta))= h^{-d}\,P K_h(w-W(\beta)) + O_P(\sqrt{\log n / (n h ^{d})}) $ uniformly in $(w,\beta)\in\mathcal W \times B$, and by standard bias manipulations Assumption \ref{Assumption: smoothness}(i) $h^{-d}\,P K_h(w-W(\beta))=f_{W(\beta)}(w)+O(h^r)$ uniformly in $(w,\beta)\in\mathcal W \times B$. Gathering results and using Lemma \ref{lem: trimming and convergence rates for 1st step}(iii) gives 
\begin{equation}\label{eq: uniform convergence of f W tilde beta}
    \sup_{\beta\in B}\| \widehat f _{\widetilde{W}(\beta)}-f_{W(\beta)}\|_{\infty,\mathcal W} = O_P(\widetilde d_ G)\, ,
\end{equation}
with $\widetilde d _G$ defined in (\ref{eq: definition of dtilde G}). Now, let us define the event $\mathcal{A}_n^C:=\{\sup_{\beta\in B}\|\widehat f _{\widetilde{W}(\beta)}-f_{W(\beta)}\|_{\infty,\cal W}\leq C \widetilde{d}_G\}$. Since $\widetilde{d}_G/\tau_n=o(1)$ (see Equation (\ref{eq: n 1/4 convergence of dtilde G / tau 2})), over the event $\mathcal{A}_n^C$ for each $n$ large enough we have $1-2 [\widehat f _{\widetilde W (\beta)}(w)-f{W(\beta)}(w)]/(\delta \tau_n)\leq 2$ for all $(w,\beta)\in\mathcal{W}\times B$. Thus, 
\begin{align*}
    \mathbb{I}\{\widehat{f}_{\widetilde W (\beta)}(w)\geq \delta \tau_n /2\}= &\mathbb I \left\{f_{W(\beta)}(w)\geq \frac{\delta \tau_n}{2}\left[1-2\frac{\widehat{f} _{\widetilde W (\beta)}(w)-f_{W(\beta)}(w)}{\delta \tau_n}\right]\right\}\\
    \geq & \mathbb I \{f_{W}(\beta)(w)\geq \delta \tau_n\}
\end{align*}
for all $(w,\beta)\in\mathcal{W}\times B$. By (\ref{eq: uniform convergence of f W tilde beta}) and (\ref{eq: n 1/4 convergence of dtilde G / tau 2}) we can choose $C$ large enough to make $\Pr(\mathcal{A}_n^C)$ arbitrarily close to 1 for each large $n$. This and the previous display deliver the desired result. The proof for $\widehat W$ or $\widetilde W^*$ proceed along the same lines. 

$(iv)$ Define the event $\mathcal{A}^C_n:=\{\sup_{x,\beta,z}|q(\beta,x,\widetilde H (z))-q(\beta,x, H (z))|t^\delta(x,z) \leq C \widetilde d_H \}$. By a Mean-Value expansion of $f_{W(\beta)}(q(\beta,x,\widetilde H(z)))$ around $q(\beta,x, H(z))$, over the event $\mathcal{A}_n^C$ we obtain 

\begin{equation*}
    t^\delta (x,z) |f_{W(\beta)}(q(\beta,x,\widetilde H(z))) - f_{W(\beta)}(q(\beta,x,H(z)))| \leq C \sup_{\beta\in B}\|\partial f_{W(\beta)}\|_{\infty,\mathbb{R}^d} \,\,\widetilde d_H
\end{equation*}
uniformly in $(x,z,\beta)\in Supp(X,Z)\times B$. Moreover, from (\ref{eq: n 1/4 convergence of dtilde H}) we have $\widetilde d_H/\tau_n=o(1)$. So, for each large $n$ we have that
$$1-2 t^\delta (x,z)\frac{[f_{W(\beta)}(q(\beta,x,\widetilde H(z))) - f_{W(\beta)}(q(\beta,x, H(z))) ]}{\delta_2 \tau_n}\leq 2 $$ for all $(x,z,\beta)\in Supp(X,Z)\times B$. This implies that 
\begin{align*}
     t^\delta (x,z) \mathbb I & \{f_{W(\beta)}(q(\beta,x,\widetilde H (z)))\geq \delta_2 \tau_n / 2\}\\
    =&  t^\delta (x,z) \mathbb I  \Big\{f_{W(\beta)}(q(\beta,x,H(z)))\geq \\
    & \quad  \frac{\delta_2 \tau_n}{2}\Big[1-2  t^\delta (x,z) \frac{f_{W(\beta)}(q(\beta,x,\widetilde H (z)))- f_{W(\beta)}(q(\beta,x, H (z)))}{\delta_2 \tau_n}\Big]\Big\}\\
    \geq & t^\delta (x,z) \mathbb I \{f_{W(\beta)}(q(\beta,x, H(z)))\geq \delta_2 \tau_n\}
\end{align*}
for all $(x,z,\beta)\in Supp(X,Z)\times B$. Now, by the Lipschitz conditions in Assumption \ref{Assumption: smoothness}(v) we have that $|q(\beta,x,\widetilde H (z))-q(\beta,x, H (z))|t^\delta(x,z)$ $=O_P(\|(\widetilde H - H)t^\delta\|_\infty)$ uniformly in $(x,z,\beta)\in Supp(X,Z)\times B$, and from Lemma \ref{lem: trimming and convergence rates for 1st step}(iii)  $\|(\widetilde H - H)t^\delta\|_\infty=O_P(\widetilde d _H)$. So, by choosing $C$ large enough we can make $\Pr(\mathcal{A}_n^C)$ arbitrarily close to 1 for each large $n$. This and the previous display deliver the desired result. This proof for $\widehat H$ and $\widetilde H ^*$ proceeds along the same lines. 

$(v)-(vi)$ The proofs  proceed along the same lines as the proofs of $(iii)$ and $(iv)$ by using the event $\{\|\widehat \beta - \beta_0 \|\leq C n^{-1/2}\}$.

\end{proof}

\begin{lem}\label{lem: convergence rates for Ghat fhat and iotahat}
Let Assumptions \ref{Assumption: iid and phi}-\ref{Assumption: trimming} hold. Then, 

\begin{enumerate}[label=(\roman*)]
 \item $\sup_{\beta\in B}\max_{i=1,\ldots,n}|\widehat{G}_{\widetilde W (\beta)}(\widetilde W _i (\beta))-G_{W(\beta)}(W_i(\beta))|( t_i + \widehat t _i)=O_P(\widetilde d_G \tau_n^{-1})=o_P(n^{-1/4})$ \\
 and the same result also holds by replacing $\widetilde W$ with $\widehat W$\,,
 \item $\sup_{\beta\in B}\max_{i=1,\ldots,n}|\widehat{f}_{\widetilde W (\beta)}(\widetilde W _i (\beta))-f_{W(\beta)}(W_i(\beta))|( t_i + \widehat t _i)=O_P(\widetilde d_G)=o_P(n^{-1/4})$\,,
 \item 
 $\sup_{s\in \mathcal{S}}\max_{i=1,\ldots,n}\,|\,\widehat \iota_s(\widetilde W_i (\widehat \beta))-\mathbb{E}\{\varphi_s(X,H(Z))|W(\beta)=W_i(\beta)\}|_{\beta=\widehat\beta}\,|\,(t_i+\widehat t _i)=o_P(1)$\,,
 \item $\sup_{\beta\in B}\max_{i=1,\ldots,n}|\widehat{f}_{\widetilde W^* (\beta)}(\widetilde W _i^* (\beta))-f_{W(\beta)}(W_i(\beta))|( t_i + \widehat t _i)=O_P(\widetilde d_G )=o_P(n^{-1/4})$\,,
 \item $\sup_{\beta\in B}\max_{i=1,\ldots,n}|\widehat{G}^*_{\widetilde W^* (\beta)}(\widetilde W _i^* (\beta))-\mathbb E \{G_{W(\beta_1)}(W( \beta_1))|W(\beta)=W_i(\beta)\}|_{\beta_1=\widehat \beta}\,|\,( t_i + \widehat t _i)=o_P(1)$\,,
  \item 
  $\sup_{s\in \mathcal{S},}\max_{i=1,\ldots,n}|\widehat \iota_s^*(\widetilde W^* _i (\widehat \beta^*))-\mathbb{E}\{\varphi_s(X,H(Z))|W( \beta)=W_i(\beta)\}|_{\beta=\widehat \beta^*}\,|\,(t_i+\widehat t _i)=o_P(1)$\,.
\end{enumerate}
If moreover $\widehat \beta - \beta _0 = O_P(n^{-1/2})$ and $\mathcal{H}_0$ holds, we have 
\begin{enumerate}[label=(\roman*)]
        \setcounter{enumi}{6}
   \item $\max_{i=1,\ldots,n}|\widehat f _{\widetilde W (\widehat \beta)}(W_i(\beta_0))-f_{W(\beta_0)}(W_i(\beta_0))|(t_i+\widehat t _i)=O_P(\widetilde d_G)=o_P(n^{-1/4})$\,,
   \item  $\max_{i=1,\ldots,n}|\widehat G _{\widetilde W (\widehat \beta)}(W_i(\beta_0))-G_{W(\beta_0)}(W_i(\beta_0))|(t_i+\widehat t _i)=O_P(\widetilde d_G \tau_n^{-1})=o_P(n^{-1/4})$ and the same result also holds by replacing $\widetilde W$ with $\widehat W$\,,
   \item  $\sup_{s\in\mathcal{S}} \max_{i=1,\ldots,n}|\widehat \iota_s (W_i(\beta_0))-\iota_s(W_i(\beta_0))|(t_i+\widehat t _i)=O_P(\widetilde d_G \tau_n^{-1})=o_P(n^{-1/4})$\,,
   \item $\|(\widehat{f}_{\widetilde W (\widehat \beta)}-f_{W(\beta_0)})t^\eta_{W(\beta_0)}\|_{\infty,\mathcal{W}}=O_P(\widetilde d_G)=o_P(n^{-1/4})$\,,
   \item $\|(\widehat{G}_{\widetilde W (\widehat \beta)}-G_{W(\beta_0)})t^\eta_{W(\beta_0)}\|_{\infty,\mathcal{W}}=O_P(\widetilde d_G \tau_n^{-1})=o_P(n^{-1/4})$ and the same result also holds by replacing $\widetilde W $ with $\widehat W$\,,
  \item  $\sup_{s\in\mathcal{S}} \|(\widehat{\iota}_s-\iota_s)t^\eta_{W(\beta_0)}\|_{\infty,\mathcal{W}}=O_P(\widetilde d_G \tau_n^{-1})=o_P(n^{-1/4})$\,,
   \item $\max_{i=1,\ldots,n}|\widehat f _{\widetilde W (\widehat \beta)}(\widetilde W_i(\widehat \beta))-f_{W(\beta_0)}(W_i(\beta_0))|(t_i+\widehat t _i)=O_P(\widetilde d_G)=o_P(n^{-1/4})$\,,
   \item  $\max_{i=1,\ldots,n}|\widehat G _{\widetilde W (\widehat \beta)}(\widetilde W_i(\widehat \beta))-G_{W(\beta_0)}(W_i(\beta_0))|(t_i+\widehat t _i)=O_P(\widetilde d_G \tau_n^{-1})=o_P(n^{-1/4})$\,,
   \item  $\sup_{s\in\mathcal{S}} \max_{i=1,\ldots,n}|\widehat \iota_s (\widetilde W_i(\widehat \beta))-\iota_s(W_i(\beta_0))|(t_i+\widehat t _i)=O_P(\widetilde d_G \tau_n^{-1})=o_P(n^{-1/4})$\,,
   \item $\sup_{\beta\in B}\max_{i=1,\ldots,n}|\widehat G ^* _{\widetilde W ^* (\beta)}(\widetilde W ^* _i (\beta))- \mathbb E \{G_{W(\beta_0)}(W(\beta_0))|W(\beta)=W_i(\beta)\}|(t_i+\widehat t _i)=O_P(\widetilde d_G \tau_n^{-2})$ $=o_P(n^{-1/4})$\,.
\end{enumerate}
If furthermore $\widehat \beta ^* -\beta_0 = O_P(n^{-1/2})$, then 
\begin{enumerate}[label=(\roman*)]
        \setcounter{enumi}{16}
    \item $\max_{i=1,\ldots,n}|\widehat f _{\widetilde W ^* (\widehat \beta ^*)}(W_i(\beta_0))-f_{W(\beta_0)}(W_i(\beta_0))|(t_i+\widehat t_i)=O_P(\widetilde d_G)=o_P(n^{-1/4})$\,,
    \item $\max_{i=1,\ldots,n}|\widehat G^* _{\widetilde W ^* (\widehat \beta ^*)}(W_i(\beta_0))-G_{W(\beta_0)}(W_i(\beta_0))|(t_i+\widehat t_i)=O_P(\widetilde d_G \tau_n^{-2})=o_P(n^{-1/4})$\,,
    \item $\sup_{s\in\mathcal{S}} \max_{i=1,\ldots,n}|\widehat \iota _s^*(W_i(\beta_0))-\iota_s(W_i(\beta_0))|(t_i+\widehat t_i)=O_P(\widetilde d_G \tau_n^{-1})=o_P(n^{-1/4})$\,,
    \item $\|(\widehat f_{\widetilde W ^* (\widehat \beta ^*)}-f_{W(\beta_0)})t^\eta_{W(\beta_0)}\|_{\infty,\mathcal{W}}=O_P(\widetilde d_G)=o_P(n^{-1/4})$\,,
    \item $\|(\widehat G^*_{\widetilde W ^* (\widehat \beta ^*)}-G_{W(\beta_0)})t^\eta_{W(\beta_0)}\|_{\infty,\mathcal{W}}=O_P(\widetilde d_G \tau_n^{-2})=o_P(n^{-1/4})$\,,
    \item $\sup_{s\in\mathcal{S}}\|(\widehat \iota^*_s-\iota_s)t^\eta_{W(\beta_0)}\|_{\infty,\mathcal{W}}=O_P(\widetilde d_G \tau_n^{-1})=o_P(n^{-1/4})$\,,
    \item $\max_{i=1,\ldots,n}|\widehat f_{\widetilde W ^*(\widehat \beta ^*)}(\widetilde W ^* _i (\widehat \beta ^*))-f_{W(\beta_0)}(W_i(\beta_0))|(t_i+\widehat t _i)=O_P(\widetilde d_G)=o_P(n^{-1/4})$\,,
    \item $\max_{i=1,\ldots,n}|\widehat G^*_{\widetilde W ^*(\widehat \beta ^*)}(\widetilde W ^* _i (\widehat \beta ^*))-G_{W(\beta_0)}(W_i(\beta_0))|(t_i+\widehat t _i)=O_P(\widetilde d_G \tau_n^{-2})=o_P(n^{-1/4})$\,,
    \item $\sup_{s\in\mathcal{S}} \max_{i=1,\ldots,n}|\widehat \iota^*_s(\widetilde W ^* _i (\widehat \beta ^*))-\iota_s(W_i(\beta_0))|(t_i+\widehat t _i)=O_P(\widetilde d_G \tau_n^{-1})=o_P(n^{-1/4})$\,.
\end{enumerate}

\end{lem}

\begin{proof}
$(i)$ Let us define 
\begin{equation}\label{eq: definition of T Y W(beta)}
    T^Y_{W(\beta)}(w):=\mathbb E \{Y|W(\beta)=w\}\,f_{W(\beta)}(w)\, .
\end{equation}
Arguments similar to those used for (\ref{eq: uniform convergence of f W tilde beta}) lead to \begin{equation}\label{eq: uniform convergence of T W tilde beta}
    \sup_{\beta \in B}\|\widehat T ^Y _{\widetilde W (\beta)}- T^Y_{W(\beta)}\|_{\infty ,\mathcal W}=O_P(\widetilde d_G)\, .
\end{equation}
From Lemma \ref{lem: trimming implications}(iii) wpa1 $t^\eta_{W(\beta)}(w)\leq \mathbb{I}\{\widehat{f}_{\widetilde W (\beta)}(w)\geq \eta \tau_n /2\}$ for all $(w,\beta)\in \mathcal{W}\times B$. So, by letting $\widehat A (w,\beta):=\widehat T^{Y}_{\widetilde W (\beta)}(w) - G_{W(\beta)}(w)\widehat f_{\widetilde W (\beta)}(w)$ we have that wpa1
\begin{align*}
    |\widehat{G}_{\widetilde{W}(\beta)}(w)-&G_{W(\beta)}(w)|t^\eta_{W(\beta)}(w)\\
    =& \Big| \frac{\widehat A (\beta,w)}{f_{W(\beta)}(w)} + \frac{\widehat A (\beta,w)}{f_{W(\beta)}(w)}\frac{f_{W(\beta)}(w)-\widehat f _{\widetilde W (\beta)}(w)}{\widehat f _{\widetilde W (\beta)}(w)}\Big|t^\eta_{W(\beta)}(w)\,\mathbb I \{\widehat f _{\widetilde W (\beta)}(w)\geq \eta \tau_n / 2\}
\end{align*}
for all $(w,\beta)\in \mathcal{W}\times B$. Combining (\ref{eq: uniform convergence of f W tilde beta}), (\ref{eq: uniform convergence of T W tilde beta}), and the above display leads to 
\begin{equation}\label{eq: uniform convergence of Ghat uniform in beta and w}
    \sup_{\beta\in B}\|(\widehat G _{\widetilde W (\beta)}-G_{W(\beta)})t^\eta_{W(\beta)}\|_{\infty,\mathcal{W}}=O_P(\widetilde d_G \tau^{-1}_n)\, .
\end{equation}
Now, wpa1 we have (see the comments ahead) $t^\delta(x,z)\leq t^\delta(x,z) \mathbb I \{f_{W(\beta)}(q(\beta,x,H(z)))\geq \eta(\delta) \tau_n\}$  $\leq t^\delta(x,z) \mathbb I \{f_{W(\beta)}(q(\beta,x,\widetilde H (z)))\geq \eta(\delta) \tau_n/2\}$ for all $(x,z,\beta)\in Supp(X,Z)\times B$, where the first inequality is ensured by Assumption \ref{Assumption: trimming}(i), while the second inequality is obtained from Lemma \ref{lem: trimming implications}(iv). Hence, wpa1\footnote{For notational simplicity we use $\partial G_{W(\beta)}(w):=\partial_w G_{W(\beta)}(w)$. }
\begin{align*}
    |\widehat G_{\widetilde W (\beta)}(\widetilde W _i (\beta))- G_{W(\beta)}(W_i (\beta))|t^\delta_i\leq & |\widehat G_{\widetilde W (\beta)}(\widetilde W _i (\beta)) - G_{W(\beta)}(\widetilde W _i (\beta))|\mathbb I \{f_{W(\beta)}(\widetilde W _i (\beta))\geq \eta(\delta)\tau_n/2\} \\
    & + |G_{W(\beta)}(\widetilde W _i (\beta))-G_{W(\beta)}(W_i (\beta))|t_i^\delta \\
    \leq & \|(\widehat G _{\widetilde W (\beta)}-G_{W(\beta)})t^{\eta(\delta)/2}_{W(\beta)}\|_{\infty,\mathcal{W}} \\
    &+ \|\partial G_{W(\beta)}\|_{\infty,\mathbb{R}^d} \|\widetilde W_i(\beta) - W_i(\beta) \|t^\delta _i\, .
\end{align*}
By (\ref{eq: uniform convergence of Ghat uniform in beta and w}) the first term on the RHS is $O_P(\widetilde d_G \tau_n^{-1})$ uniformly in $\beta$. For the second term on the RHS, by Assumption \ref{Assumption: smoothness}(i) $\|\partial G_{W(\beta)}\|_{\infty,\mathbb{R}^d}\leq C$ uniformly in $\beta\in B$, while by the Lipschitz condition in Assumption \ref{Assumption: smoothness}(v)  $\sup_{\beta\in B}\max_{i=1,\ldots,n}\|\widetilde W_i (\beta) - W_i (\beta)\|t^\delta _i =O_P(\|(\widetilde H - H)t^\delta\|)=O_P(\widetilde d_H)$, where the last equality follows from Lemma \ref{lem: trimming and convergence rates for 1st step}(iii). Hence, by using the definition of $\widetilde d_G$ in (\ref{eq: definition of dtilde G}) we get
$$\sup_{\beta \in B}\max_{i=1,\ldots,n}|\widehat G _{\widetilde W (\beta)}(\widetilde W _i (\beta))- G_{W(\beta)}(W_i (\beta))|t_i^{\delta} =O_P(\widetilde d_G \tau_n^{-1})=o(n^{-1/4})$$
where the last equality follows from (\ref{eq: n 1/4 convergence of dtilde G / tau 2}). 
Taking $\delta=1$ gives the part of result $(i)$ with $t_i$. To get the convergence rate with $\widehat t_i$, from Lemma \ref{lem: trimming implications}(ii) wpa1 $\widehat t(x,z)\leq t^{(1/2)}(x,z)$ for all $(x,z)\in Supp(X,Z)$. So, the result with $\widehat t_i$ follows from this inequality and the previous display evaluated at $\delta=1/2$.

$(ii)$ The result follows from (\ref{eq: uniform convergence of f W tilde beta}) and by the same arguments as in the proof of $(i)$.


$(iii)$ Let us recall that 
\begin{equation*}
    \widehat{\iota}_s(w)=\frac{\widehat T^{\varphi_s}_{\widetilde{W}(\widehat \beta)}(w)}{\widehat{f}_{\widetilde{W}(\widehat \beta)}(w)}\,\text{ with }\, \widehat{T}^{\varphi_s}_{\widetilde{W}(\widehat \beta)}(w)=\frac{1}{ n h^d}\sum_{i=1}^n \varphi_s(X_i,\widetilde{H}(Z_i))\,\widehat t_i\, K\left(\frac{w-\widetilde{W}_i(\widehat{\beta})}{h}\right)\, .
\end{equation*}
By the arguments used in the proof of Lemma \ref{lem: trimming implications}(iii) we get 
\begin{align*}
    \widehat{T}^{\varphi_s}_{\widetilde{W}(\widehat \beta)}(w)&= h^{-d}\,P\,K\left(\frac{w-W(\widehat \beta)}{h}\right)\,\varphi_s(X,H(Z))+O_P(\widetilde{d}_G)\\
    =& \int K(u)\,\mathbb{E}\{\varphi_s(X,H(Z))|W(\widehat \beta)=w+u h\}\,f_{W(\widehat \beta)}(w+ u h)\,du\\
    =& \int K(u)\,\Big[\,\mathbb{E}\{\varphi_s(X,H(Z))|W(\widehat \beta)=w+u h\}\,f_{W(\widehat \beta)}(w+ u h)\,\\
    &- \mathbb{E}\{\varphi_s(X,H(Z))|W(\widehat \beta)=w\}\,f_{W(\widehat \beta)}(w) \,\Big]\,du\\
    &+ \mathbb{E}\{\varphi_s(X,H(Z))|W(\widehat \beta)=w\}\,f_{W(\widehat \beta)}(w) + O_P(\widetilde{d}_G)\\
    =& \mathbb{E}\{\varphi_s(X,H(Z))|W(\widehat \beta)=w\}\,f_{W(\widehat \beta)}(w) + O_P(\widetilde{d}_G+h)\, ,
\end{align*}
uniformly in $(w,s)\in\mathcal{W}\times\mathcal{S}$,
where the expectation $\mathbb{E}\{\varphi_s(X,H(Z))|W(\widehat \beta)=w\}$ considers as random only $(X,Z)$ but not $\widehat \beta$, and in the last equality of the above display we have used the Lipschitz condition from Assumption \ref{Assumption: smoothness}(vi).  
Thus, the previous display gives
\begin{equation*}  \sup_{s\in\mathcal{S}}\|\widehat{T}^{\varphi_s}_{\widetilde{W}(\widehat \beta)}- \mathbb{E}\{\varphi_s(X,H(Z))|W(\widehat \beta)=\cdot\}\,f_{W(\widehat \beta)}\|_{\infty,\mathcal{W}}=O_P(\widetilde{d}_G+h)\, .
\end{equation*}
Next, by arguing similarly as for (\ref{eq: uniform convergence of Ghat uniform in beta and w}) we obtain 
\begin{equation}\label{eq: conv rate of iotahat without root n normality of betahat}
\sup_{s\in\mathcal{S}}\|(\widehat{\iota}_s-\mathbb{E}\{\varphi_s(X,H(Z))|W(\widehat \beta)=\cdot\})t_{W(\widehat \beta)}^{\eta}\|_{\infty,\mathcal{W}}=O_P\left(\frac{\widetilde{d}_G+h}{\tau_n}\right)\, .
\end{equation}
Also, wpa1 $t^{\delta}(x,z)\leq t^\delta(x,z)\,\mathbb{I}\{f_{W(\widehat \beta)}(q(\widehat \beta,x,H(z)))$ $\geq \eta(\delta)\tau_n\}$ $\leq t^\delta(x,z)\,\mathbb{I}\{f_{W(\widehat \beta)}(q(\widehat \beta,x,\widetilde H(z)))$ $\geq \eta(\delta)\tau_n/2\}$ for all $(x,z)\in Supp(X,Z)$, where the first inequality follows from Assumption \ref{Assumption: trimming}(i), while the second inequality is due to Lemma \ref{lem: trimming implications}(iv). By  these inequalities, we obtain that wpa1 
\begin{align*}
    \Big|\widehat \iota_s(\widetilde{W}_i(\widehat \beta))-\mathbb{E}\{\varphi_s&(X,H(Z))|W(\widehat \beta)=W_i(\widehat \beta)\}\Big|\,t_i^\delta\\
    \leq& \left|\widehat \iota_s(\widetilde{W}_i(\widehat \beta))-\mathbb{E}\{\varphi_s(X,H(Z))|W(\widehat \beta)=\widetilde{W}_i(\widehat \beta)\}\right|\,\mathbb{I}\{f_{W(\widehat \beta)}(\widetilde{W}_i(\widehat \beta))\geq \eta(\delta)\tau_n/2\}\\
    &+ \left|\mathbb{E}\{\varphi_s(X,H(Z))|W(\widehat \beta)=\widetilde{W}_i(\widehat \beta)\} - \mathbb{E}\{\varphi_s(X,H(Z))|W(\widehat \beta)=W_i(\widehat \beta)\}\right|\,t_i^\delta\\
    \lesssim & \|(\widehat \iota _s - \mathbb{E}\{\varphi_s(X,H(Z))|W(\widehat \beta)=\cdot\})\,t_{W(\widehat \beta)}^{\eta(\delta)/2}\|_{\infty,\mathcal{W}}\\
    &+ \|\widetilde{W}_i(\widehat \beta)-W_i(\widehat \beta)\|\,t_i^\delta\\
    \lesssim & \|(\widehat \iota _s - \mathbb{E}\{\varphi_s(X,H(Z))|W(\widehat \beta)=\cdot\})\,t_{W(\widehat \beta)}^{\eta(\delta)/2}\|_{\infty,\mathcal{W}}\\
    &+ \|(\widetilde H - H)t^\delta\|_\infty\\
    =& O_P\left(\frac{\widetilde d _G +h}{\tau_n}\right)\,
\end{align*}
uniformly in $i=1,\ldots,n$ and $s\in\mathcal{S}$,
where for the second inequality we have used the Lipschitz condition from Assumption \ref{Assumption: smoothness}(vi), for the third inequality we have used the Lipschitz condition from Assumption \ref{Assumption: smoothness}(v), and the for the last equality we have used (\ref{eq: conv rate of iotahat without root n normality of betahat}) and Lemma \ref{lem: trimming and convergence rates for 1st step}(iii). By (\ref{eq: n 1/4 convergence of dtilde G / tau 2}) and Assumption \ref{Assumption: trimming}(ii) 
we get that that the RHS of the above display is $o_P(1)$. So, result $(iii)$ with $t_i$ is obtained by taking $\delta=1$ in the previous display. To get the result with $\widehat t_i$, we first notice that wpa1 $\widehat t(x,z)\leq t^{(1/2)}(x,z)$   for all $(x,z)\in Supp(X,Z)$, as obtained earlier. Then, we evaluate the previous display at $\delta=1/2$.


$(iv)$ By the same reasoning as in (\ref{eq: uniform convergence of f W tilde beta}) we get 

\begin{equation}\label{eq: uniform convergence of fhat Wtilde star beta}
    \sup_{\beta\in B}\|\widehat f_{\widetilde W ^* (\beta)}-f_{W(\beta)}\|_{\infty,\mathcal W}=O_P(\widetilde d_G)\, .
\end{equation}
Combining the above display with the arguments used in the proof of $(i)$ leads to the desired result.

$(v)$ Let us recall that $\widehat G ^*_{\widetilde W ^*(\beta)}(\widetilde W ^* (\beta))=\widehat T ^{Y^*}_{\widetilde W (\beta)}/\widehat{f}_{\widetilde W ^* (\beta)}$. For the numerator we have 
    \begin{align*}
    \widehat T ^{Y^*}_{\widetilde W ^* (\beta)}(w)=& h^{-d}\mathbb P _n [G_{W(\widehat \beta)}(W(\widehat \beta)+\xi (Y-G_{W(\widehat \beta)}(W(\widehat \beta))]\,\widehat t\, K_h(w-\widetilde W ^*(\beta))\\
    &+ h^{-d}\mathbb P _n  (1-\xi)[\widehat G _{\widehat W (\widehat \beta)}(\widehat W (\widehat \beta))-G_{W(\widehat \beta)}(W(\widehat \beta))]\,\widehat t \, K_h(w-\widetilde W ^*(\beta))\\
    =&\mathbb E \{G_{W(\widehat \beta)}(W(\widehat \beta))|W(\beta)=w\}\, f_{W(\beta)}(w) \\
    &+ O_P(\widetilde d_G \tau_n^{-1}+h)\, , 
\end{align*}
uniformly in  $(w,\beta)\in \mathcal W \times  B$, where the expectation $\mathbb E \{G_{W(\widehat \beta)}(W(\widehat \beta))|W(\beta)=w\}$ considers as random only $(X,Z)$ but not $\widehat \beta$.
By reasoning as in the proof of Lemma \ref{lem: trimming implications}(iii), 
the second term after the first equality is $O_P(\max_{i=1,\ldots,n}$ $|\widehat G _{\widehat W (\widehat \beta)}(\widehat W _i (\widehat \beta))$ $-G_{W(\widehat \beta)}(W_i(\widehat \beta))|$ $\widehat t_i)$ which in turn is of order $\widetilde d_G \tau_n^{-1}$ by $(i)$ of the present lemma. To handle the first term after the first equality, we proceed similarly as in the proof of $(iii)$ of the present lemma.
So, we obtain that
\begin{align*}
 \sup_{\beta\in B}  \| \widehat T ^{Y^*}_{\widetilde W ^* (\beta)}-\mathbb E \{G_{W(\widehat \beta)}(W(\widehat \beta))|W(\beta)=w\}\|_{\infty,\mathcal W}=O_P(\widetilde d_G \tau_n^{-1}+h)\, .
\end{align*}
The desired result then follows by combining the above display, (\ref{eq: uniform convergence of fhat Wtilde star beta}), and the arguments in the proof of $(i)$.

$(vii)$ By arguing as in (\ref{eq: uniform convergence of f W tilde beta}), we get

\begin{equation}\label{eq: uniform convergence of fhat W tilde betahat}
\|\widehat f _{\widetilde W (\widehat \beta)}-f_{W(\beta_0)}\|_{\infty, \mathcal{W}}=O_P(\widetilde d_G)    
\end{equation}
 Also, from Assumption \ref{Assumption: trimming}(i) for all $n$ large enough $t^\delta(x,z)\leq \mathbb I \{f_{W(\beta_0)}(q(\beta_0,x,H(Z)))\geq \eta(\delta)\tau_n\}$ for all $(x,z)\in Supp(X,Z)$. Hence, 
\begin{equation}
    \max_{i=1,\ldots,n}|\widehat f _{\widetilde W (\widehat \beta)}( W _i (\beta_0))-f_{W(\beta_0)}(W_i(\beta_0))|t^\delta_i \leq \|\widehat f _{\widetilde W (\widehat \beta)}-f_{W(\beta_0)}\|_{\infty,\mathcal W}=O_P(\widetilde d_G)\, .
\end{equation}
Taking $\delta=1$ delivers the part of the result with $t_i$. To get the part of the result involving $\widehat t_i$, as already noticed earlier wpa1 $\widehat t(x,z)\leq t^{(1/2)}(x,z)$ for all $(x,z)\in Supp(X,Z)$. So the result follows from the previous display.

$(viii)\text{  and } (ix)$ Arguments similar to those used in (\ref{eq: uniform convergence of f W tilde beta}) lead to
\begin{equation}\label{eq: uniform convergence of That Wtilde betahat}
    \|\widehat T^{Y}_{\widetilde W (\widehat \beta)}-T^{Y}_{W(\beta_0)}\|_{\infty,\mathcal W}=O_P(\widetilde d_G)\, .
\end{equation}
Combining Lemma \ref{lem: trimming implications}(v) with (\ref{eq: uniform convergence of fhat W tilde betahat}), (\ref{eq: uniform convergence of That Wtilde betahat}), and the arguments used for (\ref{eq: uniform convergence of Ghat uniform in beta and w}) gives
\begin{equation}\label{eq: uniform convergene of Ghat Wtilde betahat}
    \|(\widehat G _{\widetilde W (\widehat \beta)}-G_{W(\beta_0)})t^\eta_{W(\beta_0)}\|_{\infty,\mathcal{W}}=O_P(\widetilde d_G \tau_n^{-1})\, .
\end{equation}
Moreover, for each $n$ large enough $t^\delta(x,z)$ $\leq \mathbb I \{f_{W(\beta_0)}(q(\beta_0,x,H(z)))\geq \eta(\delta) \tau_n\} $   for all $(x,z)\in Supp(X,Z)$, see Assumption \ref{Assumption: trimming}(i). Hence, 
\begin{align*}
    |\widehat G_{\widetilde W (\widehat \beta)}(W_i(\beta_0))-G_{w(\beta_0)}(W_i(\beta_0))|t^\delta _i \leq& |\widehat G_{\widetilde W (\widehat \beta)}(W_i(\beta_0))-G_{w(\beta_0)}(W_i(\beta_0))|t^{\eta(\delta)}_{W(\beta_0)}(W_i(\beta_0))\\
    \leq & \|\widehat G_{\widetilde W (\widehat \beta)}-G_{W(\beta_0)})t_{W(\beta_0)}^{\eta(\delta)}\|_{\infty,\mathcal{W}}=O_P(\widetilde d_G \tau_n^{-1})\, ,
\end{align*}
where  in the last equality we have used (\ref{eq: uniform convergene of Ghat Wtilde betahat}). Taking $\delta=1$ gives the desired result with $t_i$. To get the result with $\widehat t_i$, we recall that wpa1 $\widehat t (x,z)\leq t^{(1/2)}(x,z)$ for all $(x,z)\in Supp(X,Z)$. Using this and the previous display evaluated at $\delta=1/2$ gives the result for $\widehat t_i$. The result for $(ix)$ follows from analogous arguments.

$(x)\, , \, (xi)\, , \, \text{ and }(xii)$ These results have already been shown earlier in this proof.  

$(xiii)$ We have that wpa1 (see the comments ahead) $t^\delta(x,z)\leq  t^\delta(x,z) \mathbb I \{f_{W(\beta_0)}(q(\beta_0,x,H(z)))$ $\geq \eta(\delta) \tau_n\}$ $\leq t^\delta(x,z) \mathbb I \{f_{W(\beta_0)}(q(\widehat \beta, x , \widetilde H (z)))\geq \eta(\delta)\tau_n/2\}$ for all $(x,z)\in Supp(X,Z)$, where the first inequality follows from Assumption \ref{Assumption: trimming}(i), while the second inequality is obtained from Lemma \ref{lem: trimming implications}(vi). This implies
\begin{align*}
    |\widehat f _{\widetilde W (\widehat \beta)}(\widetilde W _i (\widehat \beta))-f_{W(\beta_0)}(W_i(\beta_0))|t^\delta_i \leq & |f _{\widetilde W (\widehat \beta)}(\widetilde W _i (\widehat \beta)) - f_{W(\beta_0)}(\widetilde W _i(\widehat \beta))|\mathbb I \{f_{W(\beta_0)}(\widetilde W _i(\widehat \beta))\geq \eta (\delta) \tau_n/2\}\\
    &+ |f_{W(\beta_0)}(\widetilde W _i(\widehat \beta)) - f_{W(\beta_0)}(W_i(\beta_0)) |t^\delta _i \\
    \leq & \|\widehat f _{\widetilde W (\widehat \beta)}-f_{W(\beta_0)})t^{\eta(\delta)/2}_{W(\beta_0)}\|_{\infty,\mathcal{W}} \\ 
    &+ \|\partial f_{W(\beta_0)}\|_{\infty,\mathbb{R}^d} t^\delta_i \|\widetilde W_i (\widehat \beta)-W_i(\beta_0)\|\, .
\end{align*}
By (\ref{eq: uniform convergence of fhat W tilde betahat}) the first term on the RHS is $O_P(\widetilde d_G)$. For the second term, by the Lipschitz condition in Assumption \ref{Assumption: smoothness}(v) we have $\max_{i=1,\ldots,n}\|\widetilde W_i (\widehat \beta) - W_i (\beta_0)\|t^\delta _i =O_P(\|(\widetilde H - H) t^\delta\|_\infty+\|\widehat \beta - \beta_0\|)$ that is of order $\widetilde{d}_H$ thanks to Lemma \ref{lem: trimming and convergence rates for 1st step}(iii) and $\widehat \beta-\beta_0=O_P(n^{-1/2})$. Gathering results and using the definition of $\widetilde{d}_G$ from (\ref{eq: definition of dtilde G}) gives  
\begin{equation*}
    \max_{i=1,\ldots,n}|\widehat f _{\widetilde W (\widehat \beta)}(\widetilde W _i (\widehat \beta))-f_{W(\beta_0)}(W_i(\beta_0))|t^\delta_i=O_P(\widetilde d_G)\, .
\end{equation*}
Taking $\delta=1$ gives the result for $t_i$. The result for $\widehat t _i$ is obtained by using the previous display and the fact that wpa1 $\widehat t(x,z)\leq t^{(1/2)}(x,z)$ for all $(x,z)\in Supp(X,Z)$, as already seen earlier.

$(xiv) \text{ and } (xv)$ As already noticed in the proof of $(xiii)$, wpa1 $t^\delta(x,z)\leq t^\delta(x,z)$ $\mathbb I \{f_{W(\beta_0)}(q(\widehat \beta,$ $x,\widetilde H(z)))\geq \eta(\delta) \tau_n/2\}$ for all $(x,z)\in Supp(X,Z)$. So, 
\begin{align*}
    |\widehat G _{\widetilde W (\widehat \beta)}(\widetilde W _i (\widehat \beta)  - G_W(\beta_0)(W_i (\beta_0))|&t^\delta _i \\
    =& |\widehat G _{\widetilde W (\widehat \beta)}(\widetilde W _i (\widehat \beta) - G_{W(\beta_0)}(\widetilde W_i(\widehat \beta))|\, \mathbb I \{f_{W(\beta_0)}(\widetilde W _i (\beta_0))\geq \eta(\delta)\tau_n/2\} \,  \\
    &+ | G_{W(\beta_0)}(\widetilde W_i(\widehat \beta)) - G_W(\beta_0)(W_i (\beta_0)) |\, t^\delta _i\\
    \leq & \|(\widehat G _{\widetilde W (\widehat \beta)})-G_{W(\beta_0)})t^{\eta(\delta)/2}_{W(\beta_0)}\|_{\infty,\mathcal{W}}\\
    &+ \|\partial G _{W(\beta_0)}\|_{\infty,\mathbb{R}^d}\, \|\widetilde W _i (\widehat \beta) - W_i(\beta_0)\|t^\delta_i\, .
\end{align*}
By (\ref{eq: uniform convergene of Ghat Wtilde betahat}), the first term on the RHS is $O_P(\widetilde d_G \tau_n^{-1})$, while as already obtained earlier in this proof the second term is $O_P(\widetilde d_H)$ uniformly in $i=1,\ldots,n$. Hence, setting $\delta=1$ gives the result for $t_i$. The result for $\widehat t_i$ can be obtained by the previous display and the fact that wpa1 $\widehat t(x,z)\leq t^{(1/2)}(x,z)$ for all $(x,z)\in Supp(X,Z)$. The proof for $(xv)$ follows from analogous arguments.

$(xvi)$ Let us recall that $\widehat G ^* _{\widetilde W ^* (\beta)}(w)=\widehat T ^{Y^*}_{\widetilde W ^* (\beta)}(w)/\widehat f _{\widetilde W ^* (\beta)}(w)$. 
For the numerator of $\widehat G ^* _{\widetilde W ^* (\beta)}(w)$, we have (see the comments below)
\begin{align*}
    \widehat T ^{Y^*}_{\widetilde W ^* (\beta)}(w)=& h^{-d}\mathbb P _n [G_{W(\beta_0)}(W(\beta_0)+\xi (Y-G_{W(\beta_0)}(W(\beta_0))]\,\widehat t\, K_h(w-\widetilde W ^*(\beta))\\
    &+ h^{-d}\mathbb P _n  (1-\xi)[\widehat G _{\widehat W (\widehat \beta)}(\widehat W (\widehat \beta))-G_{W(\beta_0)}(W(\beta_0))]\,\widehat t \, K_h(w-\widetilde W ^*(\beta))\\
    =&\mathbb E \{G_{W(\beta_0)}(W(\beta_0))|W(\beta)=w\}\, f_{W(\beta)}(w) \\
    &+ O_P(\widetilde d_G \tau_n^{-1})\, , 
\end{align*}
uniformly in  $(w,\beta)\in \mathcal W \times  B$.
By reasoning as in (\ref{eq: uniform convergence of f W tilde beta}) and using $(xiv)$ of the present lemma, the second term after the first equality is $O_P(\max_{i=1,\ldots,n}$ $|\widehat G _{\widehat W (\widehat \beta)}(\widehat W _i (\widehat \beta))$ $-G_{W(\beta_0)}(W_i(\beta_0))|$ $\widehat t_i)$, and hence of order $\widetilde d_G \tau_n^{-1}$. For the first term after the first equality, we proceed similarly as in (\ref{eq: uniform convergence of f W tilde beta}). So, we obtain that 
\begin{equation}\label{eq: uniform convergence of That star Wtilde star beta}
    \sup_{\beta\in B}\|\widehat T ^{Y^*}_{\widetilde W ^* (\beta)}-\mathbb{E}\{G_{W(\beta_0)}(W(\beta_0))|W(\beta)=\cdot\}\,f_{W(\beta)}\|_{\infty,\mathcal W}=O_P(\widetilde d_G \tau_n^{-1})\, .
\end{equation}
Finally, given (\ref{eq: uniform convergence of fhat Wtilde star beta}) and (\ref{eq: uniform convergence of That star Wtilde star beta}) we can follow the same reasoning as in the proof of $(i)$ to obtain the desired result. 

$(xvii)-(xxv)$ These remaining results follow from arguments analogous to the proof of $(i)-(xvi)$.

\end{proof}

Before introducing the following lemma, we show that 
\begin{equation}\label{eq: negligibility of convergence rate for the derivatives}
    \frac{\widetilde{d}_G}{h^{|l|} \tau_n^{|l|+1}}=o(1)\text{ for all }|l|=1,\ldots,\lceil(d+1)/2\rceil+1\, ,
\end{equation}
where $\widetilde{d}_G=d_G+p_n/h^d+\sum_{j=1}^4(\widetilde{d}_H/h)^j+\widetilde{d}_H^5/h^{d+5}$ from (\ref{eq: definition of dtilde G}). First, 
$$\frac{d_G}{h^{\lceil(d+1)/2\rceil+1}\tau_n^{\lceil(d+1)/2\rceil+2}}\leq \frac{d_G}{h^{\lceil(d+1)/2\rceil+1}\tau_n^{2(\lceil(d+1)/2\rceil+1)}}=o(1)\, ,$$
where the last equality follows from Assumption \ref{Assumption: bandwidth}(i).
Also, from Assumption \ref{Assumption: trimming}(iii)
$$\frac{p_n}{h^{d+\lceil(d+1)/2\rceil+1}\tau_n^{\lceil(d+1)/2\rceil+2}}=o(1)\, .$$
Next, since $\widetilde{d}_H=d_H/\tau_n^3+p_n/(h_H^p \tau_n^2)$ from (\ref{eq: definition of dtilde H}), we have 
\begin{align*}
    \frac{\widetilde{d}_H}{h^{\lceil(d+1)/2\rceil+2}\tau_n^{\lceil(d+1)/2\rceil+2}}=&\frac{d_H}{\tau_n^5 h^2}\cdot\frac{1}{\tau_n^{\lceil(d+1)/2\rceil} h^{\lceil(d+1)/2\rceil}}+ \frac{p_n}{h_H^p \tau_n^{\lceil(d+1)/2\rceil+4} h^{\lceil(d+1)/2\rceil+2}}\\
    =& o(n^{-1/4})\cdot \frac{1}{\tau_n^{\lceil(d+1)/2\rceil} h^{\lceil(d+1)/2\rceil}} +o(1)\\
    =& o(1)\, .
\end{align*}
where the second line is a direct consequence of  Assumptions \ref{Assumption: bandwidth}(ii) and \ref{Assumption: trimming}(iii), while the third line follows from Assumption \ref{Assumption: bandwidth}(i).\footnote{Specifically, to obtain the third line, it is sufficient to show that $n h^{4\lceil(d+1)/2\rceil}\tau_n^{4\lceil(d+1)/2\rceil}\rightarrow \infty$. Notice first that from Assumption \ref{Assumption: bandwidth}(i) $d_G h^{-\lceil(d+1)/2\rceil-1}\tau_n^{-2\lceil(d+1)/2\rceil-2}=o(1)$ which implies $n h^{d+2+2\lceil(d+1)/2\rceil}\tau_n^{4\lceil(d+1)/2\rceil+4}\rightarrow \infty$. Thus, when $d$ is odd, $n h^{4\lceil(d+1)/2\rceil}\tau_n^{4\lceil(d+1)/2\rceil}$ $=n h^{2 d +2} \tau_n^{2d+2}$ $\geq n h^{2 d +3} \tau_n^{2d + 6} $ $=n h^{d+2+2\lceil(d+1)/2\rceil}\tau_n^{4\lceil(d+1)/2\rceil+4}\rightarrow \infty$. When $d$ is even,  $n h^{4\lceil(d+1)/2\rceil}\tau_n^{4\lceil(d+1)/2\rceil}$  $=n h^{2 d + 4} \tau_n^{2 d + 4}$ $\geq n h^{2 d + 4} \tau_n^{2 d + 8}$ $=n h^{d+2+2\lceil(d+1)/2\rceil}\tau_n^{4\lceil(d+1)/2\rceil+4}\rightarrow \infty$.}
The above display also implies that 
\begin{align*}
    \sum_{j=1}^4\left(\frac{\widetilde{d}_H}{h}\right)^j\frac{1}{h^{\lceil(d+1)/2\rceil+1}\tau_n^{\lceil(d+1)/2\rceil+2}}=o(1)\, .
\end{align*}
Finally, since $\widetilde{d}_H/(h \tau_n^2)=o(n^{-1/4})$ (see Footnote \ref{footnote: 1st negligibility of convergence rate}), we have
\begin{align*}
    \frac{\widetilde{d}^5_H}{h^{d+6+\lceil(d+1)/2\rceil}\tau_n^{\lceil(d+1)/2\rceil+2}}=&\left(\frac{\widetilde{d}_H}{h \tau_n^2}\right)^4\cdot \frac{\widetilde{d}_H}{h \tau_n^2}\cdot \frac{\tau_n^8}{h^{d+\lceil(d+1)/2\rceil+1}\tau_n^{\lceil(d+1)/2\rceil}}\\
    =& \frac{o(n^{-1/4})}{ n h^{d+\lceil(d+1)/2\rceil+1}\tau_n^{\lceil(d+1)/2\rceil}}\\
    =&o(1)\, ,
\end{align*}
where the last equality is a direct consequence of $n h^{d+2+2\lceil(d+1)/2\rceil}\tau_n^{4\lceil(d+1)/2\rceil+4}\rightarrow \infty$ which is implied by $d_G h^{-\lceil(d+1)/2\rceil-1}\tau_n^{-2\lceil(d+1)/2\rceil-2}=o(1)$ (see Assumption \ref{Assumption: bandwidth}(i)). Gathering results gives Equation (\ref{eq: negligibility of convergence rate for the derivatives}).


\begin{lem}\label{lem: uniform convergence rates for derivatives}
Let Assumptions \ref{Assumption: iid and phi}-\ref{Assumption: trimming} hold. Then, for $|l|=1,\ldots,\lceil(d+1)/2\rceil+1$ we have
\begin{enumerate}[label=(\roman*)]
    \item $\sup_{\beta\in B}\|(\partial^l \widehat G _{\widetilde W ( \beta)}-\partial^l G_{W(\beta)})t^\eta_{W(\beta)}\|_{\infty,\mathcal W}=O_P(\widetilde d_G h^{-|l|}  \tau_n^{-|l|})=o_P(1)$\,,
\end{enumerate}
and for $|l|=0,1,\ldots,\lceil(d+1)/2\rceil$ we have
\begin{enumerate}[label=(\roman*)]
\setcounter{enumi}{1}
    \item $\sup_{\beta \in B}\|(\partial_\beta \partial^l \widehat G _{\widetilde W (\beta)}- \partial_\beta \partial^l G_{W(\beta)})t_{W(\beta)}^\eta\|_{\infty,\mathcal{W}}=O_P(\widetilde d_G h^{-|l|-1} \tau_n^{-|l|-1})=o_P(1)$\,, 
    \item $\sup_{\beta\in B}\|(\partial^2_{\beta,\beta^T}\widehat G _{\widetilde W (\beta)}-\partial^2_{\beta,\beta^T} G _{ W (\beta)})t^\eta_{W(\beta)}\|_{\infty,\mathcal{W}}=O_P(\widetilde d_G h^{-|2|} \tau_n^{-2})=o_P(1)$\,.
\end{enumerate}
If moreover $\widehat \beta-\beta_0=O_P(n^{-1/2})$ and $\mathcal{H}_0$ holds, then for $|l|=1,\ldots,\lceil (d+1)/2 \rceil$ we have
\begin{enumerate}[label=(\roman*)]
\setcounter{enumi}{3}
    \item $\|\partial^l \widehat f _{\widetilde W (\widehat \beta)}-\partial^l f_{W(\beta_0)}\|_{\infty,\mathcal W}=O_P(\widetilde d_G h^{-|l|})=o_P(1)$\,,
    \item $\|(\partial^l \widehat G _{\widetilde W (\widehat \beta)}-\partial^l G_{W(\beta_0)})t^\eta_{W(\beta_0)}\|_{\infty,\mathcal{W}}=O_P(\widetilde d_G h^{-|l|} \tau_n^{-|l|})=o_P(1)$\,,
    \item $\sup_{s\in\mathcal{S}} \|(\partial^l \widehat \iota _s-\partial^l \iota_s)t^\eta_{W(\beta_0)}\|_{\infty,\mathcal{W}}=O_P(\widetilde d_G h^{-|l|}\tau_n^{-|l|})=o_P(1)$\,,
    \end{enumerate}
    and for $|l|=0,1,\ldots,\lceil (d+1)/2 \rceil$ we have 
    \begin{enumerate}[label=(\roman*)]
\setcounter{enumi}{6}
    \item $\sup_{\beta\in B}\|(\partial^l \partial_\beta \widehat G ^*_{\widetilde W ^*(\beta)}-\partial^l \partial_\beta G_{W(\beta)})t^\eta_{W(\beta)}\|_{\infty,\mathcal{W}}=O_P(\widetilde d_G h^{-|l|-1}\tau_n^{-|l|-2})=o_P(1)$\,,
    \item $\|(\partial^l \partial_\beta \widehat G ^*_{\widetilde W ^*(\widehat \beta)}-\partial^l \partial_\beta G_{W(\beta_0)})t^\eta_{W(\beta_0)}\|_{\infty,\mathcal{W}}=O_P(\widetilde d_G h^{-|l|-1}\tau_n^{-|l|-2})=o_P(1)$\,,
    \item $\sup_{\beta\in B}\|\partial^2_{\beta \beta^T}\widehat G ^*_{\widetilde W ^*(\beta)}-\partial^2_{\beta \beta^T}G_{W(\beta)})t^\eta_{W(\beta)}\|_{\infty,\mathcal{W}}=O_P(\widetilde d_G h^{-2}\tau_n^{-3})=o_P(1)$\,,
    \end{enumerate}
    and for $|l|=1,\ldots,\lceil (d+1)/2\rceil+1$ we have 
    \begin{enumerate}[label=(\roman*)]
\setcounter{enumi}{9}
        \item $\sup_{\beta\in B}\|(\partial^l \widehat G ^*_{\widetilde W^*(\beta)}-\partial^l G_{W(\beta)})t^\eta_{W(\beta)}\|_{\infty,\mathcal{W}}=O_P(\widetilde d_G h^{-|l|}\tau_n^{-|l|-1})=o_P(1)$\,,
        \item $\|(\partial^l \widehat G^*_{\widetilde W ^*(\widehat \beta)}-\partial^l G_{W(\beta_0)})t^\eta_{W(\beta_0)}\|_{\infty,\mathcal{W}}=O_P(\widetilde d_G h^{-|l|}\tau_n^{-|l|-1})=o_P(1)$\,.
    \end{enumerate}
    If moreover $\widehat \beta^*-\beta_0=O_P(n^{-1/2})$ then for $|l|=1,\ldots,\lceil(d+1)/2\rceil+1$ we have

    \begin{enumerate}[label=(\roman*)]
    \setcounter{enumi}{11}
    \item $\|\partial^l \widehat f_{\widetilde W ^*(\widehat \beta^*)}-\partial^l f_{W(\beta_0)}\|_{\infty,\mathcal{W}}=O_P(\widetilde d_G h^{-|l|})=o_P(1)$\,,
     \item $\sup_{s\in\mathcal{S}} \|(\partial^l \widehat \iota^*_s-\partial^l \iota_s)\,t^{\eta}_{W(\beta_0)}\|_{\infty,\mathcal{W}}=O_P(\widetilde d_G h^{-|l|} \tau_n^{-|l|})=o_P(1)$\,,
     \item $\|(\partial^l \widehat G^* _{\widetilde W^* (\widehat \beta^*)}-\partial^l G_{W(\beta_0)})t^\eta_{W(\beta_0)}\|_{\infty,\mathcal{W}}=O_P(\widetilde d_G h^{-|l|} \tau_n^{-|l|-1})=o_P(1)$\,.
\end{enumerate}

\end{lem}

\begin{proof}
$(i)$ Let us define 
\begin{equation*}
    K_h^{(l)}(w):=\partial^{l}K\left(\frac{w}{h}\right)\, .
\end{equation*}
As already noticed in the proof of Lemma \ref{lem: trimming implications}(iii), $\sup_{\beta\in B}\max_{i=1,\ldots,n}\|\widetilde W _i (\beta) - W_i(\beta)\|\widehat t_i=O_P(\|(\widetilde H - H)\widehat t\|_\infty)$. So, for any $|l|=1,\ldots,\lceil (d+1)/2 \rceil$, by a 4th order Taylor expansion and Lemma \ref{lem : Li Racine}(ii) we have 
\begin{align*}
\partial^l \widehat T ^{Y}_{\widetilde W (\beta)}(w)=& h^{-(d+l)}\mathbb P _n Y \widehat t K_h^{(l)}(w-W(\beta)) \\
&+ O_P\left(\frac{\|(\widetilde H - H)\widehat t\|_\infty^5}{h^{d+l+5}}+\sum_{j=1}^4 \frac{\|(\widetilde H - H)\widehat t\|_\infty^j}{h^{j+l}}\right)
\end{align*}
uniformly in $(w,\beta)\in\mathcal W\times B $. By Lemma \ref{lem: trimming and convergence rates for 1st step}(iii) $\|(\widetilde H - H)\widehat t\|_\infty=O_P(\widetilde d _H)$. Thus, by Equation (\ref{eq: definition of dtilde G}) the reminder term is $O_P(\widetilde{d}_G/h^{|l|})$. Let us now consider the first leading term on the RHS. By Lemma \ref{lem: trimming and convergence rates for 1st step}(i) the trimming $\widehat t$ in the first leading term can be replaced with 1 at the cost of an $O_P(p_n h^{-(d+l)})$ reminder. Then, by Lemma \ref{lem : Li Racine}(i)
\begin{equation*}
    h^{-(d+l)}\mathbb P _n Y  K_h^{(l)}(w-W(\beta))=h^{-(d+l)} P  Y  K_h^{(l)}(w-W(\beta))+O_P\left(\sqrt{\frac{\log n}{n h^{d+2l}}}\right)\,
\end{equation*}
uniformly in $(w,\beta)\in\mathcal W\times B$. For the leading term we have (see the comments ahead)
\begin{align*}
    h^{-(d+l)} P  Y  K_h^l&(w-W(\beta))=  h^{-d} \partial^l P  Y  K_h(w-W(\beta))\\
    = &  h^{-d}\partial^l \int_{ } (G_{W(\beta)}f_{W(\beta)})(\overline w) K_h(w-\overline w)d \overline w
    = \partial^l \int_{ } (G_{W(\beta)}f_{W(\beta)})(w+u h) K(u) d u\\
    &= \int_{ } (G_{W(\beta)}f_{W(\beta)})^{(l)}(w+u h) K(u) d u
    = \partial^l T^{Y}_{W(\beta)}(w)+O(h^{r-|l|})\, 
\end{align*}
uniformly in $(w,\beta)\in\mathcal W\times B$, where $T^Y_{W(\beta_0)}$ is defined in (\ref{eq: definition of T Y W(beta)}). To obtain the first equality we have exchanged integral and differentiation, thanks to the Lebesgue Dominated Convergence Theorem. For the second equality we have applied the Law of Iterated Expectations. For the third equality we have made the change of variable  
typical of kernel bias manipulations. The fourth equality is obtained by exchanging again integral and differentiation. Finally, for the last equality we have used a Taylor expansion of order $r-|l|$. Gathering  results,
\begin{equation}\label{eq: uniform convergence of partial That Wtilde beta0}
   \sup_{\beta \in B} \|\partial^l \widehat T ^{Y}_{\widetilde W (\beta)}-\partial^l T ^{Y}_{W (\beta)} \|_{\infty,\mathcal W}=O_P(\widetilde d_G h^{-|l|} )\, .
\end{equation}
Similar arguments lead to
\begin{align}\label{eq: uniform convergence of partial fhat Wtilde beta0}
   \sup_{\beta \in B} \|\partial^l \widehat f_{\widetilde W (\beta)}-\partial^l f_{W (\beta)} \|_{\infty,\mathcal W}=&O_P(\widetilde d_G h^{-|l|})\,.
\end{align}
Let us now recall the set 
\begin{equation}
\mathcal{W}_{n,\beta}^{\eta/2}:=\left\{w\,:\, f_{W(\beta)}(w)\geq \eta \tau_n/2 \right\}\,.
\end{equation}
To prove $(i)$ it is sufficient to show that (a) wpa1 $\widehat G _{\widetilde W (\beta)}$ is $\lceil (d+1)/2\rceil+1$ times continuously differentiable over  $\mathcal{W}_{n,\beta}^{\eta/2}$ for all $\beta\in B$ and (b) for $|l|=1,\ldots,\lceil (d+1)/2 \rceil+1$ we have $\sup_{w\in\mathcal{W},\beta\in B}|\partial^l \widehat G _{\widetilde W (\beta)}(w)-\partial^l G_{W(\beta)}(w)|\,\mathbb I \{f_{W(\beta)}(w)\geq \eta \tau_n\}\leq C \widetilde d_G h^{-|l|}\tau_n^{-|l|}$ for a certain constant $C$, over an event having probability tending to one.
To show (a) we use Lemma \ref{lem: trimming implications}(iii) to get that
\begin{equation}\label{eq: inclusion of cal W n}
    \text{ wpa1 }\mathcal{W}_{n,\beta}^{\eta/2} \subset \left\{w\,:\, \widehat{f}_{\widetilde W (\beta)}(w)\geq \eta \tau_n/4 \right\}\text{ for all }\beta\in B\, .
\end{equation}
Given the smoothness of the kernel $K$ from Assumption \ref{Assumption : kernels}(i), $\widehat{G}_{\widetilde W (\beta)}$ is $\lceil (d+1)/2 \rceil+1$ times continuously differentiable over the set $\{w\,:\, \widehat{f}_{\widetilde W (\beta)}(w)\geq \eta \tau_n/2 \}$. So, (a) is proved. To prove (b), let us define the event
\begin{align}\label{eq: definition of cal D}
    \mathcal{D}^{C}_n :=& \{ \sup_{\beta \in B}\|(\widehat G _{\widetilde W (\beta)}-G_{W(\beta)})t_{W(\beta)}^\eta\|_{\infty,\mathcal{W}}\leq C \widetilde d_G\tau_n^{-1}\}\nonumber \\
    \cap & \Big\{\sup_{\beta \in B}\|\partial^l \widehat T ^{Y}_{\widetilde W (\beta)}-\partial^l T^{Y}_{W(\beta)}\|_{\infty,\mathcal W }\leq C \widetilde d_G h^{-|l|}  \text{ for }|l|\leq \lceil (d+1)/2 \rceil+1 \Big\}\nonumber \\
    \cap & \Big\{\sup_{\beta \in B}\| \partial^l \widehat f_{\widetilde W (\beta)}-\partial^l f_{W(\beta)}\|_{\infty,\mathcal W }\leq C \widetilde d_G h^{-|l|}\text{ for }|l|\leq \lceil (d+1)/2 \rceil +1\Big\}\nonumber \\
     \cap &  \Big\{\mathcal{W}_{n,\beta}^{\eta/2} \subset \{w:\widehat f _{\widetilde W (\beta)}(w)\geq \eta \tau_n/4\}\text{ for all }\beta\in B\Big\}
\end{align}
By (\ref{eq: uniform convergence of Ghat uniform in beta and w}),  (\ref{eq: uniform convergence of partial That Wtilde beta0}), (\ref{eq: uniform convergence of partial fhat Wtilde beta0}), and (\ref{eq: inclusion of cal W n}) by choosing $C$ large enough $\Pr(\mathcal{D}^C_n)$ can be made arbitrarily close to 1 for each large $n$. So, to prove (b) it is sufficient to show that $\mathcal{D}^C_n$ implies $\sup_{w\in\mathcal{W},\beta\in B}|\partial^l \widehat G _{\widetilde W (\beta)}(w)-\partial^l G_{W(\beta)}(w)|\,\mathbb I \{f_{W(\beta)}(w)\geq \eta \tau_n\}\leq C \widetilde d_G h^{-|l|} \tau_n^{-|l|}$ for all $|l|=1,\ldots,\lceil(d+1)/2\rceil+1$  for a certain constant $C$. To this end, let us assume that the event $\mathcal{D}^C_n$ holds. Then, by the Leibniz rule for derivation we get 
\begin{align}\label{eq: Leibniz rule with beta0}
    \partial^l \widehat G _{\widetilde W (\beta)}=& \frac{\partial^l \widehat T ^Y_{\widetilde W (\beta)}}{\widehat f _{\widetilde W (\beta)}} - \sum_{0\leq |j|\leq |l|-1}\binom{l_1}{j_1}\ldots \binom{l_d}{j_d}\frac{\partial^{l-j}\widehat f _{\widetilde W (\beta)}}{\widehat{f} _{\widetilde W (\beta)}}\partial^j\widehat G _{\widetilde W (\beta)} \nonumber \\
     \partial^l  G _{ W (\beta)}=& \frac{\partial^l  T ^Y_{ W (\beta)}}{ f _{ W (\beta)}} - \sum_{0\leq |j|\leq |l|-1}\binom{l_1}{j_1}\ldots \binom{l_d}{j_d}\frac{\partial^{l-j} f _{ W (\beta)}}{f _{ W (\beta)}}\partial^j G _{ W (\beta)} \,.
\end{align}
To show that $\sup_{w\in\mathcal{W},\beta\in B}|\partial^l \widehat G _{\widetilde W (\beta)}(w)-\partial^l G_{W(\beta)}(w)|\,\mathbb I \{f_{W(\beta)}(w)\geq \eta \tau_n\}\leq C \widetilde d_G h^{-|l|} \tau_n^{-|l|}$ for a certain constant $C$ we can now use (\ref{eq: definition of cal D}), (\ref{eq: Leibniz rule with beta0}), and $\widetilde{d}_G/(h^{|l|}\tau_n^{|l|})=o(1)$ for all $|l|=1,\ldots,\lceil(d+1)/2\rceil+1$ (see (\ref{eq: negligibility of convergence rate for the derivatives})), and proceed as in page 4  of the Supplementary Material of \citet{lapenta_encompassing_2022}. Finally, using (\ref{eq: negligibility of convergence rate for the derivatives}) concludes the proof of $(i)$. 

$(ii)$ For any $|l|=0,1,\ldots,\lceil (d+1)/2\rceil$ we have (see the comment ahead) 
\begin{align}\label{eq: partial l partial beta That Wtilde beta}
\partial^l \partial_\beta \widehat T ^{Y}_{\widetilde W (\beta)}(w)=& -h^{-(d+l+1)}\mathbb P _n Y \widehat t K_h^{(l+1)}(w-\widetilde W (\beta))\,\partial_\beta \widetilde W (\beta)\nonumber \\
=& -h^{-(d+l+1)}\mathbb P _n Y \widehat t K_h^{(l+1)}(w-W(\beta))\partial_\beta \widetilde W(\beta)\nonumber  \\
&+ O_P\left(\frac{\|(\widetilde H - H)\widehat t\|_\infty^5}{h^{d+l+6}}+\sum_{j=1}^4 \frac{\|(\widetilde H - H)\widehat t\|_\infty^j}{h^{j+l+1}}\right)\nonumber \\
=& -h^{-(d+l+1)}\mathbb P _n Y \widehat t K_h^{(l+1)}(w-W(\beta))\partial_\beta W(\beta) \nonumber \\
&+ O_P\left(\frac{\|(\widetilde H - H)\widehat t\|_\infty^5}{h^{d+l+6}}+\sum_{j=1}^4 \frac{\|(\widetilde H - H)\widehat t\|_\infty^j}{h^{j+l+1}}\right)\nonumber\\
=& -h^{-(d+l+1)}\mathbb P _n Y \widehat t K_h^{(l+1)}(w-W(\beta))\partial_\beta W(\beta) \nonumber \\
&+ O_P\left(\frac{\widetilde{d}_H^5}{h^{d+l+6}}+\sum_{j=1}^4 \frac{\widetilde{d}_H^j}{h^{j+l+1}}\right)\, 
\end{align}
uniformly in $(w,\beta)\in\mathcal W\times B$. The first equality is immediate. The second equality is obtained by a 4th order Taylor expansion of $K_h(w-\widetilde W _i (\beta))$ around $W_i(\beta)$, Lemma \ref{lem : Li Racine}(ii), $\sup_{\beta\in B}\max_{i=1,\ldots,n}| \widetilde W _i (\beta)- W_i(\beta)|\widehat t_i\leq C \|(\widetilde H - H )\widehat t\|_\infty$ (see Assumption \ref{Assumption: smoothness}(v)), and $\sup_{\beta\in B}\max_{i=1,\ldots,n}| \partial_\beta \widetilde W _i (\beta)-\partial_\beta W_i(\beta)|\widehat t_i\leq C \|(\widetilde H - H )\widehat t\|_\infty$ (see Assumption \ref{Assumption: smoothness}(v)). The third equality follows from $\sup_{\beta\in B}\max_{i=1,\ldots,n}| \partial_\beta \widetilde W _i (\beta)-\partial_\beta W_i(\beta)|\widehat t_i\leq C \|(\widetilde H - H )\widehat t\|_\infty$ and Lemma \ref{lem : Li Racine}(ii). Finally, the last equality is obtained from  $\|(\widetilde H - H)\widehat t\|_\infty=O_P(\widetilde d _H)$, see Lemma \ref{lem: trimming and convergence rates for 1st step}(iii). To handle the leading term on the RHS, we use arguments analogous to those used for (\ref{eq: uniform convergence of partial That Wtilde beta0}). This gives 
\begin{equation}\label{eq: uniform convergence of partial partial beta That Y Wtilde beta0}
   \sup_{\beta\in B} \|\partial^l \partial_\beta \widehat T ^Y _{\widetilde W (\beta)}-\partial^l \partial_\beta T^Y_{W(\beta)} \|_{\infty,\mathcal W }=O_P(\widetilde d_G h^{-|l|-1})\, \text{ for }|l|=0,1,\ldots,\lceil (d+1)/2\rceil\, .
\end{equation}
Similarly, 
\begin{equation}\label{eq: uniform convergence of partial partial beta fhat Y Wtilde beta0}
   \sup_{\beta \in B} \|\partial^l \partial_\beta \widehat f  _{\widetilde W (\beta)}-\partial^l \partial_\beta f_{W(\beta)} \|_{\infty,\mathcal W }=O_P(\widetilde d_G h^{-|l|-1})\, \text{ for }|l|=0,1,\ldots,\lceil (d+1)/2\rceil \, .
\end{equation}
From Lemma \ref{lem: trimming implications}(iii) wpa1 we have $t^\eta_{W(\beta)}(w)\leq \mathbb{I}\{\widehat f_{\widetilde W (\beta)}(w)\geq \eta \tau_n /2\}$ for all $(w,\beta)\in \mathcal W\times B$. So, wpa1
\begin{align*}
    |\partial_\beta \widehat G _{\widetilde W (\beta)}(w) -& \partial_\beta G_{W(\beta)}(w)|t^\eta_{W(\beta)}(w)= \Big| \frac{\partial_\beta \widehat T ^Y_{\widetilde W (\beta)}(w)}{\widehat f_{\widetilde W (\beta)}(w)}-\widehat G _{\widetilde W (\beta)}(w)\,\frac{\partial_\beta \widehat f _{\widetilde W (\beta)}(w)}{\widehat f _{\widetilde W (\beta)}(w)}\\
    &+ \frac{\partial_\beta  T ^Y_{ W (\beta)}(w)}{ f_{ W (\beta)}(w)}- G _{ W (\beta)}(w)\,\frac{\partial_\beta  f _{ W (\beta)}(w)}{ f _{ W (\beta)}(w)}\Big|t^\eta_{W(\beta)}(w)\,\mathbb I \{\widehat f_{\widetilde W (\beta)}(w)\geq \eta \tau_n/2\}\, 
\end{align*}
uniformly in $(w,\beta)\in\mathcal{W}\times B$. By the above equality, (\ref{eq: uniform convergence of f W tilde beta}), (\ref{eq: uniform convergence of Ghat uniform in beta and w}), (\ref{eq: uniform convergence of partial partial beta That Y Wtilde beta0}), and (\ref{eq: uniform convergence of partial partial beta fhat Y Wtilde beta0}) we obtain that 
\begin{equation}\label{eq: uniform convergence of partial beta Ghat Wtilde beta uniform in w and beta}
    \sup_{\beta\in B}\|(\partial_\beta \widehat G_{\widetilde W (\beta)}-\partial_\beta  G_{W(\beta)})t^\eta_{W(\beta)}\|_{\infty,\mathcal{W}}=O_P(\widetilde d_G \tau_n^{-1}h^{-1})\, .
\end{equation}
Now, the proof proceed similarly as the proof of $(i)$ by defining the set 
\begin{align*}
    \mathcal{E}^C_n :=&\{\sup_{\beta\in B}\|(\widehat G _{\widetilde W (\beta)}-G_{W(\beta)})t_{W(\beta)}^\eta \|_{\infty,\mathcal{W}} \leq C \widetilde d_G \tau_n^{-1}\}\\
    \cap&\left\{\sup_{\beta\in B}\|\widehat{f}_{\widetilde{W}(\beta)}-f_{W(\beta)}\|_{\infty,\mathcal{W}}\leq C \widetilde{d}_G\right\}\\
    \cap&\left\{\sup_{\beta\in B}\|(\partial_\beta \widehat G _{\widetilde W (\beta)}-\partial_\beta G_{W(\beta)})t^\eta_{W(\beta)}\|_{\infty,\mathcal{W}}\leq C \widetilde d_G h^{-1} \tau_n^{-1}\right\}\\
    \cap&\left\{\sup_{\beta\in B}\|(\partial^l \widehat G _{\widetilde W (\beta)}-\partial^l G_{W(\beta)})t^\eta_{W(\beta)}\|_{\infty,\mathcal W }\leq C \widetilde d_G h^{-|l|}\tau_n^{-|l|}\text{ for }1\leq |l|\leq \lceil (d+1)/2 \rceil\right\}\\
    \cap & \left\{\sup_{\beta\in B}\|\partial^l  \widehat f _{\widetilde W (\beta)}-\partial^l  f_{W(\beta)}\|_{\infty,\mathcal W }\leq C \widetilde d_G h^{-|l|}\text{ for }1\leq |l|\leq \lceil (d+1)/2 \rceil \right\}\\
    \cap & \left\{\sup_{\beta\in B}\|\partial^l \partial_\beta \widehat T ^Y _{\widetilde W (\beta)}-\partial^l \partial_\beta T ^Y_{W(\beta)}\|_{\infty,\mathcal W }\leq C \widetilde d_G h^{-|l|-1}\text{ for }|l|\leq \lceil (d+1)/2 \rceil \right\}\\
    \cap & \left\{\sup_{\beta \in B}\|\partial^l \partial_\beta \widehat f _{\widetilde W (\beta)}-\partial^l \partial_\beta f_{W(\beta)}\|_{\infty,\mathcal W }\leq C \widetilde d_G h^{-|l|-1}\text{ for }|l|\leq \lceil (d+1)/2 \rceil \right\}\\
    \cap& \left\{\mathcal{W}^{\eta/2}_{n,\beta} \subset \{w\,:\, \widehat f_{\widetilde W (\beta)}(w)\geq \eta \tau_n/4\}\text{ for all }\beta\in B\right\}
\end{align*}
and using the following equations (obtained by the Leibniz rule of derivation)
\begin{align*}
    \partial_\beta \partial^l \widehat G _{\widetilde W (\beta)}=& \frac{\partial_\beta \partial^l \widehat T^Y_{\widetilde W (\beta)}}{\widehat f_{\widetilde W (\beta)}} - \frac{\partial^l \widehat G _{\widetilde W (\beta)}\partial_\beta \widehat f _{\widetilde W (\beta)}}{\widehat f _{\widetilde W (\beta)}}\\
    & - \sum_{0\leq |j|\leq |l|-1}\binom{l_1}{j_1}\ldots\binom{l_d}{j_d}\frac{\partial_\beta \partial^{l-j}\widehat f _{\widetilde W (\beta)}\partial^j \widehat G _{\widetilde W (\beta)} + \partial^{l-j}\widehat f _{\widetilde W (\beta)}\partial_\beta \partial^j \widehat G _{\widetilde W (\beta)}}{\widehat f _{\widetilde W (\beta)}}\\
    \partial_\beta \partial^l  G _{ W (\beta)}=& \frac{\partial_\beta \partial^l  T^Y_{ W (\beta)}}{ f_{ W (\beta)}} - \frac{\partial^l  G _{ W (\beta)}\partial_\beta  f _{ W (\beta)}}{ f _{ W (\beta)}}\\
    & - \sum_{0\leq |j|\leq |l|-1}\binom{l_1}{j_1}\ldots\binom{l_d}{j_d}\frac{\partial_\beta \partial^{l-j} f _{ W (\beta)}\partial^j  G _{ W (\beta)} + \partial^{l-j} f _{ W (\beta)}\partial_\beta \partial^j  G _{ W (\beta)}}{ f _{ W (\beta)}}\, .
\end{align*}

$(iii)$ By proceeding similarly as for (\ref{eq: uniform convergence of partial partial beta That Y Wtilde beta0}) and (\ref{eq: uniform convergence of partial partial beta fhat Y Wtilde beta0}) we obtain
\begin{align*}
    \sup_{\beta\in B}\|\partial^2_{\beta \beta^T}\widehat T ^Y _{\widetilde W (\beta)}-\partial^2_{\beta \beta^T}T^Y_{W(\beta)}\|_{\infty,\mathcal W}=&O_P(\widetilde d_G h^{-2})\\     \sup_{\beta\in B}\|\partial^2_{\beta \beta^T}\widehat f _{\widetilde W (\beta)}-\partial^2_{\beta \beta^T}f_{W(\beta)}\|_{\infty,\mathcal W}=&O_P(\widetilde d_G h^{-2})\, .
\end{align*}
Using the above display and proceeding as for (\ref{eq: uniform convergence of partial beta Ghat Wtilde beta uniform in w and beta}) gives the desired result.

$(iv)$ The proof follows from arguments analogous to those used for Equation (\ref{eq: uniform convergence of partial That Wtilde beta0}).

$(v)$ This proof is similar to the proof of $(i)$, with some minor modifications. For completeness, we highlight such modifications here. By reasoning similarly as in (\ref{eq: uniform convergence of partial That Wtilde beta0}) and (\ref{eq: uniform convergence of partial fhat Wtilde beta0}) we obtain 
\begin{equation}\label{eq: uniform convergence of partial That Wtilde betahat and fhat Wtilde betahat}
    \|\partial^l \widehat T ^Y _{\widetilde W (\widehat \beta)}-\partial^l T^Y_{W(\beta_0)}\|_{\infty,\mathcal W}=O_P(\widetilde d_G h^{-|l|})\quad \text{ and }\quad \|\partial^l \widehat f _{\widetilde W (\widehat \beta)}-\partial^l f_{W(\beta_0)}\|_{\infty,\mathcal W}=O_P(\widetilde d_G h^{-|l|})\, 
\end{equation}
for $|l|=1,\dots,\lceil (d+1)/2 \rceil$.
Let us recall that
\begin{equation*}
    \mathcal{W}^{\eta/2}_{n,\beta_0}:=\{w\,:\,f_{W(\beta_0)}(w)\geq \eta\tau_n/2\}\, .
\end{equation*}
By Lemma \ref{lem: trimming implications}(v) we have that
\begin{equation}\label{eq: inclusion of Wtilde into fhat Wtilde betahat}
    \text{ wpa1 }\mathcal{W}^{\eta/2}_{n,\beta_0} \subset \{w:\widehat f_{\widetilde W (\widehat \beta)}(w)\geq \eta \tau_n /4\}\, .
\end{equation}
By the above inclusion and the smoothness of the kernel $K$ (see Assumption \ref{Assumption : kernels}(i)), $\widehat G _{\widetilde W (\widehat \beta)}$ is differentiable over $\mathcal{W}^{\eta/2}_{n,\beta_0}$ wpa1. Next, let us define the event 
\begin{align*}
    \mathcal{F}^{C}_n :=& \{ \|(\widehat G _{\widetilde W (\widehat \beta)}-G_{W(\beta_0)})t_{W(\beta_0)}^\eta\|_{\infty,\mathcal{W}}\leq C \widetilde d_G \tau_n^{-1}\}\\
    \cap & \Big\{\|\partial^l \widehat T ^{Y}_{\widetilde W (\widehat \beta)}-\partial^l T^{Y}_{W(\beta_0)}\|_{\infty,\mathcal W }\leq C \widetilde d_G h^{-|l|}  \text{ for }|l|\leq \lceil (d+1)/2 \rceil \Big\}\\
    \cap & \Big\{\| \partial^l \widehat f_{\widetilde W (\widehat \beta)}-\partial^l f_{W(\beta_0)}\|_{\infty,\mathcal W }\leq C \widetilde d_G h^{-|l|}\text{ for }|l|\leq \lceil (d+1)/2 \rceil \Big\}\\
     \cap &  \Big\{\mathcal{W}^{\eta/2}_{n,\beta_0} \subset \{w:\widehat f _{\widetilde W (\widehat \beta)}(w)\geq \eta \tau_n/4\}\Big\}\,.
\end{align*}
By (\ref{eq: uniform convergence of partial That Wtilde betahat and fhat Wtilde betahat}), (\ref{eq: inclusion of Wtilde into fhat Wtilde betahat}), and Lemma \ref{lem: convergence rates for Ghat fhat and iotahat}(xi), by choosing $C$ large enough we can make $\Pr(\mathcal{F}^C_n)$ arbitrarily close to 1 for each large $n$. So, to prove $(iii)$ it is sufficient to show that the event $\mathcal{F}_n^C$ implies $\sup_w |\partial^l \widehat G _{\widetilde W (\widehat \beta)}(w)-\partial ^l G_{W(\beta_0)}(w)|\mathbb I \{f_{W(\beta_0)}(w)\geq  \eta \tau_n \}\leq C \widetilde d_G h^{-|l|}\tau_n^{-|l|}$ for $|l|=1,\ldots,\lceil(d+1)/2\rceil$, for a certain constant $C$. This can be obtained by following the same line of reasoning as in the proof of $(i)$. 

$(vi)$ Let us define
\begin{equation*}
    T^{\varphi_s}_{W(\beta_0)}(w):=\mathbb E \{\varphi_s(X,H(Z)) | W(\beta_0)=w\}\, f_{W(\beta_0)}(w)\, .
\end{equation*}
By proceeding similarly as for (\ref{eq: uniform convergence of partial That Wtilde beta0}) we get 
\begin{equation*}
    \sup_s \|(\partial^l \widehat T^{\varphi_s}_{\widetilde W (\widehat \beta)} - \partial^l T^{\varphi_s}_{ W ( \beta_0)})t^\eta_{W(\beta_0)}\|_{\infty,\mathcal{W}}=O_P(\widetilde d_G h^{-|l|})\, 
\end{equation*}
for $l=1,\ldots,\lceil (d+1)/2\rceil$. So, the proof proceeds similarly as the proof of $(v)$.

$(vii)$ First, let us notice that 
\begin{align*}
    \partial^l \partial_\beta \widehat T ^{Y^*}_{\widetilde W ^* (\beta)}(w)=& -h^{-(d+l+1)}\mathbb P _n [G_{W(\beta_0)}(W(\beta_0))+\xi (Y-G_{W(\beta_0)}(W(\beta_0))]\\
    &\qquad\hspace{7cm} \cdot\widehat t\, K_h^{(l+1)}(w-\widetilde W ^*(\beta))\,\partial_\beta \widetilde W^* (\beta)\\
    &- h^{-(d+l+1)}\mathbb P _n  (1-\xi)[\widehat G _{\widehat W (\widehat \beta)}(\widehat W (\widehat \beta))-G_{W(\beta_0)}(W(\beta_0))]\,\\
    &\hspace{8cm}\cdot\widehat t \, K_h^{(l+1)}(w-\widetilde W ^*(\beta))\,\partial_\beta \widetilde W^* (\beta)\, 
\end{align*}
for $|l|=0,1,\ldots,\lceil (d+1)/2 \rceil$. Now, the first term on the RHS can be handled by arguments analogous to those used for (\ref{eq: uniform convergence of partial That Wtilde beta0}) and  (\ref{eq: partial l partial beta That Wtilde beta}). Also, by proceeding as in (\ref{eq: partial l partial beta That Wtilde beta}) and then using Lemma \ref{lem: convergence rates for Ghat fhat and iotahat}(i) and Lemma \ref{lem : Li Racine}(ii), we obtain that the second term on the RHS is $O_P(\widetilde d_G h^{-|l|-1}\tau^{-1})$. So, 
\begin{equation*}
    \sup_{\beta\in B}\|\partial^l \partial_\beta \widehat T ^{Y^*}_{\widetilde W ^* (\beta)}-\partial^l \partial_\beta T^Y_{W(\beta)}\|_{\infty,\mathcal{W}}=O_P(\widetilde d_G h^{-|l|-1}\tau^{-1})\, .
\end{equation*}
Similarly, we get 
\begin{equation*}
    \sup_{\beta\in B}\|\partial^l \partial_\beta \widehat f_{\widetilde W ^* (\beta)}-\partial^l \partial_\beta f_{W(\beta)}\|_{\infty,\mathcal{W}}=O_P(\widetilde d_G h^{-|l|-1})\, .
\end{equation*}
Finally, proceeding as in the proof of $(ii)$ leads to the desired result.

$(viii)$ The proof of $(viii)$ follows from arguments similar as the proof of $(vii)$.

$(ix)$ Similarly as in the proof of $(vii)$ we get 
\begin{align*}
    \sup_{\beta\in B}\|\partial^2_{\beta \beta^T}\widehat T ^{Y^*}_{\widetilde W ^* (\beta)}-\partial^2_{\beta \beta^T}T^{Y}_{W(\beta)}\|_{\infty,\mathcal W}=&O_P(\widetilde d_G h^{-2} \tau^{-1})\,,\\    \sup_{\beta\in B}\|\partial^2_{\beta \beta^T}\widehat f_{\widetilde W ^* (\beta)}-\partial^2_{\beta \beta^T}f_{W(\beta)}\|_{\infty,\mathcal W}&=O_P(\widetilde d_G h^{-2})\,.
\end{align*}
So, we can proceed as in the proof of (\ref{eq: uniform convergence of partial beta Ghat Wtilde beta uniform in w and beta}) to get the desired result.

$(x)-(xiv)$ The results in $(x)-(xiv)$ are  obtained from arguments analogous to those used in the proofs of the previous parts of this lemma.

\end{proof}
Before introducing the following lemma, we show that 
\begin{equation}\label{eq: n 1/4 convergence rate of dtildeG/taun 2 h}
    \frac{\widetilde{d}_G}{h \tau_n^2}=o(n^{-1/4})\,.
\end{equation}
To this end, let us recall that $\widetilde{d}_G=d_G+p_n/h^d+\sum_{j=1}^4(\widetilde{d}_H/h)^j+\widetilde{d}_H^5/h^{d+5}$, see Equation (\ref{eq: definition of dtilde G}). From Assumption \ref{Assumption: bandwidth}(i) $d_G/(h\tau_n^2)=o(n^{-1/4})$ and from Assumption \ref{Assumption: trimming}(iii) $p_n/(h^{d+1}\tau_n^2)=o(n^{-1/4})$. Since $d_H/(h^2\tau_n^5)=o(n^{-1/4})$ (see Assumption \ref{Assumption: bandwidth}(ii)) and $p_n/(h_H^ph^2\tau_n^4)=o(n^{-1/4})$ (see Assumption \ref{Assumption: trimming}(iv)), we also have that $\widetilde{d}_H/(h^2\tau_n^2)=d_H/(h^2\tau_n^5)+p_n/(h_H^ph^2\tau_n^4)=o(n^{-1/4})$. Thus, $\sum_{j=1}^4(\widetilde{d}_H/h)^j/(h\tau_n^2)=o(n^{-1/4})$. Finally, $\widetilde{d}^5_H/(h^{d+6}\tau_n^2)$$=[\widetilde{d}_H/(h^2\tau_n^2)]^4 [\widetilde{d}_H/(h^2\tau_n^2)]\tau_n^{8}/(h^{d-4})$ $=o(n^{-1}) o(n^{-1/4})\tau_n^8/h^{d-4}$ $=o(n^{-1/4})\tau_n^8h^4/(n h^d)=o(n^{-1/4})$. Gathering results gives  (\ref{eq: n 1/4 convergence rate of dtildeG/taun 2 h}).

\begin{lem}\label{lem: stochastic expansions}
Let Assumptions \ref{Assumption: iid and phi}-\ref{Assumption: trimming} hold, $\mathcal{H}_0$ hold, and assume that $\widehat \beta - \beta_0=O_P(n^{-1/2})$. Then, 
\begin{enumerate}[label=(\roman*)]
    \item $\max_{i=1,\ldots,n}|\widehat G _{\widetilde W (\widehat \beta)}(\widetilde W_i (\widehat \beta))- \widehat G _{\widetilde W (\widehat \beta)}(W_i (\beta_0))-\partial^T G_{W(\beta_0)}(W_i(\beta_0))\,[\widetilde W_i (\widehat \beta)-W_i(\beta_0)]|(t_i+\widehat t_i)=o_P(n^{-1/2})$  and the same result holds by replacing $\widetilde W$ with $\widehat W$\,,
    \item $\max_{i=1,\ldots,n}|\widehat f _{\widetilde W (\widehat \beta)}(\widetilde W_i (\widehat \beta))- \widehat f _{\widetilde W (\widehat \beta)}(W_i (\beta_0))-\partial^T f_{W(\beta_0)}(W_i(\beta_0))\,[\widetilde W_i (\widehat \beta)-W_i(\beta_0)]|(t_i+\widehat t_i)=o_P(n^{-1/2})$\,,
    \item $\sup_{s\in\mathcal{S}}\max_{i=1,\ldots,n}|\widehat \iota _s(\widetilde W_i (\widehat \beta))- \widehat \iota_s(W_i (\beta_0))-\partial^T \iota_s(W_i(\beta_0))\,[\widetilde W_i (\widehat \beta)-W_i(\beta_0)]|(t_i+\widehat t_i)=o_P(n^{-1/2})$\,,
    \item $\sup_{s\in\mathcal{S}} \max_{i=1,\ldots,n}|\varphi_s(X_i,\widetilde H (Z_i)) - \varphi_s(X_i,H(Z_i))-\partial_H^T \varphi_s (X_i,H(Z_i))\, [\widetilde H(Z_i) - H(Z_i)]|(\widehat t _i + t_i)=o_P(n^{-1/2})$\,,\footnote{For simplicity, we are using the notation $\partial_H \varphi_s(x,H(z)):=\partial_u \varphi_s(x,u)|_{u=H(z)}$.}
    \item $\max_{i=1,\ldots,n}|\widehat G^* _{\widetilde W^* (\widehat \beta)}(\widetilde W_i^* (\widehat \beta))- \widehat G^*_{\widetilde W^* (\widehat \beta)}(W_i (\beta_0))-\partial^T G_{W(\beta_0)}(W_i(\beta_0))\,[\widetilde W_i^* (\widehat \beta)-W_i(\beta_0)]|(t_i+\widehat t_i)=o_P(n^{-1/2})$\,.
\end{enumerate}
If in addition $\widehat \beta ^* - \beta_0 =O_P(n^{-1/2})$, then 
\begin{enumerate}[label=(\roman*)]
\setcounter{enumi}{5}
    \item $\max_{i=1,\ldots,n}|\widehat G^* _{\widetilde W^* (\widehat \beta^*)}(\widetilde W_i^* (\widehat \beta^*))- \widehat G^*_{\widetilde W^* (\widehat \beta^*)}(W_i (\beta_0))-\partial^T G_{W(\beta_0)}(W_i(\beta_0))\,[\widetilde W_i^* (\widehat \beta^*)-W_i(\beta_0)]|(t_i+\widehat t_i)=o_P(n^{-1/2})$\,, 
    \item $\max_{i=1,\ldots,n}|\widehat f _{\widetilde W^* (\widehat \beta^*)}(\widetilde W_i^* (\widehat \beta^*))- \widehat f _{\widetilde W^* (\widehat \beta^*)}(W_i (\beta_0))-\partial^T f_{W(\beta_0)}(W_i(\beta_0))\,[\widetilde W_i^* (\widehat \beta^*)-W_i(\beta_0)]|(t_i+\widehat t_i)=o_P(n^{-1/2})$\,,
    \item $\sup_{s\in\mathcal{S}} \max_{i=1,\ldots,n}|\widehat \iota^*_s(\widetilde W_i^* (\widehat \beta^*))- \widehat \iota^*_s(W_i (\beta_0))-\partial^T \iota_s(W_i(\beta_0))\,[\widetilde W_i^* (\widehat \beta^*)-W_i(\beta_0)]|(t_i+\widehat t_i)=o_P(n^{-1/2})$\,,
    \item $\sup_{s\in\mathcal{S}} \max_{i=1,\ldots,n}|\varphi_s(X_i,\widetilde H ^*(Z_i)) - \varphi_s(X_i,H(Z_i))-\partial_H^T \varphi_s(X_i,H(Z_i))\, [\widetilde H ^* (Z_i) - H(Z_i)]\, |\, (\widehat t_i $ $+ t_i)=o_P(n^{-1/2})$\,.
\end{enumerate}

\end{lem}

\begin{proof}
We will provide a detailed proof only for $(i)$, as the proof for the remaining parts follow from analogous arguments. Let us fix $\delta\in(0,1]$. Then, wpa1 (see the comments below)
\begin{align}\label{eq: trimming implication at the data points for stochastci expansion}
t^{\delta}(X_i,Z_i)\leq & t^{\delta}(X_i,Z_i)\,\,\mathbb I \{f_{W(\beta_0)}(W_i(\beta_0))\geq \eta(\delta)\tau_n\}\nonumber \\
\leq&  t^{\delta}(X_i,Z_i)\,\,\mathbb I \{f_{W(\beta_0)}(\widetilde W _i (\widehat \beta)\geq \eta(\delta)\tau_n/2\}    
\end{align}
 for all $i=1,\ldots,n$. The first inequality follows from Assumption \ref{Assumption: trimming}(i), while the second inequality from Lemma \ref{lem: trimming implications}(vi). Lemma \ref{lem: trimming implications}(v) ensures that wpa1
 \begin{align*}
     \{w:f_{W(\beta_0)}(w)\geq \eta(\delta)\tau_n / 2\}\subset \{w:\widehat{f}_{\widetilde W (\widehat \beta)}(w)\geq \eta(\delta)\tau_n/4\}\, .
 \end{align*}
Now, by the smoothness of the kernel $K$ (see Assumption \ref{Assumption : kernels}(i)) $\widehat G _{\widetilde W (\widehat \beta)}$ is differentiable over the set on the RHS, and for $n$ large enough the set on the LHS is convex (see Assumption \ref{Assumption: trimming}(v)). So, from (\ref{eq: trimming implication at the data points for stochastci expansion}) we have that wpa1 $t^\delta (X_i,Z_i)=1$ implies $\widetilde W _i(\widehat \beta) , W_i(\beta_0),\overline {W_i(\beta)}\in\{w:f_{W(\beta_0)}(w)\geq \eta(\delta)\tau_n/2\}$, for any $\overline{W_i(\beta)}$ lying on the segment joining  $\widetilde W_i (\widehat \beta)$ and $W_i(\beta_0)$. Thus, we can apply a Mean-Value expansion of $\widehat G_{\widetilde W (\widehat \beta)}(\widetilde W _i(\widehat \beta))$ around $W_i(\beta_0)$ and obtain
\begin{align*}
    \Big|\widehat G_{\widetilde W (\widehat \beta)}(\widetilde W _i(\widehat \beta)) - \widehat G_{\widetilde W (\widehat \beta)}( W _i(\beta_0))-&\partial^T G_{W(\beta_0)}(W_i(\beta_0))\,[\widetilde W_i (\widehat \beta)-W_i(\beta_0)]\Big|\,t_i^\delta\\
    \leq & \|\partial \widehat G_{\widetilde W (\widehat \beta)}(\overline{W_i(\beta)})-\partial G_{W(\beta_0)}(W_i(\beta_0))\|t_i^\delta \\
    &\cdot \|\widetilde W _i(\widehat \beta)-W_i(\beta_0)\|t^\delta_i\, .
\end{align*}
Next, from the Lipschitz condition on $q$ in Assumption \ref{Assumption: smoothness}(v), $\widehat \beta - \beta_0=O_P(n^{-1/2})$, and Lemma \ref{lem: trimming and convergence rates for 1st step}(iii) we have
\begin{equation}\label{eq: n 1/4 rate for Wtildei betahat - Wi beta0}
\max_{i=1,\ldots,n}\|\widetilde W _i(\widehat \beta)-W_i(\beta_0)\|t^\delta_i=O_P(\|(\widetilde H - H)\, t^\delta\|_\infty)=O_P(\widetilde{d}_H)=o_P(n^{-1/4})    \, .
\end{equation}
Also (see the comments below), wpa1
\begin{align*}
    \|\partial \widehat G_{\widetilde W (\widehat \beta)}(\overline{W_i(\beta)})-\partial G_{W(\beta_0)}&(W_i(\beta_0))\|t_i^\delta\\
    \leq & \|\partial \widehat G_{\widetilde W (\widehat \beta)}(\overline{W_i(\beta)})-\partial G_{W(\beta_0)}(\overline{W_i(\beta)})\|\,\mathbb I \{f_{W(\beta_0)}(\overline{W_i(\beta)})\geq \eta(\delta)\tau_n/2\}\\
    &+ \|\partial G_{W(\beta_0)}(\overline{W_i(\beta)})-\partial G_{W(\beta_0)}(W_i(\beta_0))\|t^\delta_i\\
    \leq & \|(\partial \widehat G_{\widetilde{W}(\widehat \beta)}-\partial G_{W(\beta_0)})t_{W(\beta_0)}^{\eta(\delta)/2}\|_{\infty,\mathcal{W}}\\
    &+ C \max_{i=1,\ldots,n}\|\widetilde W _i(\widehat \beta)-W_i(\beta_0)\|t^\delta_i\\
    =&o_P(n^{-1/4})\, 
\end{align*}
uniformly in $i=1,\ldots,n$. The first inequality follows from the fact that wpa1 $t^\delta_i=1$ implies $\overline{W_i(\beta)}\in\{w:f_{W(\beta_0)}(w)\geq \eta(\delta)\tau_n/2\}$, as noticed earlier. To obtain the second inequality, we have used 
$\|\partial G_{W(\beta_0)}(\overline{W_i(\beta)})-\partial G_{W(\beta_0)}(W_i(\beta_0))\|t^\delta_i\leq C \|\widetilde W _i(\widehat \beta)-W_i(\beta_0) \|t^\delta_i$ which is implied by Assumption \ref{Assumption: smoothness}(i). 
Finally, to obtain the last equality, we have used (\ref{eq: n 1/4 rate for Wtildei betahat - Wi beta0}), $\|(\partial \widehat G_{\widetilde{W}(\widehat \beta)}-\partial G_{W(\beta_0)})t_{W(\beta_0)}^{\eta(\delta)/2}\|_{\infty,\mathcal{W}}=O_P(\widetilde{d}_G/(h \tau_n))$ from Lemma \ref{lem: uniform convergence rates for derivatives}(v), and $\widetilde d _G/(h \tau_n)=o(n^{-1/4})$ from (\ref{eq: n 1/4 convergence rate of dtildeG/taun 2 h}). 
Gathering results, 
\begin{align*}
    \max_{i=1,\ldots,n}|\widehat G _{\widetilde W (\widehat \beta)}(\widetilde W_i (\widehat \beta))- \widehat G _{\widetilde W (\widehat \beta)}(W_i (\beta_0))-\partial G_{W(\beta_0)}(W_i(\beta_0))\,[\widetilde W_i (\widehat \beta)-W_i(\beta_0)]|t_i^\delta=o_P(n^{-1/2})\ .
\end{align*}
Taking $\delta=1$ gives the first part of $(i)$  with $t_i$. The part of $(i)$ with $\widehat t_i$ follows by noticing that wpa1 $\widehat t(x,z)\leq t^{(1/2)}(x,z)$ for all $(x,z)\in Supp(X,Z)$, see Lemma \ref{lem: trimming implications}(ii), and by evaluating the previous display at $\delta=1/2$. 

\end{proof}

For the following lemma, we recall that for any $\eta>0$ 
\begin{equation*}
    \mathcal{W}_{n,\beta_0}^\eta=\{w\,:\, f_{W(\beta_0)}(w)\geq \eta \tau_n /2\}\, 
\end{equation*}
and 
\begin{equation*}
\mathcal{G}_\lambda(\mathcal{W}_{n,\beta_0}^\eta):=\left\{g:\mathcal{W}_{n,\beta_0}^\eta\mapsto \mathbb R \text{ s. t.} \sup_{w\in\mathcal \mathcal{W}_{n,\beta_0}^\eta}|\partial^l g(w)|<M\text{ for all }|l|\leq \lambda\right \}
\end{equation*}

\begin{lem}\label{lem: belonging conditions}
Let Assumptions \ref{Assumption: iid and phi}-\ref{Assumption: trimming} hold. Then, for $\eta\in(0,1]$ and $\lambda=\lceil (d+1)/2\rceil$ we have
\begin{enumerate}[label=(\roman*)]
    \item $\Pr\left(\widehat G _{\widetilde W (\beta_0)}\in\mathcal{G}_\lambda(\mathcal{W}_{n,\beta_0}^\eta) \right)\rightarrow 1$\,,
    \item $\Pr \left(\partial \widehat G _{\widetilde W (\beta_0)}\in\mathcal{G}_\lambda(\mathcal{W}_{n,\beta_0}^\eta) \right)\rightarrow 1$\,,
    \item $\Pr \left(\partial_{\beta_j} \widehat G _{\widetilde W ( \beta_0)}\in\mathcal{G}_\lambda(\mathcal{W}_{n,\beta_0}^\eta) \right)\rightarrow 1$\,.
        \end{enumerate}
    If moreover $\widehat \beta - \beta_0=O_P(n^{-1/2})$ and $\mathcal{H}_0$ holds, then 
    \begin{enumerate}[label=(\roman*)]
    \setcounter{enumi}{3}
    \item $\Pr \left(\widehat f _{\widetilde W (\widehat \beta)}\in\mathcal{G}_\lambda(\mathcal{W}_{n,\beta_0}^\eta) \right)\rightarrow 1$\,,
    \item $\Pr \left(\widehat G _{\widetilde W (\widehat \beta)}\in\mathcal{G}_\lambda(\mathcal{W}_{n,\beta_0}^\eta) \right)\rightarrow 1$ and $\Pr \left(\widehat G _{\widehat W (\widehat \beta)}\in\mathcal{G}_\lambda(\mathcal{W}_{n,\beta_0}^\eta) \right)\rightarrow 1$\,,
    \item $\Pr \left(\widehat \iota _s\in\mathcal{G}_\lambda(\mathcal{W}_{n,\beta_0}^\eta)\text{ for all }s\in \cal S \right)\rightarrow 1$\,,
    
        \item $\Pr \left(\widehat G^* _{\widetilde W^* ( \beta_0)}\in\mathcal{G}_\lambda(\mathcal{W}_{n,\beta_0}^\eta)\right)\rightarrow 1$\,,
        \item $\Pr \left(\partial \widehat G^* _{\widetilde W^* ( \beta_0)}\in\mathcal{G}_\lambda(\mathcal{W}_{n,\beta_0}^\eta) \right)\rightarrow 1$\,,
        \item $\Pr \left(\partial_{\beta_j} \widehat G^* _{\widetilde W^* ( \beta_0)}\in\mathcal{G}_\lambda(\mathcal{W}_{n,\beta_0}^\eta) \right)\rightarrow 1$\,.
\end{enumerate}
If moreover $\widehat \beta ^* - \beta_0 = O_P(n^{-1/2})$, then 
\begin{enumerate}[label=(\roman*)]
    \setcounter{enumi}{9}
    \item $\Pr \left( \widehat G^* _{\widetilde W^* (\widehat \beta^*)}\in\mathcal{G}_\lambda(\mathcal{W}_{n,\beta_0}^\eta) \right)\rightarrow 1$\,,
    \item $\Pr \left( \widehat f _{\widetilde W^* (\widehat \beta^*)}\in\mathcal{G}_\lambda(\mathcal{W}_{n,\beta_0}^\eta) \right)\rightarrow 1$\,,
    \item $\Pr \left( \widehat \iota^* _s\in\mathcal{G}_\lambda(\mathcal{W}_{n,\beta_0}^\eta) \text{ for all }s\in\cal S\right)\rightarrow $ $1$\,.
\end{enumerate}
Finally, under Assumptions \ref{Assumption: iid and phi}-\ref{Assumption: trimming} we have 
\begin{enumerate}[label=(\roman*)]
    \setcounter{enumi}{12}
    \item $\Pr \left(\widetilde H \in \mathcal{G}_\lambda (\mathcal{Z}_{n}^\delta)\right)\rightarrow 1$ with $\lambda=\lceil (p+1)/2 \rceil $\,,
    \item $\Pr \left(\widetilde H^* \in \mathcal{G}_\lambda (\mathcal{Z}_{n}^\delta)\right)\rightarrow 1$ with $\lambda=\lceil (p+1)/2 \rceil $\,.
\end{enumerate}

\end{lem}

\begin{proof}
We will provide a proof only for $(i)$, as the results in $(ii)-(xii)$ follow from the same arguments. By Assumption \ref{Assumption: smoothness}(i) there exists $\delta>0$ and $M>0$ such that $\|\partial^l G_{W(\beta_0)}\|_{\infty,\mathbb{R}^d}+\delta< M$  for all $l$ such that $|l|\leq \lceil (d+1)/2 \rceil$. So, to show $(i)$ it suffices to prove that wpa1 (a) $\widehat G _{\widetilde W (\beta_0)}$ is $\lceil (d+1)/2 \rceil$ times continuously differentiable over $\mathcal{W}_{n,\beta_0}^\eta$ and (b) $\|(\partial^l \widehat G _{\widetilde W (\beta_0)}-\partial^l G_{W(\beta_0)})t^\eta_{W(\beta_0)}\|_{\infty,\mathcal{W}} \leq \delta$ for all $|l|\leq \lceil (d+1)/2 \rceil$.   From Lemma \ref{lem: trimming implications}(iii) we have that wpa1 $\mathcal{W}^\eta_{n,\beta_0}\subset \{w:\widehat f _{\widetilde W (\beta_0)}(w)\geq \eta \tau_n /4\}$, and by Assumption \ref{Assumption : kernels}(i) $\widehat G _{\widetilde W (\beta_0)}$ is $\lceil(d+1)/2\rceil$ times continuously differentiable over $\{w:\widehat f _{\widetilde W (\beta_0)}(w)\geq \eta \tau_n /4\}$.  So, (a) is proved. (b) is obtained from Lemma \ref{lem: uniform convergence rates for derivatives}(i). This gives the desired result.\\
The proof of $(xiii)-(xiv)$ is contained in the proof of Lemma 8.3 in the Supplementary material of \citet{lapenta_encompassing_2022}.
\end{proof}

\begin{lem}\label{lem: convergence of the derivatives at the data points}
Let Assumptions \ref{Assumption: iid and phi}-\ref{Assumption: trimming} hold. Then, 
\begin{enumerate}[label=(\roman*)]
    \item $\sup_{\beta\in B}\max_{i=1,\ldots,n}\|\nabla_\beta \widehat G_{\widetilde W (\beta)}(\widetilde W_i (\beta)) - \nabla_\beta G_{W(\beta)}(W_i(\beta))\|\,(\widehat t _i + t_i)=O_P(\widetilde d_G h^{-1}  \tau_n^{-1})=o_P(n^{-1/4})$\,,
    \item $\sup_{\beta\in B}\max_{i=1,\ldots,n}\|\nabla^2_{\beta \beta^T} \widehat G _{\widetilde W (\beta)}(\widetilde W _i (\beta))-\nabla^2_{\beta \beta^T} G _{W(\beta)}(W_i(\beta))\,\|(\widehat t _i + t_i)=O_P(\widetilde d_G h^{-2} \tau_n^{-2})=o_P(1)$\,,
    \item $\sup_{(x,z)\in Supp(X,Z)}\|\partial \widehat G _{\widetilde W (\beta_0)}(q(\beta_0,x,\widetilde H (z))) - \partial G _{W(\beta_0)}(q(\beta_0,x,H(z)))\|t^\delta(x,z)=o_P(1)$\,,
    \item $\sup_{(x,z)\in Supp(X,Z)}\|\partial_\beta \widehat G _{\widetilde W (\beta_0)}(q(\beta_0,x,\widetilde H (z))) - \partial_\beta G _{W(\beta_0)}(q(\beta_0,x,H(z)))\|t^\delta(x,z)=o_P(1)$.
\end{enumerate}
If moreover $\widehat \beta - \beta_0=O_P(n^{-1/2})$ and $\mathcal{H}_0$ holds, then 
\begin{enumerate}[label=(\roman*)]
    \setcounter{enumi}{4}
    \item $\sup_{\beta \in B}\max_{i=1,\ldots,n}\|\nabla_\beta \widehat G^*_{\widetilde W ^* (\beta)}(\widetilde W_i^* (\beta)) - \nabla_\beta G_{W(\beta)}(W_i(\beta))\|\,(\widehat t _i + t_i)=O_P(\widetilde{d}_G h^{-1}\tau_n^{-2})=o_P(n^{-1/4})$\,, 
    \item $\sup_{\beta\in B}\max_{i=1,\ldots,n} \|\nabla^2_{\beta \beta^T} \widehat G^* _{\widetilde W ^* (\beta)}(\widetilde W _i^* (\beta))-\nabla^2_{\beta \beta^T} G _{W(\beta)}(W_i(\beta))\,\|(\widehat t _i + t_i)=O_P(\widetilde{d}_G h^{-2}\tau_n^{-3})$ $=o_P(1)$\,, 
    \item $\sup_{(x,z)\in Supp(X,Z)}\|\partial \widehat G^* _{\widetilde W^* (\beta_0)}(q(\beta_0,x,\widetilde H ^*(z))) - \partial G _{W(\beta_0)}(q(\beta_0,x,H(z)))\|t^\delta(x,z)=o_P(1)$\,,
    \item $\sup_{(x,z)\in Supp(X,Z)}\|\partial_\beta \widehat G^* _{\widetilde W ^* (\beta_0)}(q(\beta_0,x,\widetilde H ^* (z))) - \partial_\beta G _{W(\beta_0)}(q(\beta_0,x,H(z)))\|t^\delta(x,z)=o_P(1)$\,.
\end{enumerate}

\end{lem}

\begin{proof}
$(i)-(ii)$. We start by proving $(i)$. Notice first that 
\begin{align*}
\nabla_{\beta^T} \widehat G _{\widetilde W (\beta)}(\widetilde W _i (\beta))=& \partial^T \widehat G _{\widetilde W (\beta)}(\widetilde W _i (\beta))\, \partial_{\beta^T} q(\beta,X_i, \widetilde H (Z_i)) \\
&+  \partial_{\beta^T} \widehat G _{\widetilde W (\beta)}(\widetilde W _i (\beta))\, .
\end{align*}
By using Lemma \ref{lem: uniform convergence rates for derivatives}(i)(ii) and arguing as in the proof of Lemma \ref{lem: convergence rates for Ghat fhat and iotahat}(i) we obtain that  
\begin{align*}
   \sup_{\beta\in B}\max_{i=1,\ldots,n}\| \partial \widehat G _{\widetilde W (\beta)}(\widetilde W _i (\beta)) - \partial G_{W(\beta)}(W_i(\beta)) \|\,(\widehat t _i + t_i) =&O_P(\widetilde d_G h^{-1}  \tau_n^{-1})\,,\\
   \sup_{\beta\in B}\max_{i=1,\ldots,n}\| \partial_\beta \widehat G _{\widetilde W (\beta)}(\widetilde W _i (\beta)) - \partial G_{W(\beta)}(W_i(\beta)) \|\,(\widehat t _i + t_i) =&O_P(\widetilde d_G h^{-1}  \tau_n^{-1})\,.
\end{align*}
By the Lipschitz continuity of $\partial_\beta q$ from Assumption \ref{Assumption: smoothness}(v) and  Lemma \ref{lem: trimming and convergence rates for 1st step}(iii), we get 
\begin{align*}
    \sup_{\beta\in B}\max_{i=1,\ldots,n} \|\partial_\beta q(\beta,X_i, \widetilde H (Z_i)) - \partial_\beta q(\beta,X_i, H(Z_i))\|\,(\widehat t _i + t_i)=O_P(\widetilde d _H)\, .
\end{align*}
Putting together the previous displays and using (\ref{eq: n 1/4 convergence rate of dtildeG/taun 2 h}) gives $(i)$. The proof of $(ii)$ proceeds along similar lines, by using Lemma \ref{lem: uniform convergence rates for derivatives}(i)(ii)(iii), the Lipschitz conditions on $\partial_\beta q$ and $\partial^2_{\beta \beta^T} q$ from Assumption \ref{Assumption: smoothness}(v), and Equation (\ref{eq: negligibility of convergence rate for the derivatives}).

$(iii)-(iv)$ Let us prove $(iii)$.  Wpa1 we have (see the comments ahead) $t^\delta(x,z)\leq t^\delta(x,z) \mathbb I \{f_{W(\beta_0)}(q(\beta_0,x,H(z)))\geq \eta(\delta) \tau_n\}$  $\leq t^\delta(x,z) \mathbb I \{f_{W(\beta_0)}(q(\beta_0,x,\widetilde H (z)))\geq \eta(\delta) \tau_n/2\}$ for all $(x,z,\beta)\in Supp(X,Z)\times B$, where the first inequality is ensured by Assumption \ref{Assumption: trimming}(i), while the second inequality is obtained from Lemma \ref{lem: trimming implications}(iv). Hence, wpa1
\begin{align*}
    \sup_{(x,z)\in Supp(X,Z)}|\partial \widehat G _{\widetilde W (\beta_0)}(q(\beta_0,x,\widetilde H (z))) -  \partial G _{W(\beta_0)}&(q(\beta_0,x,H(z)))|t^\delta(x,z) \\
    \leq  \sup_{(x,z)\in Supp(X,Z)}|\partial \widehat G _{\widetilde W (\beta_0)}&(q(\beta_0,x,\widetilde H (z))) - \partial G_{W(\beta_0)}(q(\beta_0,x,\widetilde H (z))) |\\ &\cdot \mathbb{I}\{f_{W(\beta_0)}(q(\beta_0,x,\widetilde H (z)))\geq \eta(\delta)\tau_n/2\}\\
    + \sup_{(x,z)\in Supp(X,Z)} |\partial G_{W(\beta_0)}(q(\beta_0,x,&\widetilde H (z)))  - \partial G_{W(\beta)}(q(\beta,x,H(z))) |\, t^\delta (x,z)\, .
\end{align*}
  The first term on the RHS is bounded by $\|(\partial \widehat G _{\widetilde W (\beta_0)} - \partial G_{W(\beta_0)})\,t^{\eta(\delta)/2}_{W(\beta_0)}\|_{\infty,\mathcal{W}}$  $=O_P(\widetilde d_G h^{-1}\tau^{-1}_n)$, see Lemma \ref{lem: uniform convergence rates for derivatives}(i). By using the Lipschitz continuity of $\partial G_{W(\beta_0)}$ from Assumption \ref{Assumption: smoothness}(i) and the Lipschitz continuity of $q$ from Assumption \ref{Assumption: smoothness}(v), the second term on the RHS is  
  $O_P(\|(\widetilde H - H)t^\delta\|_\infty)=O_P(\widetilde d_H)$, see Lemma \ref{lem: trimming and convergence rates for 1st step}(iii). Gathering results and using Equation (\ref{eq: negligibility of convergence rate for the derivatives}) gives $(iii)$.  $(iv)$ is obtained by a similar reasoning.
  
  $(v)-(viii)$ The results are obtained by combining the same arguments as in the previous parts of this proof with Lemma \ref{lem: uniform convergence rates for derivatives}(vii)(ix)(x) and Lemma \ref{lem: trimming and convergence rates for 1st step}(iv).

\end{proof}

\begin{lem}
\label{lem : Li Racine} Let $\{U_i\}_{i=1}^n$ be a sequence of i.i.d. random variables taking values in $\mathbb{R}^q$.
Let $\mathcal{S}$ be a compact set and let $K$ be a kernel satisfying Assumption \ref{Assumption : kernels}(i). Assume that
$\{\varphi_{n,s}\text{ : }s\in\mathcal{S}\}$ is
a sequence of classes of real-valued functions defined on the support of $U_1$
such that for any $n\in\mathbb{N}$ : $\sup_{s\in\mathcal{S}}\|\varphi_{n,s}\|_{\infty}<L_{\varphi}$
and $ \|\varphi_{n,s_{1}}-\varphi_{n,s_{2}}\|_{\infty}\leq L_{\varphi}\|s_{1}-s_{2}\|$
for all $s_{1},s_{2}\in\mathcal{S}$. Then, for any compact set $\mathcal{A}\subset\mathbb{R}^{d}$ and for $|l|=0,1,\ldots,\lceil (d+1)/2 \rceil + 6$ we have 
\begin{align*}
&(i) \sup_{(w,\beta,s)\in\mathcal{A}\times B \times \cal S} \left| h^{-d} (\mathbb{P}_{n}-P) \varphi_{n,s}(U) K^{(l)}\left(\frac{W(\beta)-w}{h}\right) \right|=O_{P}\Big(\sqrt{\frac{\log n}{n\,h^{d}}}\Big)\,, \\
&(ii) \sup_{(w,\beta,s)\in\mathcal{A}\times B \times \cal S}  h^{-d} (\mathbb{P}_{n}-P) \left|  K^{(l)}\left(\frac{W(\beta)-w}{h}\right) \right|=O_{P}\Big(\sqrt{\frac{\log n}{n\,h^{d}}}\Big)\, ,
\end{align*}
where $K^{(l)}:=\partial^l K$.
\end{lem}
\begin{proof} The result is a minor modification of the proof of Theorem 1.4 in \citet{li_nonparametric_2006}.
\end{proof}

For the next result, let us recall that the empirical process operator is $\mathbb{G}g:=\sqrt{n}(\mathbb{P}_n g-Pg)$, where $P g:=\int g (Y,X,Z,D) d P(Y,X,Z,D)$. Also, the set $\mathcal{U}_n^\delta$ is defined in  Equation (\ref{eq: definition of cal W, cal Z and cal U}) of the paper.

\begin{lem}
\label{lem: ASE}Let Assumptions \ref{Assumption: iid and phi}-\ref{Assumption: trimming} hold, let $\mathcal W _n^{\eta}:=\mathcal{W}_{n,\beta_0}^\eta$ for some $\eta>0$, let $\mathcal{W}$ be defined as in (\ref{eq: definition of cal W}), and let $\zeta$ be a bounded random variable.
\begin{enumerate}[label=(\roman*)]
    \item  If $\{\widehat{f}_{s}:s\in\mathcal{S}\}$ is a collection of stochastic real-valued functions defined on $\mathcal{W}$ such that $ \sup_{s\in\mathcal{S}}\| \widehat{f}_{s}t_{W(\beta_0)}^{2\eta}\|_{\infty,\mathcal{W}}=o_P(1)$ and $\Pr(\widehat{f}_{s}\in \mathcal{G}_{l}(\mathcal{W}_{n}^{\eta}) \text{ for all }s\in\mathcal{S})\rightarrow1$ for $l =\lceil (d+1)/2 \rceil$, then
\begin{equation*}
 \sup_{s\in\mathcal{S}}\Big| \mathbb{G}_{n} \zeta t^{2\eta}_{W(\beta_0)} \widehat{f}_{s} \varphi_{s}\Big|=o_{P}(1)\, .
\end{equation*}
The same result also holds when $\varphi_s$ is replaced by 
$\iota_s$, $\varphi_s-\iota_s$, 
or a fixed bounded random variable.
\item Let $\delta>0$. 
Let $(\mathcal I _n)_n$ be a sequence of classes of functions defined on $Supp(X,Z)$ such that $\log N(\epsilon,\mathcal I _n , \|\cdot\|_{\infty,\mathcal{U}_{n}^\delta})\lesssim \epsilon^{-\upsilon}$ for any $\epsilon \in (0,1)$, with $\upsilon\in (0,2)$. Assume that $\widehat f$ is a stochastic real-valued function defined on $Supp(X,Z)$ and that $f_0$ is a fixed bounded function on $Supp(X,Z)$. If $\|\widehat f -f_0\|_{\infty,\mathcal{U}_{n}^\delta}=o_P(1)$ and $\Pr (\widehat f \in \mathcal I _n)\rightarrow 1$, then 
\begin{equation*}
    \mathbb G _n \zeta (\widehat f-f_0) t^{2 \delta}=o_P(1)\, .
\end{equation*}
 \end{enumerate}
  
\end{lem}
\begin{proof} $(i)$ is Lemma 8.4(i) of \citet{lapenta_encompassing_2022}. The proof of $(ii)$ follows from \citet[Lemma 19.34]{vaart_asymptotic_1998}. 
    
\end{proof}

\bigskip
\bibliographystyle{ecta}
\footnotesize
\renewcommand{\baselinestretch}{1}
\bibliography{biblio}
\normalsize

\end{document}